\documentclass[a4paper]{lipics-v2021}
\usepackage[T1]{fontenc}

\nolinenumbers
\usepackage{ifthen}
\usepackage{multicol}
\usepackage{pifont}
\usepackage{stmaryrd}
\usepackage{latexsym}
\usepackage{amsfonts}
\usepackage{alltt}
\usepackage{tabularx}
\usepackage{proof}
\usepackage{xfrac}
\usepackage{DotArrow}
\usepackage{xspace}
\usepackage{ulem} 
\usepackage{mathabx}
\usepackage{mynatproof}
\renewcommand{\emph}[1]{\textit{#1}}

\title{\mbox{Strong Priority and Determinacy in Timed CCS}}

\titlerunning{\mbox{Strong Priority and Determinacy in Timed CCS}}

\author{Luigi Liquori}{Inria, France}{}{}{Partially funded by ETSI.} 
\author{Michael Mendler}{University of Bamberg, Germany}{}{}{Partially funded by UniCA/I3S.} 
\authorrunning{Liquori and Mendler}  

\ccsdesc{Theory of computation, Process algebras}

\keywords{Timed process algebras, determinacy, priorities, synchronous programming}

\newtheorem{Conjecture}{Conjecture}

\newcommand{\store}[1]{}

\newcommand{\version}{short} %\newcommand{\version}{long}

\newcommand{\longshort}[3]{%
\ifthenelse{\equal{#1}{short}}{#2}{}\ifthenelse{\equal{#1}{long}}{#3}{}}

\newcommand{\cin}[2]
{#1.} 
\newcommand{\cout}[2]
{\ol{#1}.} 
\newcommand{\eqdef}{\stackrel{_\mathit{def}}{=}}
\newcommand{\pdef}{\eqdef}

\newcommand{\myr}[1]{{\color[rgb]{0.9,0,0}{#1}}}
\newcommand{\myb}[1]{{\color[rgb]{0,0,0.7}{#1}}}

\newcommand{\gll}[1]{}%{{\color{quest}{\scriptsize [L: #1]}}}
\newcommand{\rll}[1]{}%{{\color{red}{#1}}}

\newcommand{\rmm}[1]{}%{{\color{red}{\scriptsize [M: #1]}}}

\newcommand{\mrev}[2]{#1}

\newcommand{\gmm}[1]{}% {{\color{quest}{\scriptsize [M: #1]}}}

\newcommand{\bag}[1]{}
\newcommand{\cone}{\ding{172}\xspace}
\newcommand{\ctwo}{\ding{173}\xspace}
\newcommand{\cthree}{\ding{174}\xspace}
\newcommand{\cfour}{\ding{175}\xspace}
\newcommand{\cfive}{\ding{176}\xspace}
\newcommand{\csix}{\ding{177}\xspace}
\newcommand{\cseven}{\ding{178}\xspace}

\newcommand{\free}{\xspace}   

\newcommand{\ccs}{\textsf{CCS}\xspace}   		 
\newcommand{\csp}{\textsf{CSP}\xspace}   		 
\newcommand{\csa}{\textsf{CSA}\xspace}   		 
\newcommand{\tpl}{\textsf{TPL}\xspace}   
\newcommand{\pmc}{\textsf{PMC}\xspace}   
\newcommand{\sccs}{\textsf{SCCS}\xspace}

\newcommand{\ccscc}{\ensuremath{\mathsf{CCS}^{\mathsf{cc}}}\xspace}   
   
\newcommand{\ccscw}{\ensuremath{\mathsf{CCS}^{\mathsf{CW}}}\xspace}   
\newcommand{\ccsp}{\ensuremath{\mathsf{CCS}^{\mathsf{Ph}}}\xspace}   
\newcommand{\ccslm}{\ensuremath{\mathsf{CCS}^{\mathsf{spt}}}\xspace}

\newcommand{\T}{\tcode{T}}

\newcommand{\tsf}[1]{\textsff{#1}}

\newcommand{\ol}[1]{\ensuremath{\overline{#1}}}
\renewcommand{\ul}[1]{\ensuremath{\underline{#1}}}

\newcommand{\textsff}[1]{\textsf{#1}}

\newcommand{\sestep}[1]%
  {\ensuremath{\mathrel{{\relbar}{#1}{\twoheadrightarrow}}}}
  
\newcommand{\spstep}[1]%
  {\ensuremath{\mathrel{{\relbar}{#1}{\rightarrow}}}}
  
\newcommand{\fmstep}[1]%
  {\ensuremath{\mathrel{\stackrel{#1}{\twoheadrightarrow}_{^{\free}}}}}

\newcommand{\pstep}[1]%
  {\ensuremath{\mathrel{\stackrel{#1}{\longrightarrow}_{\textsf{c}}}}}
\newcommand{\sstep}[1]%
  {\ensuremath{\mathrel{\stackrel{#1}{\longrightarrow}}}}

\newcommand{\tcode}[1]{{\mathsf{#1}}}

\newcommand{\of}{\col}
\newcommand{\col}{\ensuremath{{:}}}
\newcommand{\rcol}{\ensuremath{\mathrel{:}}}
\newcommand{\cnf}{\ensuremath{\Vdash}}
\newcommand{\ncnf}{\ensuremath{\nVdash}}

\newcommand{\csum}{\ensuremath{\mathrel{\smalltriangleright}\xspace}}

\newcommand{\wpsum}{\ensuremath{\mathrel{+\hspace{-1.5mm}\rangle}}\xspace}

\newcommand{\cpar}{\,|\,}

\newcommand{\cseq}{\ensuremath{\tcode{.}}\xspace}

\newcommand{\bang}[1]{\ensuremath{{#1}^{\ast}}}

\newcommand{\crec}{\ensuremath{\mathop{\textsf{rec}}}\xspace}

\newcommand{\zero}{{\tcode{0}}\xspace}
\newcommand{\one}{{\tcode{1}}\xspace}

\newcommand{\A}{\mathcal{A}}
\newcommand{\I}{\mathcal{I}}
\newcommand{\R}{\mathcal{R}}

\renewcommand{\L}{\mathcal{L}}
\newcommand{\C}{\mathcal{C}}
\renewcommand{\P}{\mathcal{P}}

\newcommand{\Act}{\mathit{Act}}

\newcommand{\coA}{\overline{\mathcal{A}}}

\newcommand{\olwilA}[1]{\ensuremath{\ol{\tsf{iA}}^{_{#1}}}\xspace}

\newcommand{\olwiA}{\ensuremath{\ol{\tsf{iA}}^{_\ast}}\xspace}

\newcommand{\wilA}[1]{\ensuremath{\tsf{iA}^{#1}}\xspace}

\newcommand{\oliA}{\ensuremath{\ol{\tsf{iA}}}\xspace}

\newcommand{\olfwiA}{\ensuremath{\ol{\tsf{fiA}}^{_\ast}}\xspace}
\newcommand{\wiA}{\ensuremath{\tsf{iA}^\tau}\xspace}
\newcommand{\wiS}{\wiA} %{\ensuremath{\tsf{iS}^{_\ast}}\xspace}

\newcommand{\fwiA}{\ensuremath{\tsf{fiA}^{_\ast}}\xspace}

\newcommand{\iA}{\ensuremath{\tsf{iA}}\xspace}
\newcommand{\vA}{\ensuremath{\tsf{fA}}\xspace}

\newcommand{\ordpre}{\raisebox{-1pt}{\ensuremath{\dotarrow{}}}}
\newcommand{\sordpre}[1]{\raisebox{-1pt}{\ensuremath{\dotarrow{}\!\!_{#1}}}\xspace}

\newcommand{\indep}{\ensuremath{\mathrel{\diamond}}}

\newcommand{\eset}{\ensuremath{\{\,\}}\xspace}

\newcommand{\sRDerives}[2]{\underset{#2}{\xleftarrow{#1}}}

\newcommand{\oldDerives}[3]{\ensuremath{\xrightarrow[#3]{#1}\!\!{_{#2}}\, }}

\newcommand{\Derives}[3]
{\mbox{ \ensuremath{\xrightarrow[{\protect\raisebox{-0.25em}[0em][0em]{$^{#3}$}}]
{\protect\raisebox{-0.5em}[0em][0em]{$^{#1}$}}\!\!{_{#2}}\, }}}

\newcommand{\LDerives}[3]
{\ensuremath{\xleftarrow[{\protect\raisebox{-0.2em}[0em][0em]{$^{#3}$}}]{#1}\!\!{_{#2}}\, }}

\newcommand{\fstep}[1]%
  {\ensuremath{\mathrel{\stackrel{#1}{\rightarrow}_{^{\free}}}}}

\newcommand{\astep}[1]%
  {\ensuremath{\mathrel{\stackrel{#1}{\rightsquigarrow}}}}

\newcommand{\nfstep}[1]%
  {\ensuremath{\mathrel{\stackrel{#1}{\not\rightarrow\!\!}_{^{\free}}}}}

\newcommand{\snDerives}[2]{\stackrel{#1}{\nrightarrow{}}}

\newcommand{\ActR}{\ensuremath{\textit{Act}}\xspace}
\newcommand{\ClkR}{\ensuremath{\textit{Clk}}\xspace}

\newcommand{\HideR}{\ensuremath{\textit{Hide}}\xspace}

\newcommand{\RestrR}{\ensuremath{\textit{Restr}}\xspace}

\newcommand{\StructR}{\ensuremath{\textit{Struct}}\xspace}
\newcommand{\RepR}{\ensuremath{\textit{Rep}}\xspace}

\newcommand{\PriSumR}{\ensuremath{\textit{PriSum}}\xspace}

\newcommand{\SumR}{\ensuremath{\textit{Sum}}\xspace}

\newcommand{\ParR}{\ensuremath{\textit{Par}}\xspace}
\newcommand{\SParR}{\ensuremath{\textit{Par}^{_\ast}}\xspace}

\newcommand{\ConR}{\ensuremath{\textit{Con}}\xspace}
\newcommand{\ComR}{\ensuremath{\textit{Com}}\xspace}

\newcommand{\restrict}{\ensuremath{\,\backslash\, }}
\newcommand{\hide}{\ensuremath{\,\slash\,}}

\newcommand{\inp}{\ensuremath{{?}}\xspace}
\newcommand{\out}{\ensuremath{{!}}\xspace}

\newcommand{\scong}{\ensuremath{\equiv}}

\newcommand{\wAND}{\ensuremath{\mathsf{wAND}}}

\newcommand{\emit}{\ensuremath{\mathit{emit}}\xspace}
\newcommand{\pres}{\ensuremath{\mathit{pres}}\xspace}
\newcommand{\abs}{\ensuremath{\mathit{abs}}\xspace}

\newcommand{\ABRO}{\ensuremath{\mathsf{ABRO}}\xspace}
\newcommand{\ABO}{\ensuremath{\mathsf{ABO}}\xspace}
\newcommand{\AB}{\ensuremath{\mathsf{AB}}\xspace}
\newcommand{\sA}{\ensuremath{\mathsf{A}}\xspace}

\newcommand{\sB}{\ensuremath{\mathsf{B}}\xspace}
\newcommand{\sR}{\ensuremath{\mathsf{R}}\xspace}
\newcommand{\sO}{\ensuremath{\mathsf{O}}\xspace}
\newcommand{\sT}{\ensuremath{\mathsf{T}}\xspace}

\newcommand{\mem}{\ensuremath{\mathsf{mem}}\xspace}
\newcommand{\prg}{\ensuremath{\mathsf{prg}}\xspace}
\newcommand{\psig}{\ensuremath{\mathsf{psig}}\xspace}

\begin{document}
\maketitle

%!TEX root = synpatick-lipics.tex

\begin{abstract}
Building on the standard theory of process algebra with priorities, we identify a new scheduling mechanism, called \textit{constructive reduction} which is designed to capture the essence of synchronous programming. The distinctive property of this evaluation strategy is to achieve determinacy-by-construction for multi-cast concurrent communication with shared memory.  
In the technical setting of \ccs extended by clocks and priorities, we prove for a large class of \textit{c-coherent} processes a confluence property for constructive reductions. We show that under some restrictions, called \textit{pivotability}, c-coherence is preserved by the operators of prefix, summation, parallel composition, restriction and hiding. Since this permits memory and sharing, we are able to cover a strictly larger class of processes compared to those in Milner's classical confluence theory for \ccs without priorities.
\end{abstract}

\section{Introduction}\label{sec:intro}
%----------------------------------------
%\vspace{-2mm}

Concurrency, by expression or by implementation, is both a convenient and unavoidable feature of modern software systems. However, this does not mean that we must necessarily give up the requirement of functional determinism which is crucial for maintaining predictability and to manage design complexity by simple mathematical models~\cite{Lee06,BocchinoAAS09,Edwards18}. 
While the pure $\lambda$-calculus is naturally deterministic by design, it cannot model shared memory. Process algebras can naturally model shared objects, but do not guarantee determinism out of the box.
In this paper we use the standard mathematical formalism \ccs~\cite{Milner:CCS} to study methods for reconciling concurrency and determinism.
%
%%%%%%%%%%%%%%%%%% Discussing Nondeterminism %%%%%%%%%%%%%%
%

%---------------------------------------------------
\paragraph*{Determinacy in Process Algebra}
%---------------------------------------------------

In \ccs, the interaction of concurrent processes $P \cpar Q$ arises from the rendezvous synchronisation of an \textit{action} $\alpha$ of $P$ and an associated \textit{co-action} $\ol{\alpha}$ from $Q$, generating a \textit{silent} action $\tau$, also called a \textit{reduction}. 
A process is in \textit{normal form} if cannot reduce any more and \textit{determinate} if it reduces to at most one normal form, up to some notion of structural congruence. 
\longshort{\version}{}{Our goal is to find compositional techniques to construct determinate processes that permit shared resources.}
\longshort{\version}
{%
 As a simplified example for the scenarios that we are interested in, consider a triple of concurrent processes $R \cpar S \cpar W$, where $S \eqdef r \cseq S + w \cseq \ol{r} \cseq \zero$ 
 %and $S' = r \cseq \zero$ 
 acts as a shared resource, such as a (write-once) store, while $R \eqdef \ol{r} \cseq \zero$, $W \eqdef \ol{w} \cseq \zero$ are read and write processes, respectively, sharing $S$ with each other. The store $S$ offer a read action $r$ (rendezvous) leaving the state unchanged, or a write action $w$ upon which it changes its state to $\ol{r} \cseq \zero$. At this point, the store only permits a single read and then becomes the inactive process $\zero$. 
 The synchronisation between the actions $r$, $w$ and co-actions $\ol{r}$, $\ol{w}$ generates two sequences of reductions from $R \cpar S \cpar W$:\\[1ex]
\mbox{\hspace{2.75cm}
$\centering \zero \cpar \ol{r} \cseq \zero \cpar \zero \sRDerives{\tau}{} \zero \cpar S \cpar W \sRDerives{\tau}{} R \cpar S \cpar W \Derives{\tau}{}{} R \cpar \ol{r} \cseq \zero \cpar \zero \Derives{\tau}{}{} \zero \cpar \zero \cpar \zero.
$
}
\\ 
 Obviously, $R \cpar S \cpar W $ is not determinate. It reduces to two distinct normal forms $\zero \cpar \ol{r} \cseq \zero \cpar \zero$ and $\zero \cpar \zero \cpar \zero$. The %In the former, the
 final store $\ol{r} \cseq \zero$ on the left permits one more read, while $\zero$ on the right does not.
}
{As an example for the general scenario that we are going to deal with, consider a triple of concurrent processes $P \cpar Q \cpar R$, where $Q$ acts as a shared resource, such as a store, store and $P$, $R$ are distinct processes sharing the resource concurrently with each other. 
An individual access of $Q$ by $P$ will be modelled as a \ccs-style rendezvous between an action $\alpha_1$ by $Q$ and a matching co-action $\ol{\alpha}_1$ by $P$. Concurrently, process $R$ may be accessing the store, expressed as a synchronisation between an action $\alpha_2$ by $Q$ and $\ol{\alpha}_2$ by $R$. The synchronisation between an action and its co-action generates a silent $\tau$-action of the concurrent composition $P \cpar Q \cpar R$ of the three processes leading to the following reductions: 
\longshort{\version}
{$P \cpar Q_2 \cpar R_2\sRDerives{\tau}{} P \cpar Q \cpar R \Derives{\tau}{}{} P_1 \cpar Q_1 \cpar R$.}
{\[
\begin{Prooftree}
  \Bproof
    \Lproof 
      P \Derives{\overline{\alpha_1}}{}{} P_1
    \ANDproof
      Q \Derives{\alpha_1}{}{} Q_1
    \Rproof 
      P \cpar Q \cpar R \Derives{\tau}{}{} P_1 \cpar Q_1 \cpar R
  \EEproof
  \qquad 
  \Bproof
    \Lproof 
      Q \Derives{\alpha_2}{}{} Q_2
    \ANDproof
      R \Derives{\overline{\alpha_2}}{}{} R_2
    \Rproof 
      P \cpar Q \cpar R \Derives{\tau}{}{} P \cpar Q_2 \cpar R_2
  \EEproof
\end{Prooftree}
\]
}
The issue for nondeterminism here is that the two silent $\tau$ transitions (modelling an internal and unobservable computation step) make the concurrent system $P \cpar Q \cpar R$ end up in two configurations $P_1 \cpar Q_1 \cpar R$ and $P \cpar Q_2 \cpar R_2$ in which a \textit{single} component $Q$ has changed into two \textit{distinct} states $Q_1$ and $Q_2$.
The problem for confluence may now arise in two ways: 

\begin{itemize}

  \item $Q_1$ does not offer the action $\alpha_2$ to $R$, or $Q_2$ does not offer $\alpha_1$ to $P$, anymore; for instance, say $Q$ is a shared data queue and $\alpha_i$ competing dequeue actions performed by $P$ and $R$. 
  \longshort{\version}
  {}
  {Suppose the dequeue action $\alpha_i$ removes the last element in $Q$. Then $Q_i$ is empty and the dequeue action $\alpha_{3-i}$ is no longer admissible in $Q_i$. Or, $Q$ is a shared memory cell, $\alpha_1$ is a read and $\alpha_2$ a write action. Obviously, if $P$ executes a read $\alpha_1$ first on $Q$ with a specific value $v$ transmitted, then $\alpha_1$  will not be admissible anymore in $Q_2$ if the write action $\alpha_2$ by $R$ has updated the memory. Formally, we have
  \[
  \begin{Prooftree}
  \Bproof
    \Lproof 
      P \Derives{\overline{\alpha_1}}{}{} P_1
    \ANDproof
      Q_2 \snDerives{\alpha_1}{}
    \Rproof 
      P \cpar Q_2 \cpar R_2 \snDerives{\tau}{} 
  \EEproof
  \quad \text{or} \quad 
  \Bproof
    \Lproof 
      Q_1 \snDerives{\alpha_2}{}
    \ANDproof
      R \Derives{\overline{\alpha_2}}{}{} R_2
    \Rproof 
      P_1 \cpar Q_1 \cpar R \snDerives{\tau}{}
  \EEproof
  \end{Prooftree}
  \]
  where one of $P$ or $R$ has blocked the other process from proceeding with its access to $Q$. 
}
  \item Both actions $\alpha_{3-i}$ ($i \in\{1,2\}$) remain admissible in $Q_i$, say $Q_1 \Derives{\alpha_2}{}{} Q_{12}$ and $Q_2 \Derives{\alpha_1}{}{} Q_{21}$ but have different effects, so that $Q_{12}$ and $Q_{21}$ are not observationally equivalent. 
 \longshort{\version}{}{For instance, if $Q$ is a data queue and $\alpha_i$ actions to enqueue distinct values $v_i$, say. Then the queue $Q_{12}$ resulting from enqueuing first $v_1$ and then $v_2$ will be different from the queue $Q_{21}$ where $v_2$ has been enqueued first. Similarly, if $Q$ is a shared memory and $\alpha_i$ are actions to destructively overwrite the same cell in $Q$ by distinct values $v_i$, then $Q_{12}$ and $Q_{21}$ are different. In Plotkin's Structural Operational Semantics (SOS)~\cite{PlotkinSOS}, we have
  \[
  \begin{Prooftree}
  \Bproof
    \Lproof 
      P \Derives{\overline{\alpha_1}}{}{} P_1
    \ANDproof
      Q_2 \Derives{\alpha_1}{}{} Q_{21}
    \Rproof 
      P \cpar Q_2 \cpar R_2 \Derives{\tau}{}{} 
        P_1 \cpar Q_{21} \cpar R_2
  \EEproof
  \quad \text{and} \quad
  \Bproof
    \Lproof 
      Q_1 \Derives{\alpha_2}{}{} Q_{12}
    \ANDproof
      R \Derives{\overline{\alpha_2}}{}{} R_2
    \Rproof 
      P_1 \cpar Q_1 \cpar R \Derives{\tau}{}{} 
       P_1 \cpar Q_{12} \cpar R_2
  \EEproof
  \quad \text{but} \quad
   Q_{21} \neq Q_{12}
  \end{Prooftree}
  \]
under a given standard behavioural equivalence, where the order of execution has produced a nondeterministic outcome $P_1 \cpar Q_{21} \cpar R_2 \neq 
   P_1 \cpar Q_{12} \cpar R_2$.
   }
\end{itemize}
}%%%%% END LONGSHORT

\longshort{long}{Classical process algebras, like \ccs~\cite{Milner:CCS}, \csp~\cite{Hoare:CSP}, $\Pi$-calculus~\cite{Milner1:PI,Milner2:PI}, just to mention a few, have been devised as a model of concurrency where a system and its environment continuously interact to mutually control each other's behaviour through synchronisation and communication actions. This contrasts with automata theory where the interaction is single-pass: The environment provides a sequence of input actions to be fully consumed by the automaton before it outputs a result. This result is a decision value, determining if the string is accepted or not. The difference is reflected in the semantical theory: 
\begin{itemize}
\item Process algebras are based on forms of bisimulation equivalence~\cite{LeeP95,Glabbeek97}, a game-theoretic notion. Two processes are equivalent if they define the same interaction game.
\item Traditional automata theory is based on language equality, which is a functional model. Two programs are equivalent if they implement the same decision function. 
\end{itemize}
These two modes of interaction are usually treated independently. 
Specifically, process algebras are typically applied for asynchronous systems, which are inherently non-deterministic. Functional models are used for synchronous systems which are inherently deterministic.

}{} % end of short

%%%%%%%%%%%%%%%% Milner Results %%%%%%%%%%%%%%%%%%%%%%

\medskip 

\noindent According to Milner~\cite{Milner:CCS}, the notion of determinacy is tied up with predictability:

\begin{quote}
 {\it ``if we perform the same experiment twice on a determinate system, starting each time in its initial state, then we expect to get the same result, or behaviour, each time.''}
\end{quote}

\noindent A sufficient condition for determinacy is determinism.
\longshort{\version}
{A process $P$ is \emph{(structurally) deterministic} if for all its derivatives $Q$ and action $\alpha \in \Act$,
if $Q \xrightarrow{\alpha}{}{} Q_1$ and $Q \xrightarrow{\alpha}{}{} Q_2$ then $Q_1 \equiv Q_2$, where $\equiv$ is a suitable structural congruence.}
{\begin{definition}
 A process $P$ is \emph{structurally deterministic} for an action $\alpha \in \Act$, if for every derivative $Q$ of $P$, % and action $\alpha \in \Act$,
 if $Q \xrightarrow{\alpha}{}{} Q_1$ and $Q \xrightarrow{\alpha}{}{} Q_2$ we have $Q_1 \equiv Q_2$.
\label{def:strong-determinism} 
\end{definition}
}
However, determinism is too strong. For instance, the store $P \eqdef R_1 \cpar S \cpar R_2$ with two readers $R_i \eqdef \ol{r} \cseq R_i'$ ($i{=}1,2$)
% = a \cseq A \cpar a \cseq B \cpar \ol{a} \cseq \zero \cpar \ol{a} \cseq \zero$ 
is determinate but not deterministic, when the $R_i$ may
%two readers 
reach structurally distinct states 
%$R_1' \not\equiv R_2'$ 
after reading, say, $P \xrightarrow{\tau} R_1' \cpar S \cpar R_2$ and $P \xrightarrow{\tau} R_1 \cpar S \cpar R_2'$ where $R_1' \cpar S \cpar 
R_2 \not\equiv
R_1 \cpar S \cpar R_2'$. 
Moreover, determinism is not closed under parallel composition. In the above example, all process $R_1$, $R_2$, $S$ are deterministic, but their composition $P$ is not.
An insightful solution, proposed by Milner (Chap.~11 of ~\cite{Milner:CCS}), is to replace determinism by the notion of \textit{confluence} which still implies determinacy but turns out to be closed under parallel composition, under natural restrictions.
%{
\begin{definition}
 $P$ is \emph{(structurally) confluent}\footnote{A variation of Milner's ``strong confluence'' that includes $\tau$-actions and uses $\equiv$ not bisimulation $\sim$.} if for every derivative $Q$ of $P$ and reductions $Q \oldDerives{\alpha_1}{}{} Q_1$ and $Q \oldDerives{\alpha_2}{}{} Q_2$ with 
 $\alpha_1 \neq \alpha_2$ or 
 $Q_1 \not\equiv 
 Q_2$, there exist 
 $Q_1' \equiv 
 Q_2'$ such that $Q_1 \oldDerives{\alpha_2}{}{} Q_1'$ and $Q_2 \oldDerives{\alpha_1}{}{} Q_2'$.
\label{def:Church-Rosser} 
\end{definition}
%}
%

Milner shows that confluence is preserved by 
\emph{``confluent composition''} $P {\cpar}_{\! L} Q \eqdef %= 
(P \cpar Q) \restrict L$, %which
combining the parallel and restriction operators, subject to the condition that $\L(P) \cap \L(Q) = \eset$ and $\ol{\L(P)} \cap \L(Q) = L \cup \ol{L}$, where $\L(R)$ is the \textit{sort} of a process $R$.
\longshort{\version}{}{The former condition says that $P$ and $Q$ operate on disjoint actions and the latter says that every possible communication between $P$ and $Q$ is restricted by $L$.} 
%These conditions express 
This is a form of \textit{sort separation}, ensuring that every action $\alpha$ %rendezvous reduction $\tau = \alpha \cpar \ol{\alpha}$ between $P$ and $Q$ 
has at most one synchronisation partner $\ol{\alpha}$ either inside or outside of $(P \cpar Q) \restrict L$.
Many practical examples of determinate systems can be understood as
sort-separated compositions of confluent processes. However, this still 
covers only a rather limited class of applications.
Memory processes such as $S$ above  
are intrinsically not confluent, and even if they were, confluent composition
forbids direct multi-cast communication, because labels such as $r$ in $S$ could not be shared by two readers $R_1$, $R_2$ due to sort-separation. 
This means that concurrent programming languages that support shared memory %(non-confluent and multi-cast) 
and yet have  
determinate reduction semantics, cannot be handled.

%%%%%%%%%%%%%%%%% Resolving Nondeterminism %%%%%%%%%%%%%%%%

%---------------------------------------------------
\paragraph*{Determinacy in Synchronous Programming}
%---------------------------------------------------

\longshort{\version}{}{%
We can identify three key principles to resolve data races and achieve confluence: \textit{isolation}, \textit{cloning} and \textit{precedence}. Let us discuss these briefly using the notation of \ccs.
If $\alpha_1 \neq \alpha_2$ are distinct actions 
affecting two concurrently independent parts of a ``shared'' process
$%\[ 
  Q \eqdef \alpha_1.Q_1 \cpar \alpha_2.Q_2,
$ %\] 
then confluence is guaranteed, even if $Q_1$ and $Q_2$ are observably distinct. 
The parallel composition $P \cpar Q \cpar R$ with $P = \ol{\alpha_1}.\, P_1$ aqnd $R = \ol{\alpha_2}.\, R_2$ ends up in configuration $P_1 \cpar Q_1 \cpar Q_2 \cpar R_2$ no matter which of $P$ or $R$ communicates first.
Similarly, $P \cpar Q$ and $Q \cpar R$ deterministically generate $P_1 \cpar Q_1$ and $Q_2 \cpar R_2$, respectively. 
This is the method of \textit{isolation} which is implemented naturally by parallel composition inside $Q$ provided $\alpha_1$ and $\alpha_2$ are distinct. %of process algebras. 
If both actions are indistinguishable, $\alpha_1 {=} \alpha {=} \alpha_2$, then we get nondeterminism since $P \cpar Q$ can generate either $P_1 \cpar Q_1$ or $P_1 \cpar Q_2$. In this case, for $Q$ to exhibit deterministic behaviour, $Q_1$ must be indistinguishable from $Q_2$. 
This is the idea of \textit{cloning} a resource. If we wish $Q$ to behave deterministically in any environment it must clone $\alpha$ unboundedly often, e.g., via a repetition operator. 
Specifically, if $Q \eqdef \alpha^\ast.\, Q'$ means the prefix $\alpha$ is repeated infinitely often, i.e., $Q \Derives{\alpha}{}{} (Q' \cpar Q)$, 
then $Q$ will offer $\alpha$ by cloning the residual $Q'$. 
The parallel compositions $P \cpar Q$, $Q \cpar R$ and $P \cpar Q \cpar R$ will always deterministically end up in $P_1 \cpar Q' \cpar Q$, $Q' \cpar Q \cpar R_2$ and $P_1 \cpar Q' \cpar Q' \cpar Q \cpar R_2$, respectively, when both $P$ and $R$ have finished their communication with $Q$. 
In all these cases, isolation or cloning, the shared object $Q$ does not really care about the ordering in which $\alpha_1$ and $\alpha_2$ are executed. 
Yet, what if the behaviour of $Q$ must depend on the choice between $\alpha_1$ and $\alpha_2$? 
\longshort{\version}{}
{A generic shared memory component, for example, must permit both reads and writes simultaneously without knowing the order in which these might be requested in the context $P \cpar {[\cdot]} \cpar R$, where ${[\cdot]}$ represents a ``hole'' in the process expression. Also, it must permit a choice of write actions for distinct values that are determined in the environment. Isolation or repetition does not work, here.} 
Consider a process 
$ %\[ 
  Q = \alpha_1.Q_1 + \alpha_2.Q_2
 $ %\] 
which offers a conflicting choice %of actions $\alpha_i$ 
leading to distinct states $Q_i$. 
Then, the canonical solution is to enforce a \textit{scheduling precedence} on the actions of the environment in order to preserve \myb{c-coherence} of the resource.  
Specifically, we may restrict the concurrent environment $E = P \cpar {[\cdot]} \cpar R$ in the following ways: 
\begin{itemize} 
  \item If $Q$ is a store where $\alpha_1 = a \inp v$ is a data input and $\alpha_2 = a \out w$ an output, we want $\alpha_1$ take \textit{precedence} over $\alpha_2$ to implement data-flow causality. 
  \longshort{\version}{} 
  {Let us write 
    $%\begin{eqnarray*}
     Q \Derives{\alpha_2}{\{ \alpha_1 \}}{} Q_2
    \label{eqn:intro-q-step-1}
    $ %\end{eqnarray*}
    to express that the output $\alpha_2$ only synchronises with a matching $\ol{\alpha_2}$ action if the environment $E$ does not offer to engage in an input $\ol{\alpha_1}$ at the same time. 
    In contrast, the input action $Q \Derives{\alpha_1}{\{ \}}{} Q_1$ is not blocked by the output $\alpha_2$.
  }
  \longshort{\version}
  {In the configuration $P \cpar Q \cpar R$ the writer $P$ is then enabled first and then the reader $Q$ is scheduled, thereby achieving derterminism.}
  {In our case $E = P \cpar {[\cdot]} \cpar R$ in which $P$
  %$P \Derives{\ol{\alpha}_1}{}{} P_1$ 
  offers $\alpha_1$, whence so the reader $R$ offering $\alpha_2$ 
  % \Derives{\ol{\alpha}_2}{}{} R_2$ 
  is blocked, because of the precedence constraint associated with the transition.
  %~\eqref{eqn:intro-q-step-1}. 
  The transition 
  $ %\[
  P \cpar Q \cpar R \Derives{\tau}{}{} P \cpar Q_2 \cpar R_2
  $ %\]
  is not possible anymore. However, since there is no precedence on the read action $\alpha_1$, i.e.,   
  $ %\begin{eqnarray*}
     Q \Derives{\alpha_1}{\{ \}}{} Q_1
  %\label{eqn:intro-q-step-2}
  $ % \end{eqnarray*} 
  the writer $P$ executing $\ol{\alpha_1}$ can go ahead, the transition
  $ %\[
  P \cpar Q \cpar R \Derives{\tau}{}{} P_1 \cpar Q_1 \cpar R
  $ %\]
   is permitted.}

  \item Sometimes an ordering does not make sense, say if $\alpha_1 = a \inp v$ and $\alpha_2 = a \inp w$ are inputs of \emph{distinct} values $v \neq w$ to be received by $Q$. Since these actions make $Q$ take on distinct continuation states $Q_1 \neq  Q_2$, the environment $E$ must be blocked as it attempts to engage in both $\alpha_1$ and $\alpha_2$, concurrently from $P$ and from $R$. This can be modelled as a precedence of $\alpha_1$ over $\alpha_2$ and $\alpha_2$ over $\alpha_1$. 
  \longshort{\version}{}
  {which generates transitions
    %not only of~\eqref{eqn:intro-q-step-1} but also 
    $ %\begin{eqnarray*}
     Q_2  \LDerives{\alpha_2}{\{ \alpha_1 \}}{} Q \Derives{\alpha_1}{\{ \alpha_2 \}}{} Q_1.
    %\label{eqn:intro-q-step-3} 
    $ %\end{eqnarray*}
  }
  \longshort{\version}
  {In the configuration $P \cpar Q \cpar R$ both writers $P$ and $Q$ are then blocked, which is deterministic.}
  {%
   Then, since $E$ offers both $\alpha_i$, as above, both transitions
   $ %\[
    P \cpar Q \cpar R \Derives{\tau}{}{} P_1 \cpar Q_1 \cpar R \text{ and } P \cpar Q \cpar R \Derives{\tau}{}{} P \cpar Q_2 \cpar R_2
   $ %\]
   are blocked.}

  \longshort{\version}{}%
  {\item Finally, on top of the issues discussed above, $E$ must not attempt to consume any of the action prefixes offered by $Q$ several times concurrently. 
  For instance, suppose $\alpha_1 = a \out v$ is the action of dequeuing a value $v$ from (the front of) a queue $Q$. Then the concurrent environment $P \cpar P \cpar {[\cdot]} \cpar R$ would introduce a data race, because the two copies of $P$ would attempt to consume this value simultaneously.

  To avoid the data race, we can block this conflict by imposing a \textit{(reflexive) precedence} of $\alpha_1$ over itself:
  $ %\begin{eqnarray*}
     Q \Derives{\alpha_1}{\{ \alpha_1 \}}{} Q_1.
  %\label{eqn:intro-q-step-4} 
  $ % \end{eqnarray*}
  Likewise, if $\alpha_2$ is an enqueue operation on $Q$, the two copies of $R$ in an environment  $P \cpar {[\cdot]} \cpar R \cpar R$ would compete for the first value of the queue. \longshort{\version}{}{Whoever comes second might block, because the value $\alpha_2$ is only available once. 
  To protect from that, we would introduce reflexive precedence 
  $ %\begin{eqnarray*}
     Q \Derives{\alpha_2}{\{ \alpha_2 \}}{} Q_2.
  %\label{eqn:intro-q-step-4} 
  $ %\end{eqnarray*}
  }
}
\end{itemize} 
} % end of longshort

%%%%%%%%%%%%%%%% Synchronous Programmiung %%%%%%%%%%%%%%%%%

\noindent The most prominent class of shared-memory programming languages that manage to reconcile concurrency and determinacy, is known as \textit{Synchronous Programming} (SP). % we find that priorities are not enough. 
In SP, processes interact with each other asynchronously at a \textit{micro-step} level, yet synchronise in lock-step to advance jointly in iterative cycles of synchronous \textit{macro-steps}. The micro-step interactions taking place during a macro-step determine the final outcome of the macro-step. This outcome is defined at the point where all subsystem pause to wait for a global logical \textit{clock}. When this clock arrives, all subsystem proceed into the next computation cycle. 
Under the so-called \textit{synchrony hypothesis}\longshort{\version}{}{\footnote{The  ``synchrony hypothesis'' assume that the reaction of a subsystem to input happens faster than the environment is able to produce new stimulus: the time needed to compute the reaction is unobservable by the environment. Also known as the ``zero-time'' assumption \cite{KloosB95} or ``fundamental mode operation''~\cite{Unger:async-circ}.}}, the final outcome of a macro-step for each subsystem is fully determined from the stimulus provided by the environment of that component during the micro-step interactions. 
As such, the interactions inside each subsystem at macro-step level can be abstracted into a global Mealy automaton with deterministic I/O.
\longshort{long}{This synchronous model of computation has been very successful, replacing analogue and asynchronous circuits, dominating hardware design up to today.}{} 
SP started with Statecharts~\cite{HarPnuPruShe87} and has generated languages such as Signal~\cite{GuernicGBM91}, Lustre~\cite{HalbwachsCRP91}, Esterel~\cite{Berry99}, Quartz~\cite{Schneider09}, SCCharts~\cite{vonHanxledenDM+14}, just to mention a few.

When it comes to mathematical modelling it is natural to think of a clock as a broadcast action in the spirit of Hoare's \csp~\cite{Hoare:CSP} that acts as a synchronisation barrier for a set of concurrent processes. Used with some scheduling control, it can bundle each process' actions into a sequence of macro-steps that are aligned with each other in a lock-step fashion. During each macro-step, a process executes only a temporal slice of its total behaviour, at the end of which it waits for all other processes in the same \textit{clock scope} to reach the end of the current phase. When all processes have reached the barrier, they are released into the next round of computation. 
This suggests a priority-based scheduling mechanism, %matches the computational model of SP.
i.e., that the clock should fire only if the system has stabilised and no other admissible data actions are possible. This principle, called \textit{maximal progress}, is built into timed process algebras, see e.g.~\cite{TPA,CleavelandLM97}.

%---------------------------------------------------
\paragraph*{Contribution}
%---------------------------------------------------

We introduce an extension of \ccs, denoted $\ccslm$, that brings together the concepts of \textit{priorities}~\cite{CamilleriWin95,CleavelandLN01,Phillips01} and \textit{clocks}~\cite{AndersenMen94,CleavelandLM97,NortonLM03} that have previously been studied for \ccs, but independently.
We thus obtain an adequate compositional setting for SP based on standard techniques from process algebra. This is surprising since many different customised %hand-crafted 
semantics (e.g.,~\cite{LuttgenBC99,BoussinotD96,Berry99,AguadoMvHF14}) have been developed for SP over the years, few of which fit into the classical framework of \ccs.
Similar ideas might be possible for other process algebras like {\sf ATP}~\cite{ATP}, \tpl~\cite{TPA}, \csp or some synchronous variants of Milner's $\pi$-calculus~\cite{SPI}, just to mention a few. We find \ccs most plausible as a point of departure since the concept of priorities %(which is very related) 
is already well established %in \ccs~\cite{CamilleriWin95,CleavelandLN01,Phillips01}. Also, in \ccs, 
and the combination of asynchronous scheduling and communication via local rendezvous synchronisation 
%with consumable prefixes 
makes the problem of ensuring determinacy more prominent.

Our calculus \ccslm provides a natural setting 
to take a fresh look at confluence in \ccs as an adequate mathematical model for SP. Moreover, we envisage that \ccslm,  extended by value-passing, can be used as a playground to study compositional embeddings of concurrent $\lambda$-calculi. In particular, we have in mind those with sharing, as required for the lazy semantics of Haskell and with deterministic memory such as IVars/TVars~\cite{Marlow:CEFP12} and LVars~\cite{KuperN:2013}. 
At the technical level, we make the following specific contributions in this paper: 

\longshort{\version}{}{This paper is inspired by earlier work of the second author on SP~\cite{AguadoMvHF14,mvm18:ESOP} and an attempt to ground the semantics of concurrent programming with deterministic shared objects in classical concurrency theory.}
%Our approach develops these existing theories of \ccs in three novel ways: 
\longshort{\version}{}{Combining prioritised rendezvous with broadcast actions in a Plotkin SOS style and discovering that an interleaving of micro- and macro-steps leads to confluence is the main result of this paper. Specifically, we make the following contributions:}

\noindent $\bullet$ We use broadcast actions (clocks) as global synchronisation barriers to schedule processes using priorities in a more flexible way than any of the earlier approaches. While \pmc~\cite{AndersenMen94} has clocks but no priorities and in %\tpl and 
\csa~\cite{CleavelandLM97} the priorities are hardwired and expressing precedence only between the rendezvous actions and the clock, we use a fully general priority scheme as in Phillips' \ccsp~\cite{Phillips01}.
%\ccs with priorities~\cite{CamilleriWin95,CleavelandLN01,Phillips01}. 
By adding clocks to \ccsp we generalise previous \ccs extensions such as Milner's \sccs~\cite{SCCS}, \pmc or %\tpl,
\csa for expressing synchronous interactions of concurrent processes. To achieve this, we must strengthen Phillips' notion of \textit{weak enabling} (Def.~\ref{def:weakly-enabled}) by \textit{constructive enabling} (Def.~\ref{def:c-enabling}). While all classical extensions of \ccs by priorities~\cite{CamilleriWin95,CleavelandLN01,Phillips01} schedule the processes by their immediate initial actions, constructive enabling takes into account all actions potentially executable up to the next clock barrier. 

\noindent $\bullet$ We identify a strictly larger class of processes that enjoy the Church Rosser (CR) Property (and thus reduce to unique normal forms) than the ``confluence class'' discussed by Milner, see e.g., Chap.~11 of~\cite{Milner:CCS}. Since priorities are not part of the classical confluence theory, it cannot model deterministic shared objects with conflicting choice. In this paper, exploiting priorities and clocks, we enrich Milner's notion of confluence to ``confluence up-to priorities'' that we call \textit{c-coherence} (Def.~\ref{def:coherence}). We show that c-coherent processes are CR for constructive enabling (Thm.~\ref{thm:church-rosser}). 

\noindent $\bullet$ We show that c-coherence is preserved by parallel composition under reasonable restrictions that permit sharing and memory (Thm.~\ref{thm:summary-coherence-closure}). Specifically, we construct c-coherent processes in a type-directed compositional fashion using the notion of \textit{precedence policies} (Def.~\ref{def:precedence-policy}). The policy for a conformant process specifies an upper bound for its precedence-blocking behaviour (Def.~\ref{def:conformance}). A policy enriches the classical sort of a process by priority information. We define the special class of \textit{pivot policies} and call c-coherent processes conforming to them  \textit{pivotable processes} (Def.~\ref{def:pivot-policy}). These are not only c-coherent (and thus CR) but preserved by parallel composition and other operators of the language (Thm.~\ref{thm:summary-coherence-closure}). In this fashion, we are able to extend Milner's classical results for confluent processes to cover significantly more applications, specifically Esterel-style multi-cast SP with shared objects.

\noindent $\bullet$ 
We show that each c-coherent process
is \textit{clock-deterministic} (Prop.~\ref{prop:action-determinism}) and that pivotable processes satisfy \textit{maximal progress} (Prop.~\ref{prop:maximal-progress}). 
No matter in which order we execute constructively enabled reductions, when the normal form is reached, and only then, a clock can tick. Since the normal form is uniquely determined\longshort{\version}{}{by confluence, the next state reached by each clock tick is uniquely determined. Thus}, each pivotable process represents a ``synchronous stream''. As such, our work extends the classical non-deterministic theory of timed (multi-clocked) process algebras~\cite{AndersenMen94,TPA,CleavelandLM97} for applications in deterministic synchronous programming.

As a final remark, for the purposes of this paper, adding synchrony and determinism in process algebras, we are not (and indeed we do not need to be) interested in higher-order calculi, like e.g. the $\pi$-calculus~\cite{Milner1:PI,Milner2:PI}, where processes are \textit{first-class} objects, i.e. sent and received as a values, as in $a \out P \cseq Q$ and $a \inp x \cseq (x \cpar Q)$.

For lack of space, four separate appendices will present (A) auxiliary results and (B) the full proofs; (C) a collection of more tricky examples; (D) how research towards determinacy for process algebras passes through literature and our speculations.

\longshort{\version}{}
{In the next sections,  we present the syntax, the structural operational semantics à la Plotkin, we set up all the definitions and the metatheory necessary to formally prove the claimed informal properties.
}
 % clean

%!TEX root = synpatick-lipics.tex
%
 
% \vspace{-4mm}
\section{A Process Algebra with Clocks and Priorities}\label{sec:lang}
%\vspace{-2mm}

%
\longshort{\version}
{The syntax and operational semantics of our process algebra `SynPaTick', denoted $\ccslm$, is defined in this section.}
{We call our process algebra `SynPaTick', denoted $\ccslm$, whose syntax and operational semantics are defined in this section.}

%\vspace{-4mm}
\subsection{\ccslm Syntax}
%\paragraph*{\ccslm Syntax}
%\vspace{-2mm}
%
%
The terms in \ccslm form a set $\P$ of \textit{processes}. We use $P,Q,R,S \ldots$ to range over $\P$.
Let $\I$ be a set of \textit{process names} and let $A,B \ldots$ range over $\I$.
Let $\A$ be a countably infinite set of \textit{channel names} with $a,b \ldots$ ranging over $\A$. As in \ccs, the set $\coA$ is a disjoint set of \textit{co-names} with elements denoted by $\ol{a}, \ol{b}, \ol{c} \ldots$ %that are 
in bijection with $\A$; The overbar operator 
%$\ol{\cdot}$ 
switches between $\A$ and $\coA$, i.e., $\ol{\ol{a}} = a$ for all $\ol{a} \in \coA$. 
We will refer to the names 
%$\A \cup \coA$ as the \textit{action labels} and use the names 
$a \in \A$ as \textit{input} (or \textit{receiver}) labels while the co-names $\ol{a} \in \coA$ denote \textit{output} (or \textit{sender}) labels.
Let $\C$ be a countably infinite set of broadcasted \textit{clock names}, disjoint from both $\A$ and $\coA$ and let $\sigma$ range over $\C$. Every clock acts as its own co-name, $\ol{\sigma} = \sigma$ and so each $\sigma \in \C$ is also considered as a \textit{clock label}.
Let $\L = \A \cup \coA \cup \C$ be the set of input, output and clock \textit{labels} and let $\ell$ range over $\L$, while $L,H$ range over subsets of $\L$. We write $\ol{L}$ for the set $\{\ol{\ell} \mid \ell \in L\}$. Let $\Act = \L \cup \{\tau\}$ be the set of all \textit{actions} obtained by adjoining the \textit{silent} action $\tau \notin \L$ and let $\alpha$ range over $\Act$. Viewed as actions, the input and output labels $\ell \in \R \eqdef \A \cup \coA \subseteq \Act$ are referred to as \textit{rendezvous actions} 
%or \textit{channel actions} 
\longshort{\version}{to distinguish them from the \textit{clock actions} $\sigma \in \C \subseteq \Act$.}{to distinguish them from the \textit{clock actions} $\sigma \in \C \subseteq \Act$ which also called \textit{broadcast actions}.
} 
\longshort{\version}{}{Further, to distinguish the elements of $\R \cup \C$ from $\tau$ as a subset of $\Act$ we call the former \textit{visible} actions in contrast to $\tau$ that is the \textit{invisible} action.} All symbols can appear indexed.
\longshort{\version}{The process terms $\P$ are defined thus:}{The process terms $\P$ are defined by the following abstract syntax:} 

\medskip 

%\mbox{
%\noindent\hspace{-2mm}
\begin{minipage}[T]{0.5\textwidth}
$ \begin{array}{lcll@{\quad}l}
 P,Q,R,S &::=& 
 \zero
 & & \text{stop (inaction)} \\
 %&\mid& \one 
 %& & \text{halt (idling)} \\
 %&\mid& P[f] 
 %& & \text{relabelling, as in \ccs \text{with }} f: \Act \rightarrow \Act \\
 &\mid& A 
 & & \text{name, } A \in \I \\
 &\mid& \alpha \col H.P 
 & & \text{action, } \alpha \in \L \text{,  } H \subseteq \L\\
  &\mid& P + Q
 & & \text{choice}
\end{array}
$
 \end{minipage}
 \hspace{20mm}
 \begin{minipage}[t!]{0.5\textwidth}
$ \begin{array}{lcll@{\quad}l}
\\
 &\mid& P \cpar Q
 & & \text{parallel} \\ 
 &\mid& P \restrict L
 & & \text{restriction} \\
 &\mid& P \hide L
 & & \text{hiding} 
 \end{array}
$
 \end{minipage}
%}
%
\medskip 

\noindent Intuitively, $\zero$ denotes the \textit{inactive} process which stops all computations, including clock actions. %broadcasting synchronisation actions. 
%, and $A$ is 
The \textit{action prefix} $\alpha \col H.P$ denotes a process that offers to engage in the action $\alpha$ and then behaves as $P$ while the \textit{blocking} set $H \subseteq \L$ lists all actions taking precedence over $\alpha$. A prefix with $\alpha\in \R$ denotes a \ccs-style rendezvous action. In case $\alpha \in \C$, it denotes a \csp-style broadcast action that can communicate with all the surrounding processes sharing on the same clock only when the clock is universally \textit{consumed}.
We assume that in every restriction %expression 
$P \restrict L$ the set $L \subseteq \R$ consists of rendezvous actions, and in every hiding %expression 
$P \hide L$ we have $L \subseteq \C$. These \textit{scoping} operators act as name binders that introduce local scopes. % for labels. 
The acute reader will notice the absence of the relabelling operator $P[f]$, useful in \ccs to define the well-known ``standard concurrent form'' $(P_1[f_1] \cpar \cdots \cpar P_2[f_n]) \restrict L$. Because all of our examples are captured without need of relabelling, 
%this process and 
and for the sake of simplicity, relabelling will be treated in the full version of this work.
A useful abbreviation is to drop the blocking sets if they are empty or drop the continuation process if it is inactive. For instance, we will typically write $a \cseq P$ instead of $a \of \eset \cseq P$ and $a \col H$ rather than $a \col H \cseq \zero$. Also, it is useful to drop the braces for the blocking set if it is a singleton, writing $a \col b \cseq P$ instead of $a \col \{b\} \cseq P$. 
\longshort{\version}{}{With these notational shortcuts we can naturally take each action $\alpha \in \L$ also as a process $\alpha \col \eset \cseq \zero$. In this way, each action $\alpha$ is a process that offers to engage in $\alpha$, without blocking, and then stops.
} 
\longshort{\version}{}{%
To encode meaningful examples in our process algebra conveniently, we also define a generic \textit{halting} process $\one$ which does not offer any rendezvous actions but is happy to permit each clock without changing its state. For a finite set of clocks $C = \{ \sigma_1, \ldots, \sigma_n \} \subset \C$ it is definable as $\one_C \eqdef \sigma_1.\one_C + \cdots + \sigma_n.\one_C$. When $C$ is understood we also write $\one$ instead of $\one_C$.
}
In our concrete notation using the abstract syntax of processes, we %will 
assume that the binary operator $+$ is right associative and taking higher binding strength than the binary operator $\cpar$, i.e., writing $P + Q + R \cpar S$ instead of $(P + (Q + R)) \cpar S$. Also, we assume that the unary operators of prefix, restriction and hiding have higher binding strength than the binary operators. Hence, we write $P \restrict c + a \cseq b \cseq c \cseq Q \cpar R \hide \sigma $ instead of $((P \restrict c) + (a \cseq (b \cseq (c \cseq Q)))) \cpar (R \hide \sigma)$.
The \textit{sub-processes} of a $P$ are all sub-expressions of $P$ and (recursively) all sub-processes of $Q$ for \textit{definitions} $A \pdef Q$ where $A$ is a sub-expression of $P$.
Sometimes it is useful to consider fragments of $\P$. %\ccslm.
Specifically, a process is called \textit{unclocked} if $P$ does not contain any clock prefix $\sigma\col H \cseq Q$ as a sub-process. A process is called \textit{sequential} or \textit{single-threaded} if it does not contain any parallel sub-process $Q \cpar R$. 
\longshort{\version}{}{%
\begin{definition}[Visible Labels]\hfill
 The set $\vA(P) \subseteq \L$ of \textit{visible labels} of a process $P$ is defined inductively as follows:
 \[
 \begin{array}{rcll}
 \vA(\zero) & \eqdef & \eset \\
 \vA(\alpha \col H.P) & \eqdef & \{\alpha\} \cup H \cup \vA(P)
 \\
 \vA(P \wr L) & \eqdef & \vA(P) - (L \cup \ol{L}) & \text{ where $\wr \in \{\restrict, \hide\}$}
 \\
 \vA(P \star Q) & \eqdef & \vA(P) \cup \vA(Q) & 
 \text{ where $\star \in\{\cpar, + \}$}
 \\ 
 \vA(A) &\eqdef& \vA(P) & \text{ where $A \pdef P$.}
 \end{array}
 \]
\label{def:visible-labels}
\end{definition}
\myb{It is not difficult to extend this definition to the halting process $\one_C$ as follows: $\vA(\one_C) = C$.}
}
$P$ is %\textit{action-closed}, or simply 
\textit{closed} if 
\longshort{\version}{%
all sub-processes $\ell\col H\cseq Q$ in $P$ occur in the scope of a restriction or hiding operator that binds $\ell$.
}
{$\vA(P) = \eset$.} 
As usual, we denote by $\L(P) \subseteq \L$ the free labels, called \textit{(syntactic) sort} of $P$,
identify %syntactically two 
processes that differ only in the choice of bound labels.

\begin{figure}[t!]
\noindent \hspace{-3.5mm}
$\begin{array}{l@{\hspace{5mm}}rcl@{\quad}l}

({\sf Comm}) & P \star Q & \equiv & Q \star P 
\\[1mm] 

({\sf Assoc}) & (P \star Q) \star R & \equiv & P \star (Q \star R)
\\[1mm]

({\sf Zero}) & P \star \zero & \equiv & P & 
\end{array}
\quad
\begin{array}{l@{\quad}rcl@{\quad}l}
({\sf Idem}_{+}) & P + P & \equiv & P
& 
\\[1mm]

({\sf Scope}_{\wr}) & P \wr L \wr L' & \equiv & P \wr (L\cup L')
\\[1mm]

({\sf Scope}_{\zero}) & \zero \wr L & \equiv & \zero & 
\end{array}$
\\[1mm]

$\noindent \hspace{-2mm}
\begin{array}{l@{\quad}rcl@{\quad}l}
({\sf Scope}_{\alpha}) & (\alpha \col H.P) \restrict L & \equiv &
			\left\{ \begin{array}{l}
			\zero 
			\\
			\alpha \col H. (P \restrict L) 
			\end{array}
			\right.
			&
			\begin{array}{l}
			\mbox{if } \alpha \in L \cup \ol{L}\; 
			\\
			\mbox{otherwise}\; 
			\end{array}
\\[1mm]

({\sf Scope}_{\restrict}) & (P \star Q) \restrict L & \equiv &
			\left\{ \begin{array}{ll}
			P \star Q \restrict L 
				%& \mbox{ if } \vA(P) \cap L = \eset
			\\
			P \restrict L + Q \restrict L 
				%& \mbox{ if } \star = +
			\\
			P \restrict L \cpar Q \restrict L 
				%& \mbox{ if } \star = \cpar \mbox{ and } \vA(P) \cap \ol{\vA(Q)} \cap L = \eset 
			\end{array}
			\right.
			&
			\begin{array}{l}
			\mbox{if } \L(P) \cap (L \cup \ol{L}) = \eset
			\\
			\mbox{if } \star = +
			\\
			\mbox{if } \star = \cpar \mbox{ and } \L(P) \cap \ol{\L(Q)} \cap (L \cup \ol{L}) = \eset
			\end{array}
%\\[1mm]

\end{array}
$
\caption{Structural Congruence, where $\star \in\{\cpar,+\}$ and $\wr \in \{\restrict, \hide\}$.}
\label{fig:structural-cong}
\end{figure}

Following 
Milner's presentation of the $\pi$-calculus in~\cite{PI}, we define a structural congruence over processes. 
\longshort{\version}{Specifically, 
\textit{structural congruence} $\equiv$ is the smallest congruence over $\P$ that satisfies the equations in Fig.~\ref{fig:structural-cong}.}{}
The idea of the structural congruence is to split interaction rules between agents from rules governing the ``left-to-right'' representation of expressions. 
It allows to use those equations in the derivation rules of the SOS making the latter simpler and easy to use, and it allows to set up formal equivalence of processes increasing for each program the number of legal SOS transitions that, otherwise, would be blocked. 
Milner divides \ccs equations in static, dynamic and expressions laws. We choose an approach inspired by the work of Berry and Boudol~\cite{BB-cham} by adding an explicit structural rule in the operational semantics.
\longshort{\version}{}{%
\begin{definition}%[Structural Congruence $\equiv$]\hfill\\
\textit{Structural congruence} $\equiv$ is the smallest congruence over $\P$ that satisfies the equations in Fig.~\ref{fig:structural-cong}. 
\label{def:structural-cong}
\end{definition}
}
Note that we do not include in our structural rules all of Milner's $\tau$-laws, i.e. the ones related to the presence of $\tau$-actions, like e.g. $\alpha.\tau. P \equiv \alpha.P$, or $P + \tau.P \equiv \tau.P$,
and we have also dropped the 
congruence equations related to the uniqueness of solution of recursive equations and the 
\textit{expansion laws} related to the \textit{standard concurrent form} (cf. Chap. 3, Props. 2, 4(2) and 5 of \cite{Milner:CCS}).
\longshort{\version}{Since we do not propose notions of behavioural equivalence in this paper, we prefer to take a conservative approach and only consider a minumum set of purely structural rules that do not impinge on the dynamics of expression behaviour.}{
It is known that timed process algebras and process algebras with priorities enjoy extra expressiveness that invalidates some of the familiar laws, see e.g.,~\cite{AndersenMen94,TPA,CleavelandLM97,CleavelandLN01,Phillips01}.
\ccslm shares features (clocks and priorities) with timed and prioritised extensions of \ccs. Hence, we must take a conservative approach and only consider a minumum set of purely structural rules that do not impinge on the dynamics of expression behaviour.
}

%-------------------------------------
\subsection{Operational Semantics}
%-------------------------------------
%\vspace{-2mm}

The operational semantics for \ccslm is given using a labelled transition relation $P \fstep{\alpha} Q$ in the style of Plotkin's Structured Operational Semantics (SOS)~\cite{PlotkinSOS,PlotkinSOSorigin}, where $\alpha \in \Act$ is the action that labels the SOS transition. When $\alpha \in \L$, the transition corresponds to a communication step, namely a rendezvous action {\it or} a broadcasted clock action. 
When $\alpha = \tau$ is silent, then the transition corresponds to a \textit{reduction} or an \textit{evaluation} step. 
In addition, we decorate the nondeterministic SOS with annotations that provides sufficient information to express a uniform and confluent reduction strategy based on scheduling precedences. 
Including these annotations, our SOS judgment takes the form
%
%\begin{eqnarray}
$ P \Derives{\alpha}{H}{R} Q$
% \label{eqn:prio-micro-step}$ 
%end{eqnarray}
%
where $P, Q, R \in \P$ are processes, $\alpha \in \Act$ is an action, and $H \subseteq \Act$ a set of actions.
The semantics is defined inductively as the smallest relation closed under the rules in Fig.~\ref{fig:free-sos}. We call each transition
%~\eqref{eqn:prio-micro-step} 
derivable by Fig.~\ref{fig:free-sos} an \textit{admissible} transition. We denote by $\alpha \col H[R]$ a compound  \textit{c-action} consisting of the three annotations $\alpha$, $H$, and $R$. The information provided by such a c-action is sufficient to define our \textit{constructive reduction} strategy for \ccslm that satisfies a refined confluence property.  Sometimes, we write $P \xrightarrow{\alpha} Q$ to state that there is a transition for some $R$ and $H$. 

The intuition behind the annotations of admissible transitions is as follows: $\alpha$ is the usual action ``passed back'' in the SOS. The \textit{blocking set} $H$ is the set of actions that take precedence over $\alpha$; %the transition, 
$R$, called the \textit{concurrent context}, is a process that represents the behaviour of all threads in $P$ that execute concurrently with the transition, i.e., all actions in $R$ are potentially in competition with $\alpha$. 
Think of $H$ as the resources that are required by the action and of $R$ as the concurrent context that potentially competes for $H$. In a parallel composition $P \cpar Q$ of $P$ with another process $Q$, the context $R$
%~\eqref{eqn:prio-micro-step} 
might interact with the environment $Q$ to generate actions in the blocking set $H$. If this happens, the reduction
%~\eqref{eqn:prio-micro-step} 
in $P$ with blocking set $H$ will be considered blocked in our scheduled semantics (see Def.~\ref{def:strong-enabling}).
% that we are deriving from~\eqref{eqn:prio-micro-step}. 

In the rules of Fig.~\ref{fig:free-sos}, the reader can observe how the annotation for the concurrent context and the blocking sets are inductively built. 
In a nutshell: 
Rule $(\ActR)$ is the axiom annotating in the transition the action and the (empty) environment. 
Rules $(\RestrR)$ and $(\HideR$) deal with ``local restriction'' of rendezvous actions and ``local hiding'' of broadcasted clocks. 
Rules $(\ParR)$
%_{1,2})$ 
and $(\SumR)$
%_{1,2})$ 
are almost standard and deal with parallelism and choice. 
Rule $(\StructR)$ internalises in the transition the above presented structural equivalence of processes. 
Rule $(\ComR)$ constitutes the core of our
SOS. It captures, in an elegant way, both the classical \ccs rendezvous communication rule \textit{and} the classical \csp broadcasting communication; it has an extra $race$ test, enriching the blocking set of the conclusion by $\tau$, in case $P$ and $Q$ interfere. The presence of $\tau$ will block unconditionally.
Finally, rule $(\ConR)$ deals with process names and refers to a process whose behaviour is specified by a definitional equation $A \pdef P$.
\longshort{\version}{}{% 
It is not difficult to observe that 
\longshort{\version}{the transition $\one_C \Derives{\sigma}{\eset}{\zero} \one_C$ is derivable when $\sigma \in C$} 
{the following rule $(\ActR_\one)$ is derivable:
\[
\infer[(\ActR_\one)]
{\one_C \Derives{\sigma}{\eset}{\zero} \one_C}
{\sigma \in C}
\]
}
}

\begin{figure}[!t]
$ %\[
\begin{array}{l@{\hspace{2cm}}r}
\multicolumn{2}{l}{
\infer[(\ActR)]
{\alpha \col H \cseq P \Derives{\alpha}{H}{\zero} P}
{}
\qquad
\hfill \infer[(\StructR)]
{P \Derives{\alpha}{H}{R} Q}
{
P \equiv P' & 
P' \Derives{\alpha}{H}{R'} Q' & 
Q' \equiv Q 
& 
R' \equiv R}
} 
\\[3mm]
\multicolumn{2}{l}{
\infer[(\RestrR)]
{P \restrict L \Derives{\alpha}{H'}{R \restrict L} Q \restrict L}
{P \Derives{\alpha}{H}{R} Q & L' = L \cup \ol{L} & \alpha \not\in L' &
 H' = H - L'}
\qquad 
\infer[(\HideR)]
{P \hide L \Derives{\alpha \hide L}{H'}{R \hide L} Q \hide L}
{P \Derives{\alpha}{H}{R} Q & H' = H - L}
}
\\[3mm]
\multicolumn{2}{l}{
\infer[(\ParR)]
{P \cpar Q \Derives{\alpha}{H}{R \cpar Q} P' \cpar Q}
{P \Derives{\alpha}{H}{R} P' & \alpha \not\in \C } %& {\color{red} \alpha \not\in H \cap \oliA(Q)}}
\qquad \hfill
\infer[(\SumR)]
{P + Q \Derives{\alpha}{H}{R} P'}
{P \Derives{\alpha}{H}{R} P'}
}
\\[4mm]
\multicolumn{2}{l}{
\infer[(\ComR)]
{P \cpar Q \oldDerives{\ell \cpar \ol{\ell}}{H_1 \cup H_2 \cup H}{R_1 \cpar R_2} P' \cpar Q'}
{P \Derives{\ell}{H_1}{R_1} P' & Q \Derives{\ol{\ell}}{H_2}{R_2} Q' & H = race(P, Q)} 
%& H = \{ \tau \mid 
% H_2 \cap \oliA(P) \not\subseteq \{\ol{\ell}\} 
% \text{ or } 
% H_1 \cap \oliA(Q) \not\subseteq \{\ell\} \} 
\qquad \hfill
\infer[(\ConR)]
{A \Derives{\alpha}{H}{R}P'}
{A \eqdef P & P \Derives{\alpha}{H}{R} P'}
}
\\[3mm]
\multicolumn{2}{l}{
\textit{where } 
 \hfill 
\begin{array}{l}
\hspace{10mm} \alpha \hide L \eqdef \left\{\begin{array}{l@{\quad\ \ }l} \tau & \text{if } \alpha \in L \\ \alpha & \text{otherwise} \end{array} \right.
\quad 
\ell \cpar \ol{\ell} \eqdef \left\{\begin{array}{ll} \tau & \text{if } \ell \in \R \\ \ell & \text{if } \ell \in \C \end{array}\right.
\\[4mm]
 race(P,Q) \eqdef \left\{\begin{array}{ll} \{\tau\} & 
 \mbox{if } 
% \ell \notin H_1 \cap \oliA(Q) \ \ 
 H_1 \cap \oliA(Q) \not\subseteq \{\ell\} 
 \text{ or } 
 %\ol{\ell} \notin H_2 \cap \oliA(P)
 H_2 \cap \oliA(P) \not\subseteq \{\ol{\ell}\}
 \\ \eset & \mbox{otherwise}
 \end{array}
 \right.
\end{array}
} 
\end{array}
$ %\]
\caption{The admissible transitions of \ccslm processes. %, defined using Plotkin's SOS. 
See Def.~\ref{def:weakly-enabled} for the \textit{initial actions} $\iA(P)$.}\vspace{-4mm}
\label{fig:free-sos}
\end{figure}

%\vspace{-5mm}
%---------------------------------------
\subsection{Constructive Scheduling}
\label{sec:constr-sched}
%---------------------------------------
%\vspace{-2mm}

The operational semantics of Fig.~\ref{fig:free-sos} for rendezvous actions  coincides with that of \ccs and the rules for clock actions coincide with the semantics of \csp broadcast actions. The scheduling annotations $H$ and $R$ of a c-action $\alpha \col H[R]$ expose extra information about blocking and concurrent context, respectively, of the underlying action $\alpha$. We can use them to implement different scheduling strategies.
%The traditional approach in 
Prioritised process algebras~\cite{CamilleriWin95,CleavelandLN01,Phillips01} block an admissible transition $\alpha \col H[R]$ when $H$ contains the silent action $\tau$ or an action that synchronises with some initial action of the concurrent context $R$.
%the concurrent context $R$ offers an initial action that synchronises with some action blocked by $H$, or if $H$ contains the silent action. 
Otherwise, it is (weakly) enabled.

\begin{definition}[Initial actions and weakly-enabled transition]
\longshort{\version}{}{[Strong Initial Actions \& Weak Enabling]} \mbox{}
 \begin{itemize}
   \item The set $\iA(R) \subseteq \Act$ of \textit{initial actions} of a process $R$ is given as $\iA(R) = \{ \alpha \mid \exists Q.\, R \xrightarrow{\alpha} Q \}$.
   %is the smallest set closed under the inductive rules in Fig.~\ref{fig:strong-initial-actions}. %\label{def:strong-initial-actions}
   \item 
   \longshort{\version}{A transition with c-action $\alpha \col H[R]$ 
   is called \textit{weakly enabled} if $H \cap (\oliA(R) \cup \{\tau\}) = \eset$.
   }
   {An admissible transition 
   $P \Derives{\alpha}{H}{R} Q$ %\eqref{eqn:prio-micro-step} 
   % $P \Derives{\alpha \col H[R]}{}{} Q$ %\eqref{eqn:prio-micro-step} 
   is called \textit{weakly enabled} if $H \cap (\oliA(R) \cup \{\tau\}) = \eset$.
   }
 \end{itemize}
\label{def:weakly-enabled}
\end{definition}

\longshort{\version}{}{%
Extending this definition to the halting process $\one_C$ yields $\iA(\one_C) = C$.
}

\noindent We call $P$ \textit{discrete} if all its c-actions $\alpha\col H[R]$ are singleton self-blocking, i.e., if $H = \{\alpha\}$. 
Let us call a process $P$ \textit{free} if 
%the blocking sets of 
all its c-actions $\alpha\col H[R]$ %occurring as sub-processes of $P$ 
have an empty blocking set $H = \eset$. 
The weakly-enabled transitions of unclocked, free processes coincide with the semantics of Milner's \ccs~\cite{Milner:CCS}, as stated in the following proposition.
\begin{proposition}
 If $P$ is unclocked and free, then a transition 
 $P \Derives{\alpha}{H}{R} Q$  
 is weakly enabled iff $P \Derives{\alpha}{}{} Q$ is derivable in Milner's original SOS for \ccs.
\label{prop:free-ccs}
\end{proposition}

In the semantics of \ccsp \textit{self-blocking} prefixes do not contribute any transitions. 
More precisely, a process $P$ is called \textit{irreflexive} if no action prefix $\alpha \col H \cseq Q$ occurring in $P$ blocks itself, i.e., $\alpha \not\in H$. 
Thus, \ccsp captures a theory of irreflexive processes.
One can show that for unclocked, irreflexive processes a transition is weakly enabled iff is derivable in the operational semantics of \ccsp~\cite{Phillips01}\footnote{In \ccsp the blocking sets are written before the action, $H \col \alpha \cseq P$, where we write $\alpha \col H \cseq P$ in \ccslm.}.
%Thus, \ccs and \ccsp correspond to theories of free and irreflexive processes of \ccslm, respectively, under weak enabling. 

%
\begin{proposition}
 If $P$ is unclocked and irreflexive, then a transition 
 $P \Derives{\alpha}{H}{R} Q$ 
 %$P \Derives{\alpha \col H[R]}{}{} Q$ 
 %\eqref{eqn:prio-micro-step} %$P \Derives{\alpha}{H}{R} P'$ 
 is weakly enabled iff $P \Derives{\alpha}{H}{} P'$ is derivable in the operational semantics of \ccsp.
\end{proposition}

Thus, from the point of view of \ccslm, the theories of \ccs and \ccsp correspond to theories of free and irreflexive processes, respectively, under weak enabling.
Likewise the theories~\cite{CamilleriWin95,CleavelandLN01} can be subsumed; specifically, the \textit{prioritised sum} operator $P \wpsum Q$ of~\cite{CamilleriWin95} 
can be coded as $P + Q \col \iA(P)$ in \ccslm.
The theory~\cite{CleavelandLN01} proposes two \textit{levels} of labels, $a\col 0$ and $a\col 1$, where $0$ denotes higher priority and $1$ ordinary priority. 
%therefore, prefixes are annotated as (shortened by $\underline{a}$ and $a$, respectively): 
Translated into \ccslm, the prefix $a \col 0 \cseq P$ corresponds to $a\col \eset \cseq P$ 
%with empty blocking set, 
while an %(disjoint set of)
%non-prioritised 
action prefix $a \col 1 \cseq P$ maps in \ccslm to $a \col \mathsf{Prio} \cseq P$ if we define $\sf Prio$ as the set of all high-prioritised labels.

% \paragraph*{Note on Conservativity}

% % %
% % In fact, one can show that the semantics of weak enabling generates the theory of \ccsp~\cite{Phillips01}. 
% % To be more precise, let us call a process $P$ \textit{irreflexive} if no action prefix $\alpha \col H \cseq Q$ occurring in $P$ has $\alpha \in H$, i.e., blocks itself. 
%  Here we suggest two extension to this classical work. Firstly, we apply a stronger scheduling restrictions, called \textit{strong enabling} (Def.~\ref{def:strong-enabling} below) and \textit{constructive enabling} (Def.~\ref{def:c-enabling} above) to achieve confluence and secondly we make use of self-loops to control sharing. 

\medskip 

Examples~\ref{ex:parallel-blocking} and~\ref{ex:binary-blocking} below in Sec.~\ref{examples:scheduling} demonstrate how prioritised scheduling by weak enabling (Def.~\ref{def:weakly-enabled}) can indeed help to ensure determinacy in special cases. For the general case, however, weak enabling is not sufficient. 
In order to prevent non-determinacy (a global property) from priorities (a purely local property), 
we must be able to block a c-action  $\alpha \col H[R]$ not just if the initial actions $\iA(R)$ of the concurrent context $R$ can synchronise with $H$, as in weak enabling, but look at the set $\wilA{\ast}(R)$ of all actions that $R$ can potentially engage in, before it reaches a clock. 
\longshort{\version}{}{%
It obtained by closing the weak initial actions $\wilA{\tau}(P)$ by arbitrary non-clock transitions to generate \textit{potential actions} $\wilA{\ast}(P)$ of a process $P$.
}
Example~\ref{ex:weak-causality-cycle} in Sec.~\ref{examples:scheduling} illustrates this.  

%
%%--------------------------------------
%\paragraph*{Constructive Enabling}
%%--------------------------------------
%
%\\\noindent \textbf{Constructive Enabling.} 

\begin{definition}[Potential actions and constructive enabling] \mbox{}
 \begin{itemize} 
   \item The set $\wilA{\ast}(P) \subseteq \L$ of \textit{potential actions} is the smallest extension $\iA(P) \cap \L \subseteq \wilA{\ast}(P)$ such that if $\ell \in \wilA{\ast}(Q)$ and $P \xrightarrow{\alpha} Q$ for $\alpha \in \R \cup \{\tau\}$, then $\ell \in \wilA{\ast}(P)$.
   %In other words, to obtain $\wilA{\ast}(P)$ we close $\iA(P)$ under arbitrary non-clock transitions of $P$.
   \item 
   \longshort{\version}{A transition with c-action $\alpha \col H[R]$ is called \textit{constructively enabled}, or \textit{c-enabled} for short, if $H \cap (\olwilA{\ast}(R) \cup \{\tau\}) = \eset$.}
   {A transition $P \Derives{\alpha}{H}{R} Q$ is called \textit{constructively enabled}, or \textit{c-enabled} for short, if $H \cap (\olwilA{\ast}(R) \cup \{\tau\}) = \eset$.}
  \end{itemize}
\label{def:c-enabling}
\end{definition}

The following Sec.~\ref{examples:scheduling} presents some motivating examples, while Sec.~\ref{sec:properties} focuses 
on the mathematical theory underlying \ccslm.

%%% Local Variables:
%%% mode: latex
%%% TeX-master: "synpatick-lipics.tex"
%%% End:
 % clean

%!TEX root = synpatick-lipics.tex
%

%\vspace{-3mm}
%--------------------------------------------------
\section{Examples}
\label{examples:scheduling}
%--------------------------------------------------
%\vspace{-3mm}

\longshort{\version}{}{The following examples illustrate the operation of our scheduling semantics. The focus is on how the transition labels, blocking sets and action environments implement precedences.}
\longshort{\version}{}
{%
\myb{C-Coherence} (see Def.~\ref{def:coherence})
is one of our main concepts, intuitively, that races between concurrent sub-processes will not affect the final result. %
For $s \cseq s \cseq t \col t$ to be \myb{c-coherent} we need to add self-blocking to both $s$ prefixes.
Firstly, $s \cseq t \col t$ is not \myb{c-coherent} because 
%$s \cseq t \col t \fsstep{s[\zero]} t \col t$ 
$s \cseq t \col t \Derives{s}{\{\}}{\zero} t \col t$, but the successor process $t \col t$ does not permit us to repeat the label $s$. 
\longshort{\version}{}{%
Instead, $s \col s \cseq t \col t$ is \myb{c-coherent}.
For the same reason, the process $s \cseq s \col s \cseq t \col t$ is not \myb{c-coherent}, because 
%$s \cseq s \col s \cseq t \col t \fsstep{s[\zero]} s \col s \cseq t \col t$ 
$s \cseq s \col s \cseq t \col t \Derives{s}{}{\zero} s \col s \cseq t \col t$ 
but the successor which repeats the label $s$, is self-blocking: 
%$s \col s \cseq t \col t \fsstep{s\col s[\zero]} t \col t$
$s \col s \cseq t \col t \Derives{s}{}{\zero} t \col t$, i.e., the blocking set has increased. 
The process $s \col s \cseq s \col s \cseq t \col t$ is \myb{c-coherent}, but the parallel composition $\ol{s} \cpar \ol{s} \cpar s \col s \cseq s \col s \cseq t \col t$ is blocking. None of the senders is permitted to handshake with the initial label $s \col s$ of the receiver.
}
Coherence and confluence are not the same thing: the process $\ol{s} \cpar \ol{s} \cpar s \cseq s \cseq \ol{t}$ is not \myb{c-coherent} but confluent. Confluence requires that there are at least as many receivers as there are senders, so it does not matter which is synchronising first. If we add another sender as in $\ol{s} \cpar \ol{s} \cpar \ol{s} \cpar s \cseq s \cseq \ol{t}$ we lose confluence. Now it depends which two of the three senders $\ol{s} \cpar \ol{s} \cpar \ol{s} $ actually synchronises with the receiver $s \cseq s \cseq \ol{t}$. 
} % end of longshort 
The first example shows the basic mechanism of how priorities can be used to schedule processes under weak enabling.
\begin{example}[Read Before Write]
\label{ex:parallel-blocking}
 Consider a ``store'' $S \eqdef w\cseq r + r \col w$ that offers a write action $w$ followed by a read $r$, or a read action $r$ but prefers the write over the read. The fact that $w$ takes priority over $r$ is recorded in the blocking set $\{w\}$ which contains the action $w$. If we put $S$ in parallel with a reader $R \eqdef \ol{r}$, the read actions $r$ and $\ol{r}$ synchronise in a reduction 
 $ % \[
 S \cpar R 
 %= w + r \col w \cpar \ol{r} 
 \Derives{\tau}{\{ w \}}{\zero} \zero.
 $ %\]
 The blocking set $\{w\}$ will block $R$'s transition when a parallel writer $W \eqdef \ol{w}$ is added, as well. Specifically, 
 %the composite process $S \cpar R \cpar W$ generates the reduction
 %
 the reduction $S \cpar R \cpar W \Derives{\tau}{\{w\}}{\ol{w}} W$ 
 %$w + r \col w \cpar \ol{r} \cpar \ol{w} \Derives{\tau}{\{w\}}{\ol{w}} \ol{w}$
 is not weakly enabled, because $\{w\} \cap \oliA(\ol{w}) = \{w\} \cap \{w\} \neq \eset$. 
 The only enabled reduction %of process $S \cpar R \cpar W$ 
 is that between $S$ and $W$ giving 
 %to produce the transition 
 %
 $S \cpar R \cpar W \Derives{\tau}{\{w \}}{\ol{r}} r \cpar R$
 %
 % which is even strongly enabled, since $\{w \} \cap \olwiA(\ol{r}) = \eset$.
 and only then continue with the read $r \cpar R \Derives{\tau}{\{\}}{\zero} \zero$. Thus, weak enabling makes the reduction determinate.
 Observe that $S$ is not confluent (Def.~\ref{def:Church-Rosser}): 
 %(Def.~\ref{def:Church-Rosser}):
 We have $S \xrightarrow{w} r$ and $S \xrightarrow{r} \zero$ but there is no $P_1 \equiv P_2$ such that $r \xrightarrow{r} P_1$ and $\zero \xrightarrow{w} P_2$. However, $S$ is c-coherence (Def.~\ref{def:coherence}), because the transition $S \Derives{r}{\{w\}}{\zero} \zero$ interferes (Def.~\ref{def:interference-free}) with the former $S \Derives{w}{\{\}}{\zero} r$ as $w \in \{w\}$, whence c-coherence does not require confluence.
\qed
\end{example}
Processes can block each other from making progress. This happens under weak enabling if processes offer two-way rendez-vous synchronisation with contradicting precedences. %, as the following example illustrates.
\begin{example}[Binary Blocking]
\label{ex:binary-blocking} 
 Take processes $P \eqdef a \col b \cseq A + b$ and $Q \eqdef \ol{b} \col \ol{a} \cseq B + \ol{a}$. 
 Their composition $P \cpar Q$ has two admissible reductions, through $a \col b \cpar \ol{a}$ to $A$ or through $b \cpar \ol{b} \col \ol{a}$ to $B$.
 So, disregarding the precedences, it is non-deterministic if $A \not\equiv B$.  Under weak enabling, however, $P \cpar Q$ creates a circular deadlock: $P$ offers $a$ unless synchronisation with $\ol{b}$ is possible, while $Q$ offers $\ol{b}$ provided there is no synchronisation with $a$. None of the reductions is enabled. By the ``race'' side-condition of $(\ComR)$ the synchronisation $a \col b \cpar \ol{a}$ generates the c-action $\tau \col H[\zero]$ with $\tau \in H$ because $\{ b \} \cap \oliA(Q) \not\subseteq \{ a \}$. Similarly, since $\{ \ol{a}\} \cap \oliA(P) \not\subseteq \{ \ol{b} \}$, the synchronisation $b \cpar \ol{b} \col \ol{a}$ is blocked. 
  We note that neither $P$ nor $Q$ is confluent, but they are c-coherent.
\qed
\end{example}

\begin{example}[Reflexive Blocking]
\label{ex:reflexive-blocking}
  Consider the free process $s \cpar \ol{s} \cseq A \cpar \ol{s} \cseq B$ where action $s$ can be consumed by $\ol{s} \cseq A$ or by $\ol{s} \cseq B$, generating non-deterministic $\tau$-transitions to $A \cpar \ol{s} \cseq B$ or to $\ol{s} \cseq A \cpar B$. However, by adding self-blocking, as in $s \col s \cpar \ol{s} \cseq A \cpar \ol{s} \cseq B$, weak enabling protects the prefix $s$ from being consumed when both $\ol{s} \cseq A$ and by $\ol{s} \cseq B$ compete for it. If only one is present, as in $s \col s \cpar \ol{s} \cseq A$ and $s \col s \cpar \ol{s} \cseq B$, no blocking occurs.
\end{example}

Blocking under weak enabling ensures determinism in the special case where the circular dependency involves only initial actions. However, parallel threads can create dependency chains through sequences of actions, e.g., when accessing a shared process that is acting as a resource. This is the general case handled by c-enabling.
\begin{example}[Transitive Blocking]
\label{ex:weak-causality-cycle}
 Consider the process 
 $  %\[
 S \eqdef w_0 + r_0 \col w_0 \cpar w_1 + r_1 \col w_1
 $ %\]
 which offers receiver actions $w_i$ and $r_i$ for $i \in \{0,1\}$ such that $w_i$ has precedence over $r_i$. 
 Now take processes $P_0 \eqdef \ol{r_0} \cseq \ol{w_1}$ and $P_1 \eqdef \ol{r_1} \cseq \ol{w_0}$ which first aim to synchronise with $r_i$ and then sequentially afterwards with $w_{1-i}$. The parallel composition $P_0 \cpar S \cpar P_1$ has two initial $\tau$-re\-ductions
 \[ 
  P_0 \cpar w_0 + r_0 \col w_0 \cpar \ol{w_0}  \LDerives{\tau}{\{ w_1 \}}{P_0 \cpar (w_0 + r_0 \col w_0)}P_0 \cpar S \cpar P_1
 \Derives{\tau}{\{ w_0 \}}{w_1 + r_1 \col w_1 \cpar P_1} 
 \ol{w_1} \cpar w_1 + r_1 \col w_1 \cpar P_1.
\]
These reductions are both weakly enabled. If we continue under weak enabling then the residual process $\ol{w_1} \cpar w_1 + r_1 \col w_1 \cpar P_1$ for the first reduction must first execute the synchro\-nisation on $w_1$ and thus proceed as $\ol{w_1} \cpar w_1 + r_1 \col w_1 \cpar P_1 \fstep{\tau} P_1$ while the residual process $P_0 \cpar w_0 + r_0 \col w_0 \cpar \ol{w_0}$ can only reduce to $P_0$. Hence, the resulting final outcome of executing $P_0 \cpar S \cpar P_1$ is non-deterministic.
Observe that none of the two reductions displayed above is actually c-enabled. For the right reduction we find that its concurrent environment $R_0 \eqdef w_1 + r_1 \col w_1 \cpar P_1$ has a reduction sequence $R_0 \fstep{\tau} \ol{w_0}$ and thus 
$ %\[
\{ w_0 \} \cap \olwiA(R_0) \subseteq \{ w_0 \} \cap \olwiA(\ol{w_0}) = \{ w_0 \} \cap \{ w_0 \} \neq \eset.
$ %\]
Symmetrically, for the concurrent environment $R_1 \eqdef P_0 \cpar w_0 + r_0 \col w_0$ of the left reduction we have $R_1 \fstep{\tau} \ol{w_1}$ and hence $\{ w_1 \} \cap \olwiA(R_1) \neq \eset$. The parallel composition $P_0 \cpar S \cpar P_1$ is deterministic under c-enabling. We note again, $S$ is not confluent but c-coherence (Def.~\ref{def:coherence}).
\qed
\end{example}
The next example generalises Ex.~\ref{ex:parallel-blocking} to a full memory cell.
\begin{example}[Read/Write Memory]
\label{ex:rwmem}
Let $w_0, w_1, r_0, r_1$ the labels of writing/reading a boolean value $0$ or $1$, respectively. A memory would be modelled by two mutually recursive processes 
 \begin{eqnarray*}
 \wAND_0 & \pdef & 
 w_1 \cseq \wAND_1 + w_0 \col w_1 \cseq \wAND_0 + r_0 \col \{w_0, w_1 \} \cseq \wAND_0 \\
 \wAND_1 & \pdef & 
 w_1 \cseq \wAND_1 + w_0 \col w_1 \cseq \wAND_0 + r_1 \col \{ w_0, w_1 \} \cseq \wAND_1
 \end{eqnarray*}
\begin{minipage}{0.6\textwidth}
  The process $\wAND_v$ represents the memory cell in a state where it has stored the value $v \in \{ 0, 1 \}$. In either state, it offers to synchronise with a write action $w_u$ to change its state to $\wAND_u$, and also with a read action $r_v$ to transmit its stored value $v$ to a reader. The write actions $w_u$ do not only change the state of the memory, they also block the read actions since they appear in their blocking sets $r_v \col \{ w_0, w_1\}$. 
\end{minipage}
$\quad$
\begin{minipage}{0.35\textwidth}
   \includegraphics[scale=.5]{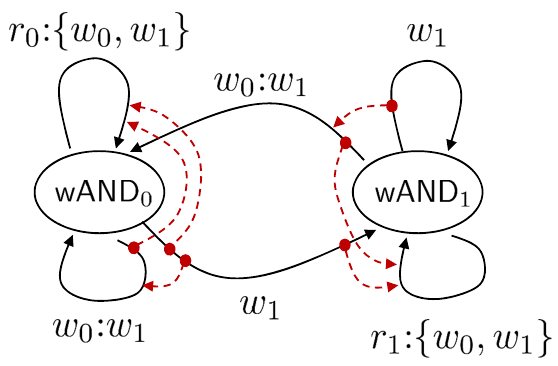}
\end{minipage}\\[1mm]
Between the write actions, the blocking sets implement a precedence of $w_1$ over $w_0$. In the picture on the right the process state transitions are visualised with the blocking sets as {\color{red} red} dashed arrows.
\qed
\end{example}

\begin{example}[Esterel Signals]
\label{examples:Esterel-signal}
Concurrent Esterel threads communicate via \textit{signals}~\cite{Berry00,Esterelv705}. A (pure, temporary) signal can have two statuses, \textit{present} and \textit{absent}. At the beginning of a synchronous macro-step, each signal is initialised to be absent, by default. The signal becomes present as soon as some thread emits it. It remains present thereafter for throughout the current macro-step, i.e., until the next clock cycle is started. The clock is a global synchronisation, taken by all threads simultaneously. The value of the signal is broadcast to all concurrent threads which can test it and branch their control-flow according to the signal's status.  To maintain coherency of data-flow and deterministic behaviour, the signal status can only be read when no writing is possible any more. Esterel compilers conduct a causality analysis on the program and obtain a suitable static schedule when they generate imperative code. In Esterel hardware compilers, the scheduling is achieved dynamically by 
 %(so-called \emph{constructive}) dual-rail coding 
propagating signal statuses.  Programs that cannot be scheduled to satisfy the emit-before-test protocol, or hardware that that does not stabilise are called \textit{non-constructive} and rejected by the compiler. A pure, temporary signal is modelled by the following recursive behaviour:

\noindent 
\begin{minipage}[b]{.6\textwidth}
 \begin{eqnarray*}
 S_0 &\pdef& 
 \emit \cseq (S_{11} \cpar S_{12}) +
 \abs\col\emit \cseq S_0 + 
 \sigma \col \{ \abs, \emit  \} \cseq S_0 \\ 
  S_{11} &\pdef& \pres \cseq S_{11} + \sigma \col \pres \cseq \zero \\
 S_{12} &\pdef& \emit \cseq S_{12} + \sigma \col \emit \cseq S_0 
% S_1 &\pdef& 
% \pres \cseq S_1 + 
% \emit \cseq S_1 + 
% \sigma \col \{ \pres, \emit \} \cseq S_0 
\end{eqnarray*}
which is depicted on the right
%in Fig.~\ref{fig:esterel-psig} 
and c-coherent according to Def.~\ref{def:coherence}. 
\end{minipage}
\qquad
\begin{minipage}[b]{.37\textwidth}
\centering \includegraphics[scale=0.65]{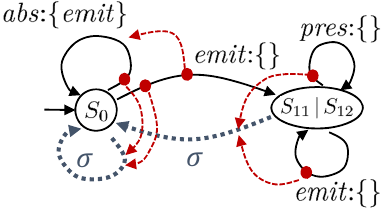}
{\small (states up to structural equivalence)}
\end{minipage}
% A typical program accessing the signal would be of form $$R \eqdef \ol{\pres} \col \ol{\pres} \cseq R_1 + \ol{\abs} \col \{ \ol{\abs}, \ol{\pres} \} \cseq R_0,$$ expressing a choice: if the signal is present we execute $R_1$, otherwise if it is absent we execute $R_0$. The two sender actions $\ol{p}$ and $\ol{a}$ are self-blocking, because these prefixes can only be consumed once. In essence, what we want is that there can be many senders of the value labels $\ol{p}$, $\ol{a}$ (readers) but at most one receiver for $p$, $a$ (signal). 
\mrev{}{%
\noindent In the original form:
 \begin{eqnarray*}
S_0 &\pdef& 
 \emit \cseq S_1 +
 \abs\col\emit \cseq S_0 + 
 \sigma \col \{ \abs, \emit  \} \cseq S_0 \\ 
 S_1 &\pdef& 
 \pres \cseq S_1 + 
 \emit \cseq S_1 + 
 \sigma \col \{ \pres, \emit \} \cseq S_0, 
\end{eqnarray*}
Esterel signals are not c-coherent: We have c-enabled and interference-free transitions 
$$
  S_1 \Derives{\pres}{\eset}{\zero} S_1
  \text{ and } 
  S_1 \Derives{\emit}{\eset}{\zero} S_1
$$
which form a diamond reconvergence. But since $\pres \neq \emit$ we would need strong residual transitions $\zero \fstep{\emit} \zero$ and $\zero \fstep{\pres} \zero$ in the concurrent environments, as per the definition of c-coherence (in the correct form as in Technical Report). Such residual transitions do not exist here, however. The problem shows up when we compose $S_1 \cpar P$ with a (c-coherent) program $P = \ol{\pres}\cseq P_1 + \ol{\emit}\col\ol{\pres}\cseq P_2$ which tests for presence and if the signal is not present then emits it. Of course, such program will not occur in practice but they depend on the race test in parallel composition. Both processes $S_1$ and $P$ are pivot processes, because they are conformant to $\pres \indep \emit$ and $\ol{\pres} \ordpre \ol{\emit}$. But we have a non-determinism unless we schedule $S_1 \cpar P$ with the race test in rule $(\ComR)$. However, the idea behind the main theorem was that pivotable processs do not need the ``ménage-a-deux'' race test at all. It seems they do. So, the notion of pivotability is too generous.
}
\qed
\end{example}
\begin{example}[\ABRO]
\label{ex:abro}
The \textit{hello word} in synchronous languages, see \cite{Berry00}, \ABRO has the following behaviour: 
%\begin{quote}
 {\it Emit the output $o$ as soon as both inputs $a$ and $b$ have been received. Reset this behaviour each time the input $r$ occurs, without emitting $o$ in the reset cycle.~\cite{Esterelv705}}
%\end{quote}
The behaviour of \ABRO is expressed as a Mealy machine and a SyncChart, displayed in Fig.~\ref{fig:ABRO} on the right. %Fig.~\ref{fig:abro-mealy}. 
The synchronous behaviour of Esterel lies in the assumption that each transition of the Mealy machine summarises a single \textit{macro-step} interaction of the machine with its environment.
The transitions of parallel machines synchronise so that state changes proceed in lock-step. The simultaneous execution of these transitions that make up a macro-step consists of the sending and receiving of signals. The propagation of individual signals between
\begin{minipage}[t]{0.6\textwidth}
 machines is modeled in the \textit{micro-step} operational semantics of Esterel. 
In \ccslm we represent the micro-steps as sequences of $\tau$-transitions, i.e., %local and asynchronous
rendez-vous synchronisations, whereas the macro-step is captured as a global clock synchronisation.  Unfortunately, the Mealy machine formalism is not scalable because the number of states can grow exponentially as the number of signals grow, while in Esterel can be expressed in a linear number of code lines as seen on the right.
\end{minipage}
$\quad$
\begin{minipage}[t]{0.35\textwidth}
{\small
\begin{verbatim}

module ABRO:
 input A, B, R;
 output O;
 loop
  await A || await B;
  emit O
 each R
end module
\end{verbatim}
}
\end{minipage}
\\[1mm]
\noindent  The
\texttt{loop P each R} construct of Esterel executes the behaviour of program \verb!P! for an unbounded number of clock cycles, but restarts its body \verb!P! when the signal \verb!R! becomes present. The reset \verb!R! takes priority over the behaviour of \verb!P! which is preempted at the moment that the reset occurs. Note that in Esterel  each signal appears exactly once, as in the specification and unlike in the Mealy machine.
Other parts of Esterel also naturally arise from precedence-based scheduling. Esterel's \verb!loop P each R! mechanism, for instance, can be coded using the static parallel composition operator that combines the process \verb!P! to be aborted and reset with a ``watchdog'' process \verb!R!. The latter waits for the reset signal from the environment in each clock cycle. The choice between the reset and continuing inside \verb!P! is resolved by giving \verb!R! higher priority. When the reset occurs, the watchdog sends ``kill'' signals to all threads running in \verb!P! and waits for these to terminate. Only then the watchdog \verb!R! restarts \verb!P! from scratch. This is a form of precedence that can be expressed in the blocking sets of \ccslm. As such, Fig.~\ref{fig:ABRO} on the left translates the Esterel code of \ABRO compositionally in \ccslm. 
\begin{figure}[t]
\begin{minipage}{0.9\textwidth}
 \includegraphics[width=.55\textwidth]{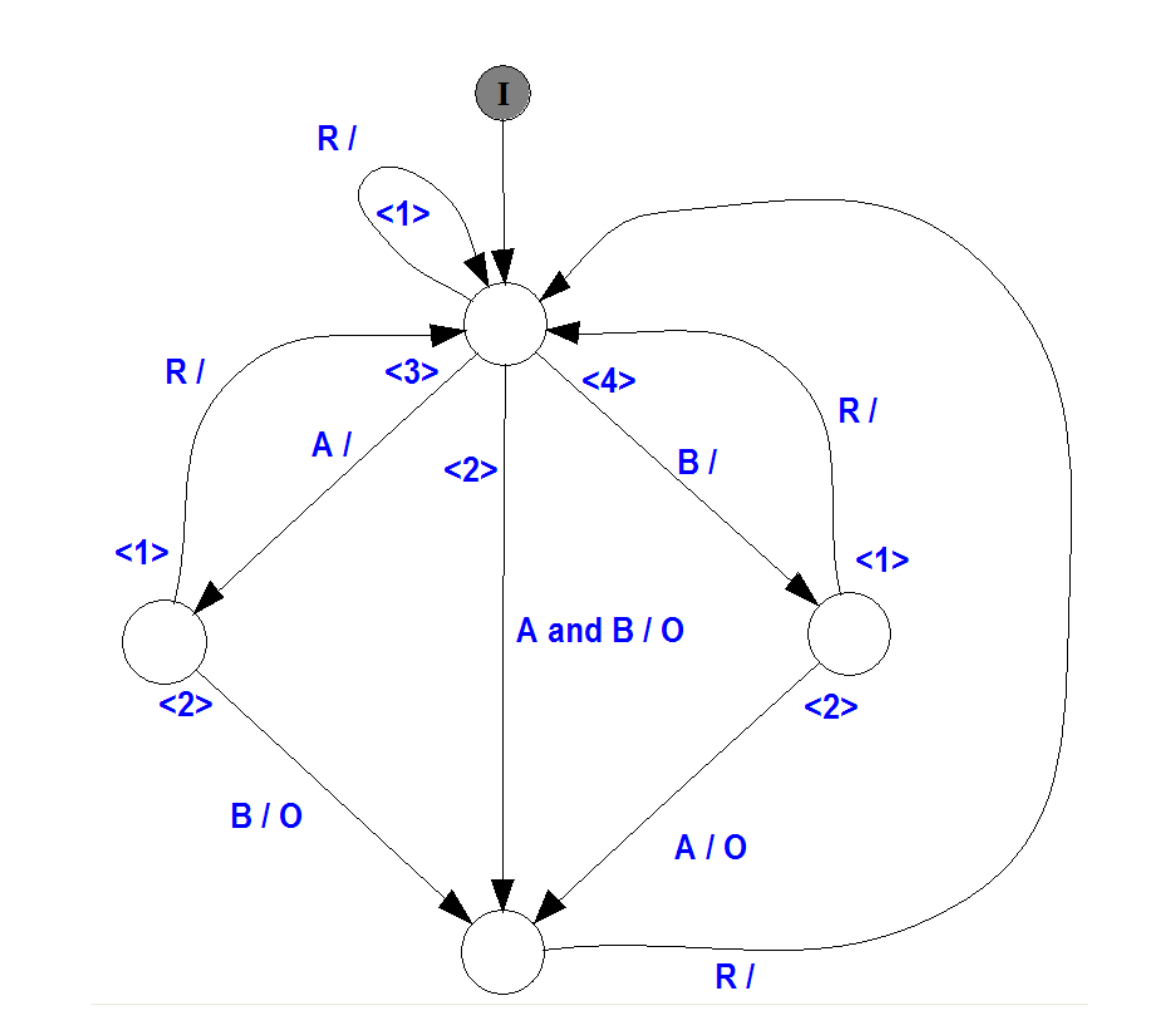}
 \hfill %\\
 \includegraphics[width=.5\textwidth]{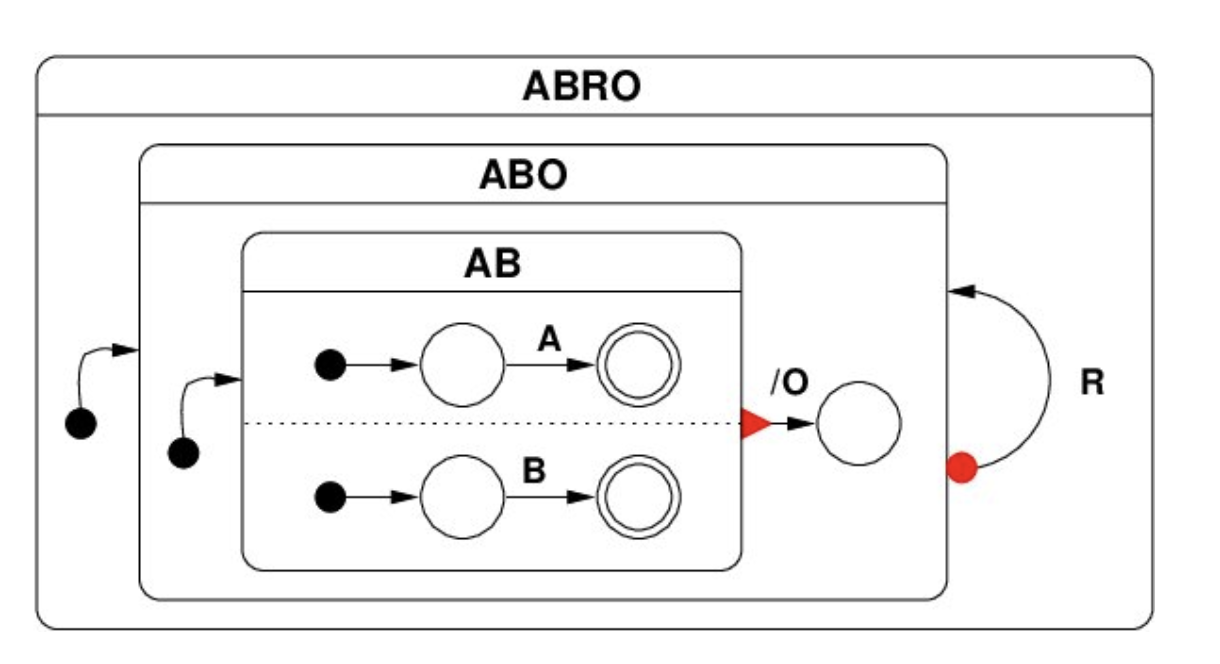}
\end{minipage}

\begin{minipage}{0.6\textwidth}
\[ \begin{array}{rcl}
 \ABRO &\pdef& 
    \sigma \cseq (\sA \cpar \sB \cpar \sR \cpar \sO \cpar \sT)  \restrict \{ s,t \} \\[2mm] 

\sA &\pdef& 
 k \col k + a \col \{k, a \} \cseq \ol{s} \col \ol{s} + 
 \tau \col \{k, a \} \cseq \sigma \cseq \sA \\[2mm]
 
 \sB &\pdef& 
 k \col k + b \col \{k, b \} \cseq \ol{s} \col \ol{s}  + 
 \tau \col \{k, b \} \cseq \sigma \cseq \sB \\[2mm]
 
 %\sC &\pdef& 
 % k_3 \col k_3 \cseq \zero + c \col \{c, k_3\} \cseq \ol{s} \col \ol{s} \cseq \zero + 
 % \sigma \col \{c, k_3\} \cseq \sC \\
 
 \sR &\pdef& 
   r \col r \cseq \sR' + 
   %\ol{k}_a \col \ol{k}_a \cseq \ol{k}_b \col \ol{k}_b %\cseq \ol{k}_3 \col \ol{k}_3 
   %\cseq \ol{k}_s \col \ol{k}_s \cseq \ol{k}_t \col \ol{k}_t \cseq \ABRO + 
   \tau \col {r} \cseq \sigma \cseq \sR \\[2mm]

 \sR' &\pdef& 
   \ol{k} \cseq \sR' + 
    \tau \col \ol{k} \cseq \ABRO \\[2mm]

 \sO &\pdef& 
   k \col k + t \col \{k, t \} \cseq \ol{o} \col \ol{o} + \tau \col \{ k, t \} \cseq \sigma \cseq \sO \\[2mm]
 
 \sT &\pdef& 
   k \col k + s \col k \cseq \sT 
    + \ol{t} \col \{k, s, \ol{t} \} +
    \tau \col \{k, s, \ol{t} \} \cseq \sigma \cseq \sT

\end{array}
\]
\end{minipage}
%\hspace{10mm}

\caption{\ABRO as a \ccslm process, Mealy machine and SyncChart. The $\tau$-prefixes $\tau \col H \cseq P$ can be obtained as $(c \col H \cpar \ol{c} \cseq P) \restrict c$ where $c$ does not occur in $P$.}
\label{fig:ABRO}
\vspace{-4mm}
\end{figure}
\ABRO is obtained as a parallel composition $(\sR \cpar \ABO) \restrict \{ s,t \}$ combining a watchdog $\sR$ with the behaviour $\ABO \eqdef \sA \cpar \sB \cpar\sO \cpar \T$. The $\sR$ is responsible to implement the reset loop at the superstate called $\ABO$, as in the state charts.
%Fig.~\ref{fig:abro-synccharts}. 
The $\ABO$ is the behaviour of the interior of this superstate, which corresponds to the statement \verb!(await A || await B); emit O! in Esterel:
% of Fig.~\ref{fig:abro-esterel}. 
it is the parallel composition of $\sA \cpar \sB$ making up the behaviour of the region called $\AB$ with code \verb!await A || await B!. The processes $\sO \cpar \sT$ implement the synchronisation that waits for termination of $\AB$ to start the output process $\sO$. 
We assume that the channels $a$, $b$, $r$, $o$ modelling interface signals \verb!A!, \verb!B!, \verb!R!, \verb!O! 
%of Fig.~\ref{fig:abro} 
all have a unique sender and receiver. 
This is seen in our coding from the fact that the prefixes for these channels are self-blocking. 
%
%Consider the process \ABRO from Ex.~\ref{ex:abro} again:
\begin{itemize}

\item The $\sA$ and $\sB$ processes are analogous to each other with the role of $a$ and $b$ swapped. Looking at $\sA$, we see that it has a choice of three actions, viz to receive a kill signal $k$, an input signal $a$ or a silent action $\tau \col \{ k, a \}$, which preempts the other prefixes $k \col k$ and $\ol a \col \{ k, a \}$ to perform the clock action $\sigma$ and start the next macro step. The kill $k$ has highest priority, action $a$ second and the clock can only be activated if neither $k$ nor $a$ can synchronise. When there is no kill but an $a$ signal, then the process $\sA$ sends an $\ol{s}$ signal to indicate its termination. When the kill $k$ occurs, it enters into the halt state $\zero$ where it is inactive but permits any number of clock ticks. Since this is the neutral element for parallel composition, $\sA$ essentially drops out of the picture. Note that we cannot use the inactive $\zero$ here, because then the terminated process $\sA$ would block the clock for all other threads. Finally, when there is neither a kill $k$ nor signal $a$, the $\sA$ process can execute the silent $\tau \col \{ k, a \}$ and then synchronise on $\sigma$ and repeat as $\sA$ in the next macro-step. In this fashion, $\sA$ will patiently wait for a clock tick in which $k$ or $a$ occur. If both occur, the kill takes priority: For instance, the transition $\sA \cpar \ol{k} \cpar \ol{a} \Derives{\tau}{\{k\}}{\ol{a}} \ol{a}$
is constructively enabled while 
$\sA \cpar \ol{k} \cpar \ol{a} \Derives{\tau}{\{k, a \}}{\ol{k}} \ol{k}$ is not enabled. 
% $\sA \cpar \ol{k} \cseq \zero \cpar \ol{a} \cseq \zero \Derives{\tau}{\{k\}}{\zero\cpar\zero\cpar \ol{a}\cseq\zero} \zero \cpar \zero \cpar \ol{a} \cseq \zero$
%  is constructively enabled while 
%  $\sA \cpar \ol{k} \cseq \zero \cpar \ol{a} \cseq \zero \Derives{\tau}{\{k, a \}}{\zero \cpar \ol{k}_1 \cseq \zero \cpar \zero} \zero \cpar \ol{k}_1 \cseq \zero \cpar \zero$ 
%  is not enabled. 

\item The reset watchdog $\sR$ can receive a reset signal $r$ or silently move to state $\sigma \cseq R$ where it synchronises on the clock $\sigma$ to repeat in the next macro-step. When a reset $r$ is received by the environment, then $\sR$ offers the infinite repetition $\ol{k} \cseq \sR'$,
 % \begin{eqnarray} 
 %$ r \col r \cseq \ol{k}_1 \col \ol{k}_1 \cseq \ol{k}_2 \col \ol{k}_2 
 %\cseq \ol{k}_3 \col \ol{k}_3 \cseq \ol{k}_4 \col \ol{k}_4 \cseq \ABRO
 %\label{eqn:reset-prefix}
 %\end{eqnarray} 
 whereby it sends a ``kill'' signal $\ol{k}$ to each of the four threads $\sA$, $\sB$, $\sO$ and $\sT$, in some arbitrary order. Only when all kills have been delivered, i.e., no more consumption on $\ol{k}$ is possible, then the watchdog $\sR$ can silently preempt the kill loop and restart $\ABRO$. 
 %Until such time, there is no clock possible. 
 Note how the second transition 
 %$R \fsstep{\tau \col r[\zero]} \sigma \cseq \sR$ 
 $\sR \Derives{\tau}{\{r\}}{\zero} \sigma \cseq \sR$ of $\sR$
 preempts the reset choice
 %~\eqref{eqn:reset-prefix} 
 but with lower priority as it is blocked by $r$. It is only taken when there is no reset possible, i.e., when the reset signal $r$ is absent. Note that with this preemption of the prefix $r \ol r \cseq \sR'$, 
 %prefix~\eqref{eqn:reset-prefix} 
 the kill signal $\ol{k}$ is removed from the potential, because $\olwilA{\ast}(\sigma \cseq R) = \eset$, since we do not look past the clock prefix. 
 
 \item The process $\sT$ is the termination detector implementing the principle discussed above. It forms a clock-patient, killable counter that waits for any possible reception of signal $s$ before it sends termination $\ol{t}$ to process $\sO$. Note that $\sT$ itself does not need to know how many possible termination signals $s$ there are. It waits for the termination count (implemented by $\wilA{\ast}$) to reach zero. Only then the prefix $ \ol{t} \col \{k, s, \ol{t} \} $ can fire and inform $\sO$. If there is not kill $k$ or temrinatikon $s$, then $\sT$ can silently preempt the summation and reach the state $\sigma \cseq \sT$ where it waits for the clock to repeat.
 
 \item When process $\sO$ receives the termination signal $t$ it sends the output signal $\ol{o}$ to the environment and then halts. Alternatively, with highest priority, it can be killed by reception of $k$ or engage in the clock $\sigma$ if neither a kill $k$ nor a termination $t$ is present. 
 Note that the reduction $\tau = \ol{t} \col \{k, s, \ol{t} \} \cpar t \col \{ k, t \}$ between $\sT$ and $\sO$ will be blocked by precedence, for as long as at least one of the processes $\sA$ or $\sB$ has a choice term that sends $\ol{s}$, i.e., as long as $s \in \olwilA{\ast}(\sA \cpar \sB)$. Only when both prefix sequences $a \col \{k, a \} \cseq \ol{s} \col \ol{s} \cseq \zero$ and $b \col \{k, b \} \cseq \ol{s} \col \ol{s} \cseq \zero$ have been eliminated by the $\tau$-reductions, then the termination can occur.  
 Note that we can eliminate $\sT$ and integrate the termination test into the output process $\sO$ if we let $\sO$ loop on consumption of $s$ rather than $\sT$. 
\end{itemize}
%
%Using multi-cast, we can even improve on this. 
% as seen in Fig.~\ref{fig:simplified-ABRO}:
%We can integrate the termination test into the output process $\sO$ if we let $\sO$ loop on consumption of $s$ rather than $\sT$. 
%Similarly, we can use only a single (multi-cast) kill signal $k$ shared between $\sA$, $\sB$, $\sO$: The reset process $\sR$ keeps sending $\ol{k}$ unboundedly often, until all target processes have received a copy and all receptions $k$ have been consumed. At that point, it takes a clock step to repeat $\ABRO$ from the start.  
\qed
\end{example}

\longshort{\version}{}{
All the above examples, except \ABRO 
have the special property that 
%are what we call \textit{input-scheduled}, because 
precedences only exist between receiver actions $\A$, while sender actions $\coA$ never appear in any choice contexts. These are the processes covered by~\cite{CamilleriWin95} which restrict prioritised choice to receiver actions (``input-guarded, prioritised choice'') in the same way as %the programming language 
Occam~\cite{Barrett89} with its ``PRIALT'' construct does.
%
%}
%
\longshort{\version}{}{What can be shown, however, is that the parallel composition of \textit{two} \emph{single-threaded} processes is always deadlock-free, even under strong enabling: this could be proved formally, although it would be a non trivial exercise of style. Strong enabling produces deterministic behaviour for PRIALT processes, because these have unique sender-receiver connections. A receiver $\ell$ can match and synchronise with exactly one sender $\ol{\ell}$. If we want to model sharing, where one (repetitive) sender can service multiple receivers, and vice versa, the notion of enabling must be strengthened further.}
This restriction is an example of a class of more general precedence schemes, that we call \textit{pivot policies} (Def.~\ref{def:pivot-policy}).
}

%%% Local Variables:
%%% mode: latex
%%% TeX-master: "synpatick-lipics.tex"
%%% End:
   % clean

%!TEX root = synpatick-lipics.tex

%\vspace{-6mm}
%---------------------------------------
\section{Coherence and Determinacy for \ccslm}
\label{sec:properties}
%---------------------------------------
%\vspace{-2mm}

This section presents the core elements 
%contains the mathematical developments needed to prove the main properties 
of our scheduling theory %for \ccslm 
under constructive enabling: coherence, policies, and policy-conformant processes. 
\textit{Coherence} (Def.~\ref{def:coherence}) corresponds to Milner's notion of confluence, but refined ``up-to'' priorities. It requires structural confluence for c-enabled and diverging c-actions only if they do not interfere with each other considering their blocking sets and concurrent environments. 
%Specifically, interference-free c-actions must not block each other. 
\longshort{\version}{}{This ensures that they are not part of a conflicting choice in which each preempts the other.}

% %\longshort{\version}{
% \begin{definition}[Transition Interference]
%   Two c-enabled transitions with c-actions $\alpha_1 \col H_1[E_1]$ and $\alpha_2 \col H_2[E_2]$ 
%   are \emph{interference-free} if $\alpha_1 = \alpha_2$ or  $\alpha_i \not\in H_{3-i}$ for all $i \in \{ 1,2 \}$.
% \label{def:interference-free}
% \end{definition}
% %

%-----------------------------------------
\begin{definition}[Transition Non-Interference]
  \mrev{Two transitions $Q \Derives{\alpha_1}{H_1}{E_1} Q_1$ and $Q \Derives{\alpha_2}{H_2}{E_2} Q_2$ %with $\alpha_i \in \Act$ 
  are called \emph{interference-free} if the following holds for all $i \in \{ 1,2 \}$:}{re-establish original definition}
  \begin{bracketenumerate}
     %\item $\alpha_1 \neq \alpha_2$ or $\{ \alpha_1, \alpha_2 \} \cap \C = \eset$ 
     \item If $\alpha_1 \neq \alpha_2$ then $\alpha_i \not\in H_{3-i}$.
     \item \mrev{If $\alpha_{3-i} = \tau$ then $H_{i} \cap (\olwilA{\ast}(E_i) \cup \{ \tau \}) = \eset$.}{added, then c-enabledness in c-coherence not needed, and we align again with proofs.}
     %$H_{i} \cap (\olwiA(E_i) \cup \{ \tau \}) = \eset$.
     \item \mrev{$\alpha_1 \neq \alpha_2$ or $Q_1 \not\equiv Q_2$ or both $\{ \alpha_1, \alpha_2 \} \subseteq \R$ and $\alpha_i \not\in H_{3-i}$.}{moved here to simplify statement of c-coherence below.}
  \end{bracketenumerate}
\label{def:interference-free}
\end{definition}
%-----------------------------------------

\mrev{The first condition (1) of Definition~\ref{def:interference-free} says that distinct c-actions must not block each other. This ensures that they are not part of a conflicting choice in which each preempts the other. 
%This is always the case for silent actions, since they cannot occur in any blocking set. 
By the second condition, non-interference with a silent action requires that the action is actually enabled, i.e., it does not get blocked by computations in its own reduction context. 
%However, the visible actions from which silent actions arise by synchronisation may block an action $\alpha_i$.
%Confluence of a c-action $\alpha_i \col H_i[E_i]$ across a silent action $\alpha_{3-i} = \tau$ is only guaranteed if the c-action $\alpha_i \col H_i[E_i]$ is enabled, i.e., it does not get blocked by computations in its own reduction context.   
Note that two c-actions $\ell \col H_1[E_1]$ and $\ell \col H_2[E_2]$ for the same visible label $\ell \in \R$ are always considered to be interference-free, regardless $H_i$ and $E_i$.}{added original explanations}

\medskip

Observe that for free processes (i.e., with empty blocking sets) all c-actions are interference-free and they are c-enabled iff they are admissible.
\longshort{\version}{}{Note that two c-actions $\ell \col H_1[E_1]$ and $\ell \col H_2[E_2]$ for the same visible label $\ell \in \R$ are always considered to be interference-free, regardless $H_i$ and $E_i$.}
Confluence of interference-free c-actions is formulated using an auxiliary form of \textit{residual transition} $P \astep{\alpha} Q$. This is a transition from $P$ to $Q$ with an action that is a ``factor'' of $\alpha$. Formally, we define $P \astep{\alpha} Q$ if %one of the following two cases holds: 
\longshort
{\version}
{$(i)$ $\alpha \in \L$ with $P \Derives{\alpha}{}{} Q$ or $P \scong Q$, or $(ii)$ $\alpha = \tau$ with $P \Derives{\tau}{}{} Q$ or \mrev{$P \xrightarrow{\ell}{}{} Q$ for some $\ell \in \R$.}{simplified, by removing redundancy.}}
{\begin{enumerate}
  \item If $\alpha \in \L$ then $P \xrightarrow{\alpha}{}{} Q$ or $P \scong Q$
  \item If $\alpha = \tau$ then 
  $P \xrightarrow{\tau}{}{} Q$,  
  \mrev{$P \xrightarrow{\ell}{}{} Q$ for some $\ell \in \R$.}{Just a simplification}
\end{enumerate}
}
Note that for visible labels $\ell \in \L$ a residual transition %$P \astep{\ell} Q$ 
does not require any change of state, i.e., we can have $P \scong Q$.  

\begin{definition}[Process Coherence]
\longshort{\version}{}{[Coherence]}
 A process $P$ is \mrev{\emph{c-coherent}}{change terminology for forward compatibility. This notion of ``c-coherence'' here is a restricted class that we are going to replace by a larger class, that we call ``coherent'' in our TCS paper. We do not want to confuse the two. The notion of c-coherence here does not take the local initial choices of each thread into account and works for an SOS without the race condition in rule $(\ComR)$.} if for all its derivatives $Q$ the following holds:
 %, where $\alpha_1, \alpha_2 \in \Act$: 
 For any two \longshort{\version}{interference-free}{} transitions
 \[
 Q \Derives{\alpha_1}{H_1}{E_1} Q_1
 \text{ and }
 Q \Derives{\alpha_2}{H_2}{E_2} Q_2
 %\] 
 %
 % such that $\alpha_1 \neq \alpha_2$ or $Q_1 \not\equiv Q_2$ or both $\{ \alpha_1, \alpha_2 \} \subseteq \R$ and $\alpha_i \not\in H_{3-i}$.
 %
 % \longshort{\version}{}{Further, suppose the c-actions $\alpha_i \col H_i[E_i]$ are interference-free.}
 \text{ there exist \longshort{\version}{transitions}{$H_i'$ and processes $E_i'$, $Q'$ such that} } 
 %\[ 
 Q_1 \Derives{\alpha_{2}}{H_{2}'}{E_{2}'} Q'
 \text{ and }
 Q_2 \Derives{\alpha_{1}}{H_{1}'}{E_{1}'} Q'
 \] 
 with $H_{i}'\subseteq H_{i}$ and
 $E_1 \astep{\alpha_{2}} E_1'$ and $E_2 \astep{\alpha_{1}} E_2'$. 
 \longshort{\version}{}{Moreover, 
 $ % \[
 E_1 \Derives{\alpha_{2}}{}{} E_1' 
 \text{ and } 
 E_2 \Derives{\alpha_{1}}{}{} E_2'
 $ % \]
 if $\{ \alpha_1, \alpha_2 \} \subseteq \R$ and $\alpha_1 \neq \alpha_2$ or $Q_1 \not\equiv Q_2$, or $\{ \alpha_1, \alpha_2 \} \cap \C \neq \eset$.} 
 \mrev{Further, if $\{\alpha_1, \alpha_2\} \subseteq \R$ and one of $\alpha_1 \neq \alpha_2$ or $Q_1 \not\equiv Q_2$ or $\{\alpha_1, \alpha_2\} \cap \C \neq \eset$ then the residuals must be strong, i.e., $E_1 \Derives{\alpha_{2}}{}{} E_1'$ and $E_2 \Derives{\alpha_{1}}{}{} E_2'$.}{This was in the technical report but got dropped by mistake.}
\label{def:coherence}
\end{definition}
\mrev{A consequence of Def.~\ref{def:coherence} is that for any derivative $Q$ of a c-coherent process, if $Q \Derives{\ell_1}{R_1}{H_1} Q_1$ and $Q \Derives{\ell_2}{R_2}{H_2} Q_2$ with $\ell_1 \neq \ell_2$: if $\ell_i \not\in H_{3-i}$ for all $i$, then $\ell_i \in R_{3-i}$ for all $i$.}{This is another way of expressing the residual property. This was not a property of our original coding of Esterel signals. With the above change, this is fixed.}

\medskip 

For free processes, c-coherence implies confluence (Def.~\ref{def:Church-Rosser}), but is strictly stronger. 

\begin{example} 
  For instance, $Q \eqdef  a\col\{\,\} \cseq \zero$ %for $a \in \R$ 
  is confluent but not c-coherent. Since the prefix not self-blocking, c-coherence ensures  that the label must \emph{consumable} by arbitrarily many concurrent processes. Technically, since $Q \xrightarrow{a} Q_1$ and $Q \xrightarrow{a} Q_2$ for 
  $Q_1 \equiv %= 
  \zero \equiv % = 
  Q_2$, there should be a transition $Q_i \xrightarrow{a} Q'$ which is not the case. However, for discrete processes, c-coherence is identical with confluence: With self-blocking, $Q \eqdef a\col\{a\} \cseq \zero$ is both confluent and c-coherent. It is important to note that all examples in Sec.~\ref{examples:scheduling} are c-coherent but not confluent. 
\qed
\end{example}

\begin{theorem}[Coherence implies Determinacy] 
 Every c-coherent process is determinate under 
 c-enabled reductions. 
\label{thm:church-rosser}
\end{theorem}
\longshort{\version}{}
{\begin{proof}[Of Thm.~\ref{thm:church-rosser}]
 Let $P$ be c-coherent and $Q$ a derivative with c-enabled reductions 
$ % \[
Q \Derives{\tau}{H_1}{R_1} Q_1 \text{ and }$  $Q \Derives{\tau}{H_2}{R_2} Q_2.
$ %\]
If $Q_1 \equiv Q_2$ we are done immediately. Suppose $Q_1 \not\equiv Q_2$. Then c-enabling implies $\tau \not \in H_i$ as well as $H_i \cap \olwilA{\ast}(R_i) = \eset$, which implies $H_i \cap (\olwilA{\ast}(R_i) \cup \{ \tau \}) = \eset$. But then the c-actions $\tau \col H_1[R_1]$ and $\tau \col H_2[R_2]$ are interference-free, whence c-coherence Def.~\ref{def:coherence}(2) implies there exist $H_i'$ and processes $R_i'$, $Q'$ such that 
$ % \[ 
 Q_1 \Derives{\tau}{H_{2}'}{R_{2}'} Q'
 \text{ and }$ 
\mbox{$Q_2 \Derives{\tau}{H_{1}'}{R_{1}'} Q'$}
$\text{and }
 R_1 \astep{\tau} R_1' 
 \text{ and } 
 R_2 \astep{\tau} R_2'
 $ %\]
 with $H_{i}'\subseteq H_{i}$. 
 Finally, observe that $\olwilA{\ast}(R_i') \subseteq \olwilA{\ast}(R_i)$ by Lem.~\ref{lem:c-enabled-aux} from which we infer 
$ %\[
H_i' \cap (\olwilA{\ast}(R_i') \cup \{ \tau \}) \subseteq 
 H_i \cap (\olwilA{\ast}(R_i) \cup \{ \tau \}).
 $ %\]
Thus, the reconverging reductions 
%$ Q_i \fsstep{\tau \col H_{3-i}'[R_{3-i}']} Q'$ 
$Q_i \Derives{\tau}{H_{3-i}'}{R_{3-i}'} Q'$
are again c-enabled. This was to be shown.
\end{proof}
}

A c-coherent process cannot simultaneously offer both a clock and another distinct action without putting them in precedence order. Moreover, whenever a clock is enabled, then every further reduction of the process is blocked by the clock.
\longshort{\version}{}{%
Transitions
$ %\[
P \Derives{\ell_1}{H_1}{R_1} P_1 \text{ and } P \Derives{\ell_2}{H_2}{R_2} P_2
$ %\]
with distinct visible labels $\ell_1 \neq \ell_2$ must stem from distinct action prefixes in $P$ for syntactic reasons. Structural c-coherence requires that the two transitions are either interfering, $\ell_1 \in H_1$ or $\ell_2 \in H_1$, and thus will disable each other, or they are offered in concurrent threads, $\ell_1 \in \iA(R_{2})$ and $\ell_2 \in \iA(R_{1})$ and thus are concurrently independent. 
}
\longshort{\version}
{\noindent Thm.~\ref{thm:church-rosser} concerns the $\tau$-transitions of a c-coherent process. For clock transitions we have an even stronger result.}
{We can show that c-coherent processes are structurally deterministic for clock transitions. This is a consequence of a weaker form of structural determinism, that we call \textit{action determinism}. To be precise, %by Def.~\ref{def:coherence}, 
a c-coherent process does
not offer a choice on the same visible label to structurally different successor states unless these actions also occur in concurrent threads.}
\longshort{\version}{}{Since clock actions are global, they are not executed concurrently with other threads. As a consequence, a clock cannot occur in a choice context with itself.}

\begin{proposition}[Clock Determinism]
 \longshort{\version}
 {Every c-coherent process}
 {Let $P$ be c-coherent with 
 $% \[
 P \Derives{\ell}{H_1}{R_1} P_1
 $ % \]
 such that $\ell \in \L$ and $\ell \not\in \iA(R_1)$. Then, for every transition 
 $ % \[
 P \Derives{\ell}{H_2}{R_2} P_2
 $ %\]
 we have $P_1 \equiv P_2$.  
 In particular,}
 $P$ is structurally clock deterministic, i.e., if 
 $ %\[
 P \xrightarrow{\sigma} P_1 \text{ and } P \xrightarrow{\sigma} P_2
 $ %\]
 then $P_1 \equiv P_2$. 
\label{prop:action-determinism}
\end{proposition}
\longshort{\version}{}
{\begin{proof}[Of Prop. \ref{prop:action-determinism}]
 The first part of the statement follows directly from Def.~\ref{def:coherence}(2): The c-actions $\ell \col H_i[R_i]$ are trivially interference-free. Now if $P_1 \nequiv P_2$, then \myb{c-coherence} gives us a strong environment shifts implying $\ell \in \iA(R_1)$. 
 The second part is because by Lem.~\ref{lem:prio-basic-1}(1) we have $R_i = \zero$ and thus $\iA(R_i) = \eset$ when $\ell \in C$.
\end{proof}
}

The question is now how we can obtain c-coherent processes by composition. In view of Milner's classical result for confluence, we expect to have to impose some form of ``sort-separation'' for parallel processes. 
It turns out that we can express such a condition and even avoid the restriction operator (thereby permit sharing) if we enrich the sorts $\L(P)$, $\L(Q)$ by priority information. We call the resulting enriched sorts $\pi(P)$, $\pi(Q)$ \textit{precedence policies}. 

\longshort{\version}{}{%
The problem of causality cycles that occurs in the parallel composition of two single-threaded processes stems from the symmetry between input and output actions (see, e.g., Ex.~\ref{ex:not-pivotable}). The standard trick to eliminate the problem as used e.g., Occam's ``PRIALT'', is to introduce a \textit{causality order} so that outputs (sender actions) are deterministic by construction. All remaining choices are choices between inputs and these are resolved by the environment selecting the matching output.
To express such causality properties of processes by static interfaces we refine the notion of a \textit{sort} known from \ccs to include precedences in \ccslm.
}

\begin{definition}[Precedence Policy]\mbox{}
  A \emph{precedence policy} is a pair $\pi = (L, \ordpre)$ where $L \subseteq \L$ is a subset of visible actions, called the \emph{alphabet} of $\pi$ and $\ordpre \subseteq L \times L$ a binary relation on $L$, called the \emph{precedence relation}. 
\label{def:precedence-policy} 
\end{definition}

We will write $\ell \in \pi$ to state that $\ell$ is in the alphabet of $\pi$ and we write $\ell_1 \ordpre \ell_2 \in \pi$ to express that the pair $(\ell_1, \ell_2)$ is 
in the precedence relation of the policy.
Further, if $\ell_1, \ell_2 \in \pi$ and both $\ell_1 \ordpre \ell_2 \not\in \pi$ as well as $\ell_2 \ordpre \ell_1 \not \in \pi$, we say that $\ell_1$ and $\ell_2$ are \emph{concurrently independent} under $\pi$, written (by abuse of the set-notation) $\ell_1 \indep{} \ell_2 \in \pi$.

\medskip 

For any policy $\pi = (L, \ordpre)$ and set of labels $K \subseteq \L$ we denote by $\pi \restrict K$ the policy $\pi$ restricted to the alphabet $L - (K \cup \ol{K})$ in the standard way. 
There is also the normal inclusion ordering $\pi_1 \subseteq \pi_2$ between policies defined by $L_1 \subseteq L_2$ and $\sordpre 1 \subseteq \sordpre 2$.

\longshort{\version}{}{%
We obtain a partial ordering $\pi_1 \preceq \pi_2$ on precedence policies by subset inclusion. Specifically, we have $(L_1, \sordpre{1}) \preceq (L_2, \sordpre 2)$ if $L_1 \subseteq L_2$ 
such that for all $\ell_1, \ell_2 \in L_1$, 
if $\ell_1 \sordpre 2 \ell_2$ then $\ell_1 \sordpre 1 \ell_2$. 
Intuitively, if $\pi_1 \preceq \pi_2$ then $\pi_2$ is a tightening of $\pi_1$ in the sense that it exports more labels (resources) subject to possibly fewer precedences (causality constraints) than $\pi_1$.
\longshort{\version}{}{There is also the normal inclusion ordering $\pi_1 \subseteq \pi_2$ between policies defined by $L_1 \subseteq L_2$ and $\sordpre 1 \subseteq \sordpre 2$. It differs from $\preceq$ in the inclusion direction of the precedence alphabets. 
Another notation we need is restriction.}
}

\medskip 

The next Def.~\ref{def:conformance} on conformance captures how policies statically specify the externally observable scheduling behaviour of a process. 
 
\begin{definition}[Conformance] 
\longshort{\version}{}{[Conformance]}
 $P$ \emph{conforms} to a policy $\pi$ if for each derivative $Q$ of $P$: If 
 %$Q \fsstep{\ell \col H[R]} Q'$ 
 $Q \Derives{\ell}{H}{R} Q'$  
 with $\ell \in \L$, then $\ell \in \pi$ and for all $\ell' \in H \cap \L$, then
 $\ell' \ordpre \ell \in \pi$.
\label{def:conformance}
\end{definition}
For a process $P$ and policy $\pi$, let us write\footnote{In Milner's notation we would write $P : \pi$, but we already use the colon `:' for blocking sets as in \ccsp.} $P \cnf \pi$ to express that $P$ is c-coherent and conforms to policy $\pi$. We will say that $P$ is \textit{c-coherent for} $\pi$.
If $P \cnf \pi$ then the alphabet of $\pi$ corresponds to the sort $\L(P)$ in \ccs. It gives an upper bound on the labels that can be used by the process, and the precedence relation in $\pi$ provides a static over-approximation of the blocking, i.e., which labels can be in a choice conflict with each other. As a notational shortcut we will write $P \cnf \ell_1 \ordpre \ell_2$ to say that there is a policy $\pi$ with $P \cnf \pi$ and $\ell_1 \ordpre \ell_2 \in \pi$. Analogously, we write $P \cnf \ell_1 \indep \ell_2$ if there is a policy $\pi$ with $\ell_1 \indep \ell_2 \in \pi$ and $P \cnf \pi$. 

\begin{example}
  The most permissive policy for a given sort $L \subseteq \L$, written $\pi_{max}(L)$, contains all labels from $L$ as its alphabet but no precedences, i.e., $\ell_1 \indep \ell_2 \in \pi_{max}(L)$ for all $\ell_1, \ell_2 \in L$. If $P \cnf \pi_{max}$ then $P$ has sort $L$ and will permit arbitrarily many concurrent processes to consume its labels $L$ under c-enabling. 
  The most restrictive policy, denoted $\pi_{min}(L)$, also has alphabet $L$, but full precedences, i.e., $\ell_1 \ordpre \ell_2 \in \pi_{min}$ for all $\ell_1, \ell_2 \in L$. If $P \cnf \pi_{min}(L)$, then $P$ has sort $L$ but does not permit any concurrency; it can only communicate with a single sequential thread.
\qed
\end{example}

\begin{example}
  Reflexive precedences $\ell \ordpre \ell \in \pi$ play a special role in our theory. They permit a conformant $P \cnf \pi$ and its derivatives to execute $\ell$ only with a single partner. For instance, $\ABRO$ from Ex.~\ref{ex:abro} satisfies $\ABRO \ncnf r \indep r$ and so $\ABRO  \cpar \sigma\cseq \ol{r}$ reduces but $\ABRO  \cpar \sigma \cseq \ol{r} \cpar \sigma \cseq \ol{r}$ will block after the initial clock transition. 
  In contrast, if the policy specifies $\ell \indep{} \ell \in \pi$, then in a process $P \cnf \pi$ any $\ell$-transition  is consumable arbitrarily often. For instance, the Esterel signal from Ex.~\ref{examples:Esterel-signal} has $S_1 \cnf \pres \indep \pres$ and so $S_1 \cpar \ol{pres} \cpar \ol{\pres}$ will not block. 
\qed
\end{example}

\longshort{\version}{}{%
By definition, notion of conformance is closed under taking transitions. Specifically, if $P$ conforms to $\pi$ and $P \Derives{\alpha}{}{} Q$, then $Q$ conforms to $\pi$, too. Conformance is not closed under taking sub-processes, however. 
For instance, $(a + b \col \{ a,b \}) \restrict a$ conforms to a policy $\pi = (L, \ordpre)$ with $L = \{b\}$ and $\pi \Vdash b \ordpre b$. The policy does not need to mention $a$ because it is a bound local signal that is restricted. The (proper) sub-process $a + b \col \{a, b\}$, however, does not conform to $\pi$, because $L$ does not permit the local action $a$.
Since conformance is not closed under sub-expressions, we could assume a system of typing rules that associates policies with sub-processes. This would be a topic of further work.
}
 
Here we are interested in
processes that conform to policies with special properties. 

\longshort{\version}{%
\begin{definition}[Pivot and precedence-closed policy]
%
%[Pivotable, Precedence-closed, Input-scheduled Processes] \mbox{}
  \mbox{} A policy $\pi$ with alphabet $L$ is called
  \begin{itemize} 
   \item  \emph{pivotable} if\longshort{\version}{ $\ol{L} \subseteq L$ and for all $\ell_1, \ell_2 \in L$ with $\ell_1 \neq \ell_2$ we have $\ell_1 \indep{} \ell_2 \in \pi$ or $\ol{\ell}_1 \indep{} \ol{\ell}_2 \in \pi$.}
  {$\ol{\pi} \preceq \pi$.} 
   \item \emph{precedence-closed} for $K \subseteq \L$ if $\ell_1 \in K$ and $\ell_1 \ordpre \ell_2 \in \pi$ implies $\ell_2 \in K$. 
  \end{itemize} 
  \longshort{\version}{A process is \emph{pivotable}/\emph{precedence-closed} if $P$ conforms to pivot/precedence-closed policy.}{A process is \emph{pivotable}/\emph{input-scheduled}/\emph{precedence-closed} if $P$ conforms to pivot/input-scheduled/precedence-closed policy.}
\label{def:special-policy}
\end{definition}
}
{%
\begin{definition}[Pivot Policy]
 A policy $\pi$ is a \emph{pivot} policy if $\ol{\pi} \preceq \pi$. 
 A process $P$ is \emph{pivotable} if $P$ conforms to a pivot policy. 
\label{def:pivot-policy}
\end{definition}
\begin{proposition}[Pivot policy]
 A policy $\pi = (L, \ordpre)$ is a pivot policy if 
 $\ol{L} \subseteq L$ and for all $\ell_1, \ell_2 \in L$ with $\ell_1 \neq \ell_2$ we have $\ell_1 \indep{} \ell_2 \in \pi$ or $\ol{\ell}_1 \indep{} \ol{\ell}_2 \pi$.
\label{prop:pivot-policy}
\end{proposition}
}
\longshort{\version}{}
{\begin{proof}[Of Prop.~\ref{prop:pivot-policy}]
Let $\pi$ be a pivotable process. Since the alphabet of $\ol{\pi}$ is $\ol{L}$ and the alphabet of $\pi$ is $L$ the assumption $\ol{\pi} \preceq \pi$ implies $\ol{L} \subseteq L$ by definition of $\preceq$. Given labels $\ell_1, \ell_2 \in L$ with $\ell_1 \neq \ell_2$ suppose $\pi \nVdash \ell_1 \indep{} \ell_2$. Hence $\pi \Vdash \ell_1 \ordpre \ell_2$ or $\pi \Vdash \ell_2 \ordpre \ell_2$. Consider the first case, i.e., $\pi \Vdash \ell_1 \ordpre \ell_2$. Since $\ol{L} \subseteq L$ and $L = \ol{\ol{L}}$ it follows that also $L \subseteq \ol{L}$, i.e., the labels $\ell_i$ are also in the alphabet of $\ol{\pi}$. But then by definition of the ordering and the fact that $\ell_i = \ol{\ol{\ell}}_i$, the assumption $\ol{\pi} \preceq \pi$ implies that $\ol{\pi} \Vdash \ol{\ol{\ell}}_1 \ordpre \ol{\ol{\ell}}_2$. Now, the definition of the dual policy means that $\pi \nVdash \ol{\ell}_1 \ordpre \ol{\ell}_2$ and $\pi \nVdash \ol{\ell}_2 \ordpre \ol{\ell}_1$. This is the same as $\pi \Vdash \ol{\ell}_1 \indep{} \ol{\ell}_2$ as desired.

Suppose the properties (1) and (2) of the proposition hold for a policy $\pi$. We claim that $\ol{\pi} \preceq \pi$. The inclusion of alphabets is by assumption (1). Then, suppose $\pi \Vdash \ol{\ell}_1 \ordpre \ol{\ell}_2$ which implies 
 $\pi \nVdash \ol{\ell}_1 \indep{} \ol{\ell}_2$. If $\ell_1 = \ell_2$ we have $\ol{\pi} \Vdash \ol{\ell}_1 \ordpre \ol{\ell}_2$ directly by definition. So, let $\ell_1 \neq \ell_2$. 
 Because of $\ol{L} \subseteq L$ and (2) this gives us $\pi \Vdash \ell_1 \indep{} \ell_2$ considering again that $\ell_i = \ol{\ol{\ell}}_i$. 
 But $\pi \Vdash \ell_1 \indep{} \ell_2$ is the same as $\ol{\pi} \Vdash \ol{\ell}_1 \ordpre \ol{\ell}_2$ by definition. 
\end{proof}
}

\begin{example} 
  An Esterel signal (Ex.~\ref{examples:Esterel-signal}) has the alphabet $\{ \emit, \pres, \abs \} = \L(S_i)$ ($i\in \{1,2\}$ and for simplicity we ignore the clock $\sigma$) and is conformant to the policy $\psig$ with $\emit \ordpre \abs \in \psig$. 
  A typical program $R$ accessing the signal has alphabet $\{\ol{\pres}, \ol{\abs} \} \subseteq \L(S_i)$ and is of form $R = \ol{\pres} \col \ol{\pres} \cseq R_1 + \ol{\abs} \col \{ \ol{\abs}, \ol{\pres} \} \cseq R_0$. If the signal is present then $R_1$ is executed, if it is absent then $R_0$ is executed.
  Such $R$ conforms to a policy $\prg$ with $\ol{\pres} \ordpre \ol{\abs} \in \prg$ and $\ell \ordpre \ell \in \prg$ for $\ell \in \{ \ol{\pres}, \ol{abs} \}$.  Both $\psig$ and $\prg$ as well as their union $\psig \cup \prg$ are pivot policies, so that $R \cpar S_i$ is pivotable.
  In a pivotable process, choices \mrev{between two admissible synchronisations are choices that involve at least three threads.}{revised for better explanation} These are resolved by the precedences in one of the involved threads, acting as the ``pivot''.
  \mrev{The choices of synchronisation between two threads do not need precedence but are instead handled by admissibility.}{added for explanation} For instance, in $R \cpar S_i$, admissibility in the signal $S_i$ resolves the choice $\ol{\pres} + \ol{\abs}$ in $R$. \mrev{When the signal thread is in state $S_0$, then only $\abs$ is admissible, when it is in state $S_1$, then only $\pres$ is admissible.}{for better explanation} Admissibility in the program $R$ resolves the choice $\abs + \emit$ in $S_0$ and the \mrev{choice $\pres + \emit$ in $S_1$: There is no derivative state in a single thread of $P$ where both $\pres$ and $\emit$ and both $\abs$ and $\emit$ are simultaneously admissible.}{filled in missing case}
\qed
\end{example}

\longshort{\version}{}{
\begin{example} 
  As another simple example, the process $a + \ol{b} \col a \cpar b \cpar \ol{a}$ is pivotable as it conforms to the pivot policy with $\pi \Vdash a \ordpre \ol{b}$ and $\pi \Vdash \ol{a} \indep{} b$. Here, the choice between synchronising on $b$ or on $\ol{a}$ is uniquely resolved by the first thread $a + \ol{b} \col a$ which, being offered $b$ and $\ol{a}$ concurrently by its environment, takes the synchronisation $a \cpar \ol{a}$ rather than $\ol{b} \col a \cpar a$.
\end{example} 
}

Non-pivotable processes, in general, may block because of two threads entering into a circular precedence deadlock with each other.

\begin{example} 
  For example, recall the processes $P \eqdef a \col b \cseq A + b$ and $Q \eqdef \ol{b} \col \ol{a} \cseq B + \ol{a}$ from Example~\ref{ex:binary-blocking}.
  Their parallel composition $P \cpar Q$  creates a deadlock under (even weak) enabling because of the contradicting precedences in the two possible interactions $a \col b \cpar \ol{a}$ and $b \cpar \ol{b} \col \ol{a}$.
  The deadlock is reflected in the policy $\pi$ of $P$ satisfying $b \ordpre a \in \pi$ and $\ol{a} \ordpre \ol{b} \in \pi$ which is not a pivot policy that would instead require $a \indep b \in \pi$ or $\ol{a} \indep \ol{b} \in \pi$. 
\qed
\end{example}

\longshort{\version}{}
{Input-scheduled processes %}{%
 are processes in which choices between distinct actions need to be resolved only between receiver actions and clocks. In other words, each thread executes its send actions always deterministically, rather than in a preemption context. 
%\longshort{\version}{
The policy $\pi_{is}$ %}{%
forces all output actions to be concurrently independent from each other. On the other hand, the policy permits a process to introduce arbitrary priorites between receiver actions and clocks.
}

\longshort{\version}{}{%
Processes conformant to $\pi_{is}$ are also called \textit{input-scheduled}.
Note that $\pi_{is}$ is a pivot policy and thus every input-scheduled process is pivotable. }

\longshort{\version}{}
{%
\begin{proposition}
 Suppose $P$ is input-scheduled and $P \Derives{\alpha}{H}{R} Q$ for some $\alpha$, $H$, $R$, $Q$. Then: 
 \begin{enumerate}
 \item If $\alpha \in \coA$ then $H \subseteq \{ \alpha, \tau \}$
 \item If $\alpha \in \A \cup \C \cup \{ \tau \}$ then $H \subseteq \A \cup \C \cup \{ \tau \}$.
 \end{enumerate}
\label{prop:input-scheduled}
\end{proposition}

\begin{proof}[Of Prop.~\ref{prop:input-scheduled}]
 Obvious by definition of conformance and the construction of $\pi_{is}$.
\end{proof}

Finally, the following definition is relevant for creating c-coherent processes under action restriction (App.~\ref{sec:restriction}).

\begin{definition}
 A policy $\pi$ is called \emph{precedence-closed} for $L \subseteq \L$ if $\ell_1 \in L$ and $\pi \Vdash \ell_1 \ordpre \ell_2$ implies $\ell_2 \in L$.
\label{def:dep-closed}
\end{definition}
}

Reductions of pivotable processes never block because of two threads entering into a circular precedence deadlock with each other. \longshort{\version}{}{Therefore, it suffices to check enabling relative to the concurrent environment in which a rendez-vous synchronisation is happening.}
One can show that for pivotable processes of at most two threads, the three operational semantics of \ccs (admissible), of \ccsp (weak enabling) and our \ccslm (c-enabling) all coincide. 
This is a consequence of 
\longshort{\version}{the fact (Lem.~\ref{lem:pivot-no-local-block-1} and \ref{lem:pivot-no-local-block-2}) that}{%
the following Lem.~\ref{lem:pivot-no-local-block-1}. It states that for c-coherent pivotable processes,}
the race test in the rule $(\ComR)$ is never introducing the silent action $\tau$ into the blocking sets. 
\longshort{\version}{}{%
%---------------------------------------------
\begin{lemma} \mbox{}
 Let $P_1 : \pi$ and $P_2 : \pi$ be c-coherent for the same pivot policy $\pi$ and $\ell \in \L$ such that
 $ %\[ 
 P_1 \Derives{\ell}{H_1}{R_1} P_1'
 \text{ and }
 P_2 \Derives{\ol{\ell}}{H_2}{R_2} P_2'.
 $ %\] 
 Then, $H_2 \cap \oliA(P_1) \subseteq \{\ol{\ell}\}$ and $H_1 \cap \oliA(P_2) \subseteq \{\ell\}$.
\label{lem:pivot-no-local-block-1}
\end{lemma}
%----------------------------------------------
%
\begin{proof}
 Let $\ell \in \L$ and the transitions be given as in the statement of the Lemma. Suppose $\alpha \in H_2$ and $\alpha \in \oliA(P_1) \subseteq \L$, in particular $\alpha \neq \tau$. The former means $\pi \Vdash \alpha \ordpre \ol{\ell}$ by Def.~\ref{def:coherence}(1). Since $\pi$ is pivot we must thus have $\pi \Vdash \ol{\alpha} \indep{} \ell$, specifically $\pi \nVdash \ol{\alpha} \ordpre \ell$, i.e., $\ol{\alpha}\not\in H_1$, again by Def.~\ref{def:coherence}(1).
 But since $\ol{\alpha} \in \iA(P_1)$ there is a transition 
$ %\[
P_1 \Derives{\ol{\alpha}}{H}{E} Q.
$ %\]
Since also $\pi \nVdash \ell \ordpre \ol{\alpha}$, it follows $\ell \not\in H$ by Def.~\ref{def:coherence}(1). Now both $\ol{\alpha} \not\in H_1$ and $\ell \not\in H$ mean that the c-actions $\ell \col H_1[R_1]$ and $\ol{\alpha}\col H[E]$ are interference-free. Note that $\{ \ell, \ol{\alpha} \} \subseteq \L$.
 Thus, the \myb{c-coherence} Def.~\ref{def:coherence} for $P_1$, implies $\ol{\alpha} = l$ or $\ol{\alpha} \in \iA(R_1)$. The latter is impossible, because 
 $H_2 \cap \oliA(R_1) \subseteq H_2 \cap \oliA(R_1 \cpar R_2) = \eset$ by assumption. Therefore we have $\alpha = \ell$ as desired. The inclusion $H_1 \cap \oliA(P_2) \subseteq \{\ell\}$ is obtained by a symmetric argument.
\end{proof}
}
In general, pivotable processes may consist of three or more interacting threads, and  blocking may occur and the semantics of \ccs, \ccsp and \ccslm are different. 
However, one can show that the blocking of a synchronisation is always due to the external environment. 
In other words, the computation of $race(P,Q)$ in rule $(\ComR)$ becomes redundant. Two single-threaded pivotable processes never block each other. 
\longshort{\version}{}
{
\begin{proof}[Of Lem.~\ref{lem:coherent-pivot-coincidence}]
 Consider an admissible transition of $P_1 \cpar P_2$ where $P_1$ and $P_2$ are single-threaded. Firstly, by Lem.~\ref{lem:prio-basic-1}, each of the single-threaded processes can only generate c-actions $\ell_i \col H_i[R_i]$ where $\tau \not\in H_i$ and $R_i = \zero$. Such transitions are always c-enabled for trivial reasons. Further, any rendez-vous synchronisation of such an $\ell_1 \col H_1[R_1]$ and $\ell_2 \col H_2[R_2]$ generates a c-action $\tau \col (H_1 \cup H_2 \cup B)[R_1 \cpar R_2]$ with $R_1 \cpar R_2 = \zero$. 
 Because of Lem.~\ref{lem:pivot-no-local-block-1} and \ref{lem:pivot-no-local-block-2}, the set
\[B = \{\tau \mid H_1 \cap \oliA(P_2) \subseteq \{\ell\} \text{ or } H_2 \cap \oliA(P_1) \subseteq \{\ol{\ell}\}\}\]
must be empty. Thus, the reduction $\ell_1 \cpar \ell_2$ is c-enabled, too. 
\end{proof}
}
\longshort{\version}{}
{
\begin{example}[On coherence, pivotable and policies]
The process $P_1 \pdef a \cpar \ol{b} \cpar \ol{a} \col b + b$ is c-coherent and pivotable with a policy $\pi$ such that $\pi \Vdash b \ordpre \ol{a}$ and $\pi \Vdash \ol{b} \indep a$. The rendez-vous synchronisation $\tau = a \cpar \ol{a}$ 
 %$\tau \col b = a \cpar \ol{a}\col b$ 
 is blocked by the presence of the offering $\ol{b}$ from a \textit{third} thread. 
 We have
$ %\[
P_1 \Derives{\tau}{\{b\}}{\ol{b}} \ol{b},
$ %\]
which is not weakly enabled, because $\{b\} \cap \oliA(\ol{b}) \neq \eset$. The blocking is due to the reaction environment $\ol{b}$ of the transition.  
 For contrast, the process $P_2 \pdef (a + \ol{b}) \cpar (\ol{a} \col b + b)$ is pivotable but not c-coherent.
 In our semantics, the synchronisation of $a$ and $\ol{a} \col b$ in $P_2$ is blocked by the side-condition of rule $(\ParR_3)$ which introduces $\tau$ into the blocking set
$ % \[
P_2 \Derives{\tau}{\{b,\tau\}}{\zero} \zero
$ %\]
because $\{ b \} \cap \oliA(a + \ol{b}) \not\subseteq \{ \ol{a} \}$. This transition is not weakly enabled, because of $ \{b,\tau\} \cap (\iA(\zero) \cup \{ \tau \}) \neq \eset$. We can make the pivotable $P_1$ also c-coherent by forcing a sequential ordering between $a$ and $\ol{b}$ in the first thread, say $P_3 \pdef a \cseq \ol{b} \cpar (\ol{a} \col b + b)$, then the synchronisation is even strongly enabled
$ %\[
P_3 \Derives{\tau}{\{b\}}{\zero} \ol{b}.
$ %\]
If we introduce a precedence into the first thread, say $P_4 \pdef (a \col b + \ol{b}) \cpar (\ol{a} \col b + b)$, then we are also c-coherent but no longer pivotable. Like for $P_2$ the synchronisation $\tau = a \cpar \ol{a}$ is blocked and not weakly enabled, because of the introduction of $\tau$ into the blocking set by the side condition of $\ParR_3$. \qed
\end{example} 
}

\medskip 

\longshort{\version}
{Pivotable process cannot offer a clock action and exhibit reduction at the same time. This means that in this class of processes, clocks and reductions are sequentially scheduled.}
{The following proposition says that a pivotable process cannot offer a clock action and exhibit another clock or a rendez-vous synchronisation at the same time. This means that in this class of processes, clocks and reductions are sequentially scheduled. 
}

\begin{proposition}[Clock Maximal Progress]
 \longshort{\version}
 {Suppose $P$ is c-coherent and pivotable. If $\sigma \in \iA(P)$ then $P$ is in normal form, i.e., there is no reduction $P \fstep{\tau} P'$.}
 {Suppose $P$ is c-coherent and pivotable and $\sigma \in \iA(P)$. Then, for all $\ell \in \L$ with $\ell \neq \sigma$ we have $\ell \not\in \iA(P)$ or $\ol{\ell} \not\in \iA(P)$. In particular, $P$ is in normal form, i.e., there is no reduction $P \fstep{\tau} P'$.} 
\label{prop:maximal-progress}
\end{proposition}
\longshort{\version}{}
{
\begin{proof}[Of Prop.~\ref{prop:maximal-progress}]
 The proof proceeds by contradiction. Let $P : \pi$ for pivot policy $\pi$. Suppose $\ell \in \L$, $\ell \neq \sigma$ and
 $ %\[ 
 P \Derives{\sigma}{H}{\zero} P_1 \text{ and }
 P \Derives{\ell}{N}{F} P_2 \text{ and }
 P \Derives{\ol{\ell}}{M}{G} P_2.
 $ %\] 
By Prop.~\ref{prop:clock-interference} we infer $\sigma \in N$ or $\ell \in H$, i.e., Conformance Def.~\ref{def:conformance} we have $\pi \Vdash \sigma \ordpre \ell$ or $\pi \Vdash \ell \ordpre \sigma$. Hence, $\pi \nVdash \sigma \indep{} \ell$. For the same reason we have $\pi \nVdash \sigma \indep{} \ol{\ell}$. 
 However, this contradicts the pivot property of $\pi$, which requires $\pi \Vdash \sigma \indep{} \ell$ or $\pi \Vdash \sigma \indep{} \ol{\ell}$. Finally, if $P \fstep{\tau} P'$ then the reduction must arise from a rendez-vous synchronisation inside $P$, i.e., there is $\ell \in \R$ with $\ell \in \iA(P)$ and $\ol{\ell} \in \iA(P)$. But then $\sigma \in \iA(P)$ is impossible as we have just seen. 
\end{proof}
}
Prop.~\ref{prop:maximal-progress} can be seen as saying that in pivotable processes all clocks take lowest precedence. As a result, clock transitions satisfy \textit{maximal progress}, i.e., a clock is only enabled on normal forms, i.e., if there is no reduction possible. Thus clocks behave like time steps in the standard timed process algebras \cite{TPA,CleavelandLM97}. More importantly (Thm.~\ref{def:coherence}), all reductions are confluent. So, a pivotable process, no matter in which order the reductions are executed, when it terminates, it computes a unique normal form. At this point it either stops if there is no clock possible, or it offers a clock step (pausing) to a continuation process from which a new normal form is produced. In this way, coherent pivotable processes correspond to synchronous streams.

\medskip 

%\noindent \textbf{Policy-coherent Processes.} 
It remains to be seen how we can obtain c-coherent processes systematically by construction. The following Thm.~\ref{thm:summary-coherence-closure} identifies a set of closure operators for c-coherent processes based on policy-conformance. 
This is our main theorem that summarises Propositions~\ref{prop:zero-coherent}--\ref{prop:restrict-coherent} (Appendix), which
\longshort{\version}{}{
cover idling and prefixes (SSec.~\ref{sec:stop-prefix}), summation (SSec.~\ref{sec:summation}), parallel composition (SSec.~\ref{sec:parallel}), repetition (SSec.~\ref{sec:repetition}), hiding (SSec.~\ref{sec:hiding}) and restriction (SSec.~\ref{sec:restriction}).
}
provide in \ccslm a generalisation of Milner's Confluence Class, Prop.~15-17 in~\cite{Milner:CCS}.

%---------------------------------------------
\begin{theorem}[Milner's Confluence Class Generalised]
\mbox{}
 \begin{bracketenumerate} 
   \item \textit{(Idling)} $\zero \cnf \pi$ for all policies $\pi$.
   \item \textit{(Rendezvous Prefix)} If $Q \cnf \pi$, then for every $\ell \in \R$ we have $\ell \col H \cseq Q \cnf \pi$ if $\ell \in H \subseteq \{ \ell' \mid \ell' \ordpre \ell \in \pi\}$.
   \item \textit{(Clock Prefix)} If $Q \cnf \pi$ and clock $\sigma \in \C$, then $\sigma \col H \cseq Q \cnf \pi$ if $H \subseteq \{ \beta \mid \beta \ordpre \sigma \in \pi \}$. 
   \item \textit{(Choice)} 
     \mrev{Let $Q \cnf \pi_1$ and $R \cnf \pi_2$ be c-coherent and for all pairs of initial transitions 
     $ %$$ 
     Q \Derives{\alpha_1}{H_1}{F_1} Q' \text{ and }
      R \Derives{\alpha_2}{H_2}{F_2} R'
      $ %$$ 
      we have $\alpha_1 \neq \alpha_2$ as well as $\alpha_1 \in H_2$ or $\alpha_2 \in H_1$. 
      Then $Q + R$ is c-coherent for $\pi$ if $\pi_1 \subseteq \pi$ and $\pi_2 \subseteq \pi$.}{corrected statement} 
          %If $P_1 \cnf \pi_1$ and $P_2 \cnf \pi_2$ and for all pairs of initial transitions  
    %%$P \fsstep{\alpha_1 \col H_1[F_1]} P_1$ and
    %%$P \fsstep{\alpha_2 \col H_2[F_2]} P_2$
    %$P \Derives{\alpha_1}{H_1}{F_1} P_1$ and
    %$P \Derives{\alpha_2}{H_2}{F_2} P_2$
    %such that $\alpha_1 \neq \alpha_2$ as well as $\alpha_1 \in H_2$ or $\alpha_2 \in H_1$. 
    %Then $P_1 + P_2 \cnf \pi$ if $\pi_1 \subseteq \pi$ and $\pi_2 \subseteq \pi$. 
   \item \textit{(Parallel)} If $Q \cnf \pi$ and $R \cnf \pi$ where $\pi \restrict \R$ is a pivot policy, then $(Q \cpar R) \cnf \pi$. 
  \item \textit{(Hiding)} If $Q \cnf \pi$ and $\sigma \ordpre \ell \not\in \pi$ for all $\sigma \in L$ and $\ell \in \pi$, then $Q \hide L \cnf \pi$. 
  \item \textit{(Restriction)} If $Q \cnf \pi$ and $\pi$ precedence-closed for $L \cup \ol{L}$, then $({Q}\restrict{L}) \cnf \pi \restrict L$.
 \end{bracketenumerate} 
\label{thm:summary-coherence-closure} 
\end{theorem}
%---------------------------------------------
%

The theorem is the main result of the paper, namely that under a suitable ``policy decoration'', a process can be run deterministically in the scheduled SOS semantics of \ccslm.  

\begin{example}
  To give a concrete example of the property stated by the above theorem, let us revisit Ex.~\ref{ex:abro}: the processes $\mathsf{A}$, $\mathsf{B}$, $\mathsf{R}$, $\mathsf{O}$, $\mathsf{T}$ are all c-coherent and conform to the pivot policy $\pi_{\mathsf{ABRO}}$ with precedences (ignoring the clock, so $\pi \setminus \R = \pi$)
  $k \ordpre \ell$ for $\ell \in \{ k, a, b, s, t, \ol{t} \}$, $\ell \ordpre \ell$ for $\ell \in \{ r, k, a, b, \ol{s}, t, \ol{o} \}$ and $s \ordpre \ol{t}$.
  % $k_s \ordpre \ol{t}$, $s \ordpre \ol{t}$, 
  % $k_x \ordpre \{ k_x, x \}$ for $x \in \{ a, b, s, t \}$ and 
  % $\ell \ordpre \ell$ for $\ell \in \{ r, a, b, \ol{s}, t, \ol{t}, \ol{o} \}$.
  Thus, by Thm.~\ref{thm:summary-coherence-closure}(5), their composition satisfies $\sA \cpar \sB \cpar \sR \cpar \sO \cpar \sT \cnf \pi_{\mathsf{ABRO}}$. Since the policy is also precedence-closed for $\{s, t, \ol{s}, \ol{t} \}$, we have $\ABRO \cnf \pi_{\mathsf{ABRO}}'$ by Thm.~\ref{thm:summary-coherence-closure}(3,7), where $\pi_{\mathsf{ABRO}}' = \pi_{\mathsf{ABRO}} \restrict \{s, t \}$.
\qed
\end{example} 

\paragraph*{On Milner' Confluence Class}

As we have noted above, \ccs corresponds to the unclocked and (blocking) free processes of \ccslm (Prop.~\ref{prop:free-ccs}). For these processes, admissible transitions and weakly/constructively enabled transitions coincide. 
The key result of Milner~\cite{Milner:CCS} (Chap.~11, Prop.~19) is that confluence is preserved by inaction $\zero$, prefixing $\ell \cseq P$ and 
\textit{confluent composition} defined as 
%\begin{eqnarray} 
$ P \mathrel{|_L} Q = (P \cpar Q) \restrict L$
%\label{eqn:confluent-composition}
%\end{eqnarray} 
where $\L(P) \cap \L(Q) = \eset$ and $\ol{\L(P)} \cap \L(Q) \subseteq L \cup \ol{L}$. 
Thm.~\ref{thm:summary-coherence-closure}(5,7) is our replacement of Milner's notion of confluent composition $(P \cpar Q) \restrict L$.
Specifically, let us see how confluence for the fragment of \ccs processes follows from our more general results on coherence. From now on, for the rest of this section, we assume that all processes are unclocked, i.e., each prefix $\ell \col H \cseq Q$ has $H \cup \{ \ell \}\subseteq \A \cup \coA$. 

\medskip 

Recall, a process $P$ is \textit{discrete} if all prefixes occurring in $P$ are of form $\ell\col\{\ell\}\cseq Q$, i.e., they are self-blocking. Under c-enabling, a discrete process behaves like a \ccs process with the restriction that each action $\ell$ only synchronises if there is a matching partner $\ol{\ell}$ in a \textit{unique} other parallel thread. If the matching partner is not unique, then the scheduling blocks. For instance, $\ell \col \ell \cpar \ol{\ell}\col\ol{\ell} \cpar \ol{\ell}\col\ol{\ell}$ blocks because there are two concurrent senders $\ol{\ell}$ matching the receiver $\ell$. 
Now consider the fragment $\ccscc$ of discrete processes built from inaction $\zero$, self-blocking prefixes $\ell\col\{\ell\}\cseq Q$ and confluent composition $(P \cpar Q) \restrict L$.
%equation ~\eqref{eqn:confluent-composition}. 
Milner's result says that the processes in $\ccscc$ are confluent under admissible scheduling. 

\medskip 

To emulate Milner's result in our setting, we consider the \textit{sort} $\L(P)$ of a process $P$ as the admissible actions of a \textit{discrete} policy $\pi_P$ with only reflexive precedences. In other words, $\ell_1 \ordpre \ell_2 \in \pi_P$ iff $\ell_1 = \ell_2$.
It is easy to see that a discrete process always conforms to $\pi_P$. 
Also, discrete policies are always pivot policies (Def.~\ref{def:pivot-policy}) and dependency-closed for all label sets (Def.~\ref{def:pivot-policy}), as one shows easily. Thus, discrete processes, which are c-coherent for discrete policies, fullfill all conditions of the preservation laws for inaction $\zero$ (Thm.~\ref{thm:summary-coherence-closure}(1)), self-blocking action prefix $\ell\col\{\ell\}\cseq Q$ (Thm.~\ref{thm:summary-coherence-closure}(2)), parallel composition $P \cpar Q$ (Thm.~\ref{thm:summary-coherence-closure}(5)) and restriction $P \restrict L$ (Thm.~\ref{thm:summary-coherence-closure}(7)).
As a consequence and in particular, all processes in $\ccscc$ are c-coherent under strong enabling by our results.

\medskip 

Our final observation now is that, for Milner's fragment $\ccscc$, the transition semantics under admissible and c-enabling, as well as the notions of confluence and coherence, coincide.
Firstly, if $P$ is discrete then in every c-action $\ell \col H[E]$ executed by a derivative $Q$ of $P$ we must have $H = \{ \ell \}$. 
%Thus, every c-action is self-interfering as it violates condition (3) of Def.~\ref{def:interference-free}. 
% Further, any two c-actions $\ell_i \col H_i[E_i]$ of transitions 
% %
% $ %$$ 
%  Q \Derives{\ell_1}{H_1}{E_1} Q_1 
%  \text{ and }
%  Q \Derives{\ell_2}{H_2}{E_2} Q_2 
% $ %$$ 
% %
% for $\ell_1 \neq \ell_2$ or $Q_1 \not\equiv Q_2$ are trivially interference-free. 
The restictions imposed by Milner's confluent composition
has the further effect that for every c-action $\ell \col H[E]$ we must have $\ell \not\in \olwiA(E)$. This is proven by induction on the structure of $P \in \ccscc$. Further, the race condition of $(\ComR)$ is vacuous, and thus generally $\tau \not\in H$. 
%This is a result of the side-conditions of the confluent composition.
%equation~\eqref{eqn:confluent-composition}. 
Hence, for $\ccscc$, every c-action $\ell \col H[E]$ is c-enabled since $H \cap (\olwilA{\ast}(E) \cup \{\tau\}) = \eset$. In other words, c-enabling coincides with plain admissibility, and thus with admissibility in \ccs. 

\medskip 

Notice that the class of processes c-coherent for sorts $\L(P)$ considered as discrete policies $\pi_P$ is larger than the confluent processes of $\ccs$. In particular, we can explain coherence of broadcast actions with repetitive prefixes. Hence, our results are more general even for the restricted class of discrete behaviours. 
Finally, to highlight our claim that we extend the classical result, we offer the following semantic definition of confluent behaviours and an associated conjecture to catch Milner's arguments in \ccslm.

\begin{definition}[Milner's Confluence Class]
  \textit{Milner's confluence class} $\ccscc \subseteq \P$ is the set of processes $P$ such that for all derivatives $Q$ of $P$, if $Q \Derives{\alpha}{H}{E} Q'$, the following holds:
  \begin{itemize}
    \item If $\alpha \neq \tau$ then $H \subseteq \{\alpha\}$ and $\alpha \not\in \olwilA{\ast}(E)$.
    \item If $\alpha = \ell \cpar \ol{\ell}$ then $H \subseteq \{ \ell, \ol{\ell}\}$ and $H \cap \olwilA{\ast}(E) = \{\}$. 
  \end{itemize}
\end{definition}

\begin{Conjecture}[Milner's Confluence Class is Coherent]
  For all $P, Q \in \ccscc$:
  \begin{itemize}
    \item $P$ is c-coherent iff $P$ is confluent. \item A transition of $P$ is c-enabled iff it is admissible.
    \item If  $\L(P) \cap \L(Q) = \eset$ and $\ol{\L(P)} \cap \L(Q) = L \cup \ol{L}$, then $(P \cpar Q) \restrict L \in \ccscc$. 
  \end{itemize}
\end{Conjecture}

\longshort{\version}{}
{
%----------------------------------------------------------
\subsubsection{Stop and Prefixes}
\label{sec:stop-prefix}
%----------------------------------------------------------

%---------------------------------------------
\begin{proposition}
 The process $\zero$ is c-coherent for any policy, i.e., $\zero : \pi$ for all $\pi$.
\label{prop:zero-coherent} 
\end{proposition}
%---------------------------------------------

\begin{proof}
 The terminating process $\zero$ is locally c-coherent for trivial reasons, simply because it does not offer any transitions at all. 
\end{proof}

%---------------------------------------------
\begin{proposition}
 Let process $Q$ be c-coherent for $\pi$. Then, for every action $\ell \in \R$ the prefix expression $\ell \col H \cseq Q$ is c-coherent for $\pi$, if $\ell \in H \subseteq \{ \ell' \mid \pi \Vdash \ell' \ordpre \ell\}$. 
\label{prop:channel-prefix-coherent} 
\end{proposition}
%---------------------------------------------

\begin{proof}
 A prefix expression $P = \ell \col H \cseq Q$ generates only a single transition by rule $(\ActR)$, so that the assumption
 \[ 
 P \Derives{\alpha_1}{H_1}{E_1} Q_1
 \text{ and }
 P \Derives{\alpha_2}{H_2}{E_2} Q_2
 \] 
 implies $\alpha_i = \ell$, $H_i = H$, $E_i = \zero$ and $Q_i = Q$. 
 It is easy to see that $\ell \col H \cseq Q$ conforms to $\pi$ if $H \subseteq \{ \beta \mid \pi \Vdash \beta \ordpre \alpha \}$. 
 Since $\alpha_i = \ell \in H = H_i$ and $\{ \alpha_1, \alpha_2 \} \subseteq \R$ 
 %c-actions $\alpha_i \col H_i[E_i]$ interfere and 
 nothing needs to be proved. 
 Finally, note that the only immediate derivative of $\alpha \col H \cseq Q$ is $Q$ which is c-coherent for $\pi$ by assumption.
\end{proof}
We will later see (Prop.~\ref{def:seq-bang-prefix}) that a prefix $\ell \col H\cseq Q$ for $\ell \in \R$ can very well be coherent even if $\ell \not\in H$. Here we note that for clock prefixes, the blocking set is only constrained by the policy.

%---------------------------------------------
\begin{proposition}
 If $Q$ is c-coherent for $\pi$ and clock $\sigma \in \C$, then the prefix expression $\sigma \col H \cseq Q$ is c-coherent for $\pi$, too, if $H \subseteq \{ \beta \mid \pi \Vdash \beta \ordpre \sigma\}$. 
\label{prop:clock-prefix-coherent} 
\end{proposition}
%---------------------------------------------

\begin{proof}
 The argument runs exactly as in case of Prop.~\ref{prop:channel-prefix-coherent}. However, we do not need to require condition $\sigma \in H$ because if $\alpha_i = \sigma$ then $\alpha_1 = \alpha_2$, $Q_1 = Q_2$ and $\{ \alpha_1, \alpha_2 \} \not\subseteq \R$. So, no reconvergence is required. 
 Again, we observe that the only immediate derivative of $\sigma \col H \cseq Q$ is $Q$ which is c-coherent for $\pi$ by assumption.
\end{proof}

\subsubsection{Summation}
\label{sec:summation}

%---------------------------------------------
\begin{proposition}
 Let $P_1 : \pi_1$ and $P_2 : \pi_2$ be c-coherent and for all pairs of initial transitions 
 \[ 
 P \Derives{\alpha_1}{H_1}{F_1} P_1 \text{ and }
 P \Derives{\alpha_2}{H_2}{F_2} P_2
 \] 
 we have $\alpha_1 \neq \alpha_2$ as well as $\alpha_1 \in H_2$ or $\alpha_2 \in H_1$. 
 Then $P_1 + P_2$ is c-coherent for $\pi$ if $\pi_1 \subseteq \pi$ and $\pi_2 \subseteq \pi$. 
 \label{prop:sum-coherent} 
\end{proposition}
%---------------------------------------------
%
\begin{proof}
 Let $P = Q + R$ be given with $Q$ and $R$ locally c-coherent for $\pi_1$ and $\pi_2$, respectively and the rest as in the statement of the proposition. Every transition of $P$ is either a transition of $Q$ or of $R$ via rule $\SumR$. Conformance to $\pi$ directly follows from the same property of $Q$ or $R$, respectively. The proof of \myb{c-coherence} Def.~\ref{def:coherence} is straightforward exploiting that two competing but non-interfering transitions must either be both from $Q$ or both from $R$. More precisely, suppose 
 \[ 
 P \Derives{\alpha_1}{H_1}{F_1} P_1 \text{ and }
 P \Derives{\alpha_2}{H_2}{F_2} P_2
 \] 
 with $\alpha_1 \neq \alpha_2$ or $P_1 \neq P_2$ or $\{ \alpha_1, \alpha_2 \} \subseteq \R$ and $\alpha_i \not\in H_{3-i}$.
 Also, we assume that $\alpha_1 \col H_1[F_1]$ and $\alpha_2 \col H_2[F_2]$ are interference-free. 
 Now if the two transitions of $P$ are by $(\SumR_1)$ from $Q$ and by $(\SumR_2)$ from $R$, then $\alpha_1 \in \iA(Q)$ and $\alpha_2 \in \iA(R)$. But then by assumption, $\alpha_1 \neq \alpha_2$ and $\alpha_1 \in H_2$ or $\alpha_2 \in H_1$, which implies that $\alpha_1 \col H_1[F_1]$ and $\alpha_2 \col H_2[E_2]$ are not interference-free. This means, to prove \myb{c-coherence} Def.~\ref{def:coherence}(2) we may assume that both transitions of $P$ are either by $(\SumR_1)$ from $Q$ or both by $(\SumR_2)$ from $R$.
 Suppose both reductions of $P = Q + R$ are generated by rule $(\SumR_1)$, i.e., 
 \[ 
 Q \Derives{\alpha_1}{H_1}{F_1} Q_1 \text{ and }
 Q \Derives{\alpha_2}{H_2}{F_2} Q_2
 \] 
 where the transitions generated by $\SumR_1$ are 
 \begin{eqnarray*}
 Q + R \Derives{\alpha_1}{H_1}{F_1} Q_1 
 \text{ and } 
 Q + R \Derives{\alpha_2}{H_2}{F_2} Q_2 
 \label{eqn:prisum-dia-1}
 \end{eqnarray*}
 Under the given assumptions, the reconverging transitions of \myb{c-coherence} Def.~\ref{def:coherence}(2) are obtained directly from \myb{c-coherence} of $Q$, by induction. The case that both transitions are from $R$ by rule $(\SumR_2)$ is symmetrical. Finally, the immediate derivatives $P_1$ and $P_2$ of $P_1 + P_2$ are c-coherent for $\pi$, because they are c-coherent for $\pi_i$ by assumption and thus also for the common extension $\pi_i \subseteq \pi$.
\end{proof}

\subsubsection{Parallel}
\label{sec:parallel}

%---------------------------------------------
\begin{lemma}
 Let $Q : \pi$ and $R : \pi$ be such that $\pi \restrict \R$ is a pivot policy. Then $(Q \cpar R) \rcol \pi$.
\label{prop:cpar-coherent} 
\end{lemma}
%---------------------------------------------
%
\begin{proof}
 Let $P = Q \cpar R$ such that both $Q$ and $R$ are c-coherent for $\pi$ where $\pi \restrict \R$ is a pivot policy.
 For conformance consider a transition 
 \begin{eqnarray} 
 P \Derives{\ell}{H}{E} P'
 \label{eqn:par-step}
 \end{eqnarray}
with $\ell \in \L$ and $\ell' \in H$. First suppose $\ell \not\in \C$. Then this transition is not a synchronisation but a transition of either one of the sub-processes $Q$ by $(\ParR_1)$ or $P$ by $(\ParR_2)$ with the same blocking set $H$. In either case, we thus use \myb{c-coherence} $Q : \pi$ or $R : \pi$ to conclude $\pi \Vdash \ell' \ordpre \ell$. 
 If $\ell = \sigma$, then the transition in question is a synchronisation of $Q$ and $R$ via rule $(\ParR_3)$
 \[ 
 Q \cpar R \Derives{\sigma}{H_2 \cup N_2 \cup B_2}{\zero} Q_2 \cpar R_2
 \] 
 with $E = \zero$, $H = H_2 \cup N_2 \cup B_2$ and $P' = Q_2 \cpar R_2$ generated from transitions 
 \[ 
 Q \Derives{\sigma}{H_2}{\zero} Q_2 
 \text{ and }
 R \Derives{\sigma}{N_2}{\zero} R_2.
 \]
 and $B_2 = \{ \ell \cpar \ol{\ell} \mid H_2 \cap \oliA(R) \not\subseteq \{ \sigma \} \text{ or }N_2 \cap \oliA(Q) \not\subseteq \{\sigma\}\}$. 
 By Lem.~\ref{lem:pivot-no-local-block-1} and \ref{lem:pivot-no-local-block-2} we have $B_2 = \eset$. If we now take an arbitrary $\ell' \in H \cap \L$ then $\ell' \in H_2$ or $\ell' \in N_2$. In these cases, we obtain $\pi \Vdash \ell' \ordpre \ell$ by conformance Def.~\ref{def:conformance} from $Q : \pi$ or $R : \pi$. 

In the sequel, we tackle \myb{c-coherence} Def.~\ref{def:coherence}. First we observe that the immediate derivatives of $Q \cpar R$ are always parallel compositions $Q' \cpar R'$ of derivatives of $Q$ and $R$. Thus, we may assume that the immediate derivatives $Q' \cpar R'$ are c-coherent for $\pi$ because $P$ and $Q$ \myb{c-coherence} is preserved under transitions (i.e., by co-induction).
 Now suppose 
 \[ 
 P \Derives{\alpha_1}{H_1}{E_1} P_1 
 \text{ and }
 P \Derives{\alpha_2}{H_2}{E_2} P_2
 \] 
 such that $\alpha_1 \neq \alpha_2$ or $P_1 \neq P_2$ or $\{ \alpha_1, \alpha_2 \} \subseteq \R$ and $\alpha_i \not\in H_{3-i}$. Moreover, the c-actions $\alpha_1 \col H_1[E_1]$ and $\alpha_2 \col H_2[E_2]$ are interference-free.
We argue reconvergence by case analysis on the rules that generate these transitions. These could be $(\ParR_1)$, $(\ParR_2)$ or $(\ParR_3)$. First note that by Lem.~\ref{lem:pivot-no-local-block-1} and \ref{lem:pivot-no-local-block-2} we can drop the consideration of the synchronisation action $\ell \cpar \ol{\ell}$ for the blocking set in $(\ParR_3)$ altogether:
 \[ 
 \infer[(\ParR_{3})]
 {P \cpar Q \Derives{\ell \cpar \ol{\ell}}{H_1 \cup H_2}{R_1 \cpar R_2} P' \cpar Q'}
 {P \Derives{\ell}{H_1}{R_1} P' & Q \Derives{\ol{\ell}}{H_2}{R_2} Q' %& \ell \in \A \cup \coA
 }
 \]
 We will use the name $(\ParR_{3a})$ to refer to the instance of $(\ParR_3)$ with $\ell \in \R$ and the name $(\ParR_{3\sigma})$ for the instance with $\ell \in \C$ and we proceeds by case analysis.

\begin{itemize}

\item $\{(\ParR_1), (\ParR_2) \} \indep{} (\ParR_{3a})$. We start with the case where one of the reductions to $P_i$ is non-synchronising by $(\ParR_1)$ or $(\ParR_2)$ and the second reduction to $P_{3-i}$ is a synchronisation obtained by $(\ParR_{3a})$. By symmetry, it suffices to consider the case that the non-synchronising transition is by $(\ParR_1)$ from $Q$ and the synchronising transition by $(\ParR_{3a})$ from $R$, i.e., $\alpha_1 \not\in C$ and $\alpha_2 = \ell \cpar \ol{\ell}$ and
 \[ 
 Q \cpar R \Derives{\alpha_1}{H_1}{F_1 \cpar R} Q_1 \cpar R 
 \text{ and }
 Q \cpar R \Derives{\ell \cpar \ol{\ell}}{H_2' \cup N_2}{F_2 \cpar G_2} Q_2 \cpar R_2
 \] 
 with $E_1 = F_1 \cpar R$, $E_2 = F_2 \cpar G_2$, 
 $H_2 = H_2' \cup N_2$, $P_1 = Q_1 \cpar R$ and $P_2 = Q_2 \cpar R_2$,  
 where $\ell \in \R$ is a rendez-vous action such that 
 \[ 
 Q \Derives{\alpha_1}{H_1}{F_1} Q_1 \text{ and }
 Q \Derives{\ell}{H_2'}{F_2} Q_2 \text{ and } 
 R \Derives{\ol{\ell}}{N_2}{G_2} R_2.
 \] 
 We may safely assume $\alpha_1 \neq \ell \cpar \ol{\ell} = \tau$ or $Q_1 \cpar R \neq Q_2 \cpar R_2$. The third case $\{ \alpha_1, \alpha_2 \} \subseteq \R$ is excluded since $\alpha_2 = \tau$. Further, suppose the c-actions $\alpha_1 \col H_1[F_1\cpar R]$ and $\ell\cpar\ol{\ell}\col (H_2'\cup N_2)[F_2 \cpar G_2]$ are non-interfering. 
 Observe that $\{ \alpha_1, \ell\cpar\ol{\ell} \} \not\subseteq \R$ and also $\{ \alpha_1, \ell\cpar\ol{\ell} \} \cap \C = \eset$. Thus only a weak environment shift is needed. We wish to exploit \myb{c-coherence} of $Q$. The first observation is that since $\alpha_2 = \tau$, the assumptions on non-interference Def.~\ref{def:interference-free}(2) implies that $H_1 \cap \olwilA{\ast}(F_1 \cpar R) = \eset$. Since $\ell = \ol{\ol{\ell}} \in \oliA(R) \cap \L \subseteq \olwilA{\ast}(F_1 \cpar R)$ by Lem.~\ref{lem:prio-basic-1} and Lem.~\ref{lem:c-enabled-aux}, this implies that $\ell \not\in H_1$ up front.
 Next, if $\alpha_1 \neq \ell \cpar \ol{\ell}$, non-interference Def.~\ref{def:interference-free}(1) means 
 %given that $\alpha_1:H_1/E_1$ and $\tau:H_2/E_2$ are non-interfering, so 
 $\alpha_1 \not\in H_2 = H_2' \cup N_2$ and thus 
 $\alpha_1 \not\in H_2'$.
 The same is true if $\alpha_1 = \ell \cpar \ol{\ell} = \tau$, but then by non-interference Def.~\ref{def:interference-free}(2) we have
 \begin{eqnarray}
 (H_2' \cup N_2) \cap (\olwilA{\ast}(F_2 \cpar G_2) \cup \{ \tau \}) = 
 H_2 \cap (\olwilA{\ast}(E_2) \cup \{ \tau \}) = \eset.
 \label{eqn:aux-non-interf}
 \end{eqnarray}
 from which $\alpha_1 = \tau \not \in H_2'$ follows.
 This settles the first part of the non-interference property Def.~\ref{def:interference-free}(1) for the diverging transitions out of $Q$. 

 For the second part of non-interference, suppose $\alpha_1 = \tau$. Then the non-interference assumption~\eqref{eqn:aux-non-interf} 
 implies that $H_2' \cap \olwilA{\ast}(F_2 \cpar G_2) = \eset$.
 %\[ 
 % (H_2' \cup N_2) \cap \olwiA(F_2 \cpar G_2) = 
 % H_2 \cap \olwiA(E_2) = \eset
 %\] 
 Now since $\olwilA{\ast}(F_2) \subseteq \olwilA{\ast}(F_2 \cpar G_2)$ 
 %and $H_2' \subseteq H_2' \cup N_2$, 
 we conclude from this that $H_2' \cap \olwilA{\ast}(F_2) = \eset$. 
 This completes the proof that the c-actions $\alpha_1 \col H_1[F_1]$ and $\ell \col H_2'[F_2]$ must be non-interfering in all cases. 

Note that we always have $\alpha_1 \neq \ell$ or $\{ \alpha_1, \ell \}\subseteq \R$, and in the latter case also $\alpha_1 \not\in H_2'$ and $\ell \not\in H_1$ from the above. 
 Hence, we can use \myb{c-coherence} Def.~\ref{def:coherence} of $Q$
 and obtain processes $F_1'$, $F_2'$, $Q'$ with reconverging transitions
 \begin{eqnarray} 
 Q_1 \Derives{\ell}{H_2''}{F_2'} Q' 
 \text{ and }
 Q_2 \Derives{\alpha_1}{H_1'}{F_1'} Q' 
 \label{eqn:reconv-par}
 \end{eqnarray} 
 such that $H_2'' \subseteq H_2'$ and $H_1' \subseteq H_1$. 
 We now invoke $(\ParR_{3a})$ and $(\ParR_1)$ to obtain reconverging transitions
 \[ 
 Q_1 \cpar R \Derives{\ell \cpar \ol{\ell}}{H_2'' \cup N_2}{F_2' \cpar G_2} Q' \cpar R_2 
 \text{ and } 
 Q_2 \cpar R_2 \Derives{\alpha_1}{H_1'}{F_1' \cpar R_2} Q' \cpar R_2
 \] 
 such that $H_1' \subseteq H_1$ and $H_2'' \cup N_2 \subseteq H_2' \cup N_2$.
 Regarding the concurrent environments, if $\alpha_1 \in \R$ then the \myb{c-coherence} of $Q$ guarantees that 
 \[ 
 F_1 \Derives{\ell}{}{} F_1' 
 \text{ and }
 F_2 \Derives{\alpha_1}{}{} F_2'
 \]
 from which we infer 
 %$Q_1' \cpar R_2 Q_2' \cpar R_2$ and
 \[ 
 F_2 \cpar G_2 \Derives{\alpha_1}{}{} F_2' \cpar G_2
 \text{ and } 
 F_1 \cpar R \Derives{\ell\cpar\ol{\ell}}{}{} F_1' \cpar R_2.
 \] 
 %In particular, note that $\alpha_1 \in \iA(F_2 \cpar G_2)$.
 If $\alpha_1 = \tau$ the above guarantee from the \myb{c-coherence} of $Q$ is weaker: We only have $F_1 \astep{\ell} F_1'$ and $F_2 \astep{\alpha_1} F_2'$.
 In any case we obtain the environment shifts
 \[ 
 F_2 \cpar G_2 \astep{\alpha_1} F_2' \cpar G_2
 \text{ and }
 F_1 \cpar R \astep{\ell\cpar\ol{\ell}} F_1' \cpar R_2,
 \] 
 as required.
 
 \item $(\ParR_1) \indep{} (\ParR_2)$: Suppose the two non-interfering transitions of $P = Q \cpar R$ are by $(\ParR_1)$ from $Q$ and by $(\ParR_2)$ from $R$. Thus, we are looking at transitions
 \[ 
 Q \Derives{\alpha_1}{H_1}{F_1} Q_1 
 \text{ and } 
 R \Derives{\alpha_2}{H_2}{F_2} R_2
 \] 
 for $\{\alpha_1, \alpha_2\} \subseteq \R \cup \{\tau\}$ combined via $(\ParR_1)$ and $(\ParR_2)$, respectively, to generate
 \[ 
 Q \cpar R \Derives{\alpha_1}{H_1}{F_1 \cpar R} Q_1 \cpar R
 \text{ and } 
 Q \cpar R \Derives{\alpha_2}{H_2}{Q \cpar F_2} Q \cpar R_2
 \] 
 with $E_1 = F_1 \cpar R$ and $E_2 = Q \cpar F_2$. Then we can directly construct reconverging transitions, applying $(\ParR_2)$ and $(\ParR_1)$, respectively, for
 \[ 
 Q_1 \cpar R \Derives{\alpha_2}{H_2}{Q_1 \cpar F_2} Q_1 \cpar R_2 
 \text{ and }
 Q \cpar R_2 \Derives{\alpha_1}{H_1}{F_1 \cpar R_2} Q_1 \cpar R_2.
 \]
 Obviously, $Q_1 \cpar R_2$ reduces to $Q_1 \cpar R_2$ and by $(\ParR_1)$ and $(\ParR_2)$ we also have the strong environment shifts
 \[ 
 F_1 \cpar R \Derives{\alpha_2}{}{} F_1 \cpar R_2 
 \text{ and } 
 Q \cpar F_2 \Derives{\alpha_1}{}{} Q_1 \cpar F_2
 \] 
 which completes the claim for we have, in particular, $F_1 \cpar R \astep{\alpha_2} F_1 \cpar R_2$ and $Q \cpar F_2 \astep{\alpha_1} Q_1 \cpar F_2$.
 Notice that in this case we do not need to assume that $Q$ or $R$ are c-coherent, i.e., the assumption that the c-actions $\alpha_i \col H_i[E_i]$ are interference-free.

\item $(\ParR_i) \indep{} (\ParR_i)$ for $i \in \{1,2\}$: Suppose both non-interfering and diverging reductions are by $(\ParR_1)$ from $Q$ (or symmetrically both from $R$ by $(\ParR_2)$)
 \[ 
 Q \cpar R \Derives{\alpha_1}{H_1}{F_1 \cpar R} Q_1 \cpar R
 \text{ and }
 Q \cpar R \Derives{\alpha_2}{H_2}{F_2 \cpar R} Q_2 \cpar R
 \]
 with $\{\alpha_1, \alpha_2\} \subseteq \R \cup \{\tau\}$ and $P_i = Q_i \cpar R$, $E_i = F_i \cpar R$, where 
 \[ 
 Q \Derives{\alpha_1}{H_1}{F_1} Q_1 
 \text{ and } 
 Q \Derives{\alpha_2}{H_2}{F_2} Q_2. 
 \] 
 with $\alpha_1 \neq \alpha_2$ or $Q_1 \cpar R \neq Q_2 \cpar R$ or $\{ \alpha_1, \alpha_2 \} \subseteq \R$ and $\alpha_i \not\in H_{3-i}$. This, in particular, implies $Q_1 \neq Q_2$.
 Further, assume 
 %if $\alpha_1 = \alpha_2$ and $\{ \alpha_1, \alpha_2\} \subseteq \C \cup \{ \tau \}$. From
 non-interference of $\alpha_i \col H_i[F_i \cpar R]$.
 We obtain non-interference of $\alpha_i \col H_i[F_i]$ by Prop.~\ref{prop:interference}.
 We can therefore apply \myb{c-coherence} Def.~\ref{def:coherence} of $Q$ and obtain processes $F_1'$, $F_2'$ and $Q'$ with reconverging transitions 
 \[ 
 Q_1 \Derives{\alpha_2}{H_2'}{F_2'} Q' \text{ and }
 Q_2 \Derives{\alpha_1}{H_1'}{F_1'} Q'
 \] 
 with $H_i' \subseteq H_i$ and such that
 \begin{eqnarray} 
 F_1 \astep{\alpha_2} F_1' 
 \text{ and }
 F_2 \astep{\alpha_1} F_2'
 \label{eqn:par-aux-1}
 \end{eqnarray}
 and if $\{\alpha_1, \alpha_2\} \subseteq \R$ and $\alpha_1 \neq \alpha_2$ or $Q_1 \neq Q_2$, more strongly 
 \begin{eqnarray} 
 F_1 \Derives{\alpha_2}{}{} F_1' \text{ and }
 F_2 \Derives{\alpha_1}{}{} F_2'.
 \label{eqn:par-aux-2} 
 \end{eqnarray}
 Reapplying $(\ParR_1)$ we construct the transitions
 \[ 
 Q_1 \cpar R \Derives{\alpha_2}{H_2'}{F_2' \cpar R} Q' \cpar R
 \text{ and }
 Q_2 \cpar R \Derives{\alpha_1}{H_1'}{F_1' \cpar R} Q' \cpar R.
 \]
 Obviously, by $(\ParR_1)$ and $(\ParR_2)$ we also obtain transitions 
 \[ 
 F_1 \cpar R \astep{\alpha_2} F_1' \cpar R 
 \text{ and }
 F_2 \cpar R \astep{\alpha_1} F_2' \cpar R 
 \] 
 from~\eqref{eqn:par-aux-1} or, more strongly,
 \[ 
 F_1 \cpar R \Derives{\alpha_2}{}{} F_1' \cpar R 
 \text{ and }
 F_2 \cpar R \Derives{\alpha_1}{}{} F_2' \cpar R 
 \] 
 from~\eqref{eqn:par-aux-2} if $\{\alpha_1, \alpha_2\} \subseteq \R$ and $\alpha_1 \neq \alpha_2$ or $Q_1 \cpar R \neq Q_2 \cpar R$.
 
\item $(\ParR_{3a}) \indep{} (\ParR_{3a})$:
 This is the most interesting case in which we are going to exploit the assumption that $Q$ and $R$ are c-coherent for the same pivot policy $\pi$. Without loss of generality we assume $P = Q \cpar R$ and both the reductions 
 \[ 
 Q \cpar R \Derives{\alpha_i}{H_i \cup N_i}{E_i} Q_i \cpar R_i
 \] 
 with $P_i = Q_i \cpar R_i$ are $\tau$-actions generated by the communication rule $(\ParR_3)$, i.e. $\alpha_1 = \ell_1 \cpar \ol{\ell_1} = \tau = \ell_2 \cpar \ol{\ell_2} = \alpha_2$ for actions $\ell_i \in \R$.
 Since $\alpha_1 = \alpha_2$ and $\{ \alpha_1, \alpha_2 \} = \{ \tau \} \not\subseteq \R$ we only need to prove confluence the case that $P_1 \neq P_2$, i.e., $Q_1 \neq Q_2$ or $R_1 \neq R_2$.
 Moroever, we only need a weak environment shift, because $\{ \alpha_1, \alpha_2 \} = \{ \tau \} \not\subseteq \R$ and $\{ \alpha_1, \alpha_2 \} \cap \C = \eset$.
 Also, non-interference Def.~\ref{def:interference-free}(2) means that $(H_i \cup N_i) \cap (\olwilA{\ast}(E_i) \cup \{ \tau \}) = \eset$, i.e.,
 \begin{eqnarray}
 H_i \cap (\olwilA{\ast}(E_i) \cup \{ \tau \}) = \eset 
 \text{ and } 
 N_i \cap (\olwilA{\ast}(E_i) \cup \{ \tau \}) = \eset.
 \label{eqn:interf-assumption} 
 \end{eqnarray}
 Thus, we are looking at transitions
 \[ 
 Q \Derives{\ell_i}{H_i}{F_i} Q_i \text{ and } 
 R \Derives{\ol{\ell}_i}{N_i}{G_i} R_i.
 \]
 We claim that the two c-actions $\ell_i \col H_i[F_i]$ of $Q$, and likewise the two c-actions $\ol{\ell}_i \col N_i[G_i]$ of $R$, are interference-free. 
 Note that since $\ell_i \neq \tau$, we only need to consider the first part of Def.~\ref{def:interference-free}(1).
 We first show that at least one of these pairs of c-actions must be interference free, because of \myb{c-coherence} and conformance to $\pi$ where $\pi \restrict \R$ is a pivot policy. It will then transpire that the other must be interference-free, too, because of the assumption~\eqref{eqn:interf-assumption}. Obviously, if $\ell_1 = \ell_2$ then both $\ell_i \col H_i[F_i]$ of $Q$ and $\ol{\ell}_i \col N_i[G_i]$ are interference-free for trivial reasons. 
 We distinguish two cases depending on whether $\ell_1$ and $\ell_2$ are identical or not. So, suppose $\ell_1 \neq \ell_2$. Both $Q$ and $R$ are c-coherent for the same policy $\pi$ and $\pi \restrict \R$ is a pivot policy. Thus, $\pi \Vdash \ell_1 \indep{} \ell_2$ or $\pi \Vdash \ol{\ell}_1 \indep{} \ol{\ell}_2$, by Def.~\ref{def:pivot-policy}. 
 Hence, $\ell_i \not\in H_{3-i}$ or $\ol{\ell}_i \not\in N_{3-i}$ by conformance Def.~\ref{def:conformance}, which means that at least one of the pairs of c-actions $\ell_i \col H_i[F_i]$ of $Q$ or the two c-actions $\ol{\ell}_i \col N_i[G_i]$ must be interference-free by policy. Now we argue that if $\ell_1 \neq \ell_2$ and one of the pairs $\ell_i \col H_i[F_i]$ or $\ol{\ell}_i \col N_i[G_i]$ is interference-free the other must be, too, and so we can close the diamond for both $Q$ and $R$. By symmetry, it suffices to run the argument in case that the c-actions $\ell_i \col H_i[F_i]$ of $Q$ are non-interfering. Then, we have $\{ \ell_1, \ell_2 \} \subseteq \R$ and $\ell_i \not\in H_{3-i}$. 
 Hence, we invoke \myb{c-coherence} Def.~\ref{def:coherence} of $Q$ obtaining processes $F_1'$ and $F_2'$ with transitions (strong environment shift)
 \[ 
 % F_1 \astep{\ell_2} F_1' 
 % \text{ and }
 % F_2 \astep{\ell_1} F_2'
 % \text{ ... }
 F_1 \Derives{\ell_2}{}{} F_1' 
 \text{ and }
 F_2 \Derives{\ell_1}{}{} F_2'
 \]
 and a process $Q'$ with transitions
 \[ 
 Q_1 \Derives{\ell_2}{H_2'}{F_2'} Q' 
 \text{ and }
 Q_2 \Derives{\ell_1}{H_1'}{F_1'} Q'.
 \]
 such that $H_i' \subseteq H_i$. 
 This means $\ell_i \in \iA(F_{3-i}) \subseteq \iA(F_{3-i} \cpar G_{3-i}) = \iA(E_{3-i}) \subseteq \wilA{\ast}(E_{3-i})$ and therefore $\ol{\ell}_i \not\in N_{3-i}$ because of~\eqref{eqn:interf-assumption}. Thus, as claimed, the c-actions $\ol{\ell}_i \col N_i[G_i]$ are interference-free, too. Again $\{ \ol{\ell}_1, \ol{\ell}_2 \} \subseteq \R$ and $\ol{\ell}_i \not\in N_{3-i}$ entitles us to exploit \myb{c-coherence} Def.~\ref{def:coherence} for $R$ 
 from which we obtain processes $G_1'$ and $G_2'$ with transitions 
 (again, a strong environment shift)
 \[ 
 G_1 \Derives{\ol{\ell}_2}{}{} G_1' 
 \text{ as well as }
 G_2 \Derives{\ol{\ell}_1}{}{} G_2'
 \]
 as well as a process $R'$ with 
 \[ 
 R_1 \Derives{\ol{\ell}_2}{N_2'}{G_2'} R' 
 \text{ and } 
 R_2 \Derives{\ol{\ell}_1}{N_1'}{G_1'} R'.
 \]
 with $N_i' \subseteq N_i$.
 Finally, we invoke $(\ParR_3)$ to obtain the re-converging reductions
 \[ 
 Q_1 \cpar R_1 \Derives{\ell_2 \cpar \ol{\ell_2}}{H_2' \cup N_2'}{F_2' \cpar G_2'} Q' \cpar R' 
 \text{ and } 
 Q_2 \cpar R_2 \Derives{\ell_1 \cpar \ol{\ell_1}}{H_1' \cup N_1'}{F_1' \cpar G_1'} Q' \cpar R'
 \] 
 with $H_i' \cup N_i' \subseteq H_i \cup N_i$,
 and also 
 \[ 
 F_{i} \cpar G_{i} \Derives{\ell_{3-i} \cpar \ol{\ell_{3-i}}}{}{} F_{i}' \cpar G_{i}'
 \] 
 as required, since this implies the weaker form
 $F_{i} \cpar G_{i} \astep{\tau} F_{i}' \cpar G_{i}'$ of environment shift. It remains to consider the case that $\ell_1 = \ell = \ell_2$. By assumption $P_1 \neq P_2$ and so we must have $Q_1 \neq Q_2$ or $R_1 \neq R_2$. Suppose $Q_1 \neq Q_2$. If $R_1 \neq R_2$ we apply a symmetric argument with the role of $\ell_i$ and $\ol{\ell}_i$ interchanged. 
 As observed above, the c-actions $\ell_i \col H_i[F_i]$ are trivially interference-free according to Def.~\ref{def:interference-free}, because $\ell_1 = \ell_2$ and $\ell_i \neq \tau$. 
 Then, by \myb{c-coherence} Def.~\ref{def:coherence} applied to $Q$ we conclude that there are processes $F_1'$ and $F_2'$ with transitions
 \[ 
 F_1 \Derives{\ell}{}{} F_1' 
 \text{ and }
 F_2 \Derives{\ell}{}{} F_2'
 \]
 and a process $Q'$ with transitions (the strong environment shift arises here because we have $\ell \in \R$ and $Q_1 \neq Q_2$)
 \[ 
 Q_1 \Derives{\ell}{H_2'}{F_2'} Q' 
 \text{ and }
 Q_2 \Derives{\ell}{H_1'}{F_1'} Q'.
 \]
 such that $H_i' \subseteq H_i$. This means that $\ell \in \iA(F_{3-i}) \subseteq \iA(E_{3-i})$ and so $\ol{\ell} \not\in N_{3-i}$ because of $N_{3-i} \cap \olwilA{\ast}(E_{3-i}) = \eset$ as per~\eqref{eqn:interf-assumption}. 
 If $R_1 = R_2$ then for each $i \in \{1,2\}$, by \myb{c-coherence} Def.~\ref{def:coherence} (as $\ol{\ell} \not \in N_{3-i} \cup \C$) there is a transition 
 \[ 
 R_i \Derives{\ol{\ell}}{N_{3-i}'}{G_{3-i}'} R'
 \] 
 with $N_{i}' \subseteq N_{i}$ and $G_i = G_i'$ or $G_i \Derives{\ol{\ell}}{}{} G_i'$. Using these transitions we now recombine
 \[ 
 Q_i \cpar R_i 
 \Derives{\ell \cpar \ol{\ell}}{H_{3-i}' \cup N_{3-i}'}{F_{3-i}' \cpar G_{3-i}'} Q' \cpar R'
 \] 
 with $H_{i}' \cup N_{i}' \subseteq H_{i} \cup N_{i}$. Note that if $G_i = G_i'$ then 
 \[ 
 F_i \cpar G_i \Derives{\ell}{}{} F_i' \cpar G_i = F_i' \cpar G_i'
 \] 
 or if $G_i \Derives{\ol{\ell}}{}{} G_i'$ then 
 \[ 
 F_i \cpar G_i \Derives{\ell\cpar\ol{\ell}}{}{} F_i' \cpar G_i'
 \] 
 This means that, in all cases, $F_i \cpar G_i \astep{\ell\cpar\ol{\ell}} F_i' \cpar G_i'$ as desired. If $R_1 \neq R_2$ we can apply \myb{c-coherence} Def~\ref{def:coherence} to $R$ in the same way and close the diamond with
 processes $G_1'$ and $G_2'$ and transitions (again, the strong environment shift arises here because we have $\ol{\ell} \in \R$ and $R_1 \neq R_2$)
 \[ 
 G_1 \Derives{\ol{\ell}}{}{} G_1' 
 \text{ and }
 G_2 \Derives{\ol{\ell}}{}{} G_2'
 \]
 as well as a process $R'$ with 
 \[ 
 R_1 \Derives{\ol{\ell}}{N_2'}{G_2'} R' 
 \text{ and } 
 R_2 \Derives{\ol{\ell}}{N_1'}{G_1'} R'.
 \]
 with $N_i' \subseteq N_i$. Finally, we invoke $(\ParR_3)$ to obtain the re-converging reductions
 \[ 
 Q_1 \cpar R_1 \Derives{\ell \cpar \ol{\ell}}{H_2' \cup N_2'}{F_2' \cpar G_2'} Q' \cpar R' 
 \text{ and } 
 Q_2 \cpar R_2 \Derives{\ell \cpar \ol{\ell}}{H_1' \cup N_1'}{F_1' \cpar G_1'} Q' \cpar R'
 \] 
 with $H_i' \cup N_i' \subseteq H_i \cup N_i$, 
 and also 
 \[ 
 F_{i} \cpar G_{i} \Derives{\ell \cpar \ol{\ell}}{}{} F_{i}' \cpar G_{i}'
 \] 
 which also implies $F_{i} \cpar G_{i} \astep{\ell \cpar \ol{\ell}} F_{i}' \cpar G_{i}'$
 as required. 
 
\item $(\ParR_{3a}) \indep{} (\ParR_{3\sigma})$ and, by symmetry, the induction case $(\ParR_{3\sigma}) \indep{} (\ParR_{3a})$: We assume $P = Q \cpar R$ and the non-interfering, diverging reductions 
 \[ 
 Q \cpar R \Derives{\ell\cpar\ol{\ell}}{H_1 \cup N_1 \cup B_1}{E_1} Q_1 \cpar R_1
 \text{ and }
 Q \cpar R \Derives{\sigma}{H_2 \cup N_2 \cup B_2}{\zero} Q_2 \cpar R_2
 \] 
 for $\ell \in \L$, $\sigma \in \C$, with $P_i = Q_i \cpar R_i$ and $E_i = F_i \cpar G_i$ are generated by instances $(\ParR_{3a})$ and $(\ParR_{3\sigma})$ of the communication rules, respectively. These transitions arise from
 \[ 
 Q \Derives{\ell}{H_1}{F_1} Q_1 
 \text{ and }
 R \Derives{\ol{\ell}}{N_1}{G_1} R_1
 \text{ and }
 Q \Derives{\sigma}{H_2}{\zero} Q_2 
 \text{ and }
 R \Derives{\sigma}{N_2}{\zero} R_2.
 \]
 We assume that the c-actions $(\ell\cpar\ol{\ell}) \col (H_1 \cup N_1 \cup B_1)[E_1]$ and $\sigma \col (H_2 \cup N_2 \cup B_2)[\zero]$ are non-interfering.
 Note that $\ell \neq \sigma$ and $\ol{\ell} \neq \sigma$.
 First, observe that by Lem.~\ref{lem:pivot-no-local-block-1} and \ref{lem:pivot-no-local-block-2}, $B_1 = \eset = B_2$.
 Also, $\ell \cpar \ol{\ell} = \tau \neq \sigma$ and so the non-interference assumption %of Def.~\ref{def:interference-free}(2) 
 means that $\sigma \not\in H_1 \cup N_1 \cup B_1$ as well as $\tau \not \in H_2 \cup N_2 \cup B_2$. The other part of the non-interference, from Def.~\ref{def:interference-free}(1), does not provide any extra information. 

 We claim that $\ell \in H_2$. By contradiction, if $\ell \not\in H_2$ then, since $\sigma \not\in H_1$, the two c-actions $\ell \col H_1[F_1]$ and $\sigma \col H_2[\zero]$ of $Q$ would be interference-free according to Def.~\ref{def:interference-free}. But then, since $\{ \ell, \sigma \} \cap \C \neq \eset$, \myb{c-coherence} Def.~\ref{def:coherence} applied to $Q$ would imply 
 $\sigma \in \iA(F_1)$ and $\ell \in \iA(\zero)$. The latter, however, is impossible. Thus, as claimed, $\ell \in H_2$. Likewise, from $\sigma \not\in N_1$ we derive $\ol{\ell} \in N_2$ in the same fashion.
 
 Now we invoke the fact that $\tau \not \in B_2$, i.e. 
 $H_2 \cap \oliA(R) \subseteq \{\sigma\}$ and $N_2 \cap \oliA(Q) \subseteq \{\sigma\}$. But $\ol{\ell} \in \iA(R)$ and $\ell \in \iA(Q)$, whence we must have $\ell \not\in H_2$ and $\ol{\ell} \not\in N_2$. This is a contradiction.
 Hence, transitions obtained by rules $\ParR_{3a}$ and $\ParR_{3\sigma}$ are never interference-free.
 
\item $\{(\ParR_1), (\ParR_2) \} \indep{} (\ParR_{3\sigma})$: 
 As the representative case, we assume $P = Q \cpar R$ and the non-interfering, diverging reductions are 
 \[ 
 Q \cpar R \Derives{\alpha_1}{H_1}{F_1 \cpar R} Q_1 \cpar R 
 \text{ and }
 Q \cpar R \Derives{\sigma}{H_2 \cup N_2 \cup B_2}{\zero} Q_2 \cpar R_2
 \] 
 with $\alpha_1 \in \R \cup \{\tau\}$ so that $P_1 = Q_1 \cpar R$ and $P_2 = Q_2 \cpar R_2$ are generated by the rule $(\ParR_1)$ and $(\ParR_{3\sigma})$, respectively, from transitions 
 \[ 
 Q \Derives{\alpha_1}{H_1}{F_1} Q_1 
 \text{ and }
 Q \Derives{\sigma}{H_2}{\zero} Q_2 
 \text{ and }
 R \Derives{\sigma}{N_2}{\zero} R_2.
 \]
 where $\alpha_1 \neq \sigma$ and the c-actions $\alpha_1 \col H_1[F_1 \cpar R]$ and $\sigma \col (H_2 \cup N_2 \cup B_2)[\zero]$ are interference-free. First, Def.~\ref{def:interference-free}(1) implies $\alpha_1 \not\in H_2 \cup N_2 \cup B_2$ and $\sigma \not\in H_1$. Further, if $\alpha_1 = \tau$ then
 $(H_2 \cup N_2 \cup B_2) \cap (\wilA{\ast}(\zero) \cup \{ \tau \}) = \eset$, which in particular means $H_2 \cap (\wilA{\ast}(\zero) \cup \{ \tau \}) = \eset$. Therefore, the c-actions $\alpha_1 \col H_1[F_1]$ and $\sigma \col H_2[\zero]$ are interference-free. Note that $\{ \alpha_1, \sigma\} \cap \C \neq \eset$. Hence we can exploit \myb{c-coherence} Def.~\ref{def:coherence} for $Q$, obtaining the stronger form of environment shift. But this would imply $\alpha_1 \in \iA(\zero)$ which is impossible.
 So we find that transitions obtained by rules $\ParR_{i}$ and $\ParR_{3\sigma}$ are never interference-free in the sense of \myb{c-coherence}.

\item $(\ParR_{3\sigma}) \indep{} (\ParR_{3\sigma})$:
 We assume $P = Q \cpar R$ and the non-interfering, diverging reductions 
 $ %\[ 
 Q \cpar R \Derives{\sigma_1}{H_1 \cup N_1 \cup B_1}{\zero} Q_1 \cpar R_1
 \text{ and }
 Q \cpar R \Derives{\sigma_2}{H_2 \cup N_2 \cup B_2}{\zero} Q_2 \cpar R_2
$ % \] 
 with $P_i = Q_i \cpar R_i$ and $E_i = F_i \cpar G_i$ are generated by the communication rule $(\ParR_{3\sigma})$ and $(\ParR_{3\sigma})$. Observing that $\{ \sigma_1, \sigma_2 \} \subseteq \R$ is impossible, here, we may assume $\sigma_1 \neq \sigma_2$ or $Q_1 \cpar R_1 \neq Q_2 \cpar R_2$. This can only hold true if $Q_1 \neq Q_2$ or $R_1 \neq R_2$. Furthermore, the two clock transitions of $Q_1 \cpar Q_2$ can be assumed to be interference-free.
 From this it follows that the two underlying clock transitions 
 %$Q \fsstep{\sigma \col H_1[\zero]} Q_1$ 
 $Q \Derives{\sigma}{H_1}{\zero} Q_1$ 
 and 
 %$Q \fsstep{\sigma \col H_2[\zero]} Q_2$ 
 $Q \Derives{\sigma }{H_2}{\zero} Q_2$ 
 and the likewise the transitions 
 %$R \fsstep{\sigma \col N_1[\zero]} R_1$
 $R \Derives{\sigma}{N_1}{\zero} R_1$
  and 
  %$R \fsstep{\sigma \col N_2[\zero]} R_2$ 
  $R \Derives{\sigma}{N_2}{\zero} R_2$ 
  are interference-free. If $Q_1 \neq Q_2$ then we could apply \myb{c-coherence} of $Q$ and obtain $\sigma \in \iA(\zero)$, if $R_1 \neq R_2$ the same follows from \myb{c-coherence} of $R$. This is impossible, whence the case can be excluded. 
 \end{itemize}\vspace{-5mm}
\end{proof}

\subsubsection{Repetition}
\label{sec:repetition}

For arbitrary continuation processes $A$, a regular prefix $a \col a \cseq A$ needs to be self-blocking to be c-coherent. Such a prefix cannot receive from two senders, for $a \col a \cseq A \cpar \ol{a} \cpar \ol{a}$ will block under weak enabling. However, if the receiver is willing to engage in action $a$ repeatedly, say if $A \pdef a \cseq A$ then we can lift the self-blocking. The process $A \cpar \ol{a} \cpar \ol{a}$ is c-coherent and will happily consume both $\ol{a}$ sender actions. Repetition is the universal way of implementing multi-sender and multi-receiver scenarios. Note that we define a new operator, representing recursion via parallel replication on the same channel $\ell$, that can be easily encoded in the current syntax.

\begin{definition}[Sequential Bang Prefix]
 For every process $P$, action $\ell \in \L$ and $H \subseteq \Act$, let $\bang{\ell} \col H \cseq P$ be the process defined by the 
 SOS rule 
 \[ 
 \infer[(\RepR)]
 {\bang{\ell} \col H \cseq P \Derives{\ell}{H}{P} P \cpar \bang{\ell} \col H \cseq P}
 {}
 \] 
 We use $\bang{\ell} \col H$ as an abbreviation for $\bang{\ell} \col H \cseq \zero$ and $\bang{\ell}$ for $\bang{\ell} \col \eset \cseq \zero$.
\label{def:seq-bang-prefix}
\end{definition}
The purpose of repetition is to replicate input and output prefixes so they can be consumed multiple times and by multiple threads. To see this, let us notice the difference in the blocking sets of a regular output prefix $\ol{a} \col \ol{a} \cseq A$ and its repetition $\bang{\ol{a}} \cseq A$. 
 In the former case, the transition 
 $ %\[
 \ol{a}\col \ol{a} \cseq A \Derives{\ol{a}}{\{\ol{a}\}}{\zero} A
 $ %\] 
 includes the output label in the blocking set $\{ \ol{a} \}$, by $(\ActR_1)$. This reflects the fact that the output prefix $\ol{a}$ is fully consumed by the transition. Accordingly, (strong/weak) enabling will block multiple concurrent receivers trying to access the label $a$ at the same time. Only a single thread can engage in the synchronisation. This is necessary, since the $\ol{a}$ prefix is consumed. The continuation process $A$ may not offer output $\overline{a}$ any more, or if it does, it may be incongruent with $A$. Therefore, each synchronisation with one receiver thread would preempt another concurrent receiver thread in getting access to $\ol{a}$. For contrast, in a transition 
 $ %\[
 \bang{\ol{a}} \cseq A \Derives{\ol{a}}{\eset}{A} A \cpar \bang{\ol{a}} \cseq A
 $ %\]
 generated from $(\ActR_1)$ and $(\RepR)$, the blocking set is empty, and so does not prevent concurrent receivers. This is fine since the prefix is not consumed but repeated in the continuation process $A \cpar \bang{\ol{a}} \cseq A$.
\begin{remark}[Why is the standard bang not good enough?]
 Note that the sequential bang $\bang{\ell} \cseq A$ is not expressible as $\bang{\ell} \cseq A = \bang{(\ell \cseq A)}$ by the standard `bang' operator $\bang{P}$ of process algebra which satisfies the \textit{false} (hence not present) structural equivalence $\bang{P} \equiv P \cpar \bang{P}$ in \ccslm. The corresponding equivalence 
 $\bang{\ell} \cseq A \cong \ell \cseq A \cpar \bang{\ell} \cseq A$ does not hold in \ccslm. Let us see why. 
 Suppose the behaviour of $\bang{\ell}$ is derived from (or identified with that of) $\ell \cseq A \cpar \bang{\ell} \cseq A$. Then the initial action of $\bang{\ell} \cseq A$ would be 
 $ %\[
 \ell \cseq A \cpar \bang{\ell} \cseq A 
 \Derives{\ell}{\eset}{\bang{\ell} \cseq A} A \cpar \bang{\ell} \cseq A
 $ %\] 
 generated by rules $(\ActR_1)$ and $(\ParR_1)$. Notice that the concurrent environment $\bang{\ell} \cseq A = \bang{\ell} \cseq A$ is not empty and that we have $\ell \in \iA(\bang{\ell} \cseq A)$. Therefore, a parallel composition
 $ %\[
 \ol{\ell}\col \ol{\ell} \cpar \ell \cseq A \cpar \bang{\ell} 
 \cseq A \Derives{\tau}{\{\ol{\ell}\}}{\bang{\ell} \cseq A} A \cpar \bang{\ell} \cseq A
 $ %\]
 would block under (weak/strong) enabling because $\{\ol{\ell}\} \cap \oliA(\bang{\ell} \cseq A) \neq \eset$.
 Thus $\ell \cseq A \cpar \bang{\ell} \cseq A $ would not be able to synchronise even with a single (self-blocking) sender $\ol{\ell} \col \ol{\ell}$. 
 For contrast, the sequential bang prefix of Def.~\ref{def:seq-bang-prefix} can serve arbitrarily many receivers
 $ %\[
 \ol{\ell} \col \ol{\ell} \cpar 
 \ol{\ell} \col \ol{\ell} \cpar 
 \bang{\ell} \cseq A 
 \Derives{\tau}{\{\ol{\ell}\}}{\ol{\ell} \col \ol{\ell} \cpar A} 
 \ol{\ell} \col \ol{\ell} \cpar A \cpar \bang{\ell} \cseq A
 \Derives{\tau}{\{\ol{\ell}\}}{A \cpar A} 
 A \cpar A \cpar \bang{\ell} \cseq A,
 $ %\]
 provided that $\{\ol{\ell}\} \cap \olwilA{\ast}(A) = \{\ol{\ell}\} \cap \olwilA{\ast}(A) = \eset$.
 The problem is that the standard bang $\bang{\ell} \cseq A \equiv \ell \cseq A \cpar \bang{\ell} \cseq A$ is a \textit{parallel} repetition of label $\ell$ while the sequential bang $\bang{\ell} \cseq A$ offers all $\ell$-labels \textit{sequentially} in a single thread. Only the payload $A$ is offered in parallel. This corresponds to an \textit{async} operator that is triggered by label $\ell$, then spawns a child thread $A$ and repeates itself in the main thread. \qed
\end{remark}
%
%
%

%---------------------------------------------
\begin{proposition}
 Let $P : \pi$ be c-coherent and $\pi \restrict \R$ a pivot policy.
 Then, for each $\ell \in \pi \restrict \R$, the sequential bang prefix $\bang{\ell} \col H \cseq P$ is c-coherent for $\pi$ if $H \subseteq \{ \ell' \mid \pi \Vdash \ell' \ordpre \ell\}$. 
\label{prop:seq-bang-coherent}
\end{proposition}
%---------------------------------------------

\begin{proof}
 By definition, the process $\bang{\ell} \col H \cseq P$ offers only a single initial transition
$ % \[ 
 \bang{\ell} \col H \cseq P 
 \mbox{$\Derives{\ell}{H}{P}$} 
 P \cpar \bang{\ell} \col H \cseq P
$ % \] 
 which must be obtained by application of rule $\RepR$.
 The conformance Def.~\ref{def:conformance} follows exactly as for action prefixes from $\ell \in \pi$ the assumption $H \subseteq \{ \ell' \mid \pi \Vdash \ell' \ordpre \ell\}$. 
 Since there are no diverging transitions for $\bang{\ell} \col H \cseq P$ the only situation to consider for confluence is that $\ell \not \in H$. For reconvergence it suffices to observe that the continuation process $P \cpar \bang{\ell} \cseq P$ offers another transition engaging with $\ell$ and the same blocking set $H$:
 $ %\[ 
 P \cpar \bang{\ell} \col H \cseq P 
 \mbox{$\Derives{\ell}{H}{P}$}
 P \cpar P \cpar \bang{\ell} \col H \cseq P.
 $ %\] 
 By structural equivalence, $P \astep{\ell} \zero \cpar P$. This weak environment shift is enough to satisfy \myb{c-coherence} Def.~\ref{def:coherence}, because of the deterministic transition of a non-clock. 
 Finally, note that since $P$ is c-coherent for $\pi$ and \myb{c-coherence} for pivot policies closed under parallel composition (Prop.~\ref{prop:cpar-coherent}), we can assume that the continuation process $P \cpar \bang{\ell} \col H \cseq P$ is conformant to $\pi$, by co-induction. 
\end{proof}

\subsubsection{Hiding}
\label{sec:hiding}

%---------------------------------------------
\begin{proposition}
 If $Q$ is c-coherent for $\pi$ and $\pi \nVdash \sigma \ordpre \ell$ for all $\sigma \in L$, then $Q \hide L$ is c-coherent for $\pi$. 
\label{prop:hiding-coherent} 
\end{proposition}
%---------------------------------------------

\begin{proof}
 It suffices to prove the proposition for the special case $Q \hide \sigma$ of a single clock. We recall the semantic rule for this case: 
 \[
 \begin{array}{l@{\hspace{2cm}}l}
 \multicolumn{2}{l}{
 \infer[(\HideR)]
 {P \hide \sigma \Derives{\alpha \hide \sigma}{H'}{E \hide \sigma} Q \hide \sigma}
 {P \Derives{\alpha}{H}{E} Q & H' = H - \{ \sigma \}}} 
 \end{array}
 \]
 where $\sigma\hide\sigma = \tau$ and $\alpha\hide\sigma = \alpha$ if $\alpha \neq \sigma$.
 Let $P = Q \hide \sigma$ and the two non-interfering and diverging transitions 
 $ %\[ 
 P \Derives{\alpha_1\hide\sigma}{H_1}{E_1} P_1
 \text{ and }
 P \Derives{\alpha_2\hide\sigma}{H_2}{E_2} P_2.
$ % \]
 arise by $(\HideR)$ from transitions 
% \begin{eqnarray} 
 $ Q \Derives{\alpha_1}{N_1}{F_1} Q_1 
 \text{ and } 
 Q \Derives{\alpha_2}{N_2}{F_2} Q_2$
 %\label{eqn:hide-Q-diverge}
 %\end{eqnarray}
 with $E_i = F_i \hide \sigma$, $P_i = Q_i \hide \sigma$ and $H_i = N_i - \{ \sigma \}$.
For \myb{c-coherence} Def.~\ref{def:coherence} we assume that $\alpha_1 \hide \sigma \neq \alpha_2 \hide \sigma$ or $P_1 \neq P_2$, or $\{ \alpha_1 \hide \sigma, \alpha_2 \hide \sigma \} \subseteq \R$ and $\alpha_i \hide \sigma \not\in H_{3-i}$. Further, let the c-actions $\alpha_i \col H_i[E_i]$ be interference-free. Observe that the first case implies $\alpha_1 \neq \alpha_2$ and the second case means $Q_1 \neq Q_2$. Moreover, because of the assumption that $\pi \nVdash \sigma \ordpre \alpha_{i}$ we must have $\sigma \not \in N_1 \cup N_2$. But this means $H_i = N_i$. First, suppose that $\alpha_1 = \alpha_2 = \sigma$. But then the assumption $Q_1 \neq Q_2$ contradicts strong determinacy of clocks (Proposition~\ref{prop:action-determinism}). Thus, the actions $\alpha_i$ cannot both be the clock $\sigma$ that is hidden. The case $\alpha_i = \sigma$ and $\alpha_{3-i} \in \Act - \{\sigma\}$, for some fixed $i \in \{1,2\}$ can also be excluded as follows: By Prop.~\ref{prop:clock-interference} the clock transition and the non-clock transition must interfere in $Q$, i.e., $\alpha_{i} \in N_{3-i}$ or $\alpha_{3-i} \in N_i$. The former $\sigma = \alpha_{i} \in N_{3-i}$ is impossible because of the above. The latter is outright impossible since then $\alpha_{3-i}\hide\sigma = \alpha_{3-i} \in N_i = H_i$ which contradicts the non-interference assumption on the transitions of $P$. 
 
The only remaining case to handle is that both of the diverging transitions are by labels $\alpha_1, \alpha_2 \in \Act - \{\sigma\}$ and $\alpha_i\hide\sigma = \alpha_i$ for both $i \in \{1,2\}$. We claim that the c-actions $\alpha_i \col N_i[F_i]$ are interference-free. If $\alpha_1 \neq \alpha_2$ then non-interference assumption directly gives $\alpha_i \not\in H_i = N_i$. Next, if for some $i \in \{1,2\}$, $\alpha_i = \alpha_i\hide\sigma = \tau$, then %non-interference 
 the assumptions on $P$ not only give $\tau \not\in H_{3-i} = N_{3-i}$ but also $N_{3-i} \cap \olwilA{\ast}(F_{3-i} \hide \sigma) = H_{3-i} \cap \olwilA{\ast}(E_{3-i}) = \eset$. 
 Since by assumption on policy conformance $\sigma \not\in N_{3-i}$ the latter implies $N_{3-i} \cap \olwilA{\ast}(F_{3-i}) = \eset$. Now we have shown that the diverging transitions
 %~\eqref{eqn:hide-Q-diverge} 
 of $Q$ are non-interfering  we can use the \myb{c-coherence} Def.~\ref{def:coherence} for $Q$ by induction hypothesis and get $F_1'$, $F_2'$, $Q'$ such that 
$ %\[ 
 Q_1 \Derives{\alpha_2}{N_2'}{F_2'} Q'
 \text{ and }
 Q_2 \Derives{\alpha_1}{N_1'}{F_1'} Q'
 \text{ with }
 F_1 \astep{\alpha_2} F_1' 
 \text{ and } 
 F_2 \astep{\alpha_1} F_2' 
 $ %\] 
 such that $N_i' \subseteq N_i$. By applying $(\HideR)$ to these transitions, we obtain 
 $ %\[ 
 P_1 \Derives{\alpha_2\hide\sigma}{H_2''}{F_2' \hide \sigma} Q' \hide \sigma
 \text{ and }
 P_2 \Derives{\alpha_1\hide\sigma}{H_1''}{F_1' \hide \sigma} Q' \hide \sigma
 $ %\] 
 where $H_i'' = N_i' - \{ \sigma \} \subseteq N_i - \{ \sigma \} = H_i$. Note, by rule $(\HideR)$ we also have 
 $ %\[ 
 F_1 \hide \sigma \astep{\alpha_2\hide\sigma} F_1' \hide \sigma 
 \text{ and }
 F_2 \hide \sigma \astep{\alpha_1\hide\sigma}F_2' \hide \sigma.
 $ %\] 
 In the special case that $\{ \alpha_1\hide\sigma, \alpha_2\hide\sigma \} = \{ \alpha_1, \alpha_2 \} \subseteq \R$ and $\alpha_1 \neq \alpha_2$ or $Q_1 \hide \sigma \neq Q_2 \hide \sigma$, or $\{ \alpha_1\hide\sigma, \alpha_2\hide\sigma \} \cap \C = \{ \alpha_1, \alpha_2 \} \cap \C \neq \eset$ we obtain the stronger context shifts
 $ %\[
 F_1 \Derives{\alpha_2}{}{} F_1' 
 \text{ and } 
 F_2 \Derives{\alpha_1}{}{} F_2' 
 $ %\] 
 from which we infer
$ % \[ 
 F_1 \hide \sigma \Derives{\alpha_2\hide\sigma}{}{} F_1' \hide \sigma 
 \text{ and }
 F_2 \hide \sigma \Derives{\alpha_1\hide\sigma}{}{} F_2' \hide \sigma. 
$ % \] 
This completes the proof. 
\end{proof}

\subsubsection{Restriction}
\label{sec:restriction}

Let us next look at the restriction operator. Again, the fact that $P$ is c-coherent does not imply that $P \restrict L$ is c-coherent. Consider the process $Q \eqdef (s + a \col s \cpar b \cseq \ol{s}) \restrict s$ which is conformant and pivotable. The signal $s$ blocks the action $a$ by priority in the thread $s + a \col s$. For the communication partner $\ol{s}$ to become active, however, it needs the synchronisation with an external action $\ol{b}$. The transition
%
%\begin{eqnarray}
$ Q \Derives{a}{\eset}{(\one \cpar b \cseq \ol{s})\restrict s} (b \cseq \ol{s}) \restrict s$
%\label{eqn:ex-restrict}
%\end{eqnarray}
%
generated by $(\RestrR)$ does not provide enough information in the blocking set $H$ and environment $E$ for us to be able to characterise the environments in which the action $a$ is blocked. For instance, the parallel composition $Q \cpar \ol{b}$ will internally set the sender $\ol{s}$ free and thus block the $a$-transition. From the above transition
%~\eqref{eqn:ex-restrict} 
the $(\ParR_3)$ rule permits $Q \cpar \ol{b}$ to offer the $a$-step. This is not right as it is internally blocked and does not commute with the $b \cpar \ol{b}$ reduction. A simple solution to avoid such problems and preserve confluence is to force restricted signals so they do not block any visible actions. 
This restriction suffices in many cases.

As we have noted above, \ccs corresponds to the unclocked and (blocking) free processes of \ccslm (Prop.~\ref{prop:free-ccs}). For these processes, admissible transitions and strongly enabled transitions coincide. 
The key result of Milner~\cite{Milner:CCS} (Chap.~11, Prop.~19) is that confluence is preserved by inaction $\zero$, prefixing $\ell \cseq P$ and 
\textit{confluent composition} defined as 
%\begin{eqnarray} 
$ P \mathrel{|_L} Q \eqdef (P \cpar Q) \restrict L$
%\label{eqn:confluent-composition}
%\end{eqnarray} 
where $\L(P) \cap \L(Q) = \eset$ and $\ol{\L(P)} \cap \L(Q) \subseteq L \cup \ol{L}$. 
We will show how confluence for this fragment of \ccs processes follows from our more general results on \myb{c-coherence}. From now on, for the rest of this section, we assume that all processes are unclocked, i.e., each prefix $\ell \col H \cseq Q$ has $H \cup \{ \ell \}\subseteq \A \cup \coA$. 

Let us call a process $P$ \textit{discrete} if all prefixes occurring in $P$ are of form $\ell\col\{\ell\}\cseq Q$, i.e., they are self-blocking. Under strong enabling, a discrete process behaves like a \ccs process with the restriction that each action $\ell$ only synchronises if there is a matching partner $\ol{\ell}$ in a \textit{unique} other parallel thread. If the matching partner is not unique, then the scheduling blocks. For instance, $\ell \col \ell \cpar \ol{\ell}\col\ol{\ell} \cpar \ol{\ell}\col\ol{\ell}$ blocks because there are two concurrent senders $\ol{\ell}$ matching the receiver $\ell$. 
Now consider the fragment $\ccscc$ of discrete processes built from inaction $\zero$, self-blocking prefixes $\ell\col\{\ell\}\cseq Q$ and confluent composition.
%equation ~\eqref{eqn:confluent-composition}. 
Milner's result says that the processes in $\ccscc$ are confluent under admissible scheduling. 

To emulate Milner's result in our setting, we consider the \textit{sort} $\L(P)$ of a process $P$ as the admissible actions of a \textit{discrete} policy $\pi_P$ with only reflexive precedences. In other words, $\pi_P \Vdash \ell_1 \ordpre \ell_2$ iff $\ell_1 = \ell_2$.
It is easy to see that a discrete process always conforms to $\pi_P$. 
Also, discrete policies are always pivot policies (Def.~\ref{def:pivot-policy}) and dependency-closed for all label sets (Def.~\ref{def:dep-closed}), as one shows easily. Thus, discrete processes, which are c-coherent for discrete policies, fullfill all conditions of the preservation laws for inaction $\zero$ (Prop.~\ref{prop:zero-coherent}), self-blocking action prefix $\ell\col\{\ell\}\cseq Q$ (Prop.~\ref{prop:channel-prefix-coherent}), parallel composition $P \cpar Q$ (Prop.~\ref{prop:cpar-coherent}) and restriction $P \restrict L$ (Prop.~\ref{prop:restrict-coherent}).
As a consequence and in particular, all processes in $\ccscc$ are c-coherent under strong enabling by our results.

Our final observation now is that, for Milner's fragment $\ccscc$, the transition semantics under admissible and strong enabling, as well as the notions of confluence and \myb{c-coherence}, coincide.
Firstly, if $P$ is discrete then in every c-action $\ell \col H[E]$ executed by a derivative $Q$ of $P$ we must have $H = \{ \ell \}$. Thus, the \myb{c-coherence} Def.~\ref{def:coherence}(1) becomes redundant since it is always satisfied. 
Further, any two c-actions $\ell_i \col H_i[E_i]$ of transitions 
$ %\[ 
 Q \Derives{\ell_1}{H_1}{E_1} Q_1 
 \text{ and }
 Q \Derives{\ell_2}{H_2}{E_2} Q_2 
$ %\] 
for $\ell_1 \neq \ell_2$ or $Q_1 \neq Q_2$ are trivially interference-free. The confluent composition
%~\eqref{eqn:confluent-composition} 
of Milner has the further effect that for every c-action $\ell \col H[E]$ we must have $\ell \not\in \olwiA(E)$. This is proven by induction on the structure of $P \in \ccscc$. Further, the side condition of $\ParR_3$ becomes trivial, and thus generally $\tau \not\in H$. This is a result of the side-conditions of 
the confluent composition.
%equation~\eqref{eqn:confluent-composition}. 
Hence, for $\ccscc$, strong enabling coincides with plain admissibility \ccs. 
Notice that the class of discrete processes and the class processes c-coherent for discrete policies is larger than $\ccscw$. In particular, we can explain \myb{c-coherence} of broadcast actions with repetitive prefixes. Hence, our results are more general even for the restricted class of discrete behaviours. This will be exploited with the examples discussed in a proper section.
%
%
%---------------------------------------------
\begin{proposition}
 If $Q : \pi$ and $\pi$ is precedence-closed for $L \cup \ol{L}$, then $({Q}\restrict{L}) : \pi - L$. 
\label{prop:restrict-coherent} 
\end{proposition}
%---------------------------------------------
%
\begin{proof}
It suffices to prove the proposition for the special case $P = Q \restrict a$ for a single channel name $a \in A$. Recall the rule $(\RestrR)$ for this case: 
 \[
 \begin{array}{l@{\hspace{2cm}}l}
 \multicolumn{2}{l}{
 \infer[(\RestrR)]
 {Q \restrict a \Derives{\alpha}{H'}{E \restrict a} R \restrict a}
 {Q \Derives{\alpha}{H}{E} R %& L' = \{ a, \ol{a} \} 
 & \alpha \not\in \{ a, \ol{a} \} &
 H' = H - \{ a, \ol{a} \}}
 }
 \end{array}
 \] 
 Consider a transition
 $ %\[ 
 Q \restrict a \Derives{\ell}{H'}{E \restrict a} Q' \restrict a
 \text{ derived from }
 Q \Derives{\ell}{H}{E} Q'
 $ %\] 
 with $H' = H - \{ a, \ol{a} \}$, $\ell \not\in \{ a, \ol{a} \}$ and $\ell' \in H'$.
 Then $\ell' \not\in \{ a, \ol{a} \}$ and $\ell' \in H$. Conformance Def.~\ref{def:conformance} applied to $Q$ implies $\pi \Vdash \ell' \ordpre \ell$. But then $\pi - a \Vdash \ell' \ordpre \ell$, too, as both $\ell$ and $\ell'$ a distinct from $a$ and $\ol{a}$. This ensures conformance. If $P = Q \restrict a$ and the two non-interfering and diverging transitions 
$ %\[ 
 P_2  \LDerives{\alpha_2}{H_2}{E_2} P \Derives{\alpha_1}{H_1}{E_1} P_1
% P \Derives{\alpha_1}{H_1}{E_1} P_1
% \text{ and }
% P \Derives{\alpha_2}{H_2}{E_2} P_2
 $ %\]
arise from rule $(\RestrR)$ then we have diverging transitions
$Q_2 \LDerives{\alpha_2}{N_2}{F_2} Q \Derives{\alpha_1}{N_1}{F_1} Q_1$
with $H_i = N_i - \{ a, \ol{a} \}$, $\alpha_i \not\in \{ a, \ol{a} \}$, where $E_i = {F_i} \restrict a$ and $P_i = {Q_i}\restrict a$. We assume $\alpha_1 \neq \alpha_2$ or $P_1 \neq P_2$ or 
$\{ \alpha_1, \alpha_2 \} \subseteq \R$ and $\alpha_i \not \in H_{3-i}$. 
Further, suppose non-interference of $\alpha_1 \col H_1[E_1]$ and $\alpha_2 \col H_2[E_2]$. Note that in case $P_1 \neq P_2$ we immediately have $Q_1 \neq Q_2$. Next, notice that the premise of the $(\RestrR)$ rule ensures $\alpha_i \not\in \{ a, \ol{a} \}$. It is easy to see that this means $\{ a, \ol{a} \} \cap N_{i} = \eset$, i.e., $H_i = N_i$. By contraposition, suppose $\ell \in \{a, \ol{a} \}$ and $\ell \in N_i$. Then $\pi \Vdash \ell \ordpre \alpha_i$. But $\pi$ is precedence-closed for $\{a, \ol{a}\}$, so we would have $\alpha_i \in \{a,\ol{a}\}$ which is a contradiction. So, $H_i = N_i$.
 
 All this means that to invoke \myb{c-coherence} Def.~\ref{def:coherence} of the above $Q$
 %on reduction~\eqref{eqn:Q-trans} 
 is sufficent to show that the c-actions $\alpha_1 \col N_1[F_1]$ and $\alpha_2 \col N_2[F_2]$ do not interfere. Now, if $\alpha_1 \neq \alpha_2$ then from the non-interference assumption on the transitions out of $P$ we get $\alpha_i \not\in H_{3-i} = N_{3-i}$. This is what we need for Def.~\ref{def:interference-free}(1) to show that the c-actions $\alpha_i \col N_i[F_i]$ are interference-free.
For the second part Def.~\ref{def:interference-free}(2) suppose $\alpha_{3-i} = \tau$ for some $i \in \{1,2\}$. Then the assumed non-interference of the transitions out of $P$ means $H_i \cap (\olwilA{\ast}(E_i) \cup \{ \tau \}) = \eset$, which is the same as $\tau \not \in H_i$ and the following set condition 
 $(N_i - \{ a, \ol{a} \}) \cap \olwilA{\ast}(F_i \restrict a) = \eset$.
%%
% %
 Since $\{ a, \ol{a} \} \cap N_{3-i} = \eset$ we have 
 $N_i \cap \olwilA{\ast}(F_i) \subseteq \L - \{ a, \ol{a} \}$. 
 We claim that the $\alpha_i$-transitions of $Q$ is enabled, specifically $N_i \cap \olwilA{\ast}(F_i) = \eset$ observing that $\tau \not \in H_i = N_i$. 
 By contraposition, suppose $\beta \in N_i \cap \olwiA(F_i) \subseteq \L - \{ a, \ol{a} \}$. Hence, $\beta \not\in \{ a, \ol{a} \}$ and so $\beta \in N_i - \{ a, \ol{a} \}$ and also $\beta \in \olwilA{\ast}(F_i \restrict a)$, but this contradicts the above set-condition.
%~\eqref{eqn:aux-non-inter}. 
 Thus, we have shown that for all $i \in \{1,2\}$, $\alpha_{3-i} = \tau$ implies $N_i \cap \olwilA{\ast}(F_i) = \eset$. This proves the c-actions $\alpha_i \col N_i[F_i]$ out of $Q$ are interference-free. Hence, we can use \myb{c-coherence} of $Q$ for the diverging transition
 %~\eqref{eqn:Q-trans} 
 and get $F_1'$, $F_2'$, $Q'$ 
$ % \[ 
 Q_1 \Derives{\alpha_2}{N_2'}{F_2'} Q'
 \text{ and }
 Q_2 \Derives{\alpha_1}{N_1'}{F_1'} Q'
 \text{ with }
 F_1 \astep{\alpha_2} F_1' 
 \text{ and } 
 F_2 \astep{\alpha_1} F_2' 
$ % \] 
 such that $N_i' \subseteq N_i$. We construct the required reconvergence by application of rule $(\RestrR)$:
 $ %\[ 
 P_1 \mbox{\Derives{\alpha_2}{H_2''}{E_2'}} Q' \restrict a
 \text{ and }
 P_2 \Derives{\alpha_1}{H_1''}{E_1'} Q' \restrict a
 \text{ with }
 F_1 \restrict a \astep{\alpha_2} F_1' \restrict a
 \text{ and } 
 F_2 \restrict a \astep{\alpha_1} F_2' \restrict a 
 $ %\] 
 with $E_i' = F_i' \restrict a$ and $H_i'' = N_i' - \{ a, \ol{a} \} \subseteq N_i - \{ a, \ol{a} \} = H_i$.
 In the special case that $\{ \alpha_1, \alpha_2 \} \subseteq \R$ and $\alpha_1 \neq \alpha_2$ or $P_1 \neq P_2$ or $\{ \alpha_1, \alpha_2 \} \cap \C \neq \eset$ we can rely on the stronger context shifts
 $ %\[
 F_1 \Derives{\alpha_2}{}{} F_1' 
 \text{ and } 
 F_2 \Derives{\alpha_1}{}{} F_2' 
 $ %\]
 to conclude
$ % \[ 
 F_1 \restrict a \Derives{\alpha_2}{}{} F_1' \restrict a
 \text{ and }
 F_2 \restrict a \Derives{\alpha_1}{}{} F_2' \restrict a. 
$ % \] 
This completes the proof.
\end{proof}
}

%%% Local Variables:
%%% mode: latex
%%% TeX-master: "synpatick-lipics.tex"
%%% End:
 % clean

%!TEX root = synpatick-lipics.tex

%\vspace{-4mm}
%---------------------------------------
\section{Further work}
\label{sec:end}
%---------------------------------------
%
%\vspace{-2mm}

We list in spare order some possible future works: 
\begin{itemize}
\item We plan to study the interaction of our notion of c-coherence with forms of congruence (static, dynamic and interaction laws in Milner's jargon) and  
notions of bisimulation.
\item We plan to generalise the interaction model of \ccslm to synchronisation algebras permitting action fusion in the spirit of Milner's ASCCS and Austry and Boudol's Meije (see~\cite{SCCS,Meije}).
%Develop a standard concurrent form for \ccslm.
\item Thm.~\ref{thm:summary-coherence-closure} suggests using policies as types. Among possible type disciplines we just mention: behavioural type systems, including session-types to capture rendezvous interactions; intersection-types, including non-idempotent ones to characterise termination and resources awareness, and quantitative-types, including linear-types, to model consuming resources.
\item C-actions are somewhat related with continuation passing style, and, in some sense they are more powerful because they permit a ``forecast'' analysis of all possible future actions: studying this relation, possibly in relation with abstract interpretation techniques, would be worth of study.
\item Clocks are unordered in \ccslm: introducing a partial order could be an interesting, although economic and fine grained, way to reason about time and time priorities. 
\end{itemize}

%%% Local Variables:
%%% mode: latex
%%% TeX-master: "synpatick-lipics.tex"
%%% End: % clean

\newpage 
\pagebreak

\bibliography{references-tick} 

\appendix

%!TEX root = synpatick-lipics.tex
%

%--------------------------------------------------
\section{Auxiliary Results}\label{auxiliary}
%--------------------------------------------------

In the following we list some simple observation on the properties of $\iA(P)$, $\wilA{\tau}(P)$ and $\wilA{\ast}(P)$ that will be used in the proofs.

 \begin{lemma} \mbox{}
 Let $P$, $Q$ be processes: 
 \begin{bracketenumerate}
 \item $P \Derives{\sigma}{H}{R} Q$ implies $R \equiv \zero$.
 \item $\alpha \in \iA(P)$ iff $P \Derives{\alpha}{}{} Q$.  
 \item $\iA(P) \cap \R \subseteq \iA(P \cpar Q)$.
 \item If $P$ is single-threaded, then $P \Derives{\alpha}{H}{R} Q$ implies $\alpha \in \L$, $R \equiv \zero$ and $\tau \not\in H$. 
 \item If $P \Derives{\alpha}{H}{R} Q$ then $\iA(R) \subseteq \iA(P)$.
\end{bracketenumerate}
\label{lem:prio-basic-1}
\end{lemma}

\begin{proof}
  (2) is by definition. The other points follow by an inspection of the SOS rules, using the definition of $\iA(P)$.
\end{proof}
Lem.~\ref{lem:prio-basic-1}(1) means that clock actions always run in an empty concurrent environment.  Observe that Lem.~\ref{lem:prio-basic-1}(2) means that the set $\iA(P)$ corresponds to the initial actions of $P$.  
The presence of $\tau$ in the blocking set of a transition
%~(\ref{eqn:prio-micro-step}) 
provides information about the set of initial actions of a process. Specifically, in rule $(\ComR)$ the presence of $\tau$ in the set $H = race(P\cpar Q)$
indicates a blocking situation arising from the priorities in $P$ or $Q$ conflicting with the rendevous actions $\ell$ or $\ol{\ell}$. Thus, the condition $\tau\not\in H$ can be used to implement priority-based scheduling.
Lem.~\ref{lem:prio-basic-1}(4) implies that single-threaded processes are never blocked.
\begin{lemma}\mbox{}
 Let $P$ be an arbitrary process: 
 \begin{bracketenumerate}
 \item $\iA(P) \cap \L \subseteq \wiA(P) \subseteq \wilA{\ast}(P)$.
 \item If $P \equiv Q$ then $\wilA{\ast}(P) = \wilA{\ast}(Q)$.
 \item If $P \Derives{\alpha}{}{} Q$ for $\alpha \in \R \cup \{ \tau \}$, then $\wilA{\ast}(Q) \subseteq \wilA{\ast}(P)$.
 % and if $\alpha = \tau$ then $\wiA(P') \subseteq \wiA(P')$ 
 \end{bracketenumerate}
\label{lem:c-enabled-aux}
\end{lemma}
\begin{proof}
 All points follow by an inspection of the SOS rules, using the definition of $\wiA(P)$ and $\wilA{\ast}(P)$.
\end{proof}

\paragraph*{On Interference}
%---------------------------------------
\begin{proposition}
  Suppose $\alpha_1 \col H_1[E_1 \cpar R]$ and $\alpha_2 \col H_2[E_2 \cpar R]$ are interference-free. Then, $\alpha_1 \col H_1'[E_1]$ and $\alpha_2 \col H_2'[E_2]$ are interference-free for arbitrary subsets $H_1' \subseteq H_1$ and $H_2' \subseteq H_2$. 
\label{prop:interference}
\end{proposition}
%---------------------------------------

As a special case, choosing $R = \zero$, Proposition~\ref{prop:interference} implies non-interference of $\alpha_i \col H_i'[E_i]$ for $i \in \{1,2\}$ from non-interference of $\alpha_i \col H_i[E_i]$ for any $H_i' \subseteq H_i$.

\begin{proof}% [Of Prop.~\ref{prop:interference}]
 \mrev{First note that if $H_i' \subseteq H_i$, the assumption $\alpha_i \not\in H_{3-i}$ implies also $\alpha_i \not\in H_{3-i}'$. Second, suppose that $H_i \cap \olwilA{\ast}(E_i \cpar R) = \eset$. Since by Lem.~\ref{lem:prio-basic-1}, $\olwilA{\ast}(E_i) \subseteq \olwilA{\ast}(E_i \cpar R)$, the inclusions $H_i' \subseteq H_i$ preserve this property, i.e., $\tau \not\in H_i'$ and $H_i' \cap (\olwilA{\ast}(E_i) \cup \{\tau\}) \subseteq H_i \cap (\olwilA{\ast}(E_i \cpar R) \cup \{\tau\})$. 
 Thus, all three conditions of non-interference Def.~\ref{def:interference-free} are preserved by changing from $H_i$ to $H_i'$.}{simplified argument}
 %  First, non-interference of $\alpha_i \col H_i[E_i \cpar R]$ means $\alpha_i \not\in H_{3-i}$ by assumption. Since $H_i' \subseteq H_i$, this implies also $\alpha_i \not\in H_{3-i}'$. Second, suppose that $\alpha_{3-i} = \tau$ for some $i \in \{1,2\}$. Then, non-interference of $\alpha_1 \col H_1[E_1\cpar R]$ and $\alpha_2 \ol H_2[E_2\cpar R]$ gives $\tau \not\in H_i$ and $H_i \cap \olwilA{\ast}(E_i \cpar R) = \eset$. 
 %  Since by Lem.~\ref{lem:prio-basic-1}, $\olwilA{\ast}(E_i) \subseteq \olwilA{\ast}(E_i \cpar R)$, the inclusions $H_i' \subseteq H_i$ preserve this property, i.e., $\tau \not\in H_i'$ and $H_i' \cap \olwilA{\ast}(E_i) \subseteq H_i \cap \olwilA{\ast}(E_i \cpar R) = \eset$.
\end{proof}

\paragraph*{On Policies}

We obtain a partial ordering $\pi_1 \preceq \pi_2$ on precedence policies by subset inclusion. Specifically, we have $(L_1, \sordpre{1}) \preceq (L_2, \sordpre 2)$ if $L_1 \subseteq L_2$ 
such that for all $\ell_1, \ell_2 \in L_1$, 
if $\ell_1 \sordpre 2 \ell_2$ then $\ell_1 \sordpre 1 \ell_2$. 
Intuitively, if $\pi_1 \preceq \pi_2$ then $\pi_2$ is a tightening of $\pi_1$ in the sense that it exports more labels (resources) subject to possibly fewer precedences (causality constraints) than $\pi_1$.
By way of this ordering, we can use policies to measure the degree of concurrency exhibited by a process. 
The alphabet of the $\preceq$-maximal policy $\pi_{max}$ contains all labels $\L$ and its precedence relation $\sordpre{max}$ is empty. 
%$\ell_1 \ordpre \ell_2 \in \pi_0$ for all $\ell_1, \ell_2 \in \L$.
It is the largest policy in the $\preceq$ ordering. 
The $\preceq$-minimal policy $\pi_{min}$ has empty alphabet and empty precedences. It is self-dual, i.e., $\ol{\pi}_{min} = \pi_{min}$. 

\medskip 

There is also the normal inclusion ordering $\pi_1 \subseteq \pi_2$ between policies defined by $L_1 \subseteq L_2$ and $\sordpre 1 \subseteq \sordpre 2$. It differs from $\preceq$ in the inclusion direction of the precedence alphabets. 

\medskip 

For any policy $\pi$ on alphabet $L$ we define its \textit{dual} as a policy $\ol{\pi}$ on the set of labels $\ol{L}$ 
%$\ol{\pi} = (\ol{L}, \lordpre{-})$ 
where for all $\ol{\ell}_1, \ol{\ell}_2 \in \ol{L}$ we stipulate $\ol{\ell}_1 \ordpre \ol{\ell}_2 \in \ol{\pi}$ iff $\ell_1 = \ell_2$ or $\ell_{1} \indep{} \ell_2 \in \pi$, i.e., both $\ell_{1} \ordpre \ell_2 \not\in \pi$ and $\ell_{2} \ordpre \ell_1 \not\in \pi$. Intuitively, the dual policy $\ol{\pi}$ consists of matching copies $\ol{\ell}$ of all labels $\ell$ from $\pi$ and puts them in a precedence relation precisely if their original copies are identical or do not stand in precedence to each other, in any direction. This generates a form of ``negation'' $\ol{\pi}$ from $\pi$, which is reflexive and symmetric. 
%It is easy to see that if $\ol{\ell}_1 \ordpre \ol{\ell}_2 \ol{\pi}$ then also $\ol{\ell}_2 \ordpre \ol{\ell}_1 \in \ol{\pi}$.
Note that $\ol{\ol{\pi}}$ need not be identical to $\pi$ 
%unless it is symmetric, 
but contains more precedences. In general, we have $\ol{\ol{\pi}} \preceq \pi$ and $\ol{\ol{\pi}} = \pi$ if $\pi$ is symmetric. The policy $\ol{\pi}$ captures the independences between actions of $\pi$ and is called the associated \textit{independence alphabet}, in the sense of trace theory~\cite{DieckertR95}. Analogously, $\ol{\ol{\pi}}$ expresses the dependencies of actions in $\pi$, called the associated \textit{dependence alphabet}. 
             
\begin{example}
  Suppose $P \cnf \ell_1 \indep{} \ell_2$ then all transitions with labels $\ell_1$ and $\ell_2$ of $P$ must be from concurrent threads and thus cannot block each other. For instance, both $\ell_1 \cpar \ell_2$ and $\ell_1 + \ell_2$ conform to $\ell_1 \indep{} \ell_2$. However, the choice $\ell_1 + \ell_2$ is not c-coherent. It becomes c-coherent if we add a precedence $\ell_1 + \ell_2 \col \ell_1$ but then it is not longer conformant to $\ell_1 \indep{} \ell_2$. Instead, we now have $\ell_1 + \ell_2 \col \ell_1 \cnf \ell_1 \ordpre \ell_2$, where the policy expresses the precedence coded in the blocking sets of the prefixes.
  More generally, if $\ell_1 \ordpre \ell_2 \in \pi$ and $\ell_1 \neq \ell_2$ are distinct, then $P \cnf \pi$ may take $\ell_1$ and $\ell_2$ transitions to completely different states from the same thread. 
  \longshort{\version}{}{Necessarily, these must be transitions combined in a prioritised sum, where one of $\ell_i$ blocks the other $\ell_{3-i}$.}
  %A special case occurs if $\ell_1 = \ell_2$. 
\end{example}

\begin{definition}[Pivot Policy]
 A policy $\pi$ is called a \emph{pivot} policy if $\ol{\pi} \preceq \pi$. 
 We call a process $P$ \emph{pivotable} if $P$ conforms to a pivot policy. 
\label{def:pivot-policy}
\end{definition}
\begin{proposition}
 A policy $\pi = (L, \ordpre)$ is a pivot policy if 
 $\ol{L} \subseteq L$ and for all $\ell_1, \ell_2 \in L$ with $\ell_1 \neq \ell_2$ we have $\ell_1 \indep{} \ell_2 \in \pi$ or $\ol{\ell}_1 \indep{} \ol{\ell}_2 \in \pi$.
\label{prop:pivot-policy}
\end{proposition}
 
\begin{proof}
Let $\pi$ be a pivot policy. Since the alphabet of $\ol{\pi}$ is $\ol{L}$ and the alphabet of $\pi$ is $L$ the assumption $\ol{\pi} \preceq \pi$ implies $\ol{L} \subseteq L$ by definition of $\preceq$. Given labels $\ell_1, \ell_2 \in L$ with $\ell_1 \neq \ell_2$ suppose $\ell_1 \indep{} \ell_2 \not\in \pi$. Hence $\ell_1 \ordpre \ell_2 \in \pi$ or $\ell_2 \ordpre \ell_2 \in \pi$. Consider the first case, i.e., $\ell_1 \ordpre \ell_2 \in \pi$. Since $\ol{L} \subseteq L$ and $L = \ol{\ol{L}}$ it follows that also $L \subseteq \ol{L}$, i.e., the labels $\ell_i$ are also in the alphabet of $\ol{\pi}$. But then by definition of the ordering and the fact that $\ell_i = \ol{\ol{\ell}}_i$, the assumption $\ol{\pi} \preceq \pi$ implies that $\ol{\ol{\ell}}_1 \ordpre \ol{\ol{\ell}}_2 \in \ol{\pi}$. Now, the definition of the dual policy means that $\ol{\ell}_1 \ordpre \ol{\ell}_2 \not\in \pi$ and $\ol{\ell}_2 \ordpre \ol{\ell}_1  \not\in  \pi$. This is the same as $\ol{\ell}_1 \indep{} \ol{\ell}_2  \in  \pi$ as desired.

Suppose the properties (1) and (2) of the proposition hold for a policy $\pi$. We claim that $\ol{\pi} \preceq \pi$. The inclusion of alphabets is by assumption (1). Then, suppose $\ol{\ell}_1 \ordpre \ol{\ell}_2 \in \pi$ which implies 
 $\ol{\ell}_1 \indep{} \ol{\ell}_2 \not\in \pi$. If $\ell_1 = \ell_2$ we have $\ol{\ell}_1 \ordpre \ol{\ell}_2 \in \ol{\pi}$ directly by definition. So, let $\ell_1 \neq \ell_2$. 
 Because of $\ol{L} \subseteq L$ and (2) this gives us $\ell_1 \indep{} \ell_2 \in \pi$ considering again that $\ell_i = \ol{\ol{\ell}}_i$. 
 But $\ell_1 \indep{} \ell_2\in \pi$ is the same as $ \ol{\ell}_1 \ordpre \ol{\ell}_2 \in \ol{\pi}$ by definition. 
\end{proof}

%%GIGI INPUT-SCHEDULED RESURRECTED FROM APPENDIX D
An interesting special class of pivot policies are input-scheduled policies. 

\begin{definition}[Input-scheduled Processes] \mbox{}
  A policy $\pi$ is called
  \emph{input-scheduled} if $\pi \subseteq \pi_{is}$ where $\pi_{is} = (L_{is}, \sordpre{is})$ given by alphabet $L_{is} = \L$ and precedence relation $\ell_1 \ordpre \ell_2 \in \pi_{is}$ iff $\ell_1 = \ell_2$ or $\ell_1, \ell_2 \in \A \cup \C$.
  A process is \emph{input-scheduled} if $P$ conforms to an input-scheduled policy.
\label{def:is-policy}
\end{definition}

\begin{proposition}
 \mbox{Suppose $P$ is input-scheduled $\pi_{is}$ and $P \Derives{\alpha}{H}{R} Q$ for some $\alpha$, $H$, $R$, $Q$. Then:} 
 \begin{enumerate}
 \item If $\alpha \in \coA$ then $H \subseteq \{ \alpha, \tau \}$.
 \item If $\alpha \in \A \cup \C \cup \{ \tau \}$ then $H \subseteq \A \cup \C \cup \{ \tau \}$.
 \end{enumerate}
\label{prop:input-scheduled}
\end{proposition}

\begin{proof}
 Obvious by definition of conformance and the construction of $\pi_{is}$.
\end{proof}

The next crucial Lemma
%~\ref{lem:pivot-no-local-block-1} 
implies that two c-coherent threads cannot block each other if they are conformant to the same pivot policy. The blocking of a reduction $\ell \cpar \ol{\ell}$ must always arise from a thread in the concurrent context that is not involved in the synchronisation. \mrev{}{The following results were wrong, because the notion of coherence was missing the stronger residual property. Now with the corrected definition of c-coherence ok.} 
%pivotability is too weak, it seems. We cannot avoid the local race test in rule $(\ComR)$. As a counter example take $S_1 \cpar P$ where $S_1$ is the Esterel signal and $P = \ol{\pres}\cseq P_1 + \ol{\emit}\col\ol{\pres}\cseq P_2$. The rendevous $\emit\cpar\ol{\emit}$ of the parallel $S_1 \cpar P$ has a local race, contradicting the following Lemma.

%---------------------------------------------
\begin{lemma} \mbox{}
 Let $P_1 \cnf \pi$ and $P_2 \cnf \pi$ be c-coherent for the same pivot policy $\pi$ and $\ell \in \L$ such that
 $ %$$ 
 P_1 \Derives{\ell}{H_1}{R_1} P_1'
 \text{ and }
 P_2 \Derives{\ol{\ell}}{H_2}{R_2} P_2'.
 $ %$$ 
 Then, $H_2 \cap \oliA(R_1) = \eset$ implies $H_2 \cap \oliA(P_1) \subseteq \{\ol{\ell}\}$.
 % and $H_1 \cap \oliA(P_2) \subseteq \{\ell\}$.
\label{lem:pivot-no-local-block-1}
\end{lemma}
%----------------------------------------------

\begin{proof}%[Of Lem.~\ref{lem:pivot-no-local-block-1}]
 Let $\ell \in \L$ and the transitions be given as in the statement of the Lemma. Suppose $\alpha \in H_2$ and $\alpha \in \oliA(P_1)$. We claim that $\alpha = \ol{\ell}$, whence $H_2 \cap \oliA(P_1) \subseteq \{\ol{\ell}\}$.
 By contraposition let us assume that $\alpha \neq \ol{\ell}$, or equivalently $\ol{\alpha} \neq \ell$. Further,
 since $\iA(P_1) \subseteq \L$ we also have $\alpha \neq \tau$. The assumption $\alpha \in H_2$ means $\alpha \ordpre \ol{\ell} \in \pi$ by \mrev{Def.~\ref{def:conformance}}{wrong reference corrected} applied to $P_2$. Since $\pi$ is pivot we must thus have $\ol{\alpha} \indep{} \ell \in \pi$, specifically $\ol{\alpha} \ordpre \ell \not \in \pi$, i.e., $\ol{\alpha}\not\in H_1$, again by \mrev{Def.~\ref{def:conformance}}{wrong reference corrected}, now applied to $P_1$.
 But since $\ol{\alpha} \in \iA(P_1)$, by Lem.~\ref{lem:prio-basic-1}, there is a transition 
 $ %$$
 P_1 \Derives{\ol{\alpha}}{H}{E} Q.
 $ %$$
 Since also $\ell \ordpre \ol{\alpha} \not\in \pi$, by pivot property, it follows $\ell \not\in H$ by Def.~\ref{def:coherence}(1) again for $P_1$. Now both $\ol{\alpha} \not\in H_1$ and $\ell \not\in H$ mean that the c-actions $\ell \col H_1[R_1]$ and $\ol{\alpha}\col H[E]$ are interference-free (Def.~\ref{def:interference-free}). Since $\ell \neq \ol{\alpha}$ there must exist a reconvergent process $R$ and transitions 
$$
  P_1' \Derives{\ol{\alpha}}{H'}{E'} R
  \text{ and } 
  Q \Derives{\ell}{H_1'}{R_1'} R
  \text{ with } 
  R_1 \xrightarrow{\ol{\alpha}} R_1'
  \text{ and } 
  E \xrightarrow{\ell} E'.
$$
 The latter transitions are strengthenings of the residuals $R_1 \astep{\ol{\alpha}} R_1'$ and $E \astep{\ell} E'$ that must exist according to the statement of Coherence, because of $\{ \ell, \ol{\alpha} \} \subseteq \R$ and $\ell \neq \ol{\alpha}$. 
 %\mrev{}{This is not true. According to the definition above, for the residuals it would be enough to have $R_1 \equiv R_1'$ and $E \equiv E'$. Note that in the Technical Report the notion of Coherence is stronger, it forces the residuals to be strong, as claimed above!} 
 However, the transition $R_1 \xrightarrow{\ol{\alpha}} R_1'$ implies $\ol{\alpha} \in \iA(R_1)$ which is impossible, because then $\alpha \in H_2 \cap \oliA(R_1)$ while $H_2 \cap \oliA(R_1) = \eset$ by assumption. 
\end{proof}
We can show that pivotable processes do not have internal races. 

%---------------------------------------------
\begin{lemma} 
Let $P \cnf \pi$ be c-coherent for a pivot policy $\pi$. Then, for every reduction 
$ % $$
P \Derives{\tau}{H}{E} P',
$ %$$
if $H \cap \oliA(E) = \eset$ then $\tau \not\in H$.
\label{lem:pivot-no-local-block-2}
\end{lemma}
%----------------------------------------------
%
\begin{proof}
 Let $P$ be c-coherent for pivotable $\pi$ and a reduction
% \begin{eqnarray} 
 $P \Derives{\tau}{H}{E} P'$
% \label{eqn:tau-no-local}
 %\end{eqnarray}
 such that $H \cap \oliA(E) = \eset$ be given. Suppose, by contraposition, $\tau \in H$. Since there are no action prefixes $\tau \cseq P'$ in \ccslm, the reduction must arise from a synchronisation of two threads $P_1$ and $P_2$ inside $P$, via rule $(\ComR)$, i.e.,
 $$
 \infer[(\ComR)]
 {P_1 \cpar P_2 \oldDerives{\ell \cpar \ol{\ell}}{H_1 \cup H_2 \cup B}{R_1 \cpar R_2} P_1' \cpar P_2'}
 {P_1 \Derives{\ell}{H_1}{R_1} P_1' & P_2 \Derives{\ol{\ell}}{H_2}{R_2} P_2' 
 & B = \{ \ell \cpar \ol{\ell} \mid 
 H_2 \cap \oliA(P_1) \not\subseteq \{\ol{\ell}\} 
 \text{ or } 
 H_1 \cap \oliA(P_2) \not\subseteq \{\ell\} \}
 }
 $$
 for some $\ell \in \L$. The assumption $H \cap \oliA(E) = \eset$ for the top-level reduction implies that $(H_1 \cup H_2) \cap \oliA(R_1 \cpar R_2) = \eset$ at the point of $(\ComR)$. Here we may assume, by induction on the depth of the derivation, that $\tau \not\in H_1 \cup H_2$. 
 We must show that the synchronisation action $\tau$ cannot enter into the blocking set by $(\ComR)$, i.e., $\tau = \ell \cpar \ol{\ell} \not\in B$, i.e., which is the same as $H_2 \cap \oliA(P_1) \subseteq \{\ol{\ell}\}$ and $H_1 \cap \oliA(P_2) \subseteq \{\ell\}$.
 The sub-expressions $P_1 \cnf \pi'$ and $P_2 \cnf \pi'$ are also c-coherent for a common pivot policy $\pi'$. This policy may contain local labels that are restricted or hidden in the outer process $P$. Since by assumption we have 
 $H_1 \cap \oliA(R_1) \subseteq (H_1 \cup H_2) \cap \oliA(R_1 \cpar R_2) = \eset$ and likewise $H_2 \cap \oliA(R_2) = \eset$, we can apply Lem.~\ref{lem:pivot-no-local-block-1} to show $H_2 \cap \oliA(P_1) \subseteq \{\ol{\ell}\}$ and $H_1 \cap \oliA(P_2) \subseteq \{\ell\}$. This means $B = \eset$ and we are done.
\end{proof}

\paragraph*{\mrev{The Role of Pivotability}{Fixed}}

\newcommand{\uola}{\ensuremath{\ul{\ol{a}}}} 
\newcommand{\uolb}{\ensuremath{\ul{\ol{b}}}} 
\newcommand{\uolc}{\ensuremath{\ul{\ol{c}}}}
\newcommand{\ola}{\ensuremath{\ol{a}}} 
\newcommand{\olc}{\ensuremath{\ol{c}}}
\newcommand{\olb}{\ensuremath{\ol{b}}}

Let us give an example to show that the requirement of pivotability cannot be removed in Thm.~\ref{thm:church-rosser}, expressing preservation of c-coherence for parallel compositions. For compactness of notation let us write $\ul{\ell} \col H$ for $\ell \col (H \cup \{\ell\})$ to create reflexive prefixes. Also, we write $\ell^\ast\col H$ for the infinite loop $\ell^\ast \col H \pdef \ell\col H \cseq \ell^\ast \col H$, and $\ell^\ast$ for $\ell^\ast\col \eset$.
One can show that the process $\ell^\ast\col H$ is c-coherent for all $H$. 

\medskip 

Now, consider the process $R$, given by
$$
  R = R_0  
    \cpar \ul{\ol{b}}^\ast \cpar \ul{\ol{c}}^\ast  
  \text{ with }
  R_0 = (\ul{\ol{b}} \col \ol{a} \cseq R_1
    + \ul{\ol{c}} \col \ol{a} \cseq R_2) 
  \quad
  R_1 = \ul{\ol{a}} + \ul{\ol{c}} \col \ol{a}
  \text{ and }
  R_2 = \ul{\ol{a}} + \ol{\ul{b}} \col \ol{a},
$$
which can be shown to be c-coherent. First, note that $R$ generates four initial (c-enabled) transitions:

\begin{enumerate}
  \item $R \Derives{\olb}{\{ \ola,\olb \}}{\uolb^\ast\cpar\uolc^\ast} R_1 \cpar \uolb^\ast \cpar \uolc^\ast$
  \item $R \Derives{\uolc}{\{\ola,\olc\}}{\uolb^\ast\cpar\uolc^\ast} R_2 \cpar \uolb^\ast \cpar \uolc^\ast$ 
  \item $R \Derives{\olb}{\{\olb\}}{R_0 \cpar \uolc^\ast} R_0 \cpar \uolb^\ast \cpar \uolc^\ast$
  \item $R \Derives{\olc}{\{\olc\}}{R_0 \cpar \uolb^\ast} R_0 \cpar \uolb^\ast \cpar \uolc^\ast$.
\end{enumerate}

\noindent Among these, we must look at all pairs of interference-free transitions $R \Derives{\alpha_1}{H_1}{E_1} Q_1$ and $R \Derives{\alpha_2}{H_2}{E_2} Q_2$. In particular,  $\alpha_1 \neq \alpha_2$ or $Q_1 \not\equiv Q_2$ or both $\{ \alpha_1, \alpha_2 \} \subseteq \R$ and $\alpha_i \not\in H_{3-i}$ by Def.~\ref{def:interference-free}(3). If $\alpha_1 = \alpha_2$, because of the reflexive blocking of all prefixes, we must have $Q_1 \not\equiv Q_2$, which leaves two cases: the pair of transitions (1)+(3) and the pair of transitions (2)+(4). In both cases one shows that there is a reconvergence with strong residuals as required by c-coherence Def.~\ref{def:coherence}. If $\alpha_1 \neq \alpha_2$, then $\{\alpha_1, \alpha_2\} = \{ \olb, \olc\}$ which has three  possible cases in which $\alpha_i \not\in H_{3-i}$ as required by Def.~\ref{def:interference-free}(1), viz. the combinations (1)+(2), (1)+(4) and (2)+(3). Again, in each case one can establish the reconvergence with strong residuals. 

\medskip 

We remark that the sub-process $R_0$ by itself is not c-coherent as it is missing the strong residuals. By putting it in parallel with $\uolb^\ast \cpar \uolc^\ast$ this deficiency is fixed. 
Thus, the example shows that not every sub-expression of a c-coherent process needs to be c-coherent. Next we show that, without pivotability, c-coherence can also be lost by parallel composition.
Towards this end, we put $R$ in parallel with the c-coherent process
$$
 Q = b^\ast \cpar c^\ast \col a
$$
which has two admissible transitions
$$ 
    Q \Derives{b}{\eset}{c^\ast \col a } Q 
    \text{ and }
    Q \Derives{c}{\{a\}}{b^\ast} Q.
$$ 
Coherence follows from Thm.~\ref{thm:summary-coherence-closure}(5) and the fact that both sub-processes $b^\ast$ and $c^\ast \col a$ are c-coherent and conformant to the same pivot-policy with sole precedence $a \ordpre c$.
We show that the parallel composition $Q \cpar R$ is not c-coherent. We have admissible c-actions 
$$ 
  Q \cpar R \Derives{b \cpar \ol{b}}{\{ \ol{a}, \olb\}}{c^\ast \col a \cpar \uolb^\ast \cpar \uolc^\ast} Q \cpar R_1 \cpar \uolb^\ast \cpar \uolc^\ast
  \text{ and } 
  Q \cpar R \Derives{c \cpar \ol{c}}{\{ \ol{c}, a, \ol{a} \}}{b^\ast \cpar \uolb^\ast \cpar \uolc^\ast} Q \cpar R_2 \cpar \uolb^\ast \cpar \uolc^\ast
$$ 
which are in fact c-enabled. \mrev{Notice that the second synchronisation $c \cpar \ol{c}$ is not blocked by a local race in $(\ComR)$, since $a \not\in \oliA(R)$ and $\{\ola, \olc \} \cap \oliA(b^\ast \cpar \uolb^\ast \cpar \uolc^\ast) = \eset$}{added for clarification} However, the continuation process $Q \cpar R_1 \cpar \uolb^\ast \cpar \uolc^\ast = Q \cpar (\ul{\ol{a}} + \ul{\ol{c}} \col \ol{a}) \cpar \uolb^\ast \cpar \uolc^\ast$ on the left is binary blocked by a race condition. We find 
$$ 
  Q \cpar R_1 \cpar \uolb^\ast \cpar \uolc^\ast
   \Derives{c \cpar \ol{c}}{\{ a, \ol{c}, \ol{a}, \tau \}}{b^\ast \cpar \uolb^\ast \cpar \uolc^\ast} Q \cpar \uolb^\ast \cpar \uolc^\ast
$$ 
where the silent $\tau$ enters the blocking set, because $\{ a \} \cap \oliA(\ul{\ol{a}} + \ul{\ol{c}} \col \ol{a}) \not\subseteq \{ c \}$. This is a race computed in rule $(\ComR)$. Thus, the reduction $c \cpar \ol{c}$ of the original state $Q \cpar R$ is no longer c-enabled in the state $Q \cpar R_1 \cpar \uolb^\ast \cpar \uolc^\ast$. This breaks coherence.

\medskip 

Finally, we notice that although each of $Q$ and $R$ are pivotable, their composition $Q \cpar R$ is not pivotable any more. A common policy $\pi$ for $Q \cpar R$ must obviously include the precedences $\ol{a} \ordpre \ol{c} \in \pi$ so $Q$ is conformant, as well as $a \ordpre c \in \pi$ so $R$ is conformant. This is not pivot according to Def.~\ref{def:pivot-policy}. 

%--------------------------------------------------
\section{Proofs}\label{proofs}

\begin{proposition}
 If a c-coherent process $P$ offers \mrev{c-enabled}{term corrected} transitions on a clock $\sigma$ and another distinct action $\alpha \neq \sigma$ with 
$ % $$ 
 P \Derives{\alpha}{H_1}{R_1} P_1
 \text{ and } 
 P \Derives{\sigma}{H_2}{R_2} P_2
 $ %$$
 then 
 %$\ell \ordpre \sigma \in \pi$ or $\sigma \ordpre \ell \in \pi$. 
 $\alpha \in H_2$ or $\sigma \in H_1$. 
 In particular, if $\alpha = \tau$ then $\sigma \in H_1$. 
\label{prop:clock-interference}
\end{proposition}
\begin{proof}%[Of Prop.~\ref{prop:clock-interference}]
 Since $R_2 = \zero$ by Lem.~\ref{lem:prio-basic-1}(1) and $\iA(\zero) = \eset$, we have $\alpha \not \in \iA(R_2)$ and thus $\alpha$ cannot be an initial action 
 %$R_2 \fsstep{\alpha} R_2'$, 
 $R_2 \Derives{\alpha}{}{} R_2'$, 
 by Lem.~\ref{lem:prio-basic-1}(2). \mrev{This contradicts the existence of strong residuals which would be required because}{added explanation} of $\{ \alpha, \sigma \} \cap \C \neq \eset$. But then the coherence \mrev{Def.~\ref{def:coherence}}{corrected ref} implies that the c-actions $\alpha \col H_1[R_1]$ and $\sigma \col H_2[R_2]$ must interfere. 
 Since $\alpha \neq \sigma$ this implies $\alpha \in H_2$ or $\sigma \in H_1$, or $\alpha = \tau$ and $H_2 \cap (\olwilA{\ast}(R_2) \cup \{ \tau \}) \neq \eset$. Since $R_2 = \zero$ the latter reduces to $\alpha = \tau \in H_2$. 
 \mrev{This means that we have $\alpha \in H_2$ or $\sigma \in H_1$ in all cases.}{added explanation.}
 The last part of the Proposition now follows from the fact that the clock transition is \mrev{c-enabled}{corrected} and thus $\tau \not \in H_2$.
\end{proof}

\begin{proof}[Proof of Thm.~\ref{thm:church-rosser}]
 Let $P$ be c-coherent and $Q$ a derivative with c-enabled reductions 
$ % $$
Q \Derives{\tau}{H_1}{R_1} Q_1 \text{ and }$  $Q \Derives{\tau}{H_2}{R_2} Q_2.
$ %$$
If $Q_1 \equiv Q_2$ we are done immediately. Suppose $Q_1 \not\equiv Q_2$. Then c-enabling implies $\tau \not \in H_i$ as well as $H_i \cap \olwilA{\ast}(R_i) = \eset$, which implies $H_i \cap (\olwilA{\ast}(R_i) \cup \{ \tau \}) = \eset$. But then the c-actions $\tau \col H_1[R_1]$ and $\tau \col H_2[R_2]$ are interference-free, whence coherence \mrev{Def.~\ref{def:coherence}}{corrected ref} implies there exist $H_i'$ and processes $R_i'$, $Q'$ such that 
$ % $$ 
 Q_1 \Derives{\tau}{H_{2}'}{R_{2}'} Q'
 \text{ and }$ 
\mbox{$Q_2 \Derives{\tau}{H_{1}'}{R_{1}'} Q'$}
$\text{and }
 R_1 \astep{\tau} R_1' 
 \text{ and } 
 R_2 \astep{\tau} R_2'
 $ %$$
 with $H_{i}'\subseteq H_{i}$. 
 Finally, observe that $\olwilA{\ast}(R_i') \subseteq \olwilA{\ast}(R_i)$ by Lem.~\ref{lem:c-enabled-aux}(3) \mrev{and the definition of residual steps}{added explanation}, from which we infer 
$ %$$
H_i' \cap (\olwilA{\ast}(R_i') \cup \{ \tau \}) \subseteq 
 H_i \cap (\olwilA{\ast}(R_i) \cup \{ \tau \}).
 $ %$$
Thus, the reconverging reductions 
%$ Q_i \fsstep{\tau \col H_{3-i}'[R_{3-i}']} Q'$ 
$Q_i \Derives{\tau}{H_{3-i}'}{R_{3-i}'} Q'$
are again c-enabled. This was to be shown.
\end{proof}

\begin{proof}[Proof of Prop. \ref{prop:action-determinism}]
 \mrev{
 Suppose $P$ is c-coherent with 
 $% \[
 P \Derives{\sigma}{H_1}{R_1} P_1
 $ % \]
 and
 $ % \[
 P \Derives{\sigma}{H_2}{R_2} P_2.
 $ %\]
 %
 %we have .
 Suppose by contraposition, $P_1 \not\equiv P_2$.
 Then the c-actions $\sigma \col H_i[R_i]$ are trivially interference-free.   Then, c-coherence gives us a reconvergence and strong environment shifts implying $\sigma \in \iA(R_i)$ for $i = 1,2$. However, by Lem.~\ref{lem:prio-basic-1}(1) we have $R_i = \zero$ and thus $\iA(R_i) = \eset$. This is a contradiction. Thus, $P_1 \equiv P_2$ as claimed.
 }{corrected the proof; the original proof was for a hidden long version and did not fit the statement.}
 \end{proof}

%---------------------------------------------
\begin{lemma} \mbox{}
 Let $P_1 \cnf  \pi$ and $P_2 \cnf \pi$ be c-coherent for the same pivot policy $\pi$. If $P_1$ and $P_2$ are single-threaded, then 
 every admissible transition of $P_1 \cpar P_2$ is c-enabled. 
\label{lem:coherent-pivot-coincidence}
\end{lemma}
%----------------------------------------------

\begin{proof}
 Consider an admissible transition of $P_1 \cpar P_2$ where $P_1$ and $P_2$ are single-threaded. Firstly, by Lem.~\ref{lem:prio-basic-1}, each of the single-threaded processes can only generate c-actions $\ell_i \col H_i[R_i]$ where $\tau \not\in H_i$ and $R_i = \zero$. Such transitions are always c-enabled for trivial reasons. Further, any rendez-vous synchronisation of such an $\ell_1 \col H_1[R_1]$ and $\ell_2 \col H_2[R_2]$ generates a c-action $\tau \col (H_1 \cup H_2 \cup B)[R_1 \cpar R_2]$ with $R_1 \cpar R_2 = \zero$. 
 Because of Lem.~\ref{lem:pivot-no-local-block-1}, the set
$$B = \{\tau \mid H_1 \cap \oliA(P_2) \subseteq \{\ell\} \text{ or } H_2 \cap \oliA(P_1) \subseteq \{\ol{\ell}\}\}$$
must be empty. Thus, the reduction $\ell_1 \cpar \ell_2$ is c-enabled, too. 
\end{proof}

The following, slightly extended version of Prop.~\ref{prop:maximal-progress} says that a pivotable process cannot offer a clock action and exhibit another clock or a rendez-vous synchronisation at the same time. This means that in this class of processes, clocks and reductions are sequentially scheduled.

\begin{proposition}[Sequential Schedule and Clock Maximal Progress]
 Suppose $P$ is c-coherent and pivotable and $\sigma \in \iA(P)$. Then, for all $\ell \in \L$ with $\ell \neq \sigma$ we have $\ell \not\in \iA(P)$ or $\ol{\ell} \not\in \iA(P)$. In particular, $P$ is in normal form, i.e., there is no reduction $P \fstep{\tau} P'$. 
\label{prop:maximal-progress-x}
\end{proposition}

\begin{proof}
 The proof proceeds by contradiction. Let $P \cnf \pi$ for pivot policy $\pi$. Suppose $\ell \in \L$, $\ell \neq \sigma$ and
 $ %\[ 
 P \Derives{\sigma}{H}{\zero} P_1 \text{ and }
 P \Derives{\ell}{N}{F} P_2 \text{ and }
 P \Derives{\ol{\ell}}{M}{G} P_2.
 $ %\] 
By Prop.~\ref{prop:clock-interference} we infer $\sigma \in N$ or $\ell \in H$, i.e., by conformance Def.~\ref{def:conformance} we have $\pi \Vdash \sigma \ordpre \ell$ or $\pi \Vdash \ell \ordpre \sigma$. Hence, $\pi \nVdash \sigma \indep{} \ell$. For the same reason we have $\pi \nVdash \sigma \indep{} \ol{\ell}$. 
 However, this contradicts the pivot property of $\pi$, which requires $\pi \Vdash \sigma \indep{} \ell$ or $\pi \Vdash \sigma \indep{} \ol{\ell}$. 
 \mrev{Thus, $\ell \not\in \iA(P)$ or $\ol{\ell} \not\in \iA(P)$.}{added explanation.} Finally, if $P \fstep{\tau} P'$ then the reduction must arise from a rendez-vous synchronisation inside $P$, i.e., there is $\ell \in \R$ with $\ell \in \iA(P)$ and $\ol{\ell} \in \iA(P)$. But then $\sigma \in \iA(P)$ is impossible as we have just seen. 
\end{proof}

The proof of Thm.~\ref{thm:summary-coherence-closure} will be conducted in the following Secs.~\ref{sec:stop-prefix}--\ref{sec:hiding}. It proceeds by induction on the structure of derivations in the SOS of \ccslm, given in Fig.~\ref{fig:free-sos}. We first show that the indirections of syntactic transformations via the static laws of Fig.~\ref{fig:structural-cong}, built into the SOS by the $(\StructR)$ rule, can in fact be eliminated in terms of a finite number of standard SOS rules that do not refer to structural laws. 

\begin{definition}
  A process $P$ is called \textit{$\zero$-free} if either $P = \zero$ or all occurrences of $\zero$ in $P$ are guarded by a prefix $\ell\col H\cseq R$.
\end{definition}

\begin{proposition}
  Every process $P$ is structurally congruent to a $\zero$-free process $Q$.
\label{prop:zero-free}
\end{proposition}

\begin{proof}
  By application of the laws 
  $P \star \zero \equiv P$ and $\zero \wr L \equiv \zero$, where $\star \in\{\cpar,+\}$ and $\wr \in \{\restrict, \hide\}$ we can eliminate every occurrence of $\zero$ in an process $P$ that is not of the form $\alpha \col H \cseq \zero$.
\end{proof} 

\begin{proposition}
  If $P$ is a $\zero$-free process, and there is a derivation ${\cal D}$ for  $P \Derives{\alpha}{H}{R} Q$, then all uses of the rule $(\StructR)$ in ${\cal D}$ can be eliminated in terms of the following symmetric versions of the rules $(\ParR)$ and $(\SumR)$:
  %-------------------------------
  $$
  \begin{array}{r@{\qquad}r}
    \infer[(\SumR_1)]
    {P + Q \Derives{\alpha}{H}{R} P'}
    {P \Derives{\alpha}{H}{R} P'}
    &
    \infer[(\SumR_2)]
    {P + Q \Derives{\alpha}{H}{R} Q'}
    {Q \Derives{\alpha}{H}{R} Q'}
    \\[5mm]
    \infer[(\ParR_1)]
    {P \cpar Q \Derives{\alpha}{H}{R \cpar Q} P' \cpar Q}
    {P \Derives{\alpha}{H}{R} P' & \alpha \not\in \C }
    & 
     \infer[(\ParR_2)]
    {P \cpar Q \Derives{\alpha}{H}{R \cpar Q} P \cpar Q'}
    {Q \Derives{\alpha}{H}{R} Q' & \alpha \not\in \C }
    \end{array}
  $$
  %-------------------------------
  More specifically, there is a derivation ${\cal D'}$ for  $P \Derives{\alpha}{H}{R'} Q'$ with $R' \equiv R$ and $Q' \equiv Q$. 
\label{prop:struct-simulate}
\end{proposition}
\begin{proof}[Proof Sketch]
By induction on the structure of processes one shows that the transitions generated by the SOS with $(\StructR)$ can be 
simulated by the extended SOS without $(\StructR)$. The argument is based on the following case analysis:
%   
%The argument uses the following facts about the SOS including $(\StructR)$:
%
\begin{itemize}
    \item If $P = \zero$ then no process structurally equivalent to $P$ can generate any transition: If $P \equiv Q$ and $Q \Derives{\alpha}{H}{R} Q'$ then $P \neq \zero$. Hence, the statement of the proposition holds trivially. 
    \item For prefixes we show that if $P \equiv \ell \col H \cseq Q$ and $P \Derives{\ell'}{H'}{E'} Q'$ then $\ell' = \ell$, $H' = H$, $E' \equiv \zero$ and $Q' \equiv Q$.
    \item For sums we observe that if $P \equiv P_1 + P_2$ and $P \Derives{\ell}{H}{E} Q$ then there exists $i \in \{1,2\}$ such that $P_i \Derives{\ell}{H}{E'} Q'$ such that $E' \equiv E$ and $Q' \equiv Q$.
    
\item For all other operators, e.g.  $P \equiv  P_1 \cpar P_2$, $P \equiv Q \hide L$, $P \equiv Q \restrict L$, we can easily show that transitions  $P \Derives{\alpha}{H}{R} P'$ can be obtained from transitions of $P_i$ and $Q$.  
\end{itemize}
\end{proof}

Note that the restriction in Prop.~\ref{prop:struct-simulate} is necessary. If $P$ is not $\zero$-free, say $P = \zero \cpar Q$, then  we can have $\zero \cpar Q \xrightarrow{\sigma} Q'$ because of $(\StructR)$, because $\zero \cpar Q \equiv Q$ and $Q \xrightarrow{\sigma} Q'$, but without $(\StructR)$ we could not derive $\zero \cpar Q \xrightarrow{\sigma} Q'$.

Because of Prop.~\ref{prop:zero-free} we may assume our processes are $\zero$-free. 
Because of Prop.~\ref{prop:struct-simulate} we may prove the preservation of coherence (for $\zero$-free processes) by induction on the SOS rules without considering $(\StructR)$, but including the symmetric rules $(\ParR_i)$ and $(\SumR_i)$ for $i \in \{1,2\}$.

\subsection{Stop and Prefixes}
\label{sec:stop-prefix}
%----------------------------------------------------------

%---------------------------------------------
\begin{proposition}
 The process $\zero$ is c-coherent for any policy, i.e., $\zero \cnf \pi$ for all $\pi$.
\label{prop:zero-coherent} 
\end{proposition}
%---------------------------------------------

\begin{proof}
 The inactive process $\zero$ is c-coherent for trivial reasons, simply because it does not offer any transitions at all. There is no SOS rule applicable to it in the $(\StructR)$-free setting. Note that in the system with $(\StructR)$ rule, we would have to argue that no process structurally equivalent to $\zero$ has any transitions.
\end{proof}

%---------------------------------------------
\begin{proposition}
 Let process $Q$ be c-coherent for $\pi$. Then, for every action $\ell \in \R$ the prefix expression $\ell \col H \cseq Q$ is c-coherent for $\pi$, if $\ell \in H \subseteq \{ \ell' \mid \ell' \ordpre \ell \in \pi\}$. 
\label{prop:channel-prefix-coherent} 
\end{proposition}
%---------------------------------------------

\begin{proof}
 A prefix expression $P = \ell \col H \cseq Q$ generates only a single transition by rule $(\ActR)$, so that the assumption
 \begin{equation}
   P \Derives{\alpha_1}{H_1}{E_1} Q_1
   \text{ and }
   P \Derives{\alpha_2}{H_2}{E_2} Q_2
 \label{eqn:prefix-divergence}
 \end{equation}
 implies $\alpha_i = \ell$, $H_i = H$, $E_i \equiv \zero$ and $Q_i \equiv Q$. 
 It is easy to see that $\ell \col H \cseq Q$ conforms to $\pi$ if \mrev{$H \subseteq \{ \beta \mid \beta \ordpre \ell \in \pi \}$.}{corrected} 
 Since \mrev{$\alpha_i = \ell \in H = H_{3-i}$}{corrected} and $\{ \alpha_1, \alpha_2 \} \subseteq \R$ 
 %c-actions $\alpha_i \col H_i[E_i]$ interfere and 
 nothing needs to be proved \mrev{since~\eqref{eqn:prefix-divergence} do not count as a interprerence-free in the sense of Def.~\ref{def:interference-free}}{added explanation}. 
 Finally, note that the only immediate derivative of $\alpha \col H \cseq Q$ is $Q$ which is c-coherent for $\pi$ by assumption. 

 Again, it is important to notice that the SOS with $(\StructR)$ would not immediately warrant the conclusion that both derivations~\eqref{eqn:prefix-divergence} are by $(\ActR)$. This is because these transitions could be the transitions of two syntactically distinct (but structurally congruent) expressions $P_1 \equiv P \equiv P_2$. Instead, we would have to argue that if $P = \ell \col H \cseq Q \equiv P'$ and $P' \Derives{\alpha'}{H'}{E'} Q'$ then $\alpha' = \ell$, $H' = H$, $E' \equiv \zero$ and $Q' \equiv Q$. That this is the case is a property of our specific system of structural laws and handled by the above Prop.~\ref{prop:struct-simulate}. This would obviously be violated, e.g., if we added a structural law like $\ell \col H \cseq Q \equiv \alpha' \col H' \cseq Q'$ where $\alpha' \neq \ell$, $H' \neq H$, $E' \not\equiv \zero$ or $Q' \not\equiv Q$.
\end{proof}

As Examples~\ref{ex:rwmem} and~\ref{examples:Esterel-signal} show, 
%We will later see (Prop.~\ref{def:seq-bang-prefix}) that 
a prefix $\ell \col H\cseq Q$ for $\ell \in \R$ can very well be c-coherent even if $\ell \not\in H$. Here we note that for clock prefixes, the blocking set is only constrained by the policy.

%---------------------------------------------
\begin{proposition}
 If $Q$ is c-coherent for $\pi$ and $\sigma \in \C$ is a clock, then the prefix expression $\sigma \col H \cseq Q$ is c-coherent for $\pi$, too, if $H \subseteq \{ \beta \mid  \beta \ordpre \sigma \in \pi\}$. 
\label{prop:clock-prefix-coherent} 
\end{proposition}
%---------------------------------------------

\begin{proof}
 The argument runs exactly as in case of Prop.~\ref{prop:channel-prefix-coherent}. However, we do not need to require condition $\sigma \in H$ because if $\alpha_i = \sigma$ then $\alpha_1 = \alpha_2$, $Q_1 \equiv Q_2$ \mrev{(by clock determinism Prop.~\ref{prop:action-determinism})}{added explanation} and $\{ \alpha_1, \alpha_2 \} \not\subseteq \R$. Thus, \mrev{the clock transitions are not interference-free and so}{added explanation} no reconvergence is required. 
 %(and thus $\sigma \col H[\zero]$ may be non-interfering with itself) because coherence Def.~\ref{def:coherence}(1) does not apply to clocks. 
 Again, we observe that the only immediate derivative of $\sigma \col H \cseq Q$ is $Q$ which is c-coherent for $\pi$ by assumption.
\end{proof}

\subsection{Summation}
\label{sec:summation}
%---------------------------------------------
\begin{proposition}
 \mrev{Let $Q \cnf \pi_1$ and $R \cnf \pi_2$ be c-coherent and for all pairs of initial transitions 
 $ %$$ 
 Q \Derives{\alpha_1}{H_1}{F_1} Q' \text{ and }
 R \Derives{\alpha_2}{H_2}{F_2} R'
 $ %$$ 
 we have $\alpha_1 \neq \alpha_2$ as well as $\alpha_1 \in H_2$ or $\alpha_2 \in H_1$. 
 Then $Q + R$ is c-coherent for $\pi$ if $\pi_1 \subseteq \pi$ and $\pi_2 \subseteq \pi$.}{corrected statement} 
 \label{prop:sum-coherent} 
\end{proposition}
%---------------------------------------------
%
\begin{proof}
 Let $P = Q + R$ be given with $Q$ and $R$ c-coherent for $\pi_1$ and $\pi_2$, respectively and the rest as in the statement of the proposition. Every transition of $P$ is either a transition of $Q$ or of $R$ via rules $(\SumR_i)$. Conformance to $\pi$ directly follows from the same property of $Q$ or $R$, respectively. The proof of c-coherence Def.~\ref{def:coherence} is straightforward, exploiting that two competing but non-interfering transitions must either be both from $Q$ or both from $R$. More precisely, suppose 
 $ %$$ 
 P \Derives{\alpha_1}{H_1}{F_1} P_1 \text{ and }
 P \Derives{\alpha_2}{H_2}{F_2} P_2
 $ %$$ 
 are \mrev{interference-free, in particular}{added explanation} $\alpha_1 \neq \alpha_2$ or $P_1 \not\equiv P_2$ or $\{ \alpha_1, \alpha_2 \} \subseteq \R$ and $\alpha_i \not\in H_{3-i}$.
 %Also, we assume that $\alpha_1 \col H_1[F_1]$ and $\alpha_2 \col H_2[F_2]$ are interference-free. 
 Now if the two transitions of $P$ are by $(\SumR_1)$ from $Q$ \mrev{and}{corrected} by $(\SumR_2)$ from $R$, then $\alpha_1 \in \iA(Q)$ and $\alpha_2 \in \iA(R)$. But then by assumption, $\alpha_1 \neq \alpha_2$ and $\alpha_1 \in H_2$ or $\alpha_2 \in H_1$, \mrev{contradicting the assumption that}{better explanation} $\alpha_1 \col H_1[F_1]$ and $\alpha_2 \col H_2[E_2]$ are interference-free. This means, to prove coherence \mrev{Def.~\ref{def:coherence}}{corrected reference} we may assume that both transitions of $P$ are either by $(\SumR_1)$ from $Q$ or both by $(\SumR_2)$ from $R$.
 Suppose both reductions of $P = Q + R$ are generated by rule $(\SumR_1)$, i.e., 
 $ %$$ 
 Q \Derives{\alpha_1}{H_1}{F_1} Q_1 \text{ and }
 Q \Derives{\alpha_2}{H_2}{F_2} Q_2
 $ %$$ 
 where the transitions generated by $(\SumR_1)$ are 
% \begin{eqnarray*}
$ Q + R \Derives{\alpha_1}{H_1}{F_1} Q_1 
 \text{ and } 
 Q + R \Derives{\alpha_2}{H_2}{F_2} Q_2.
 %\label{eqn:prisum-dia-1}
 %\end{eqnarray*}
 $
Under the given assumptions, the reconverging transitions of coherence \mrev{Def.~\ref{def:coherence}}{corrected ref} are obtained directly from coherence of $Q$. \mrev{}{dropped ``by induction''} The case that both transitions are from $R$ by rule $(\SumR_2)$ is symmetrical. Finally, the immediate \mrev{subexpressions}{changed wording} $P_1$ and $P_2$ of $P_1 + P_2$ are c-coherent for $\pi$, because they are c-coherent for $\pi_i$ by assumption and thus also for the common extension $\pi_i \subseteq \pi$.
\end{proof}

\subsection{Parallel}
\label{sec:parallel}

%---------------------------------------------
\begin{lemma}[Key lemma]
 Let $Q \cnf \pi$ and $R \cnf \pi$ be such that $\pi \restrict \R$ is a pivot policy. Then $(Q \cpar R) \cnf \pi$.
\label{prop:cpar-coherent} 
\end{lemma}
%---------------------------------------------
%
%
\begin{proof}
 Let $P = Q \cpar R$ such that both $Q$ and $R$ are c-coherent for $\pi$ where $\pi \restrict \R$ is a pivot policy.
 For conformance consider a transition 
 %
% \begin{eqnarray} 
$$ P \Derives{\ell}{H}{E} P'
 %\label{eqn:par-step}
 %\end{eqnarray}
$$
with $\ell \in \L$ and $\ell' \in H$. First suppose $\ell \not\in \C$. Then this transition is not a synchronisation but a transition of either one of the sub-processes $Q$ by $(\ParR_1)$ or $P$ by $(\ParR_2)$ with the same blocking set $H$. In either case, we thus use coherence $Q \cnf \pi$ or $R \cnf \pi$ to conclude $\ell' \ordpre \ell \in \pi$. 
 If $\ell = \sigma$, then the transition in question is a synchronisation of $Q$ and $R$ via rule $(\ComR)$
\mrev{$$ 
 Q \cpar R \Derives{\sigma}{H}{\zero} P'
$$}{adjusted to fit the following text} 
 with $E \equiv \zero$, $H = H_2 \cup N_2 \cup B_2$ and $P' = Q_2 \cpar R_2$ generated from transitions 
 $$ 
 Q \Derives{\sigma}{H_2}{\zero} Q_2 
 \text{ and }
 R \Derives{\sigma}{N_2}{\zero} R_2.
 $$
 and \mrev{$B_2 = \{ \tau \mid H_2 \cap \oliA(R) \not\subseteq \{ \sigma \} \text{ or }N_2 \cap \oliA(Q) \not\subseteq \{\sigma\}\}$}{corrected}. 
 By Lem.~\ref{lem:pivot-no-local-block-1} we have $B_2 = \eset$.  If we now take an arbitrary $\ell' \in H \cap \L$ then $\ell' \in H_2$ or $\ell' \in N_2$. In these cases, we obtain $\ell' \ordpre \ell \in \pi$ by conformance Def.~\ref{def:conformance} from $Q  \cnf \pi$ or $R \cnf \pi$. 

 In the sequel, we tackle coherence Def.~\ref{def:coherence}. 
 As before, we work in the SOS without $(\StructR)$ but with symmetrical rules $(\ParR_i)$. First we observe that the immediate derivatives of $Q \cpar R$ are always parallel compositions $Q' \cpar R'$ of derivatives of $Q$ and $R$. Thus, we may assume that the immediate derivatives $Q' \cpar R'$ are c-coherent for $\pi$ because coherence of $P$ and $Q$ is preserved under transitions. \mrev{}{dropped ``by co-induction''} Now suppose 
 $$ 
 P \Derives{\alpha_1}{H_1}{E_1} P_1 
 \text{ and }
 P \Derives{\alpha_2}{H_2}{E_2} P_2
 $$ 
 \mrev{are interference-free, in particular}{shortened explanation} such that $\alpha_1 \neq \alpha_2$ or $P_1 \not \equiv P_2$ or $\{ \alpha_1, \alpha_2 \} \subseteq \R$ and $\alpha_i \not\in H_{3-i}$. 
 %Moreover, the c-actions $\alpha_1 \col H_1[E_1]$ and $\alpha_2 \col H_2[E_2]$ are interference-free.
 We argue reconvergence by case analysis on the rules that generate these transitions. These could be $(\ParR_1)$, $(\ParR_2)$ or $(\ComR)$.
 First note that by Lem.~\ref{lem:pivot-no-local-block-1} we can drop the consideration of the synchronisation action $\ell \cpar \ol{\ell}$ for the blocking set in $(\ComR)$ altogether:
 $$
 \infer[(\ComR)]
 {P \cpar Q \oldDerives{\ell \cpar \ol{\ell}}{H_1 \cup H_2}{R_1 \cpar R_2} P' \cpar Q'}
 {P \Derives{\ell}{H_1}{R_1} P' & Q \Derives{\ol{\ell}}{H_2}{R_2} Q' %& \ell \in \A \cup \coA
 }
 $$
 We will use the name $(\ComR_{a})$ to refer to the instance of $(\ComR)$ with $\ell \in \R$ and the name $(\ComR_{\sigma})$ for the instance with $\ell \in \C$ and we \mrev{proceed}{corrected typo} by case analysis.

\begin{itemize}

\item $(\ParR_1, \ParR_2) \indep{} (\ComR_{a})$. We start with the case where one of the reductions to $P_i$ is non-synchronising by $(\ParR_1)$ or $(\ParR_2)$ and the second reduction to $P_{3-i}$ is a synchronisation obtained by $(\ComR_{a})$. By symmetry, it suffices to consider the case that the non-synchronising transition is by $(\ParR_1)$ from $Q$, i.e., $\alpha_1 \not\in C$ and $\alpha_2 = \ell \cpar \ol{\ell}$ and
$$ 
 Q \cpar R \Derives{\alpha_1}{H_1}{F_1 \cpar R} Q_1 \cpar R 
 \text{ and }
 Q \cpar R \Derives{\ell \cpar \ol{\ell}}{H_2' \cup N_2}{F_2 \cpar G_2} Q_2 \cpar R_2
$$ 
 with $E_1 =  F_1 \cpar R$, $E_2 = F_2 \cpar G_2$, 
 $H_2 = H_2' \cup N_2$, $P_1 = Q_1 \cpar R$ and $P_2 = Q_2 \cpar R_2$, 
 %and $B' = \{ \tau \mid H_2' \cap \oliA(R) \not\subseteq \{\ell\} \text{ or } N_2 \cap \oliA(Q) \not\subseteq \{\ol{\ell}\} \}$, 
 where $\ell \in \R$ is a rendez-vous action such that 
 $$ 
 Q \Derives{\alpha_1}{H_1}{F_1} Q_1 \text{ and }
 Q \Derives{\ell}{H_2'}{F_2} Q_2 \text{ and }
 R \Derives{\ol{\ell}}{N_2}{G_2} R_2.
 $$ 
 \mrev{To show c-coherence we exploit the assumption that}{improved wording}
 the c-actions $\alpha_1 \col H_1[F_1\cpar R]$ and $\ell\cpar\ol{\ell}\col (H_2'\cup N_2)[F_2 \cpar G_2]$ are non-interfering. This implies, in particular, $\alpha_1 \neq \ell \cpar \ol{\ell} = \tau$ or $Q_1 \cpar R \not \equiv Q_2 \cpar R_2$. The third case $\{ \alpha_1, \alpha_2 \} \subseteq \R$ is excluded since $\alpha_2 = \tau$. 
 Observe that $\{ \alpha_1, \ell\cpar\ol{\ell} \} \not\subseteq \R$ and also $\{ \alpha_1, \ell\cpar\ol{\ell} \} \cap \C = \eset$. Thus only weak \mrev{residual steps}{adjusted wording} are needed. We wish to exploit c-coherence of $Q$. The first observation is that 
 \mrev{non-interference of the c-actions $\alpha_1 \col H_1[F_1\cpar R]$ and $\ell\cpar\ol{\ell}\col (H_2'\cup N_2)[F_2 \cpar G_2]$}{corrected} 
 %$\alpha_1 \neq \alpha_2$, 
 %the assumption on non-interference \mrev{Def.~\ref{def:interference-free}}{wrong ref fixed} 
 implies that $H_1 \cap \olwilA{\ast}(F_1 \cpar R) = \eset$. Since $\ell = \ol{\ol{\ell}} \in \oliA(R) \cap \L \subseteq \olwilA{\ast}(F_1 \cpar R)$ by Lem.~\ref{lem:prio-basic-1} and Lem.~\ref{lem:c-enabled-aux}, this implies that $\ell \not\in H_1$ up front.
 Next, if $\alpha_1 \neq \ell \cpar \ol{\ell}$, non-interference \mrev{Def.~\ref{def:interference-free}}{corrected reference} means 
 %given that $\alpha_1:H_1/E_1$ and $\tau:H_2/E_2$ are non-interfering, so 
 $\alpha_1 \not\in H_2 = H_2' \cup N_2$ and thus 
 $\alpha_1 \not\in H_2'$.
 The same is true if $\alpha_1 = \ell \cpar \ol{\ell} = \tau$, but then by non-interference we have the equation
 $$ (H_2' \cup N_2) \cap (\olwilA{\ast}(F_2 \cpar G_2) \cup \{ \tau \}) = 
 H_2 \cap (\olwilA{\ast}(E_2) \cup \{ \tau \}) = \eset
 %\label{eqn:aux-non-interf}
 $$
 from which $\alpha_1 = \tau \not \in H_2'$ follows.
 \mrev{Thus $\alpha_1 \not\in H_2'$ and $\ell \not\in H_1$.}{added explanation}
 This settles the first part of the non-interference property \mrev{Def.~\ref{def:interference-free}}{corrected reference} for the diverging transitions out of $Q$. 

 For the second part of non-interference, suppose $\alpha_1 = \tau$. Then the non-interference assumption
 %~\eqref{eqn:aux-non-interf} 
 implies that $H_2' \cap \olwilA{\ast}(F_2 \cpar G_2) = \eset$.
 Now since $\olwilA{\ast}(F_2) \subseteq \olwilA{\ast}(F_2 \cpar G_2)$ 
 %and $H_2' \subseteq H_2' \cup N_2$, 
 we conclude from this that $H_2' \cap \olwilA{\ast}(F_2) = \eset$. 

 \mrev{For the third part of non-interference}{added explanation} we note that we always have $\alpha_1 \neq \ell$ or $\{ \alpha_1, \ell \}\subseteq \R$, and in the latter case also $\alpha_1 \not\in H_2'$ and $\ell \not\in H_1$ from the above. 
 This completes the proof that the c-actions $\alpha_1 \col H_1[F_1]$ and $\ell \col H_2'[F_2]$ must be non-interfering in all cases. 

 %If $\alpha_1 \neq \ell$ or $Q_1 \neq Q_2$ or $\{ \alpha_1, \ell \}\subseteq \R$ %and $\alpha_1 \not\in H_2'$ and $\ell \not\in H_1$, 
 Hence, we can use c-coherence Def.~\ref{def:coherence} of $Q$
 and obtain processes $F_1'$, $F_2'$, $Q'$ with reconverging transitions
 %\begin{eqnarray*} 
 $$ Q_1 \Derives{\ell}{H_2''}{F_2'} Q' 
 \text{ and }
 Q_2 \Derives{\alpha_1}{H_1'}{F_1'} Q' 
 $$  
 such that $H_2'' \subseteq H_2'$ and $H_1' \subseteq H_1$. 
 We now invoke $(\ComR_{a})$ and $(\ParR_1)$ to obtain reconverging transitions
 $$ 
   Q_1 \cpar R \Derives{\ell \cpar \ol{\ell}}{H_2'' \cup N_2}{F_2' \cpar G_2} Q' \cpar R_2 
   \text{ and } 
   Q_2 \cpar R_2 \Derives{\alpha_1}{H_1'}{F_1' \cpar R_2} Q' \cpar R_2
 $$ 
 such that $H_1' \subseteq H_1$ and $H_2'' \cup N_2 \subseteq H_2' \cup N_2$.
 Regarding the concurrent environments, if $\alpha_1 \in \R$ then the coherence of $Q$ guarantees that 
 $ %$$ 
 F_1 \Derives{\ell}{}{} F_1' 
 \text{ and }
 F_2 \Derives{\alpha_1}{}{} F_2'
 $ %$$
 from which we infer 
 %$Q_1' \cpar R_2 Q_2' \cpar R_2$ and
$ % $$ 
 F_2 \cpar G_2 \Derives{\alpha_1}{}{} F_2' \cpar G_2
 \text{ and } 
 F_1 \cpar R \Derives{\ell\cpar\ol{\ell}}{}{} F_1' \cpar R_2.
 $ %$$ 

 If $\alpha_1 = \tau$ the above guarantee from the coherence of $Q$ is weaker: We only have $F_1 \astep{\ell} F_1'$ and $F_2 \astep{\alpha_1} F_2'$.
 In any case we obtain the environment shifts
 $ %$$ 
 F_2 \cpar G_2 \astep{\alpha_1} F_2' \cpar G_2
 \text{ and }
 F_1 \cpar R \astep{\ell\cpar\ol{\ell}} F_1' \cpar R_2,
$ % $$ 
 as required.
 
 \medskip 

 \item $(\ParR_1)\indep{} (\ParR_2)$: Suppose the two non-interfering transitions of $P = Q \cpar R$ are by $(\ParR_1)$ from $Q$ and $(\ParR_2)$ from $R$. Thus, we are looking at transitions
 $$ 
 Q \Derives{\alpha_1}{H_1}{F_1} Q_1 
 \text{ and } 
 R \Derives{\alpha_2}{H_2}{F_2} R_2
 $$ 
 for $\{\alpha_1, \alpha_2\} \subseteq \R \cup \{\tau\}$ combined via $(\ParR)$ to generate
 $$ 
   Q \cpar R \Derives{\alpha_1}{H_1}{F_1 \cpar R} Q_1 \cpar R
  \text{ and }
   Q \cpar R \Derives{\alpha_2}{H_2}{Q \cpar F_2} Q \cpar R_2
 $$ 
 with $E_1 = F_1 \cpar R$ and $E_2 = Q \cpar F_2$. 
 Then we can directly construct reconverging transitions, applying $(\ParR)$ for
 $$ 
   Q_1 \cpar R \Derives{\alpha_2}{H_2}{Q_1 \cpar F_2} Q_1 \cpar R_2 
   \text{ and } 
   Q \cpar R_2 \Derives{\alpha_1}{H_1}{F_1 \cpar R_2} Q_1 \cpar R_2.
 $$
 Obviously, $Q_1 \cpar R_2$ reduces to $Q_1 \cpar R_2$ and by $(\ParR_1)$ and $(\ParR_2)$ 
 we also have the strong environment shifts
$ % $$ 
 F_1 \cpar R \Derives{\alpha_2}{}{}$ $F_1 \cpar R_2$ 
$ \text{and } 
 Q \cpar F_2 \Derives{\alpha_1}{}{} Q_1 \cpar F_2
 $ %$$ 
 which completes the claim for we have, in particular, $F_1 \cpar R \astep{\alpha_2} F_1 \cpar R_2$ and $Q \cpar F_2 \astep{\alpha_1} Q_1 \cpar F_2$.
 Notice that in this case we do not need to assume that $Q$ or $R$ are c-coherent, i.e., the assumption that the c-actions $\alpha_i \col H_i[E_i]$ are interference-free.

 \medskip 

 \item $(\ParR_i) \indep{} (\ParR_i)$ for $i \in \{1,2\}$: Suppose both non-interfering and diverging reductions are by $(\ParR_1)$ from $Q$ (or symmetrically, both by $(\ParR_2)$ from $R$). Then,  
 $$ 
   Q \cpar R \Derives{\alpha_1}{H_1}{F_1 \cpar R} Q_1 \cpar R
   \text{ and } 
   Q \cpar R \Derives{\alpha_2}{H_2}{F_2 \cpar R} Q_2 \cpar R
 $$
 with $\{\alpha_1, \alpha_2\} \subseteq \R \cup \{\tau\}$ and 
 $P_i = Q_i \cpar R$, $E_i = F_i \cpar R$, where 
 $$ 
   Q \Derives{\alpha_1}{H_1}{F_1} Q_1
   \text{ and } 
   Q \Derives{\alpha_2}{H_2}{F_2} Q_2
 $$ 
 \mrev{Further, assume non-interference of $\alpha_i \col H_i[F_i \cpar R]$, in particular we have}{reworded} $\alpha_1 \neq \alpha_2$ or $Q_1 \cpar R \not\equiv Q_2 \cpar R$ or $\{ \alpha_1, \alpha_2 \} \subseteq \R$ and $\alpha_i \not\in H_{3-i}$. This implies $Q_1 \not \equiv Q_2$.
 We obtain non-interference of $\alpha_i \col H_i[F_i]$ by Prop.~\ref{prop:interference}.
 We can therefore apply coherence Def.~\ref{def:coherence} of $Q$ and obtain processes $F_1'$, $F_2'$ and $Q'$ with reconverging transitions 
 $$ 
 Q_1 \Derives{\alpha_2}{H_2'}{F_2'} Q' \text{ and }
 Q_2 \Derives{\alpha_1}{H_1'}{F_1'} Q'
 $$ 
 with $H_i' \subseteq H_i$ and such that
 \begin{eqnarray} 
 F_1 \astep{\alpha_2} F_1' 
 \text{ and }
 F_2 \astep{\alpha_1} F_2'
 \label{eqn:par-aux-1}
 \end{eqnarray}
and if $\{\alpha_1, \alpha_2\} \subseteq \R$ and $\alpha_1 \neq \alpha_2$ or $Q_1 \not \equiv Q_2$, more strongly 
 \begin{eqnarray} 
 F_1 \Derives{\alpha_2}{}{} F_1' \text{ and }
 F_2 \Derives{\alpha_1}{}{} F_2'.
 \label{eqn:par-aux-2} 
 \end{eqnarray}
 Reapplying $(\ParR_1)$ we construct the transitions
 $$ 
 Q_1 \cpar R \Derives{\alpha_2}{H_2'}{F_2' \cpar R} Q' \cpar R
 \text{ and }
 Q_2 \cpar R \Derives{\alpha_1}{H_1'}{F_1' \cpar R} Q' \cpar R.
 $$
 Obviously, by $(\ParR_1)$ and $(\ParR_2)$ we also obtain transitions 
 %$$ 
 $F_1 \cpar R \astep{\alpha_2} F_1' \cpar R$
and 
$ F_2 \cpar R \astep{\alpha_1} F_2' \cpar R$ 
%$$ 
 from~\eqref{eqn:par-aux-1} or, more strongly,
 % $$ 
 $F_1 \cpar R \Derives{\alpha_2}{}{} F_1' \cpar R$ 
 and 
 $F_2 \cpar R \Derives{\alpha_1}{}{} F_2' \cpar R$ 
% $$ 
 from~\eqref{eqn:par-aux-2} if $\{\alpha_1, \alpha_2\} \subseteq \R$ and $\alpha_1 \neq \alpha_2$ or $Q_1 \cpar R \not \equiv Q_2 \cpar R$.
 
\medskip 

\item $(\ComR_{a}) \indep{} (\ComR_{a})$:
 This is the most interesting case in which we are going to exploit the assumption that $Q$ and $R$ are c-coherent for the same pivot policy $\pi$. Without loss of generality we assume $P = Q \cpar R$ and both the \mrev{interference-free}{added} reductions 
 $$ 
   Q \cpar R \Derives{\alpha_i}{H_i \cup N_i}{E_i} Q_i \cpar R_i
 $$ 
 with $P_i = Q_i \cpar R_i$ are $\tau$-actions generated by the communication rule $(\ComR)$, i.e. $\alpha_1 = \ell_1 \cpar \ol{\ell_1} = \tau = \ell_2 \cpar \ol{\ell_2} = \alpha_2$ for actions $\ell_i \in \R$.  Since $\alpha_1 = \alpha_2$ and $\{ \alpha_1, \alpha_2 \} = \{ \tau \} \not\subseteq \R$ we only need to prove \mrev{reconvergence in}{changed word and grammar} the case that $P_1 \not\equiv P_2$, i.e., $Q_1 \not \equiv Q_2$ or $R_1 \not\equiv R_2$. Moreover, we only need weak \mrev{residual steps}{changed wording}, because $\{ \alpha_1, \alpha_2 \} = \{ \tau \} \not\subseteq \R$ and $\{ \alpha_1, \alpha_2 \} \cap \C = \eset$.
 Also, non-interference Def.~\ref{def:interference-free}(2) means that $(H_i \cup N_i) \cap (\olwilA{\ast}(E_i) \cup \{ \tau \}) = \eset$, i.e.,
 \begin{eqnarray}
 H_i \cap (\olwilA{\ast}(E_i) \cup \{ \tau \}) = \eset 
 \text{ and } 
 N_i \cap (\olwilA{\ast}(E_i) \cup \{ \tau \}) = \eset.
 \label{eqn:interf-assumption} 
 \end{eqnarray}

 Thus, we are looking at transitions
 $$ 
 Q \Derives{\ell_i}{H_i}{F_i} Q_i \text{ and } 
 R \Derives{\ol{\ell}_i}{N_i}{G_i} R_i.
 $$
 We claim that the two c-actions $\ell_i \col H_i[F_i]$ of $Q$, and likewise the two c-actions $\ol{\ell}_i \col N_i[G_i]$ of $R$, are interference-free. 
 Note that since $\ell_i \neq \tau$, we only need to consider the first \mrev{and third part of Def.~\ref{def:interference-free}.}{added to reflect changed definitions}
 We first show that at least one of these pairs of c-actions must be interference free, because of \mrev{c-coherence}{} and conformance to $\pi$ where $\pi \restrict \R$ is a pivot policy. It will then transpire that the other must be interference-free, too, because of the assumption~\eqref{eqn:interf-assumption}. 

 \medskip 

 We distinguish two cases depending on whether $\ell_1$ and $\ell_2$ are identical or not. 
 %Obviously, if $\ell_1 = \ell_2$ then both $\ell_i \col H_i[F_i]$ of $Q$ and $\ol{\ell}_i \col N_i[G_i]$ are interference-free for trivial reasons. \mrev{}{No needs refinement!}
 \mrev{First}{simplified argument} suppose $\ell_1 \neq \ell_2$. Both $Q$ and $R$ are c-coherent for the same policy $\pi$ and $\pi \restrict \R$ is a pivot policy. Thus, $ \ell_1 \indep{} \ell_2 \in \pi$ or $\ol{\ell}_1 \indep{} \ol{\ell}_2 \in \pi$, by Def.~\ref{def:pivot-policy}. 
 Hence, $\ell_i \not\in H_{3-i}$ or $\ol{\ell}_i \not\in N_{3-i}$ by conformance Def.~\ref{def:conformance}. \mrev{Since both $\{\ell_1, \ell_2\} \subseteq \R$ and $\{\ol{\ell}_1, \ol{\ell}_2\} \subseteq \R$}{added explanation} which means that at least one of the pairs of c-actions $\ell_i \col H_i[F_i]$ of $Q$ or the two c-actions $\ol{\ell}_i \col N_i[G_i]$ must be interference-free by policy. 

 \medskip 

 Now we argue that if $\ell_1 \neq \ell_2$ and one of the pairs $\ell_i \col H_i[F_i]$ or $\ol{\ell}_i \col N_i[G_i]$ is interference-free the other must be, too, and so we can close the diamond for both $Q$ and $R$. By symmetry, it suffices to run the argument in case that the c-actions $\ell_i \col H_i[F_i]$ of $Q$ are \mrev{interference-free}{changed wording}. Then, we have $\{ \ell_1, \ell_2 \} \subseteq \R$ and $\ell_i \not\in H_{3-i}$. 
 Hence, we invoke c-coherence Def.~\ref{def:coherence} of $Q$ obtaining processes $F_1'$ and $F_2'$ with \mrev{strong residual steps}{changed wording}
 $ %$$ 
 % F_1 \astep{\ell_2} F_1' 
 % \text{ and }
 % F_2 \astep{\ell_1} F_2'
 % \text{ ... }
 F_1 \Derives{\ell_2}{}{} F_1' 
 \text{ and }
 F_2 \Derives{\ell_1}{}{} F_2'
$ % $$
 and a process $Q'$ with transitions
 $$ 
 Q_1 \Derives{\ell_2}{H_2'}{F_2'} Q' 
 \text{ and }
 Q_2 \Derives{\ell_1}{H_1'}{F_1'} Q'.
 $$
 such that $H_i' \subseteq H_i$. 
 This means $\ell_i \in \iA(F_{3-i}) \subseteq \iA(F_{3-i} \cpar G_{3-i}) = \iA(E_{3-i}) \subseteq \wilA{\ast}(E_{3-i})$ and therefore $\ol{\ell}_i \not\in N_{3-i}$ because of~\eqref{eqn:interf-assumption}. \mrev{This implies}{changed wording}, as claimed, the c-actions $\ol{\ell}_i \col N_i[G_i]$ are interference-free, too. \mrev{This}{simplified argument} 
 %Again $\{ \ol{\ell}_1, \ol{\ell}_2 \} \subseteq \R$ and $\ol{\ell}_i \not\in N_{3-i}$ 
 entitles us to exploit c-coherence Def.~\ref{def:coherence} for $R$ 
 from which we obtain processes $G_1'$ and $G_2'$ with transitions 
 (again, \mrev{with strong residual steps}{adjusted})
 $ %$$ 
 G_1 \Derives{\ol{\ell}_2}{}{} G_1' 
 \text{ as well as }
 G_2 \Derives{\ol{\ell}_1}{}{} G_2'
$ % $$
 as well as a process $R'$ with 
 $$ 
 R_1 \Derives{\ol{\ell}_2}{N_2'}{G_2'} R' 
 \text{ and } 
 R_2 \Derives{\ol{\ell}_1}{N_1'}{G_1'} R'.
 $$
 with $N_i' \subseteq N_i$.
 Finally, we invoke $(\ComR)$ to obtain the re-converging reductions
 $$ 
   Q_1 \cpar R_1 \Derives{\ell_2 \cpar \ol{\ell}_2}{H_2' \cup N_2'}{F_2' \cpar G_2'} Q' \cpar R' 
   \text{ and } 
   Q_2 \cpar R_2 \Derives{\ell_1 \cpar \ol{\ell}_1}{H_1' \cup N_1'}{F_1' \cpar G_1'} Q' \cpar R'
 $$ 
 with $H_i' \cup N_i' \subseteq H_i \cup N_i$,
 and also 
 $$ 
 F_{i} \cpar G_{i} \Derives{\ell_{3-i} \cpar \ol{\ell}_{3-i}}{}{} F_{i}' \cpar G_{i}'
 $$ 
 as required, since this implies the weaker form
 $F_{i} \cpar G_{i} \astep{\tau} F_{i}' \cpar G_{i}'$ of \mrev{residual steps}{adjusted}.

 \medskip 

 It remains to consider the case that $\ell_1 = \ell = \ell_2$. By assumption $P_1 \not\equiv P_2$ and so we must have $Q_1 \not\equiv Q_2$ or $R_1 \not\equiv R_2$. Suppose $Q_1 \not\equiv Q_2$. If $R_1 \not\equiv R_2$ we apply a symmetric argument with the role of $\ell_i$ and $\ol{\ell}_i$ interchanged. 
 \mrev{We now observe that}{changed wording} the c-actions $\ell_i \col H_i[F_i]$ are trivially interference-free according to Def.~\ref{def:interference-free}, because $\ell_1 = \ell_2$, $\ell_i \neq \tau$ \mrev{and $Q_1 \not\equiv Q_2$}{improved explanation}. 
 Then, by c-coherence Def.~\ref{def:coherence} applied to $Q$ we conclude that there are processes $F_1'$ and $F_2'$ with transitions
 $ %$$ 
 F_1 \Derives{\ell}{}{} F_1' 
 \text{ and }
 F_2 \Derives{\ell}{}{} F_2'
 $ %$$
 and a process $Q'$ with transitions (the strong \mrev{residual steps}{adjusted} arise here because we have $\ell \in \R$ and $Q_1 \nequiv Q_2$)
 $$ 
 Q_1 \Derives{\ell}{H_2'}{F_2'} Q' 
 \text{ and }
 Q_2 \Derives{\ell}{H_1'}{F_1'} Q',
 $$
 such that $H_i' \subseteq H_i$. This means that $\ell \in \iA(F_{3-i}) \subseteq \iA(E_{3-i})$ and so $\ol{\ell} \not\in N_{3-i}$ because of $N_{3-i} \cap \olwilA{\ast}(E_{3-i}) = \eset$ as per~\eqref{eqn:interf-assumption}. 
 \mrev{Independently of whether $R_1 \equiv R_2$ or $R_1 \not\equiv R_2$, 
 we observe that the c-actions $\ol{\ell}_i \col N_i[G_i]$ of $R$, are interference-free, because $\{\ol{\ell}_1, \ol{\ell}_2\} = \{\ol{\ell}\} \subseteq \R$ and $\ol{\ell} \not \in N_{3-i}$. Thus, there are reconverging transitions}{corrected and simplified argument} 
 $$ 
 R_i \Derives{\ol{\ell}}{N_{3-i}'}{G_{3-i}'} R'
 $$ 
 with $N_{i}' \subseteq N_{i}$ and $G_i \equiv G_i'$ or $G_i \Derives{\ol{\ell}}{}{} G_i'$. Using these transitions we now recombine
 $$ 
 Q_i \cpar R_i 
 \Derives{\ell \cpar \ol{\ell}}{H_{3-i}' \cup N_{3-i}'}{F_{3-i}' \cpar G_{3-i}'} Q' \cpar R'
  $$ 
 with $H_{i}' \cup N_{i}' \subseteq H_{i} \cup N_{i}$. Note that if $G_i \equiv G_i'$ then 
 $$ 
 F_i \cpar G_i \Derives{\ell}{}{} F_i' \cpar G_i = F_i' \cpar G_i'
 $$ 
 or if $G_i \Derives{\ol{\ell}}{}{} G_i'$ then 
 $$ 
 F_i \cpar G_i \Derives{\ell\cpar\ol{\ell}}{}{} F_i' \cpar G_i'.
 $$ 
 This means that, in all cases, $F_i \cpar G_i \astep{\ell\cpar\ol{\ell}} F_i' \cpar G_i'$ as desired. 
 % If $R_1 \not\equiv R_2$ we can apply coherence Def~\ref{def:coherence} to $R$ in the same way and close the diamond with
 %  processes $G_1'$ and $G_2'$ and transitions (again, the strong environment shift arises here because we have $\ol{\ell} \in \R$ and $R_1 \nequiv R_2$)
%  $ %$$ 
%  G_1 \Derives{\ol{\ell}}{}{} G_1' 
%  \text{ and }
%  G_2 \Derives{\ol{\ell}}{}{} G_2'
% $ % $$
%  as well as a process $R'$ with 
%  $$ 
%  R_1 \Derives{\ol{\ell}}{N_2'}{G_2'} R' 
%  \text{ and } 
%  R_2 \Derives{\ol{\ell}}{N_1'}{G_1'} R'.
%  $$
%  with $N_i' \subseteq N_i$.  
 Finally, we invoke $(\ComR)$ to obtain the re-converging reductions
 $$ 
 Q_1 \cpar R_1 \Derives{\ell \cpar \ol{\ell}}{H_2' \cup N_2'}{F_2' \cpar G_2'} Q' \cpar R' 
 \text{ and } 
 Q_2 \cpar R_2 \Derives{\ell \cpar \ol{\ell}}{H_1' \cup N_1'}{F_1' \cpar G_1'} Q' \cpar R'
 $$ 
 with $H_i' \cup N_i' \subseteq H_i \cup N_i$, 
 and also 
 $ %$$ 
 F_{i} \cpar G_{i} \Derives{\ell \cpar \ol{\ell}}{}{} F_{i}' \cpar G_{i}'
 $ %$$ 
 which also implies $F_{i} \cpar G_{i} \astep{\ell \cpar \ol{\ell}} F_{i}' \cpar G_{i}'$
 as required. 
 
 \medskip 

\item $(\ComR_{a}) \indep{} (\ComR_{\sigma})$ and, by symmetry, the induction case $(\ComR_{\sigma}) \indep{} (\ComR_{a})$: We assume $P = Q \cpar R$ and the diverging reductions 
$$ 
 Q \cpar R \Derives{\ell\cpar\ol{\ell}}{H_1 \cup N_1 \cup B_1}{E_1} Q_1 \cpar R_1
 \text{ and } 
 Q \cpar R \Derives{\sigma}{H_2 \cup N_2 \cup B_2}{\zero} Q_2 \cpar R_2
$$ 
 for \mrev{$\ell \in \R$}{corrected}, $\sigma \in \C$, with $P_i = Q_i \cpar R_i$ and $E_i = F_i \cpar G_i$ are generated by instances $(\ComR_{a})$ and $(\ComR_{\sigma})$ of the communication rules, respectively. These transitions arise from
 $$ 
 Q \Derives{\ell}{H_1}{F_1} Q_1 
 \text{ and } 
 R \Derives{\ol{\ell}}{N_1}{G_1} R_1
 \text{ and } 
 Q \Derives{\sigma}{H_2}{\zero} Q_2 
 \text{ and } 
 R \Derives{\sigma}{N_2}{\zero} R_2.
 $$
 We assume that the c-actions $(\ell\cpar\ol{\ell}) \col (H_1 \cup N_1 \cup B_1)[E_1]$ and $\sigma \col (H_2 \cup N_2 \cup B_2)[\zero]$ are \mrev{interference-free}{changed term}.
 Note that $\ell \neq \sigma$ and $\ol{\ell} \neq \sigma$.
 %First, observe that by Lem.~\ref{lem:pivot-no-local-block-1}, $B_1 = \eset = B_2$.
 We have $\ell \cpar \ol{\ell} = \tau \neq \sigma$ and so the non-interference assumption %of Def.~\ref{def:interference-free}(2) 
 means that $\sigma \not\in H_1 \cup N_1 \cup B_1$ as well as $\tau \not \in H_2 \cup N_2 \cup B_2$. The other \mrev{two parts}{adjusted} of non-interference, from \mrev{Def.~\ref{def:interference-free}, do not}{adjusted} provide any extra information. 

 We claim that $\ell \in H_2$. By contradiction, if $\ell \not\in H_2$ then, since $\sigma \not\in H_1$, the two c-actions $\ell \col H_1[F_1]$ and $\sigma \col H_2[\zero]$ of $Q$ would be interference-free according to Def.~\ref{def:interference-free}. But then, since $\{ \ell, \sigma \} \cap \C \neq \eset$, c-coherence Def.~\ref{def:coherence} applied to $Q$ would imply 
 $\sigma \in \iA(F_1)$ and $\ell \in \iA(\zero)$, \mrev{from the required strong residual steps}{added explanation}. The latter, however, is impossible. Thus, as claimed, $\ell \in H_2$. Likewise, from $\sigma \not\in N_1$ we derive $\ol{\ell} \in N_2$ in the same fashion.
 
 Now we invoke the fact that $\tau \not \in B_2$, i.e. 
 $H_2 \cap \oliA(R) \subseteq \{\sigma\}$ and $N_2 \cap \oliA(Q) \subseteq \{\sigma\}$. But $\ol{\ell} \in \iA(R)$ and $\ell \in \iA(Q)$, whence we must have $\ell \not\in H_2$ and $\ol{\ell} \not\in N_2$. This is a contradiction.
 Hence, transitions obtained by rules $(\ComR_{a})$ and $(\ComR_{\sigma})$ are never interference-free.
 
 \medskip 

\item $\{(\ParR_1), (\ParR_2) \} \indep{} (\ComR_{\sigma})$: 
%\item $(\ParR) \indep{} (\ComR_{\sigma})$: 
 As the representative case, we assume $P = Q \cpar R$ and the diverging reductions are 
 $$ 
 Q \cpar R \Derives{\alpha_1}{H_1}{F_1 \cpar R} Q_1 \cpar R 
 \text{ and } 
 Q \cpar R \Derives{\sigma}{H_2 \cup N_2 \cup B_2}{\zero} Q_2 \cpar R_2
 $$ 
 with $\alpha_1 \in \R \cup \{\tau\}$ so that $P_1 \equiv Q_1 \cpar R$ and $P_2 \equiv Q_2 \cpar R_2$ are generated by the rule $(\ParR_1)$ and $(\ComR_{\sigma})$, respectively, from transitions 
 $$ 
 Q \Derives{\alpha_1}{H_1}{F_1} Q_1 
 \text{ and } 
 Q \Derives{\sigma}{H_2}{\zero} Q_2 
 \text{ and } 
 R \Derives{\sigma}{N_2}{\zero} R_2.
 $$
 where $\alpha_1 \neq \sigma$ and the c-actions $\alpha_1 \col H_1[F_1 \cpar R]$ and $\sigma \col (H_2 \cup N_2 \cup B_2)[\zero]$ are interference-free. First, \mrev{Def.~\ref{def:interference-free}}{simplified reference} implies $\alpha_1 \not\in H_2 \cup N_2 \cup B_2$ and $\sigma \not\in H_1$. Further, if $\alpha_1 = \tau$ then
 $(H_2 \cup N_2 \cup B_2) \cap (\wilA{\ast}(\zero) \cup \{ \tau \}) = \eset$, which in particular means $H_2 \cap (\wilA{\ast}(\zero) \cup \{ \tau \}) = \eset$. Therefore, the c-actions $\alpha_1 \col H_1[F_1]$ and $\sigma \col H_2[\zero]$ are interference-free. Note that $\{ \alpha_1, \sigma\} \cap \C \neq \eset$. Hence we can exploit c-coherence Def.~\ref{def:coherence} for $Q$, obtaining the stronger form of \mrev{residual steps}{adjusted}. But this would imply $\alpha_1 \in \iA(\zero)$ which is impossible.
 So we find that transitions obtained by rules $(\ParR_i)$ and $(\ComR_{\sigma})$ are never interference-free in the sense of c-coherence.

\item $(\ComR_{\sigma}) \indep{} (\ComR_{\sigma})$:
 We assume $P = Q \cpar R$ and the diverging reductions 
 $$ 
   Q \cpar R \Derives{\sigma_1}{H_1 \cup N_1 \cup B_1}{\zero} Q_1 \cpar R_1
   \text{ and }
   Q \cpar R \Derives{\sigma_2}{H_2 \cup N_2 \cup B_2}{\zero} Q_2 \cpar R_2
 $$ 
 with $P_i = Q_i \cpar R_i$ and $E_i = F_i \cpar G_i$ are generated by the communication rules $(\ComR_{\sigma})$. Furthermore, the two clock transitions of $Q_1 \cpar Q_2$ \mrev{are}{corrected} assumed to be interference-free.  
 
\mrev{First, note that if $\sigma_1 \neq \sigma_2$ the non-interference implies $\sigma_i \not\in H_{3-i} \cup N_{3-i} \cup B_{3-i}$ and thus $\sigma_i \not\in H_{3-i}$ as well as $\sigma_i \not\in N_{3-i}$.}{added to complete argument}
Furthermore, observing that $\{ \sigma_1, \sigma_2 \} \subseteq \R$ is impossible, here, we may assume $\sigma_1 \neq \sigma_2$ or $Q_1 \cpar R_1 \nequiv Q_2 \cpar R_2$. This can only hold true if $Q_1 \not\equiv Q_2$ or $R_1 \not\equiv R_2$. 
From this it follows that \mrev{at least one of}{corrected argument} the two underlying clock transitions 
 %$Q \fsstep{\sigma \col H_1[\zero]} Q_1$ 
 $Q \Derives{\sigma}{H_1}{\zero} Q_1$ 
 and 
 %$Q \fsstep{\sigma \col H_2[\zero]} Q_2$ 
 $Q \Derives{\sigma }{H_2}{\zero} Q_2$ 
 \mrev{or the pair of}{corrected argument} transitions 
 %$R \fsstep{\sigma \col N_1[\zero]} R_1$
 $R \Derives{\sigma}{N_1}{\zero} R_1$
  and 
  %$R \fsstep{\sigma \col N_2[\zero]} R_2$ 
  $R \Derives{\sigma}{N_2}{\zero} R_2$ 
  are interference-free. \mrev{Specifically, if}{changed wording} $Q_1 \not\equiv Q_2$ then we could apply c-coherence of $Q$ and obtain $\sigma \in \iA(\zero)$, if $R_1 \not\equiv R_2$ the same follows from c-coherence of $R$. This is impossible, whence the case can be excluded. 
 \end{itemize}
\end{proof}

\longshort{\version}{}
{%
\subsection{Repetition}
\label{sec:repetition}
For arbitrary continuation processes $A$, a regular prefix $a \col a \cseq A$ needs to be self-blocking to be c-coherent. Such a prefix cannot receive from two senders, for $a \col a \cseq A \cpar \ol{a} \cpar \ol{a}$ will block under weak enabling. However, if the receiver is willing to engage in action $a$ repeatedly, say if $A \pdef a \cseq A$ then we can lift the self-blocking. The process $A \cpar \ol{a} \cpar \ol{a}$ is c-coherent and will happily consume both $\ol{a}$ sender actions. Repetition is the universal way of implementing multi-sender and multi-receiver scenarios. Note that we define a new operator, representing recursion via parallel replication on the same channel $\ell$, that can be easily encoded in the current syntax.

\begin{definition}[Sequential Bang Prefix]
 For every process $P$, action $\ell \in \L$ and $H \subseteq \Act$, let $\bang{\ell} \col H \cseq P$ be the process defined by the 
 SOS rule 
 $$ 
 \infer[(\RepR)]
 {\bang{\ell} \col H \cseq P \Derives{\ell}{H}{P} P \cpar \bang{\ell} \col H \cseq P}
 {}
 $$  
 We use $\bang{\ell} \col H$ as an abbreviation for $\bang{\ell} \col H \cseq \zero$ and $\bang{\ell}$ for $\bang{\ell} \col \eset \cseq \zero$.
\label{def:seq-bang-prefix}
\end{definition}
The purpose of repetition is to replicate input and output prefixes so they can be consumed multiple times and by multiple threads. To see this, let us notice the difference in the blocking sets of a regular output prefix $\ol{a} \col \ol{a} \cseq A$ and its repetition $\bang{\ol{a}} \cseq A$. 
 In the former case, the transition 
 $ %$$
 \ol{a}\col \ol{a} \cseq A \Derives{\ol{a}}{\{\ol{a}\}}{\zero} A
 $ %$$ 
 includes the output label in the blocking set $\{ \ol{a} \}$, by $(\ActR_1)$. This reflects the fact that the output prefix $\ol{a}$ is fully consumed by the transition. Accordingly, (strong/weak) enabling will block multiple concurrent receivers trying to access the label $a$ at the same time. Only a single thread can engage in the synchronisation. This is necessary, since the $\ol{a}$ prefix is consumed. The continuation process $A$ may not offer output $\overline{a}$ any more, or if it does, it may be incongruent with $A$. Therefore, each synchronisation with one receiver thread would preempt another concurrent receiver thread in getting access to $\ol{a}$. For contrast, in a transition 
 $ %$$
 \bang{\ol{a}} \cseq A \Derives{\ol{a}}{\eset}{A} A \cpar \bang{\ol{a}} \cseq A
 $ %$$
 generated from $(\ActR_1)$ and $(\RepR)$, the blocking set is empty, and so does not prevent concurrent receivers. This is fine since the prefix is not consumed but repeated in the continuation process $A \cpar \bang{\ol{a}} \cseq A$.
\begin{remark}[Why is the standard bang not good enough?]
 Note that the sequential bang $\bang{\ell} \cseq A$ is not expressible as $\bang{\ell} \cseq A = \bang{(\ell \cseq A)}$ by the standard `bang' operator $\bang{P}$ of process algebra which satisfies the \textit{false} (hence not present) structural equivalence $\bang{P} \equiv P \cpar \bang{P}$ in \ccslm. The corresponding Milner's equivalence 
 $\bang{\ell} \cseq A \cong \ell \cseq A \cpar \bang{\ell} \cseq A$ does not hold in \ccslm. Let us see why. 
 Suppose the behaviour of $\bang{\ell}$ is derived from (or identified with that of) $\ell \cseq A \cpar \bang{\ell} \cseq A$. Then the initial action of $\bang{\ell} \cseq A$ would be 
 $ %$$
 \ell \cseq A \cpar \bang{\ell} \cseq A 
 \Derives{\ell}{\eset}{\bang{\ell} \cseq A} A \cpar \bang{\ell} \cseq A
 $ %$$ 
 generated by rules $(\ActR_1)$ and $(\ParR_1)$. Notice that the concurrent environment 
 $\bang{\ell} \cseq A$
 %$\bang{\ell} \cseq A = \bang{\ell} \cseq A$ 
 is not empty and that we have $\ell \in \iA(\bang{\ell} \cseq A)$. Therefore, a parallel composition
 $ %$$
 \ol{\ell}\col \ol{\ell} \cpar \ell \cseq A \cpar \bang{\ell} 
 \cseq A \Derives{\tau}{\{\ol{\ell}\}}{\bang{\ell} \cseq A} A \cpar \bang{\ell} \cseq A
 $ %$$
 would block under (weak/strong) enabling because $\{\ol{\ell}\} \cap \oliA(\bang{\ell} \cseq A) \neq \eset$.
 Thus $\ell \cseq A \cpar \bang{\ell} \cseq A $ would not be able to synchronise even with a single (self-blocking) sender $\ol{\ell} \col \ol{\ell}$. 
 For contrast, the sequential bang prefix of Def.~\ref{def:seq-bang-prefix} can serve arbitrarily many receivers
 $ %$$
 \ol{\ell} \col \ol{\ell} \cpar 
 \ol{\ell} \col \ol{\ell} \cpar 
 \bang{\ell} \cseq A$ 
 $\Derives{\tau}{\{\ol{\ell}\}}{\ol{\ell} \col \ol{\ell} \cpar A} 
 \ol{\ell} \col \ol{\ell} \cpar A \cpar \bang{\ell} \cseq A$
 $\Derives{\tau}{\{\ol{\ell}\}}{A \cpar A} 
 A \cpar A \cpar \bang{\ell} \cseq A,
 $ %$$
 provided that $\{\ol{\ell}\} \cap \olwilA{\ast}(A) = \{\ol{\ell}\} \cap \olwilA{\ast}(A) = \eset$.
 The problem is that the standard bang $\bang{\ell} \cseq A \equiv \ell \cseq A \cpar \bang{\ell} \cseq A$ is a \textit{parallel} repetition of label $\ell$ while the sequential bang $\bang{\ell} \cseq A$ offers all $\ell$-labels \textit{sequentially} in a single thread. Only the payload $A$ is offered in parallel. This corresponds to an \textit{async} operator that is triggered by label $\ell$, then spawns a child thread $A$ and repeates itself in the main thread. \qed
\end{remark}
%
%

%---------------------------------------------
\begin{proposition}
 Let $P : \pi$ be c-coherent and where $\pi \restrict \R$ is a pivot policy. 
 Then, for every $\ell \in \pi \restrict \R$ such that $P \astep{\ell} P \cpar P$, the sequential bang prefix $\bang{\ell} \col H \cseq P$ is c-coherent for $\pi$ if $H \subseteq \{ \ell' \mid  \ell' \ordpre \ell \in \pi\}$. 
\label{prop:seq-bang-coherent}
\end{proposition}
%---------------------------------------------

\begin{proof}
 By definition, the process $\bang{\ell} \col H \cseq P$ offers only a single initial transition
$ % $$ 
 \bang{\ell} \col H \cseq P 
\Derives{\ell}{H}{P}$
$ P \cpar \bang{\ell} \col H \cseq P
$ % $$ 
 which must be obtained by application of rule $\RepR$.
 The conformance Def.~\ref{def:conformance} follows exactly as for action prefixes from $\ell \in \pi$ the assumption $H \subseteq \{ \ell' \mid  \ell' \ordpre \ell \in \pi\}$. 
 Since there are no diverging transitions for $\bang{\ell} \col H \cseq P$ the only situation to consider for confluence is that $\ell \not \in H$. For reconvergence it suffices to observe that the continuation process $P \cpar \bang{\ell} \cseq P$ offers another transition engaging with $\ell$ and the same blocking set $H$:
 $ %$$ 
 P \cpar \bang{\ell} \col H \cseq P 
 \mbox{$\Derives{\ell}{H}{P \cpar P}$}
 P \cpar P \cpar \bang{\ell} \col H \cseq P.
 $ %$$ 
 By assumption, $P \astep{\ell} P \cpar P$. This weak environment shift is enough to satisfy coherence Def.~\ref{def:coherence}, because of the deterministic transition of a non-clock. 
 Finally, note that since $P$ is c-coherent for $\pi$ and coherence for pivot policies closed under parallel composition (Prop.~\ref{prop:cpar-coherent}), we can assume that the continuation process $P \cpar \bang{\ell} \col H \cseq P$ is conformant to $\pi$, by co-induction. 
\end{proof}
} % end of longshort

\subsection{Hiding}
\label{sec:hiding}

%---------------------------------------------
\begin{proposition}
 If $Q$ is c-coherent for $\pi$ and $\sigma \ordpre \ell \not\in \pi$ for all $\sigma \in L$ and $\ell \in \pi$, then $Q \hide L$ is c-coherent for $\pi$. 
\label{prop:hiding-coherent} 
\end{proposition}
%---------------------------------------------

\begin{proof}
 It suffices to prove the proposition for the special case \mrev{$P = Q \hide \sigma$}{corrected confusion of $P$ and $Q$ in the notation below} of a single clock. We recall the semantic rule for this case: 
$$
 \begin{array}{l@{\hspace{2cm}}l}
 \multicolumn{2}{l}{
 \infer[(\HideR)]
 {Q \hide \sigma \Derives{\alpha \hide \sigma}{H'}{E \hide \sigma} Q' \hide \sigma}
 {Q \Derives{\alpha}{H}{E} Q' & H' = H - \{ \sigma \}}} 
 \end{array}
$$
 where $\sigma\hide\sigma = \tau$ and $\alpha\hide\sigma = \alpha$ if $\alpha \neq \sigma$.
 Let $P = Q \hide \sigma$ and the two diverging transitions 
 $$ 
 P \Derives{\alpha_1\hide\sigma}{H_1}{E_1} P_1
 \text{ and } 
 P \Derives{\alpha_2\hide\sigma}{H_2}{E_2} P_2
 $$
 arise by $(\HideR)$ from transitions 
% \begin{eqnarray} 
 $$ Q \Derives{\alpha_1}{N_1}{F_1} Q_1 
 \text{ and } 
 Q \Derives{\alpha_2}{N_2}{F_2} Q_2
 $$
 %\label{eqn:hide-Q-diverge}
 %\end{eqnarray}
 with $E_i = F_i \hide \sigma$, $P_i = Q_i \hide \sigma$ and $H_i = N_i - \{ \sigma \}$.
 We assume that the c-actions $\alpha_i \hide \sigma \col H_i[E_i]$ are interference-free. 
 \mrev{}{dropped redundant argument} 
 %For c-coherence Def.~\ref{def:coherence} we assume 
 % that $\alpha_1 \hide \sigma \neq \alpha_2 \hide \sigma$ or $P_1 \not\equiv P_2$, or $\{ \alpha_1 \hide \sigma, \alpha_2 \hide \sigma \} \subseteq \R$ and $\alpha_i \hide \sigma \not\in H_{3-i}$. Observe that the first case implies $\alpha_1 \neq \alpha_2$ and the second case means $Q_1 \not\equiv Q_2$. 
 Moreover, because of the assumption that $\sigma \ordpre \alpha_{i} \not\in \pi$ we must have $\sigma \not \in N_1 \cup N_2$. But this means $H_i = N_i$. 

First, suppose that $\alpha_1 = \alpha_2 = \sigma$. But then the assumption $Q_1 \not\equiv Q_2$ contradicts strong \mrev{determinism}{corrected term} of clocks (Proposition~\ref{prop:action-determinism}). Thus, the actions $\alpha_i$ cannot both be the clock $\sigma$ that is hidden. The case $\alpha_i = \sigma$ and $\alpha_{3-i} \in \Act - \{\sigma\}$, for some fixed $i \in \{1,2\}$ can also be excluded as follows: By Prop.~\ref{prop:clock-interference} the clock transition and the non-clock transition must \mrev{block each other}{adjusted wording} in $Q$, i.e., $\alpha_{i} \in N_{3-i}$ or $\alpha_{3-i} \in N_i$. The former $\sigma = \alpha_{i} \in N_{3-i}$ is impossible because of the above. The latter is outright impossible since then $\alpha_{3-i}\hide\sigma = \alpha_{3-i} \in N_i = H_i$ which contradicts the non-interference assumption on the transitions of $P$. 
 
The only remaining case to handle is that both of the diverging transitions are by labels \mrev{$\{\alpha_1, \alpha_2\} \subseteq \Act - \{\sigma\}$}{more precise} and $\alpha_i\hide\sigma = \alpha_i$ for both $i \in \{1,2\}$. We claim that the c-actions $\alpha_i \col N_i[F_i]$ are interference-free. If $\alpha_1 \neq \alpha_2$ then non-interference assumption \mrev{of the c-actions $\alpha_i \hide \sigma \col H_i[E_i]$}{added for explanation} directly gives $\alpha_i = \alpha_i \hide \sigma \not\in H_i = N_i$. Next, if for some $i \in \{1,2\}$, $\alpha_i = \alpha_i\hide\sigma = \tau$, then %non-interference 
 the assumptions on $P$ not only give $\tau \not\in H_{3-i} = N_{3-i}$ but also 
 $$
  N_{3-i} \cap \olwilA{\ast}(F_{3-i} \hide \sigma) = H_{3-i} \cap \olwilA{\ast}(E_{3-i}) = \eset.
 $$ 
 Since by assumption on policy conformance $\sigma \not\in N_{3-i}$ the latter implies $N_{3-i} \cap \olwilA{\ast}(F_{3-i}) = \eset$. 
 \mrev{It remains to verify the third property of non-interference Def.~\ref{def:interference-free}: If $\alpha_1 = \alpha_2$ and $Q_1 \equiv Q_2$ then 
 both $\alpha_1 \hide \sigma = \alpha_2 \hide \sigma$ or $P_1 \equiv P_2$. Then, non-interference of the c-actions $\alpha_i \hide \sigma \col H_i[E_i]$ from $P$ implies that $\{ \alpha_1 \hide \sigma, \alpha_2 \hide \sigma \} \subseteq \R$ and $\alpha_i \hide \sigma \not\in H_{3-i}$. But this means that $\alpha_i \neq \sigma$, $\{ \alpha_1, \alpha_2 \} \subseteq \R$ and $\alpha_i \not\in N_{3-i}$.
 }{completed argument}
 Now we have shown that the diverging transitions
 %~\eqref{eqn:hide-Q-diverge} 
 of $Q$ are \mrev{interference-free}{adjusted term} we can use c-coherence Def.~\ref{def:coherence} for $Q$ by \mrev{}{dropped ``by induction hypothesis''} and get $F_1'$, $F_2'$, $Q'$ such that 
 $$ 
 Q_1 \Derives{\alpha_2}{N_2'}{F_2'} Q'
 \text{ and }
 Q_2 \Derives{\alpha_1}{N_1'}{F_1'} Q'
 \text{ with }
 F_1 \astep{\alpha_2} F_1' 
 \text{ and } 
 F_2 \astep{\alpha_1} F_2' 
 $$ 
 such that $N_i' \subseteq N_i$. By applying $(\HideR)$ to these transitions, we obtain 
 $$ 
 P_1 \Derives{\alpha_2\hide\sigma}{H_2''}{F_2' \hide \sigma} Q' \hide \sigma
 \text{ and }
 P_2 \Derives{\alpha_1\hide\sigma}{H_1''}{F_1' \hide \sigma} Q' \hide \sigma
 $$ 
 where $H_i'' = N_i' - \{ \sigma \} \subseteq N_i - \{ \sigma \} = H_i$. Note, by rule $(\HideR)$ we also have 
 $$ 
 F_1 \hide \sigma \astep{\alpha_2\hide\sigma} F_1' \hide \sigma 
 \text{ and }
 F_2 \hide \sigma \astep{\alpha_1\hide\sigma}F_2' \hide \sigma.
 $$ 
 In the special case that $\{ \alpha_1\hide\sigma, \alpha_2\hide\sigma \} = \{ \alpha_1, \alpha_2 \} \subseteq \R$ and $\alpha_1 \neq \alpha_2$ or $Q_1 \hide \sigma \not\equiv Q_2 \hide \sigma$, or $\{ \alpha_1\hide\sigma, \alpha_2\hide\sigma \} \cap \C = \{ \alpha_1, \alpha_2 \} \cap \C \neq \eset$ we obtain the stronger \mrev{residual steps}{adjusted} 
 $ %$$
 F_1 \Derives{\alpha_2}{}{} F_1' 
 \text{ and } 
 F_2 \Derives{\alpha_1}{}{} F_2' 
 $ %$$ 
 from which we infer
 $$ 
 F_1 \hide \sigma \Derives{\alpha_2\hide\sigma}{}{} F_1' \hide \sigma 
 \text{ and }
 F_2 \hide \sigma \Derives{\alpha_1\hide\sigma}{}{} F_2' \hide \sigma. 
 $$ 
This completes the proof. 
\end{proof}

\subsection{Restriction}
\label{sec:restriction}

Let us next look at the restriction operator. 
The following definition is relevant for creating coherent processes under action restriction (App.~\ref{sec:restriction}).
% %

\medskip 

The fact that $P$ is c-coherent does not imply that $P \restrict L$ is c-coherent. Consider the process $Q = (s + a \col s \cpar b \cseq \ol{s}) \restrict s$ which is conformant and pivotable. The signal $s$ blocks the action $a$ by priority in the thread $s + a \col s$. For the communication partner $\ol{s}$ to become active, however, it needs the synchronisation with an external action $\ol{b}$. The transition
%
%\begin{eqnarray}
$$
 Q \Derives{a}{\eset}{(\one \cpar b \cseq \ol{s})\restrict s} (b \cseq \ol{s}) \restrict s
$$
%\label{eqn:ex-restrict}
%\end{eqnarray}
%
generated by $(\RestrR)$ does not provide enough information in the blocking set $H$ and environment $E$ for us to be able to characterise the environments in which the action $a$ is blocked. For instance, the parallel composition $Q \cpar \ol{b}$ will internally set the sender $\ol{s}$ free and thus block the $a$-transition. From the above transition
%~\eqref{eqn:ex-restrict} 
the $(\ComR)$ rule permits $Q \cpar \ol{b}$ to offer the $a$-step. This is not right as it is internally blocked and does not commute with the $b \cpar \ol{b}$ reduction. A simple solution to avoid such problems and preserve confluence is to force restricted signals so they do not block any visible actions. 
This restriction suffices in many cases. Another, more general solution is to use the extended rule $\RestrR^*$  discussed above after Example~\ref{ex:clock-hiding}.

%

%
%---------------------------------------------
\begin{proposition}
 If $Q \cnf \pi$ and $\pi$ is precedence-closed for $L \cup \ol{L}$, then $({Q}\restrict{L}) \cnf \pi \restrict L$. 
\label{prop:restrict-coherent} 
\end{proposition}
%---------------------------------------------
%
\begin{proof}
It suffices to prove the proposition for the special case $P = Q \restrict a$ for a single channel name $a \in A$. Recall the rule $(\RestrR)$ for this case: 
 $$
 \begin{array}{l@{\hspace{2cm}}l}
 \multicolumn{2}{l}{
 \infer[(\RestrR)]
 {Q \restrict a \Derives{\alpha}{H'}{E \restrict a} R \restrict a}
 {Q \Derives{\alpha}{H}{E} R %& L' = \{ a, \ol{a} \} 
 & \alpha \not\in \{ a, \ol{a} \} &
 H' = H - \{ a, \ol{a} \}}
 }
 \end{array}
 $$
 Consider a transition
 $$ 
 Q \restrict a \Derives{\ell}{H'}{E \restrict a} Q' \restrict a
 \text{ derived from }
 Q \Derives{\ell}{H}{E} Q'
 $$ 
 with $H' = H - \{ a, \ol{a} \}$, $\ell \not\in \{ a, \ol{a} \}$ and $\ell' \in H'$.
 Then $\ell' \not\in \{ a, \ol{a} \}$ and $\ell' \in H$. Conformance Def.~\ref{def:conformance} applied to $Q$ implies $\ell' \ordpre \ell \in \pi$. But then $\ell' \ordpre \ell \in \pi - a$, too, as both $\ell$ and $\ell'$ a distinct from $a$ and $\ol{a}$. This ensures conformance. 

 If $P = Q \restrict a$ and the two diverging transitions 
 $$ 
 P_2  \LDerives{\alpha_2}{H_2}{E_2} P \Derives{\alpha_1}{H_1}{E_1} P_1
% P \Derives{\alpha_1}{H_1}{E_1} P_1
% \text{ and }
% P \Derives{\alpha_2}{H_2}{E_2} P_2
 $$
arise from rule $(\RestrR)$ then we have diverging transitions
$$Q_2 \LDerives{\alpha_2}{N_2}{F_2} Q \Derives{\alpha_1}{N_1}{F_1} Q_1$$
with $H_i = N_i - \{ a, \ol{a} \}$, $\alpha_i \not\in \{ a, \ol{a} \}$, where $E_i \equiv {F_i} \restrict a$ and $P_i \equiv {Q_i}\restrict a$. Suppose the c-actions $\alpha_1 \col H_1[E_1]$ and $\alpha_2 \col H_2[E_2]$ are interference-free. \mrev{}{dropped redundant argument}.  
%$\alpha_1 \neq \alpha_2$ or $P_1 \not\equiv P_2$ or 
%$\{ \alpha_1, \alpha_2 \} \subseteq \R$ and $\alpha_i \not \in H_{3-i}$. 
%Note that in case $P_1 \not\equiv P_2$ we immediately have $Q_1 \not\equiv Q_2$. Next, notice that the premise of the $(\RestrR)$ rule ensures $\alpha_i \not\in \{ a, \ol{a} \}$. It is easy to see that this means $\{ a, \ol{a} \} \cap N_{i} = \eset$, i.e., $H_i = N_i$. By contraposition, suppose $\ell \in \{a, \ol{a} \}$ and $\ell \in N_i$. Then $\ell \ordpre \alpha_i \in \pi$. But $\pi$ is precedence-closed for $\{a, \ol{a}\}$, so we would have $\alpha_i \in \{a,\ol{a}\}$ which is a contradiction. So, $H_i = N_i$.
 %All this means that to invoke c-coherence Def.~\ref{def:coherence} of the above $Q$
 %on reduction~\eqref{eqn:Q-trans} 
 It is sufficent to show that the c-actions $\alpha_1 \col N_1[F_1]$ and $\alpha_2 \col N_2[F_2]$ \mrev{are interference-free}{adjusted}. Now, if $\alpha_1 \neq \alpha_2$ then from the non-interference assumption on the transitions out of $P$ we get \mrev{$\alpha_i \not\in H_{3-i}$ and thus $\alpha_i \not\in N_{3-i}$ since $\alpha_i \not\in \{ a, \ol{a}\}$.}{corrected argument} This is what we need for Def.~\ref{def:interference-free}(1) to show that the c-actions $\alpha_i \col N_i[F_i]$ are interference-free.
For the second part Def.~\ref{def:interference-free}(2) suppose $\alpha_{3-i} = \tau$ for some $i \in \{1,2\}$. Then the assumed non-interference of the transitions out of $P$ means $H_i \cap (\olwilA{\ast}(E_i) \cup \{ \tau \}) = \eset$, which is the same as $\tau \not \in H_i$ and the following set condition 
$$
  (N_i - \{ a, \ol{a} \}) \cap \olwilA{\ast}(F_i \restrict a) = \eset.
$$
%%
% %
 Since $\{ a, \ol{a} \} \cap N_{3-i} = \eset$ we have 
 %The former immediately implies $\tau\not\in N_i$ and
 %$\{ a, \ol{a} \} \cap N_i \cap \olwiA(F_i) = \eset$ which is 
 $N_i \cap \olwilA{\ast}(F_i) \subseteq \L - \{ a, \ol{a} \}$. 
 We claim that the $\alpha_i$-transitions of $Q$ is enabled, specifically $N_i \cap \olwilA{\ast}(F_i) = \eset$ observing that $\tau \not \in H_i = N_i$. 
 By contraposition, suppose $\beta \in N_i \cap \olwiA(F_i) \subseteq \L - \{ a, \ol{a} \}$. Hence, $\beta \not\in \{ a, \ol{a} \}$ and so $\beta \in N_i - \{ a, \ol{a} \}$ and also $\beta \in \olwilA{\ast}(F_i \restrict a)$, but this contradicts the above set-condition.
%~\eqref{eqn:aux-non-inter}. 
 Thus, we have shown that for all $i \in \{1,2\}$, $\alpha_{3-i} = \tau$ implies $N_i \cap \olwilA{\ast}(F_i) = \eset$. 
 \mrev{For the third part of Def.~\ref{def:interference-free}(3) suppose $\alpha_1 = \alpha_2$ and $Q_1 \equiv Q_2$. Then, $P_1 \equiv P_2$ and non-interference of the c-actions $\alpha_i \col H_i[E_i]$ of $P$ implies that 
 $\{\alpha_1, \alpha_2\} \subseteq \R$ and $\alpha_i \not\in H_{3-i}$.
 %$\alpha_1 \neq \alpha_2$ or $P_1 \not\equiv P_2$ or 
 %$\{ \alpha_1, \alpha_2 \} \subseteq \R$ and $\alpha_i \not \in H_{3-i}$. 
 %Note that in case $P_1 \not\equiv P_2$ we immediately have $Q_1 \not\equiv Q_2$. 
 Next, notice that the premise of the $(\RestrR)$ rule ensures $\alpha_i \not\in \{ a, \ol{a} \}$. It is easy to see that this means $\{ a, \ol{a} \} \cap N_{i} = \eset$, i.e., $H_i = N_i$. By contraposition, suppose $\ell \in \{a, \ol{a} \}$ and $\ell \in N_i$. Then $\ell \ordpre \alpha_i \in \pi$. But $\pi$ is precedence-closed for $\{a, \ol{a}\}$, so we would have $\alpha_i \in \{a,\ol{a}\}$ which is a contradiction. So, $H_i = N_i$. But then $\alpha_i \not\in H_{3-i}$ gives us $\alpha_i \not\in N_{3-i}$.
 }{completed argument}

 This proves the c-actions $\alpha_i \col N_i[F_i]$ out of $Q$ are interference-free. Hence, we can use coherence of $Q$ for the diverging transition
 %~\eqref{eqn:Q-trans} 
 and get $F_1'$, $F_2'$, $Q'$ such that
 $$ 
 Q_1 \Derives{\alpha_2}{N_2'}{F_2'} Q'
 \text{ and } 
 Q_2 \Derives{\alpha_1}{N_1'}{F_1'} Q'
 \text{ with }
 F_1 \astep{\alpha_2} F_1' 
 \text{ and } 
 F_2 \astep{\alpha_1} F_2'
 $$ 
 such that $N_i' \subseteq N_i$. We construct the required reconvergence by application of rule $(\RestrR)$:
 $$ 
 P_1 \mbox{\Derives{\alpha_2}{H_2''}{E_2'}} Q' \restrict a
 \text{ and }
 P_2 \Derives{\alpha_1}{H_1''}{E_1'} Q' \restrict a
 \text{ with }
 F_1 \restrict a \astep{\alpha_2} F_1' \restrict a
 \text{ and } 
 F_2 \restrict a \astep{\alpha_1} F_2' \restrict a 
 $$ 
 with $E_i' = F_i' \restrict a$ and $H_i'' = N_i' - \{ a, \ol{a} \} \subseteq N_i - \{ a, \ol{a} \} = H_i$.
 In the special case that $\{ \alpha_1, \alpha_2 \} \subseteq \R$ and $\alpha_1 \neq \alpha_2$ or $P_1 \not\equiv P_2$ or $\{ \alpha_1, \alpha_2 \} \cap \C \neq \eset$ we can rely on the stronger \mrev{residual steps}{adjusted} 
 $ %$$
 F_1 \Derives{\alpha_2}{}{} F_1' 
 \text{ and } 
 F_2 \Derives{\alpha_1}{}{} F_2' 
 $ %$$
 to conclude
$ % $$ 
 F_1 \restrict a \Derives{\alpha_2}{}{} F_1' \restrict a
 \text{ and }
 F_2 \restrict a \Derives{\alpha_1}{}{} F_2' \restrict a. 
$ % $$ 
This completes the proof.
\end{proof}
%
%
%
%

   % clean 

%!TEX root = synpatick-lipics.tex
%

\section{Examples, new and reloaded}
\label{examples-again}

This section contains a number of reloaded and new examples in view of the definitions and results presented in the previous section. For didactic reasons we will introduce some new auxiliary SOS rules, like $(\HideR^*)$ and $(\RestrR^*)$, which however are not part of \ccslm.

\subsection{Priorities and Scheduling}

The first example shows how aliasing can destroy coherence. 

\begin{example}[On Def. \ref{def:coherence}]
 If $A \not\equiv B$ then $a \cseq A + a \cseq B$ is not c-coherent while $a \col a \cseq A \cpar a \col a \cseq B$ is c-coherent.
 Conflicting choices on the same label in a single thread as in $a \cseq A + a \cseq B$ cannot be externally observed and thus cannot be deterministically scheduled via their externally observable properties. In c-coherent processes such choices are assumed to be resolved internally by each thread itself, rather than through the precedence relation.
 Note that the process $a \col a \cseq A + a \col a \cseq B$ in which the two action prefixes are self-blocking is not c-coherent either. The reason is that the two identical (self-blocking) transitions lead to structurally distinct states $A \not\equiv B$. Thus, the confluence condition expressed in Def.~\ref{def:coherence} applies and $A$ and $B$ should be reconvergent with transitions $A \Derives{a}{}{} Q'$ and $B \Derives{a}{}{} Q'$. But this is not guaranteed for arbitrary $A$ and $B$. On the other hand, a single self-blocking prefix $a \col a \cseq A$ is c-coherent. Without the self-blocking, $a \cseq A$ is only c-coherent if process $A$ permits infinite repetition of $a$ as in $A \pdef a \cseq A$ or $A \pdef a \cpar A$. 
 
The precedence relation cannot resolve the choice of a prefix with itself. The precedence relation can resolve the choice between \textit{distinct} actions prefixes. For instance, if $a \neq b$ then $a \cseq A + b \cseq B$ is not c-coherent, for any $A$ and $B$. However, both $a \col a \cseq A \cpar b \col b \cseq B$ and $a \col a \cseq A + b \col \{a, b\} \cseq B$ are c-coherent. In the former case, the actions are concurrently independent and so confluence occurs naturally. In the latter case, the choice is resolved by precedence: No reconvergence is required by coherence Def.~\ref{def:coherence} as the actions are interfering with each other. If we identify $a$ and $b$ (considering them as aliases) and replace them by the same action $c$, then the resulting process $c \col c \cseq A + c \col c \cseq B$ is no longer c-coherent, in general. \qed
\end{example}

The next example shows that structural determinism is striclty stronger that the property of having unique normal forms. 

\begin{example}[On Normal form]
 For the process $P_1 \eqdef (a \cseq c \cpar b \cseq d \cpar \ol{a} \cpar \ol{b}) \restrict \{a,b\}$ we have 
 $ %\[
 (a \cseq c \cpar d \cpar \ol{a}) \restrict \{a,b\}
 \LDerives{\tau}{}{} P_1 \Derives{\tau}{}{} (c \cpar b \cseq d \cpar \ol{b}) \restrict \{a,b\} 
 %\text{ and } P_1 \Derives{\tau}{}{} (a \cseq c \cpar d \cpar \ol{a}) \restrict \{a,b\}
 $ %\] 
 where the continuation processes $(c \cpar b \cseq d \cpar \ol{b}) \restrict \{a,b\}$ and $(a \cseq c \cpar d \cpar \ol{a}) \restrict \{a,b\}$ are not structurally equivalent. They are not even behaviourally equivalent.
 However, both reductions are concurrently independent and commute with each other. We can execute both rendez-vous synchronisations $\tau = a \cpar \ol{a}$ and $\tau = b \cpar \ol{b}$ in any order, the process $P$ converges to a unique final outcome $(c \cpar d) \restrict \{a,b\}$. Similarly, in the multi-cast process $P_2 \eqdef (a \cseq c \cpar a \cseq d \cpar \ol{a} \cpar \ol{a}) \restrict \{a\}$ the reduction behaviour is not structurally deterministic but nevertheless generates a unique normal form. 
 \qed
\end{example}

The next example shows how priorities provide a compositional way of coding sequential composition. 

\begin{example}[On Sequentiality à la Esterel in \ccslm]
Like \ccs, our calculus \ccslm does not have a generic operator for sequential composition as in Esterel. This is not necessary, because sequential composition can be coded in \ccs by action prefixing and parallel composition.  More precisely, a sequential composition $P;Q$ is obtained as a parallel composition $([P]_t \cpar t \col t \cseq Q) \restrict \{t\}$ where $[P]_t$ is the process $P$, modified so when it terminates, then the termination signal $\ol{t}$ is sent. The signal $t$ will then trigger the sequentially downstream process $Q$. The operation $[P]_t$ can be implemented by a simple (linear) syntactic transformation of $P$.
In the same way we can detect termination of a parallel composition $P = P_1 \cpar P_2$. Since a parallel terminates if both threads terminate, we can define $[P]_t$ as follows 
$$ [P_1 \cpar P_2]_t \eqdef %&\eqdef& 
 ([P_1]_{s_1} \cpar [P_2]_{s_2} \cpar 
 s_1 \col s_1 \cseq s_2 \col s_2 \cseq \ol{t}) \restrict \{s_1, s_2\}
$$
where $s_1$ and $s_2$ are fresh ``local'' termination signals.
The ``termination detector'' $T_{\textit{seq}} \pdef s_1 \col s_1 \cseq s_2 \col s_2 \cseq \ol{t}$ first waits for $P_1$ to terminate and then for $P_2$, before it sends the signal $\ol{t}$ indicating the termination of the parallel composition. We could equally first wait for $P_2$ and then for $P_1$, if the difference is not observable from the outside.
In cases where the termination detector needs to permit both processes $P_1$ and $P_2$ to terminate in any order, however, we need a different solution. We need to ``wait concurrently'' for $s_1$ or $s_1$, not sequentially as in $T_{\textit{seq}}$. The obvious idea would be to use a \textit{confluent sum} 
$$T_{\textit{csum}} \pdef s_1 \col s_1 \cseq s_2 \col s_2 \cseq \ol{t} + s_2 \col s_2 \cseq s_1 \col s_1 \cseq \ol{t}$$ as discussed in Milner~\cite{Milner:CCS}~(Chap.~11.4). 
The sum must explicitly unfold all possible orders in which the termination signals can arrive and thus again creates an exponential descriptive complexity, which is not good.
Interestingly, in \ccslm we can exploit precedence scheduling to obtain a robust and compositional solution. Consider the process
$ %\[
T_{\textit{par}} \pdef s \cseq T_{\textit{par}} + \ol{t} \col \{s, \ol{t} \}.
$ %\]
It engages in an unbounded number of receptions on channel $s$ without changing its state. At the same time, it offers to send the termination signal $\ol{t}$, but only if the action $s$ is not possible. Under the weak enabling strategy, the composition $\ol{s} \cpar \ol{s} \cpar T$ will consume the two senders $\ol{s}$ and only then engage in $t$. In this way, we can use $T_{\textit{par}}$ to detect the \textit{absence} of $s$ which corresponds to the termination of the two environment processes $\ol{s} \cpar \ol{s}$, in any order. 
The signal $s$ could be the shared termination signal of each process $P_1$ and $P_2$ that is sent and consumed by $T_{\textit{par}}$ if and when the process terminates.
If we execute the parallel composition $P_1 \cpar P_2 \cpar T_{\textit{par}}$ under constructive enabling, the termination signal $\ol{t}$ will not be sent until such time that both $P_1$ and $P_2$ have been able to send $s$ to $T$. 
From the scheduling point of view, the composition $P_1 \cpar P_2$ acts like a counter that is initialised to $2$, representing the number of potential offerings of $s$ initially in $P_1 \cpar P_2$. It gets decremented each time one of $P_1$ or $P_2$ terminates. For instance, $\zero \cpar P_2$ and $P_1 \cpar \zero$ have one potential offering left, while $\zero \cpar \zero$ has no more potential to send $s$. At this moment, the counter has reached $\zero$ and $\zero \cpar \zero \cpar T_{\textit{par}}$ can take the $\ol{t}$ action and thereby send termination. 
In sum, there are two important observations here: First, we find that precedences with constructive enabling permits us to code sequential composition of Esterel in a compositional fashion. Second, the coding of Milner's confluent sum $T_{\textit{csum}}$ is possible in \ccslm with linear descriptive complexity. \qed
\end{example}

The next example illustrates that the closure property for coherence of hiding (Thm.~\ref{thm:church-rosser}(6), see also Prop.~\ref{prop:hiding-coherent}) needs to require that the hidden clock must not take precedence over rendez-vous actions.

\begin{example} [On Clock-hiding]
\label{ex:clock-hiding}
The process $P \pdef \sigma\cseq A + a \col \{\sigma, a \} \cseq B \cpar \sigma$ offers a choice of two transitions
$ %\[ 
  P \Derives{\sigma}{\eset}{\zero} A$
and
$
  P \Derives{a}{\{\sigma,a\}}{\sigma} B \cpar \sigma.
$ %\] 
These two transitions interfere with each other since $\sigma \neq a$ and $\sigma \in \{\sigma, a\}$, violating condition (1) of Def.~\ref{def:coherence}. 
Further, each of the transitions interferes with itself since it violates condition (3) of Def.~\ref{def:coherence}. 
Thus, $P$ is trivially c-coherent, because no pair of non-interfering transitions exists. 
Observe that the $a$-transition is not weakly enabled,   
% so that the $\sigma$-transition blocks the $a$-transition under weak enabling, 
given that $\{ \sigma, a \} \cap \oliA(\sigma) = \{\sigma\} \neq \eset$. 
Any policy with $P \cnf \pi$ must have $\sigma \ordpre a \in \pi$ and $a \ordpre a \in \pi$. This does not satisfy the side condition of Thm.~\ref{thm:summary-coherence-closure}(6). So we cannot conclude that $P \hide \sigma$ is c-coherent. 
Let us see what happens when the clock is hidden and thus generates a silent step. Indeed, hiding the clock, we obtain from rule $(\HideR)$ the two transitions 
$$ 
  P \hide \sigma 
    \Derives{\tau}{\eset}{\zero\hide\sigma} A \hide \sigma
  \text{ and } 
  P \hide \sigma \Derives{a}{\{a \}}{\sigma \hide \sigma} (B \cpar \sigma) \hide \sigma
$$ 
where the blocking set of the $a$-transition turns into $\{\sigma,a \} - \{\sigma\} = \{a\}$. 
Because $\sigma$ is removed from the blocking set, the $a$-transition is now constructively enabled. Moreover, both the $\tau$-transition and the $a$-transition are interference-free. However, in general, there is no chance that we can find recovergent transitions $A \hide \sigma \xrightarrow{a} Q_1$ and $(B \cpar \sigma) \hide \sigma \xrightarrow{\tau} Q_2$ with $Q_1 \equiv Q_2$. Thus, $P \hide \sigma$ is not c-coherent any more. \qed 
\end{example}

The problem behind the previous example on clock hiding may be subsumed as follows: On the one hand, the $(\HideR)$ rule closes the scope of a clock, but one the other hand it does not actually implement any blocking that comes from this clock inside the scope.  
We could now refine the hiding rule in a manner that, like $(\ComR)$, it implements a ``race test'' and introduces a $\tau$ into the blocking set whenever there is a conflict. This could suggest the following $(\HideR^*)$ rule: 
$$
  \infer[(\HideR^*)]
  {P \hide \sigma \Derives{\alpha \hide \sigma}{H'}{R \hide \sigma} Q \hide \sigma}
  {P \Derives{\alpha}{H}{R} Q & 
    H' = H - \{\sigma\} \cup 
    \{ \tau \mid \sigma \in \olwilA{\ast}(R) \cap H  \}}
$$
In this modified hiding rule, the c-enabling check for the conclusion transition verifies that $\tau \not\in H'$ and thus that $\sigma \not\in \olwilA{\ast}(R) \cap H$. In other words, if the local clock $\sigma$ is in $H$ (blocking the visible action $\alpha$), then the concurrent context $R$ inside the clock scope must not conflict with $\sigma$.

\begin{example}[Clock-hiding with built-in Scheduling Test]
Consider the process $P \pdef \sigma \cseq A + a \col \{\sigma,a\} \cseq B \cpar b \col b \cseq \sigma $ where the local clock $\sigma$ also blocks the action $a$ by priority. But now the clock is guarded in $b \col b \cseq \sigma$ and so to become active for $P$, it needs the synchronisation with an external action $\ol{b}$. 
%As in the previous example $P \hide \sigma$  would not be c-coherent. 
The extended hiding rule $(\HideR^*)$ would generate the $a$-transition
$$
   P \hide \sigma \Derives{a}{\{ \tau, a \}}
   {(b \col b \cseq \sigma) \hide \sigma} 
   (B \cpar b \col b \cseq \sigma) \hide \sigma
 \text{ from }
   P \Derives{a}{\{\sigma\}}{b \col b \cseq\sigma} B \cpar b \col b \cseq \sigma.
$$
The $\tau$ enters into the blocking set because of the race test which discovers that  
$\sigma \in \olwilA{\ast}( b \col b \cseq \sigma) \cap \{ \sigma \}$, where $\{\sigma\}$ is the blocking set of the $a$-transition in $P$.  Hence, the $a$-transition of $P \hide \sigma$ is not (constructively) enabled anymore, which saves coherence. There is no requirement to exhibit any reconvergence for it. 
\qed
\end{example}

\medskip

Not surprisingly, \textit{mutatis mutandis}, a we could also envisage the same refinement for the restriction, suggesting the following $(\RestrR^*)$ rule:
$$
  \infer[(\RestrR^*)]
  {P \restrict \ell \Derives{\alpha}{H'}{R \restrict \ell} Q \restrict \ell}
  {P \Derives{\alpha}{H}{R} Q & 
     \alpha \not\in \{\ell, \ol{\ell} \} 
     & H' = H - \{\ell, \ol{\ell} \} \cup 
     \{ \tau \mid \ell \in \olwilA{\ast}(R) \cap H \}
  }
$$
We conjecture that the side conditions on the policy in Thm.~\ref{thm:summary-coherence-closure}(6,7) can be lifted to show with a modified SOS using the above rules $(\HideR^*)$ and $(\RestrR^*)$, the operations of hiding and restriction preserve coherence.

% \medskip\noindent 
% Note that the above $(\HideR^*)$ and $(\RestrR^*)$ rules are NOT used in this paper.

\subsection{Policies} 
\label{sec:C2-Policies}

In the following let us again write $\ul{\ell} \col H$ for $\ell \col (H \cup \{\ell\})$ to create reflexive prefixes. Also, we write $\ell^\ast\col H$ for the infinite loop $\ell^\ast \col H \pdef \ell\col H \cseq \ell^\ast \col H$ and write $\ell^\ast$ for $\ell^\ast\col \eset$. 

\medskip 

The following example shows that coherence and pivotability are independent notions.

\begin{example}[On coherence and pivotable processes]
The process $$P_1 \pdef a^\ast \cpar \ol{b}^\ast \cpar \ul{\ol{a}} \col b \cseq A + \ul{b} \cseq B$$ is c-coherent and pivotable with a policy $\pi$ such that 
$\pi \Vdash b \ordpre \ol{a}$, $\pi \Vdash \ol{a} \ordpre \ol{a}$ and $\pi \Vdash b \ordpre b$, while $\pi \Vdash \ol{b} \indep a$, $\pi \Vdash a \indep a$ and $\pi \Vdash \ol{b} \indep \ol{b}$. 
The process $P_1$ admits of two reductions 
\[
P_1 \Derives{a \cpar \ol{a}}{\{a, b\}}{\ol{b}^\ast} a^\ast \cpar \ol{b}^\ast \cpar A
\text{ and }
  P_1 \Derives{b \cpar \ol{b}}{\{\}}{a^\ast} a^\ast \cpar \ol{b}^\ast \cpar B.
\]
The first reduction $a \cpar \ol{a}$ is blocked by the presence of the offering $\ol{b}$ from a \textit{third} thread $\ol{b}^\ast$. 
It is not weakly enabled, because $\{a, b\} \cap \oliA(\ol{b}^\ast) \neq \eset$. 
For contrast, the process $$P_2 \pdef (a^\ast + \tau\cseq\ol{b}^\ast) \cpar \ul{\ol{a}} \col b \cseq A + \ul{b} \cseq B$$ is pivotable, too, but not c-coherent. 
Note that the reduction $a \cpar \ol{a}$ is not blocked anymore by the race test of rule $(\ComR)$ since it only looks into the initial actions $\iA(\tau\cseq\ol{b}^\ast)$ which does not contain $b$.
%  In our semantics, the synchronisation of $a$ and $\ol{a} \col b$ in $P_2$ is blocked by the race test of rule $(\ComR)$ which introduces $\tau$ into the blocking set
%  %
% $ % \[
% P_2 \Derives{\tau}{\{b,\tau\}}{\zero} \zero
% $ %\]
%  %
% because $\{ b \} \cap \oliA(a^\ast + \ol{b}^\ast) \not\subseteq \{ \ol{a} \}$. This transition therefore is not weakly enabled, because of $ \{b,\tau\} \cap (\oliA(\zero) \cup \{ \tau \}) \neq \eset$. 
%
% We can make the pivotable $P_1$ also c-coherent by forcing a sequential ordering between $a$ and $\ol{b}$ in the first thread, say $$P_3 \pdef a \cseq (\ol{b}^\ast + a^\ast) \cpar (\ul{\ol{a}} \col b \cseq A + \ul{b} \cseq B),$$ then the only reduction is 
% %
% $ %\[
% P_3 \Derives{\tau}{\{b\}}{\zero} \ol{b}^\ast + a^\ast \cpar A
% $ %\]
% %
% which is c-enabled. 
If we introduce a precedence into the first thread, say $$P_3 \pdef (a^\ast \col \ol{b} + \tau \cseq \ol{b}^\ast) \cpar (\ul{\ol{a}} \col b \cseq A + \ul{b} \cseq B),$$ then the process $P_3$ is again c-coherent but no longer pivotable. Like for $P_2$, the synchronisation $\tau = a \cpar \ol{a}$ is blocked and not weakly enabled, because of the introduction of $\tau$ into the blocking set by the race condition of $(\ComR)$. \qed
\end{example}

\begin{example}[On Coherent and Pivotable Processes]
 The process $P_1 \pdef \sigma + \ul{a} \col \sigma \cpar \ol{a}^\ast$ is c-coherent and pivotable, with policy $\sigma \ordpre a \in \pi$ and $\sigma \indep \ol{a} \in \pi$. It has both $a, \ol{a} \in \iA(P_1)$ but $\sigma\not\in \iA(P_1)$ since the clock is not offered by the parallel process $\ol{a}$. The synchronisation $\tau = a \cpar \ol{a}$ can go ahead.
 On the other hand, 
 the process $P_2 \pdef \sigma + \ul{a} \col \sigma \cpar \ol{a}^\ast + \sigma \col \ol{a}$ has all three actions $\sigma, a, \ol{a} \in \iA(P_2)$ initial. It is c-coherent but not pivotable. Since both $a$ and $\ol{a}$ are in a precedence relationship with $\sigma$, $P_2$ does not satisfy concurrent independence of $\sigma$ with neither $a$ nor $\ol{a}$. Formally, any policy $\pi$ for $P_2$ must have $\sigma \ordpre a \in \pi$ and $\ol{a} \ordpre \sigma \in \pi$ violating the pivot requirement that if $\sigma \ordpre a \in \pi$,  then $\sigma \indep \ol{a} \in \pi$. 
%The same is true of the process $P_3 \pdef \sigma + (\ul{a}\col\sigma \cpar \ol{a}^\ast)$ which can perform both $\sigma$ and a (silent) reduction $\tau = \alpha \cpar \ol{\alpha}$ that has $\sigma$ in its blocking set. 
\qed
\end{example}

\begin{example}[On Pivotable but not Input-scheduled Processes]
 The process $a + b \col c \cpar \ol{a} + \ol{b} \col \ol{a}$ is pivotable. It conforms to $\pi$ with precedences $c \ordpre b \in \pi$ and $\ol{a} \ordpre \ol{b} \in \pi$, while at the same time $\ol{c} \indep \ol{b} \in \pi$ and $a \indep b \in \pi$. However, assuming $a \neq b$, the process is not input-scheduled, because it has $\ol{a} \ordpre \ol{b}\in \pi$ while $\ol{a} \indep \ol{b} \in \pi_{is}$. The same is true of the process in Ex.~\ref{ex:binary-blocking}. \qed
 \end{example}

\begin{figure}[t]
\hfill 
\begin{minipage}{0.5\textwidth}
\begin{subfigure}{\textwidth}
 \centering
 \includegraphics[scale=0.6]{images/wAND-mem}
\caption{A Boolean Read-Write process with c-coherent states $\wAND_0 \Vdash \mem$ and $\wAND_1 \Vdash \mem$.}
\label{fig:wAND-mem}
\end{subfigure}
\end{minipage}
\hfill
\begin{minipage}{0.4\textwidth}
\begin{subfigure}{\textwidth}
 \centering
 \includegraphics[scale=0.6]{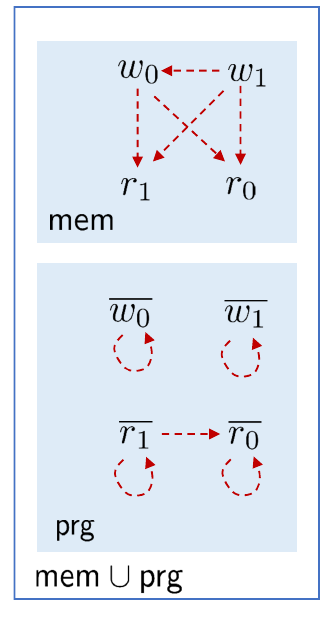}
\caption{The $\mem \cup \prg$ Pivot Policy.}
\label{fig:mem-prg-policy}
\end{subfigure}
\end{minipage}
\hfill
\caption{A Boolean Read/Write memory process (Fig.~\ref{fig:wAND-mem}) and the common pivot policy $\mem \cup \prg$ (Fig.~\ref{fig:mem-prg-policy}) governing both memory and programs.}
\label{fig:mem-prg2}
\end{figure}

\begin{example}[On Read/Write Memory, reloaded]
\label{ex:read-write-reloaded}
Fig.~\ref{fig:mem-prg2} show that the processes $\wAND_v$ conform to a policy $\mathsf{mem}$ with $w_i \ordpre r_j \in \mathsf{mem}$ and $w_1 \ordpre w_0 \in \mathsf{mem}$. Obviously, $\mathsf{mem}$ is a pivot policy. The processes $\wAND_v$ are also c-coherent. 
Consider any two interference-free transitions $Q \Derives{\alpha_1}{H_1}{\zero} Q_1$ and $Q \Derives{\alpha_2}{H_2}{\zero} Q_2$ of a derivative $Q$ of $\wAND_v$. There are two cases. In the first case, both transitions are from the same prefix $\alpha_1 = \ell = \alpha_2$ and $Q_1 \equiv Q' \equiv Q_2$. Then we must have $\ell \not\in H$ by Def.~\ref{def:coherence}. But for all transitions with a label $\ell$ that is not self-blocking the Read-Write Memory ensures that
%One verifies by simple case analysis that for any derivative $Q$ of $\wAND_v$, if 
%$Q \fsstep{\ell \col H[\zero]} Q'$ 
%$Q \Derives{\ell}{H}{\zero} Q'$ 
%and $\ell \not\in H$ 
there is a repetition 
%$Q' \fsstep{\ell \col H'[\zero]} Q''$ 
$Q' \Derives{\ell}{H'}{\zero} Q'$ 
with $H' = H$. This yields the reconvergence required by Def.~\ref{def:coherence} as one can show. For the second case suppose both prefixes are distinct, i.e., $\alpha_1 \neq \alpha_2$.
%$Q \fsstep{\ell_1 \col H_1[\zero]} Q_1$ 
%$Q \Derives{\ell_1}{H_1}{\zero} Q_1$ 
%and 
%$Q \fsstep{\ell_2 \col H_2[\zero]} Q_2$. 
%$Q \Derives{\ell_2}{H_2}{\zero} Q_2$. 
Then, obviously, %if $\ell_1 \neq \ell_2$, 
the two transitions cannot be interference-free in the Read-Write Memory, because of the choice of blocking sets which orders all summation prefixes in a linear precedence order. Thus, Def.~\ref{def:coherence} is satisfied, too, for trivial reasons. 
So, the memory, which satisfies $\wAND_v \cnf \mathsf{mem}$, is pivotable.

\medskip 

\noindent Note that $\wAND_v$ is not a reflexive process, since none of its actions is self-blocking. We can put the memory in parallel with arbitrarily many writers and readers. 
Consider a parallel composition $\wAND_v \cpar P$ where $P$ is any single-threaded process that synchronises via $\ol{w_i}$ for writing and $\ol{r_i}$ for reading. Each writing action of $P$ is a self-blocking prefix $\ol{w_i} \col \ol{w_i} \cseq Q$ and each reading is a choice $\ol{r_1} \col \ol{r_1} \cseq R_1 + \ol{r_0} \col \{\ol{r_0},\ol{r_1}\} \cseq R_0$ where the continuation processes $R_0$ and $R_1$ are selected according to the content of the memory, by matching with the $r_0$ action of $\wAND_0$ or the $r_1$ action of $\wAND_1$. If the memory would offer both a synchronisation on $r_0$ and $r_1$ (which it does not), the reader would take the match on $r_1$, because of the precedences. 
 If $Q$, $R_0$ and $R_1$ follow the same principle, such a sequential program is a discrete process $P$ that is c-coherent and conformant to the policy $\prg$ with $\ol{r_1} \ordpre \ol{r_i} \in \prg$ and $\ol{r_0} \ordpre \ol{r_0} \in \prg$.
 Now we observe that the union $\mem \cup \prg$ of policies is a pivot policy. Since $\wAND_v \cnf \mem$ and $P \cnf \prg$ we have $\wAND_v \cnf \mem \cup \prg$ as well as $P \cnf \mem \cup \prg$. Then, by Prop.~\ref{prop:cpar-coherent} we get 
 $$\wAND_v \cpar P \cnf \mem \cup \prg$$ 
 i.e., the composition is c-coherent and pivotable.
 In fact, we can put the memory in parallel with arbitrarily many sequential programs c-coherent for $\prg$ and the composition $\wAND_v \cpar P_1 \cpar P_2 \cpar \cdots \cpar P_n$ will be c-coherent for $\mem \cup \prg$. In particular, all reductions will be confluent under strong enabling. Technically, what happens is that first all writers of value $1$, synchronising as $w_1 \cpar \ol{w_1}$, will go ahead. Then all writers of $\zero$, synchronising as $w_0 \cpar \ol{w_0}$ become enabled. Only then the readers will be able to retrieve the final value via $r_0 \cpar \ol{r_0}$ or $r_1 \cpar \ol{r_1}$ rendez-vous synchronisations. \qed
\end{example}

 Observe that the pivot policy $\mem \cup \prg$ of Ex.~\ref{ex:read-write-reloaded} permits a choice (and precedence) between both the receiver actions $w_0$, $w_1$ but also between the sender actions $\ol{r_0}$, $\ol{r_1}$. As a result, the combined memory and program, which is c-coherent for $\mem \cup \prg$, is not an input-scheduled process in the sense of Def.~\ref{def:is-policy}. This breaks the assumption on which the theory of classical prioritised process algebras like~\cite{CamilleriWin95} or~\cite{CleavelandLN01} are based. In reference to the former, Phillips writes:
 \begin{quote} 
 \textit{They [Camilleri and Winskel in \ccscw] decide to sidestep this question by breaking the symmetry in CCS between inputs and outputs, and only allowing prioritised choice on input actions. We feel that this complicates the syntax and operational semantics,
 and should not be necessary.~\cite{Phillips01}}
 \end{quote}
 Our example and analysis indeed highlights the usefulness of the more general setting suggested by Phillips. In order to explain the theory of deterministic shared memory we need a larger class of processes than those considered in \ccscw.

\medskip 

Note that all the examples of Sec.~\ref{examples:scheduling}
have the ``input-scheduled'' property that 
(non-trivial, i.e. non-self-blocking) precedences only exist between receiver actions $\A$, while sender actions $\coA$ never appear in any choice contexts. These are the processes covered by~\cite{CamilleriWin95} which restrict prioritised choice to receiver actions (``input-guarded, prioritised choice'') in the same way as the programming language 
Occam~\cite{Barrett89} with its ``PRIALT'' construct does (cf. \cite{CamilleriWin95}).

\medskip 

\begin{example}
\label{ex:wired-and}
The memory cell schedules readers after the writers and among the writers prefers to engage first with with those that set the value to $1$, in the prioritised sum. This implements an Esterel-style resolution function for multi-writer convergecast, multi-reader broadcast of boolean signals. All concurrent readers receive the same value, which is the combined-AND of all values written by any concurrent writer. If any process writes a boolean $1$ it will win the competition. This is also known as a ``wired-AND''. Likewise, we can implement a memory cell that models a ``wired-OR'' if we swap the roles of $w_0$ and $w_1$ in the prioritised choices of $\wAND_v$. 

Note that the wired-AND behaviour only applies to concurrent threads. A single thread can use the $\wAND$ cell as a local memory and destructively update it sequentially as it wishes. Consider the difference: Both \textit{concurrent} writers in $\ol{w_1} \col \ol{w_1} \cpar \ol{w_0} \col \ol{w_0} \cpar \wAND_0 \cpar R$ or in $\ol{w_0} \col \ol{w_0} \cpar \ol{w_1} \col \ol{w_1} \cpar \wAND_0 \cpar R$ will necessarily make the reader $R \eqdef \ol{r_1} \col \ol{r_1} \cseq R_1 + \ol{r_0} \col \{\ol{r_0}, \ol{r_1}\} \cseq R_2$ detect a value $r_0$ and therefore end in state $R_2$ (wired-AND of $w_0$ and $w_1$).
 The two single-threaded \textit{sequential} writings $\ol{w_1} \col \ol{w_1} \cseq \ol{w_0} \col \ol{w_0} \cpar \wAND_0 \cpar R$ and $\ol{w_0} \col \ol{w_0} \cseq \ol{w_1} \col \ol{w_1} \cpar \wAND_0 \cpar R$ will generate two different readings by $R$: The first is the reading $r_0$, leading to $R_2$. The second is the reading $r_1$, leading to $R_1$. \qed
\end{example}

\begin{example}[Conformance of \ABRO] 
\label{ex:multicastabro}

% \begin{figure}[t]
% \begin{eqnarray*}
% \ABRO &\pdef& 
%     \sigma \cseq (\sA \cpar \sB \cpar \sR \cpar \sO) \restrict \{ s,t \}\\ 
%  \sA\, &\pdef& 
%  k \col k \cseq \zero + a \col \{k, a \} \cseq \ol{s} \col \ol{s} \cseq \zero + 
%  \sigma \col \{k, a \} \cseq \sA \\
%  \sB\, &\pdef& 
%  k \col k \cseq \zero + b \col \{k, b \} \cseq \ol{s} \col \ol{s} \cseq \zero + 
%  \sigma \col \{k, b \} \cseq \sB \\
%  \sR\, &\pdef& 
%  r \col r \cseq \sR' + \tau \col {r} \cseq \sigma \cseq \sR \\
% \sR'\hspace{-0.4mm}  &\pdef& 
%  \ol{k} \cseq \sR' + \tau \col \ol{k} \cseq \ABRO \\
%  \sO\hspace{0.25mm} &\pdef& 
%  k \col k \cseq \zero + s \col k \cseq \sO 
%  + t \col \{k, t, s \} \cseq \ol{o} \col \ol{o} \cseq \zero + \sigma \col \{ k, s, t \} \cseq \sO.
% \end{eqnarray*}
% \caption{Simplified version of \ABRO.}
% \label{fig:simplified-ABRO}
% \end{figure} 

% To model full multi-cast \ABRO in \ccslm, we could use sharable (i.e., non-self-blocking) prefixes on internal actions to show how they are useful to model the control-flow of Esterel programs, specifically termination and reset.
% %
% \medskip 

Let us verify that \ABRO from Fig.~\ref{fig:ABRO} is c-coherent and pivotable using Thm.~\ref{thm:summary-coherence-closure}.
%We consider the simplified version seen in Fig.~\ref{fig:simplified-ABRO}. 
%

The choices in each of the threads $\sA$, $\sB$, $\sR$, $\sO$ and $\sT$ are fully resolved by the blocking sets, i.e., every pair of actions arising from distinct prefixes in the summations are in precedence relationships. This means that any pair of distinct prefixes are interfering as condition (1) of Def.~\ref{def:interference-free} is violated. 
 Every $\tau$-prefix in a derivative of $\sA$, $\sB$, $\sR$, $\sO$ or $\sT$ is self-interfering since condition (3) of Def.~\ref{def:interference-free} is violated. Also, each self-blocking rendez-vous prefix in a derivative is self-interfering since condition (3) of Def.~\ref{def:interference-free} is violated. 
 All the rendevous prefixes in a derivative that are not self-blocking have the property that the action can be repeated with the same blocking in the continuation process. This happens for channels $\ol{k}$ and $s$ in the recursions $\sR' \pdef \ol{k} \cseq \sR' + \ldots$ and $\sT \pdef \ldots + s \col k \cseq \sT + \ldots$, respectively. Thus, the threads $\sA$, $\sB$, $\sR$, $\sO$ and $\sT$ are c-coherent. 

 A suitable policy $\pi$, for which the threads $\sA$, $\sB$, $\sR$, $\sO$ and $\sT$ are conformant, enforces the precedences 
 \begin{itemize} 
 \item $k \ordpre \ell \in \pi$ for $\ell \in \{ k, a, b, s, t, \ol{t} \}$ 
 \item $\ell \ordpre \ell \in \pi$ for $\ell \in \{ r, k, a, b, \ol{s}, t, \ol{o} \}$
 \item $s \ordpre \ol{t} \in \pi$
 %\item $\ell \ordpre \sigma \in \pi$ for $\ell \in \{ a, b, s, t \}$.
 \end{itemize}
 The dual matching pairs for these precedences on distinct labels, however, are all concurrently independent. More precisely, $\ol{k} \indep \ol{\ell} \in \pi$ for all $\ell \in \{ a, b, s, t, \ol{t} \}$
 %$\ol{\ell} \indep \ol{\ell}$ for $\ell \in \{ r, k, a, b, \ol{s}, t, \ol{o} \}$ 
 and $\ol{s} \indep t \in \pi$.
 % and $\ol{\ell} \indep \sigma \in \pi$ for $\ell \in \{ r, o, a, b, s, t \}$.
%
Thus, the policy $\pi$ (and thus $\pi \restrict \R$) is a pivot policy and so the parallel composition $\sA \cpar \sB \cpar \sR \cpar \sO$ is c-coherent for $\pi$ by Thm.~\ref{thm:summary-coherence-closure}(5). Further, $\pi$ is precedence-closed for $\{ s, t \} \cup \{ \ol{s}, \ol{t} \}$ which means that $(\sA \cpar \sB \cpar \sR \cpar \sO) \restrict \{s,t\}$ is c-coherent for $\pi \restrict \{s,t\}$ by Thm.~\ref{thm:summary-coherence-closure}(7). Finally, since $\{ \} \subseteq \{ \ell \mid \ell \ordpre \sigma \in \pi \restrict \{s,t\} \} = \{a, b \}$, the process $\ABRO$ (which starts with a prefix $\sigma \col H$ where $H = \eset$) is c-coherent for $\pi \restrict \{s,t\}$ by Thm.~\ref{thm:summary-coherence-closure}(3). \qed
\end{example}

%%% Local Variables:
%%% mode: latex
%%% TeX-master: "synpatick-lipics.tex"
%%% End:
  % clean

%!TEX root = synpatick-lipics.tex

\section{From Asynchronous Priorities to Synchronous ones, via Clocks}
\label{sec:app-pilgimage}

This final and rather non-typical section gives to the curious reader a nice view of our personal pilgrimage of previous work looking at reaching more control and determinism in asynchronous process algebra.  This lead us to consider of adding clocks to the previous literature and to study confluence in an synchronous setting. 

\medskip
\noindent 
Note that the SOS rules presented in this section are not part of \ccslm, but only presented to provide links with existing work and for the sake of recording our speculations.

%!TEX root = synpatick-lipics.tex

%----------------------------------------------------
\subsection{\ccs with (Weak) Priorities}
%----------------------------------------------------

Let us look in more detail on how a precedence mechanism such as ``write-before-read'' might be derived from the classical priorities studied for \ccs. 
As a simple example, suppose $S$ is a (synchronous write, asynchronous read) memory cell that can be read at any time offering the read action $r$, but only be written once with the write action $w$.
Ignoring data values, the cell would be modelled in plain \ccs
by the equations 
$S \eqdef w \cseq S' + r \cseq S''$ where 
$S' \eqdef r \cseq S'$ and
$S'' \eqdef r \cseq S'' + e \cseq S''$. 
\begin{figure}[t!]
  \centering
  \includegraphics[scale=0.6]{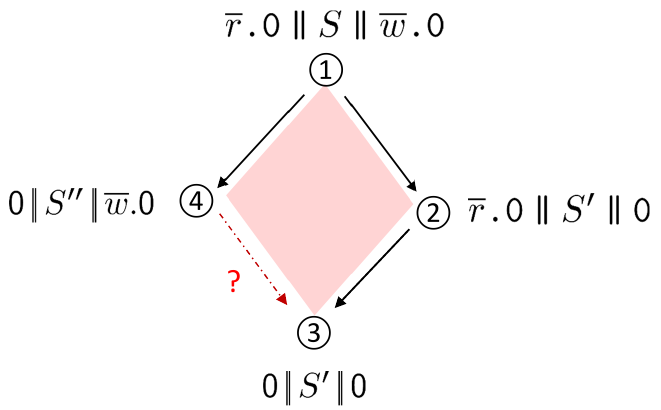}
  %$S \eqdef w \cseq S' + r \cseq S''$
\caption{Symmetric Choice. The final outcome, in states \cthree and \cfour, is non-determinate because it depends on the order of read and write. The continuations states are $S' \eqdef r \cseq S'$ and $S'' \eqdef r \cseq S'' + e \cseq S''$. For confluence, the error state \cthree should permit a (delayed) read to move forward and (eventually) meet with state \cfour, indicated by the dashed red arrow. However, this is not in general guaranteed.}
\label{fig:conf-prio-1}
\end{figure}
Initially, memory cell $S$ is empty and offers a choice between a read action $r$ and a write action $w$. If the write action is performed, the cell receives a value and changes to state $S'$. In this filled state $S'$, only read actions are admissible to access the stored value, the state does not change. 
When the read action is executed on the empty cell $S$, however, the continuation state is $S''$ which represents an error state. In the error state $S''$ only a read or an error action $e$ is possible, but no write and the cell remains in error. 
Observe that only the read action is always admissible in both states $S$ and $S'$, while the write is admissible only in state $S$. 
%This models  
If $S$ is put in parallel with a reader offering $\ol{r}$ and a writer offering $\ol{w}$, then the final state of $S$ depends on 
the order between the $\ol{r}$ and $\ol{w}$ actions. If we synchronise with the external sequence $\ol{r} \cseq \ol{w}$ on $S$ where the read action is first, the writer gets blocked on the error state $S''$, 
%we end up in state $S'$ 
while if we exercise $\ol{w} \cseq \ol{r}$ where the write action is first, both can proceed and we obtain filled state $S'$. In term rewriting system jargon~\cite{KlopTRS} this is called ``critical pair''.

Formally, let $R \eqdef \ol{r} \cseq \zero$ be a reader process with a single read and $W \eqdef \ol{w} \cseq \zero$ the writer with a single write, both terminating in the empty behaviour $\zero$. 
Then, according to the operational semantics of \ccs, we have the two execution sequences
$
  R \cpar S \cpar W \Derives{\tau}{}{} 
  \zero \cpar S'' \cpar W 
  $
  and
  $
  R \cpar S \cpar W \Derives{\tau}{}{} R \cpar S' \cpar \zero \Derives{\tau}{}{} 
  \zero \cpar S' \cpar \zero
 $,
where both configurations $\zero \cpar S'' \cpar W$ and 
$\zero \cpar S'' \cpar \zero$ are distinct normal forms, in the sense that no reductions are possible anymore. 
This is pictorially illustrated in Fig.~\ref{fig:conf-prio-1}. 
Observe that the behaviour is not only \textit{non-deterministic} in the sense that there is more than one scheduling sequence. It is also non-determinate as the final outcome depends on the scheduling decisions.
  
\medskip 

Among the first \ccs extensions with priorities is \ccscw introduced by Camilleri and Winskel~\cite{CamilleriWin95}. It provides a \textit{prioritised choice} operator
$
Q \eqdef \alpha_1.Q_1 \wpsum \alpha_2.Q_2
$
called ``prisum'' that gives the action $\alpha_1$ priority over $\alpha_2$ and blocks the second prefix $\alpha_2$ in case $\alpha_1 = \alpha_2$. 
Such priorities help us resolve the issue. 
The problem with non-determinism originates from the preemptive behaviour of the write-read choice $w \cseq S' + r \cseq S''$. It permits the read and the write actions to compete with each other for the next state of the process $S$. In data-flow and functional programming, we resolve the conflict by making the write action take precedence over the read one. First, the data-flow variable must be written by the \textit{producer} process, then it can be read by the \textit{consumers}, if any. This ``write-before-read'' scheduling policy is a \textit{prioritised choice}, represented in \ccscw as
$
 S \eqdef w \cseq S'  \wpsum r \cseq S''
$
in which the prisum $ \wpsum$ operator enforces a precedence from left to right. The store $S$ is willing to engage in a read action $r$ (to the right of $ \wpsum$) only if the environment does not offer a write $w$ (on the left of $ \wpsum$) at the same time, concurrently. 

\medskip 

Another important study in the endeavour of adding priorities to \ccs is the work of Phillips~\cite{Phillips01}. In Phillips' extension of \ccs with priority guards, referred to as \ccsp, the previous behavior is expressed by 
$ 
 S \eqdef {\eset\of w} \cseq \, S' +  {\{ w \}\of r} \cseq \, S''
$,
where each action is pre-annotated by a set of actions that take priority. Process $S$ will synchronise with the reader process $R$ to move to the error state $S''$ with \textit{guarded} action $\{w\}\of r$. The \textit{priority guard} $\{w\}$ states that this can only happen if the environment does not offer a write $w$. Otherwise, the only action available is $\eset\of w$ which makes the cell move into state $S'$. The write action has an empty guard and so cannot be blocked. 
If both $\ol{w}$ and $\ol{r}$ are offered concurrently, then only the higher-prioritised write action $\ol{w}$ will go ahead. 

\medskip 

The standard way to implement this scheduling in prioritised \ccs is via a transition relation that keeps track of blocking information. In line with the approaches of~\cite{CamilleriWin95,Phillips01,CleavelandLN01}, let us write
$% \begin{eqnarray} 
  P \Derives{\alpha}{H}{} P'
%\label{eqn:blocking-label}
$ %\end{eqnarray}
to denote that process $P$ can engage in an action $\alpha$ to continue as $P'$, provided the concurrent environment does not offer any action from $H$.
We will call $H$ the \textit{blocking set} of the transition.
%~\eqref{eqn:blocking-label}. 
In Plotkin's SOS, the rules for prefix and choice in \ccscw are the following\footnote{%
For the fragment of \ccscw that we are considering, the precise connection is the following: $\vdash_R P \Derives{\alpha}{}{} P'$ in the notation of \ccscw iff $R$ is a subset of o-actions (``output'' in our parlance)
%and there is a subset $L$ of $in$-actions 
such that $P \Derives{\alpha}{H}{} P'$ and either $\alpha$ is an o-action and $H = \eset$ or $\alpha$ is an i-action with $\ol{\alpha} \in R$ and $R \cap \ol{H} = \eset$.}:
\[ %
    \begin{Prooftree}
    \label{actpri}
        \Bproof
          \\[(\ActR)] \alpha \cseq P \Derives{\alpha}{\eset}{} P
        \EEproof
        \quad
        \Bproof
          P \Derives{\alpha}{H}{} P'
          \\[$(\PriSumR_1)$]
          P \wpsum Q \Derives{\alpha}{H}{} P'
        \EEproof
       \quad 
        \Bproof
          \Lproof
          Q \Derives{\alpha}{H}{} Q'
          \ANDproof
            \tau, \alpha \not\in \iA(P)
          \Rproof[$(\PriSumR_2)$]
            P \wpsum Q 
            \Derives{\alpha}{H \cup \iA(P)}{} Q'
        \EEproof
    \end{Prooftree}
\] 
where $\iA(P)$ represents the set of \textit{initial actions} of $P$ that are offered by $P$ in some environment and $\wpsum$ is the priority sum defined by \ccscw.
% and $\wiR^\ast(P) \eqdef \wiR(P) \setminus \{\tau\}$ \rll{not yet defined}. 
From rule $(\ActR)$, individual prefixes generate transitions without any blocking constraints. Rules $(\PriSumR_1)$ and $(\PriSumR_2)$ describe how the transitions of a choice $P \wpsum Q$ are obtained from the transitions of $P$ and of $Q$: More in details, rule $(\PriSumR_1)$ lets all transitions of the higher-prioritised process $P$ pass through to form a transition of the choice $P \wpsum Q$ without restriction.
This is different for $Q$, where, according to rule $(\PriSumR_2)$, an action of the lower-prioritised process $Q$ can only be executed, and thus $P$ preempted by it, if $P$ does not have an initial silent action and if it does not offer the same action $\alpha$, already. All initial actions of $P$ are added to the blocking set of $Q$'s transition.  This will ensure that it gets blocked in concurrent contexts in which $P$ has a communication partner. The blocking sets take effect in the rules for parallel composition displayed below:
\[ %
    \begin{Prooftree}
        \Bproof
          \Lproof 
             P_1 \Derives{\alpha_1}{H_1}{} P_1'
          \ANDproof 
             H_1 \cap \ol{\iA}(P_2) = \{\, \}
          \Rproof[$(\ParR_1)$] 
           P_1 \cpar P_2 
              \Derives{\alpha_1}{H_1}{} P_1' \cpar P_2
        \EEproof
        \qquad
        \Bproof
          \Lproof 
             P_2 \Derives{\alpha_2}{H_2}{} P_2'
          \ANDproof 
             H_2 \cap \ol{\iA}(P_1) = \{\, \}
          \Rproof[$(\ParR_2)$] 
           P_1 \cpar P_2 
              \Derives{\alpha_2}{H_2}{} P_1 \cpar P_2'
        \EEproof
    \end{Prooftree}
\]
\[ %
    \begin{Prooftree}
        \Bproof
          \Lproof 
             P_1 \Derives{\alpha_1}{H_1}{} P_1'
          \ANDproof 
             P_2 \Derives{\alpha_2}{H_2}{} P_2'
          \ANDproof
             \alpha_1 = \ol{\alpha}_2
          \ANDproof
             H_1 \cap \ol{\iA}(P_2) = \{\, \}
          \ANDproof
             H_2 \cap \ol{\iA}(P_1) = \{\, \}
          \Rproof[$(\ParR_3)$] 
           P_1 \cpar P_2 
              \Derives{\tau}{H_1 \cup H_2}{} P_1' \cpar P_2'
        \EEproof
    \end{Prooftree}
\]
More in details, rules $(\ParR_i)$ for $i \in \{1,2\}$ define an action $\alpha_i$ by $P_1 \cpar P_2$ offered in one of the sub-processes $P_i$. Such is enabled if the competing concurrent process $P_{3-i}$ has no initial action that is in the blocking set $H_i$ of action $\alpha_i$, i.e. when the condition $H_i \cap \ol{\iA}(P_{3-i}) = \{\, \}$ is satisfied.
Rule $(\ParR_3)$ implements the synchronisation between an action $\alpha_1$ in $P_1$ and the associated matching action $\alpha_2 = \ol{\alpha}_1$, originating from $P_2$. The side-condition is the same: each process $P_i$ only accepts the matching action call from $P_{3-i}$ if it does not offer any action that blocks action $\alpha_i$.
\begin{figure}[!t]
  \centering
  \includegraphics[scale=0.6]{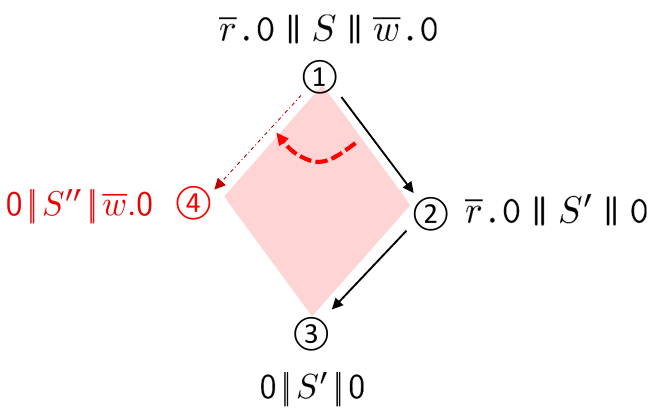}
  $S \eqdef w \cseq S' \,\myr{\csum}\, r \cseq S''$
\caption{Prioritised Choice. The read-write data race on memory cell $S$ resolved by scheduling priority.  The write takes precedence over the read, indicated by the dashed arrow. This removes the dotted transition on the top left. Since there is only one execution path, the system is deterministic. The continuations states $S' \eqdef r \cseq S'$ and $S'' \eqdef r \cseq S'' + e \cseq S''$ of the cell are the same as in Fig.~\ref{fig:conf-prio-1}. 
}
\label{fig:conf-prio-2}
\end{figure}
Let us see how this blocking semantics works for our running example of a write-once memory cell, where we now enforce a write-before-read priority on the empty cell $S$:
$ %\[ 
   S \eqdef w \cseq S' \wpsum r \cseq S'' \mbox{ and }
   S' \eqdef r \cseq S' \mbox{ and }
   S'' \eqdef r \cseq S'' + e \cseq S''.
$ %\] 
The cell $S$ has two potential transitions that can be fired, namely
\[ 
 \begin{Prooftree}
  \Bproof
    \\[$(\ActR)$]
      w\cseq S' \Derives{w}{\eset}{} S'
    \\[$(\PriSumR_1)$]
      S \Derives{w}{\eset}{} S'
  \EEproof
  \quad\mbox{ or }\quad
  \Bproof
    \Lproof 
      \\[$(\ActR)$]
        r \cseq S'' \Derives{r}{\eset}{} S''
    \ANDproof
        \tau, r \not\in \iA(w \cseq S')
    \Rproof[$(\PriSumR_2)$]
      S \Derives{r}{\{ w \}}{} S''
  \EEproof
 \end{Prooftree}
\] 
considering that $\iA(w \cseq S') = \{ w \}$.
The latter transition synchronises with the reader process' $\ol{r}$ transition
to give 
$ %\begin{eqnarray*}
  R \cpar S \Derives{\tau}{\{ w \}}{} \zero \cpar S''
%\label{eqn:send-read-comm}
$ %\end{eqnarray*}
as verified by the derivation
\[ 
 \begin{Prooftree}
  \Bproof
   \Lproof 
    \\[$(\ActR)$] 
      R \Derives{\ol{r}}{\eset}{} \zero
   \ANDproof
      \stackrel{\text{as above}}{\vdots} 
      \\ 
      S \Derives{r}{\{ w \}}{} S''
   \ANDproof
      \eset \cap \ol{\iA}(S) = \eset 
   \ANDproof 
      \{ w \} \cap \ol{\iA}(R) = \eset
   \Rproof[$(\ParR_3)$]
     R \cpar S \Derives{\tau}{\{ w \}}{} \zero \cpar S''
  \EEproof
 \end{Prooftree}
\]
in which the side-conditions $\{ w \} \cap \ol{\iA}(R) = \{ w \} \cap \{ r \} = \eset$ and $\eset \cap \ol{\iA}(S) = \eset$ of rule $(\ParR_3)$ are satisfied. 
This is fine since a single reader should be able to access the store, as long as there is no conflicting write. Note that 
%the label $\tau\of\{w\}$ of 
the resulting silent transition
%~\eqref{eqn:send-read-comm} 
is guarded by the write action $w$ to record this constraint. Now if we add the writer process, too, then we get stuck. The writer's initial action $\ol{w} \in \iA(W)$
conflicts with the silent transition.
%~\eqref{eqn:send-read-comm}. 
The execution 
%of~\eqref{eqn:send-read-comm} 
in the context $E = [\cdot] \cpar W$ is blocked by rule $(\ParR_1)$, because $\{ w \} \cap \ol{\iA}(W) = \{ w \} \cap \{ w \} \neq \eset$. The derivation tree 
\[ 
 \begin{Prooftree}
  \Bproof
    \Lproof 
      \Lproof 
        \\[$(\ActR)$]
        R \Derives{\ol{r}}{\eset}{} \zero
      \ANDproof 
        \Lproof 
          \\[$(\ActR)$]
             r \cseq S'' \Derives{r}{\eset}{} S''
        %\ANDproof 
        %  \text{side-cond. ok}
        \Rproof[$(\PriSumR_2)$]
           S \Derives{r}{\{w\}}{} S''
      %\ANDproof 
      %   \text{side-cond. ok}
      \Rproof[$(\ParR_3)$]
      R \cpar S \Derives{\tau}{\{w\}}{} \zero \cpar S''
    \ANDproof 
      \myr{\{w\} \cap \iA(W) \neq \eset}
    \Rproof[$(\ParR_1)$]
      R \cpar S \cpar W \,\,\,\slash\!\!\!\!\!\!\!\Derives{\tau}{\{w\}}{} 
      \zero \cpar S'' \cpar W
  \EEproof
 \end{Prooftree}
\] 
where the side-conditions of $(\PriSumR_2)$ and of $(\ParR_3)$ are satisfied,
blocks at side-condition of rule $(\ParR_1)$.
%\[ 
%  R \cpar S \cpar W \nsDerives{\tau:\{w\}}{} \zero \cpar S'' \cpar W.
%\] 
On the other hand, the writing of $S$ by $W$ followed by the reading of $S'$ by $R$ can be derived, i.e., we have 
$ %\[ 
   R \cpar S \cpar W \Derives{\tau}{\eset}{} 
      R \cpar S' \cpar \zero
   \text{ and }
   R \cpar S' \cpar \zero \Derives{\tau}{\eset}{} 
      \zero \cpar S' \cpar \zero. 
$ %\]
The derivation trees in the SOS are as follows
\[ 
 \begin{Prooftree}
  \Bproof
    \Lproof
       \Lproof 
         \\[$(\ActR)$]
         w \cseq S' \Derives{w}{\eset}{} S'
         \\[$(\PriSumR_1)$]
         S \Derives{w}{\eset}{} S' 
       %\ANDproof
       % \eset \cap \iA(R) = \eset 
       \Rproof[$(\ParR_2)$]
       R \cpar S \Derives{w}{\eset}{} R \cpar S'
    \ANDproof
      \\[$(\ActR)$]
      W \Derives{\ol{w}}{\eset}{} \zero
    %\ANDproof
    %   \eset \cap \iA(W) = \eset
    %\ANDproof 
    %   \eset \cap \iA(R \cpar S) = \eset
    \Rproof[$(\ParR_3)$] 
     R \cpar S \cpar W \Derives{\tau}{\eset}{} 
      R \cpar S' \cpar \zero
  \EEproof
  \quad
  \Bproof
    \Lproof
         \\[$(\ActR)$]
         R \Derives{\ol{r}}{\eset}{} \zero
    \ANDproof
      \\[$(\ActR)$]
         S' \Derives{r}{\eset}{} S'
      \\[$(\ParR_1)$]
         S' \cpar 0 \Derives{r}{\eset}{} S' \cpar \zero
        %\ANDproof
        %   \eset \cap \iA(W) = \eset
        %\ANDproof 
        %   \eset \cap \iA(R \cpar S) = \eset
    \Rproof[$(\ParR_3)$] 
     R \cpar S' \cpar \zero \Derives{\tau}{\eset}{} 
      \zero \cpar S' \cpar \zero 
  \EEproof
 \end{Prooftree}
\] 
where the blocking side-conditions of $(\ParR_1)$, $(\ParR_2)$ and $(\ParR_3)$ 
% $\eset \cap \iA(R) = \eset$ of $\ParR_1$ and $\eset \cap \iA(W) = \eset$ and $\eset \cap \iA(R \cpar S) = \eset$ of $\ParR_3$ 
are all trivially satisfied. 
The scheduling which resolves the non-determinism is depicted in Fig.~\ref{fig:conf-prio-2}. 

\medskip

\noindent Note that some of the rules in the above have the same names but are different from those used in \ccslm.

 % clean

%!TEX root = synpatick-lipics.tex

%---------------------------------------------
\subsection{Asynchronous \ccs with Strong Priorities}
%---------------------------------------------

Our example from the previous D.1 subsection shows how (weak) priority-based scheduling helps us to exercise control over the execution order in an asynchronous process calculus with rendez-vous communication. The surprising observation is that very little needs to be changed in the traditional notion of priority for \ccs to achieve confluence. For this we need to look how to ``tune'' the classical weak priorities for our purposes. Our fix will make the difference between \textit{weak} priority-based scheduling of \ccscw and \ccsp and the \textit{strong} priority-based scheduling of \ccslm (Def.~\ref{def:c-enabling}).
What we want is that every process using only strong priority satisfies a \textit{weak local confluence} property~\cite{KlopTRS}: Suppose $P$ admits two transitions to different states $P_1 \neq P_2$ that do not block each other, i.e.,
$ % \[ 
  P \Derives{\alpha_1}{H_1}{} P_1
  \text{ and } 
  P \Derives{\alpha_2}{H_2}{} P_2
$ %\] 
such that $\alpha_1 \not\in H_2$ and $\alpha_2 \not\in H_1$. 
Then, each action $\alpha_i$ (for $i\in\{1,2\}$) is still admissible in $P_{3-i}$, and there exist $P'$ and $H_i' \subseteq H_i$ such that
$ %\[ 
  P_i \Derives{\alpha_{3-i}}{H_i'}{} P'.
$ %\]
The assumption $\alpha_i \not\in H_{3-i}$ uses the blocking sets to express non-interference between transitions. 
Then, our weak\footnote{In the pure $\lambda$-calculus the diamond property holds for $\beta$-reduction without any assumptions. We use the term ``weak'' confluence, because of the side-condition of non-interference.} local confluence expresses a weak ``diamond property'' stating that every diverging choice of non-interfering transitions immediately reconverges. The subset inclusions $H_i' \subseteq H_i$ ensure that delaying an admissible action does not introduce more chances of interference.

\subsubsection{Fix 1}
The standard notion of (weak) priority for \ccs, as implemented in \ccscw or \ccsp, of course was not designed to restrict \ccs for determinism but to extend it conservatively to be able to express scheduling control. 
Even in the fragment of \ccscw where all choices are prioritised, weak confluence fails.  
For instance, although our memory cell $S$ behaves deterministically for a single reader $R$ and a single writer $W$, in the context $E \eqdef W_1 \cpar [\cdot] \cpar W_2$ of two writers $W_i \eqdef \ol{w} \cseq W_i'$, it creates a race.
The two writers compete for the single write action on $S$ and only one can win. 
Using the operational semantics of \ccscw we can derive the (non-interfering) transitions 
$ %\[
   W_1 \cpar S \cpar W_2 \Derives{\tau}{\eset}{} 
     W_1' \cpar S' \cpar W_2 
  \text{ and }
   W_1 \cpar S \cpar W_2 \Derives{\tau}{\eset}{} 
     W_1 \cpar S' \cpar W_2'. 
$ %\] 
Still, the two outcomes of the divergence are a critical pair that have no reason to reconverge $W_1' \cpar S' \cpar W_2 \Derives{}{}{} P'$ and $W_1 \cpar S' \cpar W_2' \Derives{}{}{} P'$ if $W_i'$ are arbitrary behaviours.
To preserve confluence we must block the environment from consuming the single prefix $w \cseq S'$ twice. 
The problem is that the synchronisations 
$ %\[ 
  W_1 \cpar S \Derives{\tau}{H_1}{} W_1' \cpar S'
  \text{ and }
  S \cpar W_2 \Derives{\tau}{H_2}{} S \cpar W_2'
$ %\] 
have empty blocking sets $H_i = \eset$ which does not prevent the second writer $W_{3-i}$ to be added in parallel. If $w \in H_i$, then the rules $(\ParR_{3-i})$ would automatically prevent the second reader. 
Now, suppose we extended the action rule $(\ActR)$ to expose the singleton nature of a prefix, like
%\longshort{\version}{by $(\ActR^+)$,  $\alpha\cseq P \Derives{\alpha}{\{ \alpha \}}{} P$}
\[ 
\begin{Prooftree}
  \Bproof
    \\[$(\ActR^+)$]
    \alpha\cseq P \Derives{\alpha}{\{ \alpha \}}{} P
  \EEproof 
\end{Prooftree}
\]
we would not only prevent a second concurrent writer but already a single writer. The transitions
$ %\[ 
  W_1 \Derives{\ol{w}}{\{ \ol{w} \}}{} W_1' 
  \text{ and }
  S \Derives{w}{\{ w \}}{} S'
$ %\] 
would not be able to communicate, violating both side-condition 
$\{ \ol{w} \} \cap \ol{\iA}(S) = \eset$ and $\{ w \} \cap \ol{\iA}(W_1) = \eset$ of $(\ParR_3)$.
The issue is that in the traditional setting the blocking set $H$ of a transition
\[ 
  P \Derives{\alpha}{H}{} P'
\] 
expresses a priority rather than an interference. An action cannot take priority over itself without blocking itself. Therefore, in \ccscw and \ccsp we always have $\alpha \not \in H$. 
Here, we need to apply a slightly different interpretation.
The actions in $\beta \in H$ only prevent $\alpha$ from synchronising with a matching $\ol{\alpha}$ in the environment, if $\ol{\beta}$ can be executed in a \textit{different} (and thus competing) prefix from the one that executes $\ol{\alpha}$.
Actions in $H$ \textit{interfere} with $\alpha$ in the sense that they are intial actions of $P$ sharing the same control thread in which $\alpha$ is offered. 
Thus, we want to avoid a situation in which the environment synchronises with $\ol{\alpha}$ in the presence of \textit{another distinct} communication on $\ol{\beta}$.
For then we would have a race for the behaviour of the thread shared by $\alpha$ and $\beta$ in $P$. 

To fix this, we need to refine the side-conditions $H_i \cap \ol{\iA}(P_{3-i})$ in $(\ParR_3)$ to check not \textit{all} initial actions $\ol{\iA}(P_{3-i})$ but only those initial actions of $P_{3-i}$ that are \textit{competing} with the matching transitions 
$ %\[ 
  P_{1} \Derives{\alpha_{1}}{H_{1}}{} P_{1}'
  \text{ and } 
  P_{2} \Derives{\alpha_{2}}{H_{2}}{} P_{2}'
$ %\] 
with $\alpha_1 = \ol{\alpha}_{2}$, involved in the communication $P_1 \cpar P_2 \Derives{\tau}{H_1\cup H_2}{} P_1' \cpar P_2'$ that we want to protect. To achieve this we must refine the information exposed by a transition. More specifically, we extend each transition 
\[ %\begin{eqnarray} 
  P \Derives{\alpha}{H}{K} P'
%\label{eqn:fix-1-label}
\] %\end{eqnarray}
by the set $K$ of initial actions that are offered in competing prefixes of the same thread as $\alpha$ or by threads concurrent to the thread from which $\alpha$ is taken.
Then, the set of all initial actions is simply
$\iA(P) = K \cup \{ \alpha \}$.
Let us call $K$ the \textit{initial environment} of the transition. %~\eqref{eqn:fix-1-label}. 
From the description, we expect the invariant that $H \subseteq K$.
%in~\eqref{eqn:fix-1-label}.
%
\begin{remark}
Note that we cannot simply define $K \eqdef \iA(P) - \{ \alpha \}$, because $\alpha$ could be contained in $K$. For example, take $P \eqdef \alpha\cseq\zero \cpar \alpha\cseq\zero$, where 
$ P\Derives{\alpha}{\{ \alpha \}}{\{\alpha\}}  \alpha\cseq\zero.$ 
The blocking set is $H = \{ \alpha \}$ since the prefix is volatile and using $\ActR^+$ must be protected from competing consumptions.
Also, we need $\alpha$ to be in the initial environment $K = \{ \alpha \}$, because the transition's action $\alpha$ is offered a second time in a concurrent thread. This is different from the set $\iA(P) - \{\alpha\} = \{\alpha\} - \{\alpha\} = \eset$. This shows that we must compute the initial environment $K$ explicitly as an extra label with each transition.
\qed
%\eqref{eqn:fix-1-label}. 
\end{remark}
We collect the initial environments in the rules for prefix, summation and parallel as follows:
\[ %
    \begin{Prooftree}
    \label{actpri}
        \Bproof
          \\[($\ActR^+$)] \alpha \cseq P \Derives{\alpha}{\{ \alpha\}}{\eset} P
        \EEproof
        \quad
        \Bproof
          P \Derives{\alpha}{H}{K} P'
          \\[$(\PriSumR_1^+)$]
          P \wpsum Q \Derives{\alpha}{H}{K \cup \iA(Q)} P'
        \EEproof
    \end{Prooftree}
\] 
\[ %
    \begin{Prooftree}
        \Bproof
          \Lproof
          Q \Derives{\alpha}{H}{K} Q'
          \ANDproof
            \tau, \alpha \not\in \iA(P)
          \Rproof[$(\PriSumR_2^+)$]
            P \wpsum Q 
            \Derives{\alpha}{H \cup \iA(P)}{K \cup \iA(P)} Q'
        \EEproof
        \qquad
        \Bproof
          \Lproof 
             P_2 \Derives{\alpha_2}{H_2}{K_2} P_2'
          \ANDproof 
             H_2 \cap \ol{\iA}(P_{1}) = \{\, \}
          \Rproof[$(\ParR_{2}^+)$] 
           P_1 \cpar P_2 
              \Derives{\alpha_2}{H_2}{K_2 \cup \iA(P_1)} P_1 \cpar P_2'
        \EEproof
    \end{Prooftree}
\]
%Now we can reformulate the synchronisation rule as 
\[ %
    \begin{Prooftree}
        \Bproof
          \Lproof 
             P_i \Derives{\alpha_i}{H_i}{K_i} P_i'
          %\ANDproof 
          %   P_2 \Derives{\alpha_2}{H_2 \cup L_2}{K_2} P_2'
          \ANDproof
             \alpha_1 = \ol{\alpha}_2
          \ANDproof
             H_i \cap \ol{K}_{3-i} = \{\, \}
          \ANDproof
             \text{for $i \in \{ 1, 2 \}$}
             %$H_2 \cap (\ol{K}_1 \cup \ol{L}_1) = \{\, \}
          \Rproof[$(\ParR_3^+)$] 
           P_1 \cpar P_2 
              \Derives{\tau}{H_1 \cup H_2}{K_1 \cup K_2} P_1' \cpar P_2'
        \EEproof
    \end{Prooftree}
\]
Now the composition $W \cpar S$ of the store with a single writer will not block but two writers $W \cpar S \cpar W$ create a race detected by $(\ParR_3^+)$, because of rule $(\ActR^+)$ that introduces $\ol{w}$ into the blocking set but \textit{not} into the initial environment of the write action $\ol{w}$ of $W$.
 
\medskip \noindent  Note that the above rules $(\ActR^+)$, $(\PriSumR_i^+)$ and $(\ParR_{i}^+)$ are not used in \ccslm. 

\subsubsection{Fix 2}

In the previous subsection we have seen how priorities can be beneficial to approaching to a deterministic scheduling in an intrinsic non-deterministic process algebra. 
The priority blocking eliminates the competition between mutually pre-empting actions and the remaining independent  actions naturally commute. 
Now another question naturally arises: Does this mean that if we replace the non-deterministic choice $P + Q$ by a prioritised choice $P \wpsum Q$, then we obtain a determinate process algebra? 
Indeed, we can show that in the ``asynchronous'' fragment of \ccscw, i.e., where all prefixes $\alpha \cseq P$ are of form $\alpha \cseq \zero$, confluence can be achieved with the SOS rules above. 
However, in the ``synchronous'' fragment of \ccs, where actions sequentially guard new actions, we run into a problem of non-monotonicity.  A communication between prefixes may trigger new actions that block another communication that was enabled before. In other words, it is not guaranteed that independent actions remain admissible. 

As an example, in the above reader-writer scenario, the read actions may initially be enabled but then become blocked in the continuation processes. Consider a parallel composition $R \cpar S^* \cpar W^*$ in which a modified store $S^* \eqdef S \cpar d \cseq \zero$  offers another concurrent action $d$ on the side, in parallel. Let the writer $W^*$ be delayed by a synchronisation with this concurrent action $d$, i.e., $W^* \eqdef \ol{d} \cseq W \eqdef \ol{d} \cseq \ol{w} \cseq \zero.$
The scheduling sequences are depicted in Fig.~\ref{fig:conf-prio-3}. 
We can still derive 
$ %\[ 
  S \cpar d \cseq \zero 
    \Derives{r}{\{ r, w \}}{\{w, d\}} S'' \cpar d \cseq \zero
  \text{ but also } 
  S\cpar d \cseq\zero \Derives{d}{\{ d \}}{\{r, w\}} S.
$ %\] 
% $$S \cpar d\cseq\zero \Derives{w}{\{w\}}{\{r, d\}} S' \cpar d\cseq\zero$$ but also  
% $$S\cpar d \cseq\zero \Derives{d}{\eset}{\{r, w\}} S \cpar \zero.$$ 
The reading synchronisation
\[
R \cpar S \cpar d\cseq\zero 
   \Derives{\tau}{\{ \ol{r}, r, w \}}{\{w, d\}} S'' \cpar d\cseq\zero
\]
can now take place by $(\ParR_3^+)$ before the write to give the following reduction sequence: 
\[ 
 R \cpar S \cpar d \cseq \zero \cpar W^* 
    \Derives{\tau}{\{\ol{r}, r, w\}}{\{ w,d, \ol{d} \}} 
 S'' \cpar d \cseq \zero \cpar W^*
    \Derives{\tau}{\{ d, \ol{d} \}}{\{ r, e \}} 
 S'' \cpar W.
%    \Derives{\tau}{\eset}{} 
% \zero \cpar S' \cpar \zero \cpar W'.
\] 
The first reduction is justified by rule $(\ParR_1^+)$, because the side-condition $\{r, \ol{r}, w\} \cap \ol{\iA}(W^*) = \{r, \ol{r}, w\} \cap \{ d \} = \eset$ is satisfied. The blocking side-condition only looks at the moment when the read is happening and does not see the write action $\ol{w}$ by $W^*$ hidden behind its initial action $\ol{d}$. It lets process $R$ access and read the store $S$.  
As before, however, the following reduction sequence is allowed, too:
\[ 
 R \cpar S \cpar d \cseq \zero \cpar W^* 
    \Derives{\tau}{\{ d, \ol{d} \}}{\{ \ol{r}, r, w \}} 
 R \cpar S  \cpar W
    \Derives{\tau}{\{w, \ol{w} \}}{\{ \ol{r}, r \}} 
 R \cpar S'  
    \Derives{\tau}{\{ \ol{r}, r, w \}}{\{ w \}} 
 S' 
\] 
in which the write action happens before the read action.
This creates two diverging reduction sequences (incoherent schedules) in which the final store is $S'$ or $S''$, non-deterministically. 
 
\medskip 

\begin{figure}
  \centering
  \includegraphics[scale=0.6]{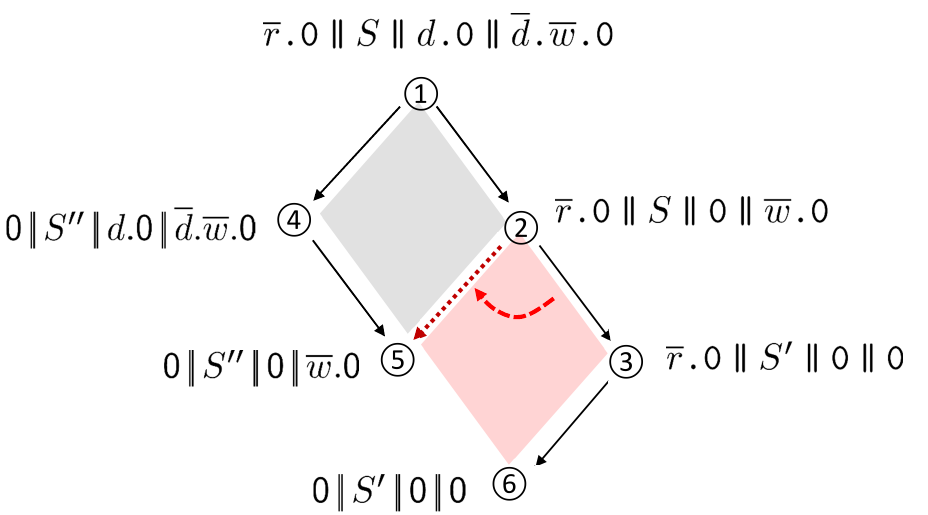}
\caption{The store $S$ in composition with a reader $\ol{r}\cseq\zero$ and a delayed writer $d \cseq \zero \cpar \ol{d}\cseq\ol{w}\cseq \zero$. 
The behaviour has two incongruent final outcomes, a store $S''$ with error in state \cfive and an error-free store $S$ in state \csix. The transition system is not confluent.
The reason is that  transition labelled $\tau \of \{ w \}$ is admissible in state \cone but not anymore in state \ctwo. A solution is to block the reading already in state \cone. In other words, the priority blocking (write-before-read) that kicks in at state \ctwo needs to be detected already at state \cone. Again, all transitons are silent with labels omitted.}
\label{fig:conf-prio-3}
\end{figure}

The new scenario is visualised in Fig.~\ref{fig:conf-prio-3}.
Our insight is that the blocking side-conditions of the standard Camilleri-Winskel theory~\cite{CamilleriWin95}, described above, are too weak in scheduling control for our purposes, which is to guarantee a weak Church-Rosser property~\cite{KlopTRS}.
%
%{\Large \bll{UNTIL HERE}}
%
The problem is that in the side-condition of rule $(\ParR_2^+)$ we are only looking at the initial actions of $Q$. This is why the read action is blocked out of state \ctwo in Fig.~\ref{fig:conf-prio-3} but not out of state \cone. In state \cone there is no pending initial write. The write action $\ol{w}$ only becomes initial when the parallel processes $\ol{d} \cseq \zero \cpar d \cseq \ol{w}\cseq \zero$ have interacted on action $d$ and advanced together with a silent action $\tau$, i.e. $\ol{d} \cseq \zero \cpar d \cseq \ol{w}\cseq \zero \xrightarrow{\tau}{}{} \zero \cpar \ol{w}\cseq \zero$. In order to observe that the read action $r$ out of state \cone should not be enabled, we must follow all sequences of interactions that can possibly happen concurrently to the read: The read is to be blocked, because after one silent $\tau$-step, in the reachable configuration of state \ctwo, the write action $w$ in the blocking set $\{w\}$ of a silent action, meets a matching initial partner $\ol{w}$. 

\medskip 

This suggests a strategy to fix up the standard theory. 
We define a refined transition relation of the form
\[ 
  P \Derives{a}{H; L}{E} P'
\]
in which we accumulate, along with the \textit{upper} blocking set $H$ of a transition (as before, the prefixes of the same thread that take higher priority due to prioritised choice $\wpsum$) the following: 
\begin{itemize} 
  \item a \textit{lower} blocking set $L$, consisting of the prefixes of the same thread that take lower priority due to prioritised choice $\wpsum$; and 
  \item the \textit{reduction context} $E$ of all concurrent processes that run in competition with the transition. From $E$ we can find out what blocking actions might be generated after any number of silent actions.
\end{itemize}
The initial environment $K$ of a transition then is $K = H \cup L \cup \iA(E)$. 
We collect the upper and lower blocking sets, as well as the reduction context in the rules for prefix and summation as follows:
\[ %
    \begin{Prooftree}
    \label{actpri}
        \Bproof
          \\[($\ActR^*$)] \alpha \cseq P \Derives{\alpha}{\{ \alpha\};\eset}{\zero} P
        \EEproof
        \quad
        \Bproof
          P \Derives{\alpha}{H;L}{E} P'
          \\[$(\PriSumR_1^*)$]
          P \wpsum Q \Derives{\alpha}{H; L \cup \iA(Q)}{E} P'
        \EEproof
    \end{Prooftree}
\] 
\[ %
    \begin{Prooftree}
        \Bproof
          \Lproof
          Q \Derives{\alpha}{H;L}{E} Q'
          \ANDproof
            \tau, \alpha \not\in \iA(P)
          \Rproof[$(\PriSumR_2^*)$]
            P \wpsum Q 
            \Derives{\alpha}{H \cup \iA(P);L}{E} Q'
        \EEproof
    \end{Prooftree}
\]
Further, let $\wiA(E)$ be informally the set of all actions that can become initial along $\tau$-sequences from $E$, which obviously includes $\iA(E)$. We will call the set $\wiA(E)$ the set of \textit{weak initial} actions of $E$.
The rules $(\ParR_1^+)$ and $(\ParR_2^+)$ need to be reformulated, too, so the reduction context $E$ is accumulated:
\[ %
    \begin{Prooftree}
        \Bproof
          \Lproof 
             P \Derives{a}{H; L}{E} P'
          \ANDproof 
             E' = E \cpar Q
          \ANDproof 
             H \cap \ol{\wiA}(E') = \{\, \}
          \Rproof[$(\SParR_1)$] 
           P \cpar Q 
              \Derives{a}{H; L}{E'} P' \cpar Q
        \EEproof
    \end{Prooftree}
\]
Rule $(\ParR_2^*)$ is symmetric to $(\ParR_1^*)$.
Rule $(\ParR_3^+)$, for the synchronisation of an action $\alpha$ and its co-action $\ol{\alpha}$ using these weak initial sets, is reformulated as follows:
\[ %
    \begin{Prooftree}
        \Bproof
          \Lproof 
             P \Derives{\alpha}{H_1; L_1}{E_1} P'
          \ANDproof
             Q \Derives{\ol{\alpha}}{H_2; L_2}{E_2} Q'
          \ANDproof 
             E = E_1 \cpar E_2
          \ANDproof
             \forall i.\, 
              H_i \cap (\ol{\wiA}(E) \cup L_{3-i}) = \eset
          \Rproof[$(\SParR_3)$] 
           P \cpar Q \Derives{\tau}{H_1 \cup H_2; L_1 \cup L_2}{E} P' \cpar Q'
        \EEproof
    \end{Prooftree}
\]
We are now ready to revisit our running example. 
By $(\SParR_2)$, since
\[ 
  R \Derives{\ol{r}}{\{ \ol{r} \}; \eset}{\zero} R' 
  \text{ and } 
  S \cpar d \cseq \zero
   \Derives{r}{\{w, r \}; \eset}{d \cseq \zero}
   S'' \cpar d \cseq \zero \text{, we get } 
  R \cpar S \cpar d \cseq \zero 
  \Derives{\tau}{\{w, r, \ol{r}\}; \eset}{d \cseq \zero} 
  S'' \cpar d \cseq \zero.
\] 
The side-condition of $(\SParR_3)$ is satisfied since $\wiA( d \cseq \zero) = \{ d \}$.
Thus, both $\{\ol{r} \} \cap (\olwilA{\tau}(d \cseq \zero) \cup \eset) = \eset$ as well as $\{ w, r \} \cap (\olwilA{\tau}(d \cseq \zero) \cup \eset) = \eset$.
If now want to shunt this past the writer process $W^*$ using rule $(\SParR_1)$, then we get stuck: 
Since 
$ %\[ 
   d\cseq \zero \cpar W^* 
  \Derives{\tau}{\{d, \ol{d} \}; \eset}{\zero}  W
  %\Derives{\ol{w}}{\{ \ol{w} \},\eset}{\zero\cpar\zero\cpar\zero} 
  % \zero \cpar \zero \cpar \zero
$ %\]
and $\ol{w} \in \iA(W)$, we find $\ol{w} \in \wiS(d\cseq \zero \cpar W^*)$
and thus $\{ w \} \cap \olwilA{\tau}(d\cseq \zero \cpar W^*) \neq \eset$. Therefore, the side-condition of $(\SParR_1)$ is violated and the reading interaction of $R \cpar S \cpar d \cseq \zero$ cannot be composed with $W^*$.
This is what we want. The transition from state \cone to state \cfour in Fig.~\ref{fig:conf-prio-3} is now blocked.
On the other hand, the reduction
$ %\[ 
  R \cpar S \cpar d \cseq \zero
  \Derives{d}{\{ d \};\eset}{R \cpar S} 
  R \cpar S
\text{ can be composed with } 
   W^* \Derives{\ol{d}}{\{ \ol{d} \};\eset}{\zero}
   W
$   
   %\text{ 
   via $(\SParR_3)$ to make a reduction 
   %} 
$  R \cpar S \cpar d \cseq \zero \cpar W^* 
  \Derives{\tau}{\{ d , \ol{d} \};\eset}{R \cpar S} 
  R \cpar S \cpar W.
$ %  \]
Regarding the side-condition, we notice that the weak initial action set $\wiA(R \cpar S)$ of the reduction context is irrelevant, because it does not contain $d$ or $\ol{d}$.
In configuration $R \cpar S \cpar W$ (state \ctwo in Fig.~\ref{fig:conf-prio-3}), the reader $R$ cannot interact with $S$. The associated reduction
$ %\[ 
  R \cpar S \cpar W 
  \Derives{\tau}{\{ r, \ol{r}, w \}; \eset}{W} 
 S'' \cpar W
$ %\] 
is blocked by the reduction environment, because 
$w \in \olwilA{\tau}(\ol{w}\cseq \zero)$.
\medskip 

For our further discussion, let us call the proposed change from $\iA$ to $\wiA$ in the side-conditions leading to $(\SParR_i)$ from $(\ParR_i^+)$ a process algebra with \textit{strong priorities}\footnote{The terminology is suggested by the fact that strong priorities use weak initial sets as rule premises while weak priorities are based on strong initial sets. Rules are implications, so weakening the antecedent makes a rule logically stronger.} and call the traditional process algebra with priorities, introduced by \cite{CamilleriWin95,CleavelandLN01}, as process algebra with \textit{weak priorities}.  
 
As an intermediate step to implement our precedence-based scheduling, the previous example suggests we look at the \textit{weak initial actions} for a process, given in the following definition.
\begin{definition}[Weak Initial  Actions \& Strong Enabling]
 The set $\wiA(P)$ of \emph{weak initial} transitions is the smallest extension $\iA(P) \subseteq \wiA(P)$ such that if $\alpha \in \wiA(Q)$ and $P \xrightarrow{\tau} Q$ 
 then $\alpha \in \wiA(P)$. 
 \longshort{\version}{}{%
 In other words, to obtain $\wiA(P)$ we close $\iA(P)$ under silent transitions of $P$.
 }
 Let $\olwilA{\tau}(P) = \ol{\wilA{\tau}(P)}$.
 A transition
 %~\eqref{eqn:prio-micro-step} 
 $P \Derives{\alpha}{H}{R} P'$ 
 is called \emph{strongly enabled} if $H \cap (\olwilA{\tau}(R) \cup \{\tau\}) = \eset$.
\label{def:strong-enabling}
\end{definition}  
It is easy to see that $\iA(P) \subseteq \wiA(P)$, so every strongly enabled transition is weakly enabled: as such, for free processes, the notions of weak and strong enabling coincide. 

\begin{remark} 
 Note also that enabling as a scheduling constraint is not monotonic under parallel composition. The enabledness property of a transition is not, in general, preserved when the process is placed into a concurrent context. If a reduction 
$ %\[
P \Derives{\tau}{H}{R} P'
$ % \]
is enabled, then we know that $H \cap \oliA(R) = \eset$ or $H \cap \olwilA{\tau}(R) = \eset$, depending on what scheduling we choose. This means that the context $R$ cannot (weakly) generate a synchronisation with a blocking initial action. This does not mean that an extension $R \cpar Q$ may not perhaps have $H \cap \oliA(R \cpar Q) \neq \eset$ or $H \cap \olwilA{\tau}(R \cpar Q) \neq \eset$. This happens if blocking actions come from $Q$ or, in the case of strong enabling, become reachable though the interactions of $R$ and $Q$. Then the parallel extension $P \cpar Q$ does not permit a reduction
$ %\[
P \cpar Q \Derives{\tau}{H}{R} P' \cpar Q.
$ %\]
\qed
\end{remark}

\medskip \noindent  Note that the above rules $(\ActR^*)$, $(\PriSumR^*_i)$ and $(\ParR_{i}^*)$ are only intermediaries and not used in \ccslm. 
As indicated by the example above, strong enabling generates confluent reductions in the reflexive (self-blocking) case, where rendez-vous actions have a unique sender and a unique receiver process. However, for multi-cast communication, i.e., irreflexive (non self-blocking) c-actions, we need the even stronger notion of \textit{constructive enabling} based on $\olwilA{\ast}(R)$ as used in \ccslm, instead of $\olwilA{\tau}(R)$.

%------------------------------------------
\subsubsection{Digression: Strong Priorities lead to Impredicativity Issues} 
%------------------------------------------

On the face of it, strong priorities may seem rather a natural modification to make. However, they amount to a significant change of the game compared to the standard weak priorities presented in the literature. For they break the standard method of defining the transition relation inductively as the least relation closed under a finite set of production rules. 
The key problem is the logical cycle in the structural definition of the operational semantics: the rules to generate the transitions of a process $P$
%in equation~\eqref{eqn:bred} 
ensure $H \cap \olwilA{\tau}(E) = \eset$, i.e., that none of the blocking actions $H$ matches with the transitively reachable initial actions of the reduction context $E$ which itself is extracted as a part of $P$.

Obviously, to compute $\wiA(E)$ we need to know the full transition relation for $E$. In other words, the existence of a transition of a process $P$ depends on the absence of certain transitions of a certain part of process $P$. If $P$ is defined recursively then the behaviour of a part of $P$ might depend on the behaviour of $P$ itself. Does such a self-referential semantical definition make sense at all? How do we ensure the absence of a Russell-paradox style relation, by defining a process $\crec p.\, P$ to have a transition to $P'$ if it does not have a transition to $P'$ for a suitable $P'$? Such would be plainly inconsistent at the meta-level, 
because it would not only define simply an empty transition relation\footnote{As an example consider the unguarded process $\crec p \cseq p$ in the standard inductive theory. Its behaviour is well-defined as the empty transition relation.}, 
but it would not define a unique semantic relation at all. 

In the standard theory of weakly prioritised process algebra as described in the cited references of \ccs with priorities, the problem does not exist, because the set of (strong) initial actions $\iA(E)$ can be determined by structural recursion on $E$ without any reference to the transition relation. Unfortunately, for strong prioritized process algebra as defined in this paper, this is not possible anymore. 
So, how do we make sense of transition-generating rules $(\SParR_i)$ that depend on negative premises? One option is to treat the negation in the side-condition $H \cap \olwilA{\tau}(E) = \eset$ constructively rather than classically, say like negation in normal logic programming. This would amount to building a second independent (inductive) rule system to generate constructive judgements about the absence of weakly reachable transitions. Trouble is that there is no canonical semantics of constructive negation, even in logic programming\footnote{Different well-known semantics are negation-by-failure, stable semantics, supported model semantics.}. 

Some progress has been made for unifying the proposals in a game-theoretic and intuitionistic setting which works well for special situations like the constructive semantics of Esterel~\cite{Berry99} or StateCharts~\cite{HarPnuPruShe87}.  However, playing with constructive negation at the meta-level will give a variety of semantics whose adequacy might be difficult to judge in practical application contexts.

A more perspicuous attack is to keep the classical set-theoretic interpretation
%$L \cap \olwiS(E) = \eset$ 
but use an over-approximation $L \cap \olfwiA(E) = \eset$ where $\fwiA(E) \supseteq \wiA(E)$ is intuitively a superset of weakly reachable initial actions obtained from a relaxed transition relation. A natural over-approximation that can be defined independently is the ``free speculative'' scheduling that ignores all priorities. The stronger side-condition $L \cap \olfwiA(E) = \eset$ forces the absence of a blocking sender action under the worst-case assumption of priority-free scheduling. Obviously, this is sound but will have the effect that some programs block that could make progress under a tighter approximation of the sets $\wiA(E)$. Between $\fwiA(E)$ and $\wiA(E)$ there are a whole range of different semantics that may or may not be equivalent to specific constructive semantics of negation. Exploring the range of options systematically will certainly be a worthwhile research exercise likely to cover much new ground in the theory of process algebras with priorities. In \ccslm the impredicativity problem can be avoided by defining the SOS (cf\ Fig.~\ref{fig:free-sos}) of admissible transitions in a ``free' fashion, without blocking side conditions. We compute local races in $(\ComR)$, possibly adding $\tau$ into the blocking set, in a constructive fashion. We express constructive scheduling by the filter condition of c-enabledness (Def.~\ref{def:c-enabling}) using $\wilA{\ast}$. The latter is defined using the SOS of admissible transitions rather than the c-enabled transitions, which would be self-referential. 

% \medskip 
% \noindent Note that $\fwiA(E)$ and $\wiA(E) $ are NOT used in this paper.

 % clean

%!TEX root = synpatick-lipics.tex

\subsection{Synchronous \ccs with (Strong) Priorities and Clocks}
\longshort{\version}{We}{%
The aim of this paper is to show that strong priorities are related to synchronous programming, independently of the choice of solving the negation problem. We adopt the over-approximating approach via free scheduling based on the formally defined notion of \textit{freely reachable initial actions} 
%$\fwiS(E)$, 
and postponing a discussion of variations to future work. 
Instead, we} argue that the use of priorities for determinacy calls for a combination of strong priorities with a second important concept that has been investigated in the literature before, but independently. This is the notion of a \textit{clock}, known from Timed Process Algebras, like e.g. \tpl~\cite{TPA} and \pmc~\cite{AndersenMen94}.
Intuitively, a clock $\sigma$ is a ``broadcast action'' in the spirit of Hoare's \csp~\cite{Hoare:CSP} that acts as a synchronisation barrier for a set of concurrent processes. It bundles each processes' actions into a sequence of \textit{macro-steps} that are aligned with each other in a lock-step fashion. During each macro-step, a process executes only a temporal slice of its total behaviour, at the end of which it waits for all other processes in the same \textit{clock scope} to reach the end of the current phase. When all processes have reached the barrier, they are released into the next round of computation. 
This scheduling principle is familiar from computer hardware like sequential circuits~\cite{Unger:async-circ} and VHDL~\cite{KloosB95}, or synchronous languages like Esterel~\cite{Berry99}, Lustre~\cite{HalbwachsCRP91}, SCCharts~\cite{vonHanxledenDM+14}, BSP~\cite{Valiant90}, just to mention a few. For our purposes, clocks are the universal trick to break causality cycles in priority scheduling, when processes need to communicate with each other in iterated cycles (so-called \textit{ticks} or \textit{macro-steps}) while maintaining determinacy.

\medskip 

Let us illustrate the use of clocks using our running example. As discussed above with Fig~\ref{fig:conf-prio-2}, 
the evaluation of $R \cpar S \cpar W$ under prioritised choice ``prunes'' all reductions that could lead to non-confluent reductions. This sorts the data race between a single read and a single write. But what if we wanted to iterate concurrent reading and writing? Then, we need to implement some form of lock-step synchronisation between $R$ and $W$. A natural synchronisation scheme for this is provided by a clock $\sigma$ that acts as a broadcast action to separate each round of communication. Let us now use the syntax of \ccslm where blocking is defined in the prefixes via explicit blocking sets as in \ccsp rather than by the prioritised sum operator $\wpsum$ as in \ccscw. Moreover, from now on now, we will use the SOS rules from \ccslm, expect that we are interested in strong enabling (Def.~\ref{def:strong-enabling}) based on the weak initial actions $\wiA$ instead of constructive enabling (Def.~\ref{def:c-enabling}) which is defined using potential actions $\wilA{\ast}$.

\medskip 

We define the recursive processes 
$W \eqdef d \cseq \zero \cpar \ol{d}\cseq\ol{w}\cseq \sigma\cseq W$
and $R \eqdef \ol{r}\cseq \sigma \cseq R$
for one write and one read per macro-step. Assuming the content of the memory cell is reset for every tick, we can define $S$ as follows
\begin{align*} 
  S &\eqdef w\cseq S'  + r\of\{w\} \cseq S'' + \sigma\of\{w,r\} \cseq S \\ 
  S' &\eqdef r \cseq S''  + \sigma\of\{r\} \cseq S' \\
  S'' &\eqdef w \cseq \mathsf{fail} + r\of\{w\} \cseq S''  + \sigma\of\{w, r\} \cseq S.
\end{align*}
\longshort{\version}{The}{%
Alternatively, the cell might be refreshed to a default value, like an Esterel ``pure signal'' that is reset to an absence at each ``clock tick''. Whether the store's state changes with the clock or not, the}
clock always takes \textit{lowest} priority among all other rendez-vous actions out of a state. 
The idea is that the clock should fire only if the system has stabilised and no other admissible data actions are possible. This is called \textit{maximal progress} and it is a standard assumption in timed process algebras~\cite{TPA,CleavelandLM97}.
The fairly standard rules for clock actions, inspired by timed process algebras would then be:
% %
\[ %
    \begin{Prooftree}
        \Bproof
          \\[$(\ClkR_1)$] 
           \sigma \cseq P 
              \Derives{\sigma}{\eset}{\zero} P
        \EEproof
        \qquad
        \Bproof
          \Lproof 
             P \Derives{\sigma}{H_1}{R_1} P'
          \ANDproof 
             P \Derives{\sigma}{H_2}{R_2} P'
          \ANDproof
             (H_1 \cup H_2) \cap \wiS(R_1 \cpar R_2) = \eset 
          \Rproof[$(\ClkR_2)$] 
           P \cpar Q 
              \Derives{\sigma}{H_1 \cup H_2}{R_1 \cpar R_2} P' \cpar Q'
        \EEproof
    \end{Prooftree}
\]
Those two rules  $(\ClkR_i)$ enforces global synchronisation and make the clock deterministic. 
The side condition of $(\ClkR_2)$ enforces maximal progress and essentially corresponds to checking strong enabling (cf. Def.~\ref{def:strong-enabling}). The rules arise essentially from the rules $(\ActR)$ and $(\ComR)$ of \ccslm for clocks and pivotable processes. We remark that $(\ClkR_2)$ does not check the local race condition in contrast to $(\ComR)$ but this condition (like adding $\tau$ into the blocking sets) is redundant for pivotable processes.
In Fig.~\ref{fig:conf-prio-4} we illustrate the resulting transition system. 

\begin{figure}[!t]
  \centering
 \includegraphics[scale=0.5]{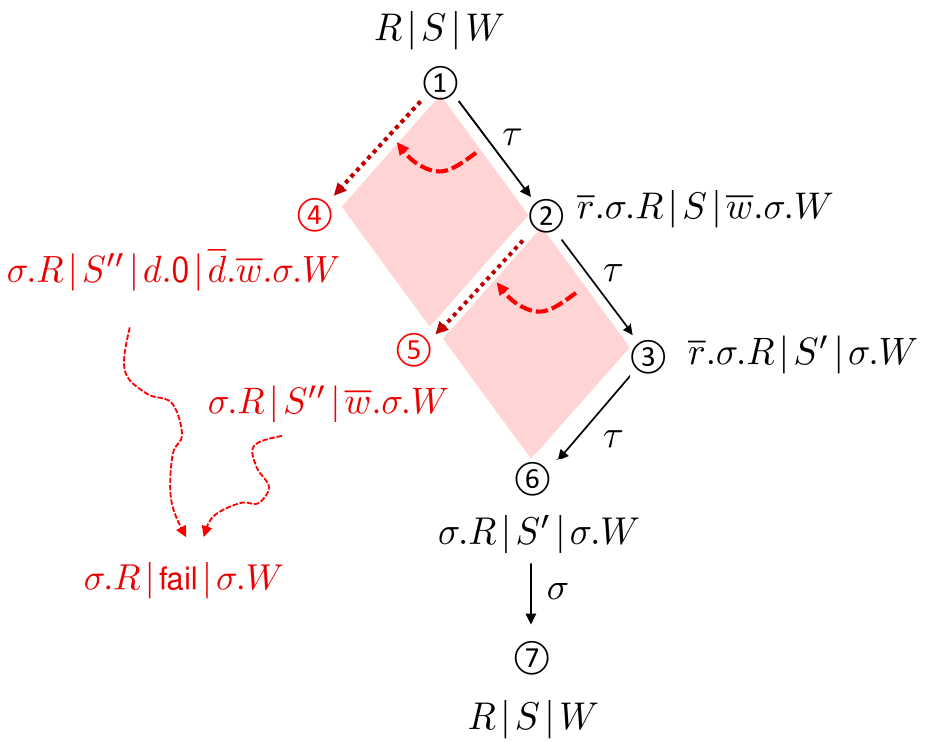}
\caption{Strong priorities introduce confluence in scheduling and clocks permit us to iterate the schedules. The transitions \cone-\ctwo-\cthree-\csix are the deterministically scheduled sequence of the initial macro-step comprising s single write and then a single read. The configuration $\sigma \cseq R \cpar S' \cpar \sigma \cseq W$ in state \csix is the completion of the macro-state, where only a clock transition is possible using rule $(\ClkR_3)$ to start the next macro-step in state \cseven. The dotted transitions \cone-\cfour and \ctwo-\cfive are blocked by the side conditions of $(\ParR_i)$ involving the weak initial actions $\wiA$. The states \cfour and \cfive would end up in the ``error'' configuration $\sigma \cseq R \cpar \mathsf{fail} \cpar \sigma \cseq W$ which does not permit the clock anymore.  
}
\label{fig:conf-prio-4}
\end{figure}

\medskip

The importance of the clock $\sigma$ is that its associated synchronisation rule $(\ClkR_2)$ enforces, in a quite natural way, the barrier explicitly in the SOS rather than making it an artefact of a specific synchronisation process, as such obliging the programmer to an extra effort. The fact that the broadcasting clock is hard-wired into the rules can now be exploited for the computation of the \textit{weak initial action sets} $\wiS(R)$. 
More precisely, we only need to predict the behaviour of $R$ for the \textit{current} macro-step, i.e., up to the next synchronization of $\sigma$ by $R$. We know that all actions taking place after $\sigma$ in $R$ are definitely observed to happen in the next macro step and thus not in competition with any given current action $a$. In this fashion, the clock acts as a \textit{stopper sentinel} for prediction and evaluation of the blocking set $L$ relative to the reduction environment $R$. If every recursion is guarded by a clock, then the search space becomes finite.

\medskip \noindent
Note that rules $(\ClkR_1)$ and $(\ClkR_2)$ are NOT used in this paper.

\subsection{Strong Enabling vs Constructive Enabling}

In this final section we aim to show why strong enabling based on $\olwilA{\tau}(R)$ is not sufficient and why we need the stronger notion of c-enabling (Def.~\ref{def:c-enabling}), based on $\olwilA{\ast}(R)$. 
%
%\begin{example}[On Multi-Cast and Constructive Enabling]
%

\medskip 

\noindent
Suppose $$S_0 \eqdef a \col e \cseq S_0 + e \cseq S_1 + \sigma \col \{a, e\} \cseq S_0$$ is an Esterel Signal (Sec.~\ref{examples:Esterel-signal}) in its absent state, with abbreviations $e = \emit$, $a = \abs$ and $p = \pres$. Its initial action 
$S_0 \Derives{a}{\{e\}}{\zero} S_0$ 
to communicate absence does not block itself to permit multi-reader scenarios. 

Consider a (perhaps non-typical) reader $C \eqdef \ol{a} \col \ol{a} \cseq \ol{e} \col \ol{e} \cseq C_1$ that wants to emit the signal when it finds it is absent. The composition $S_0 \cpar C$ is deterministic under strong enabling, with the only reduction 
$S_0 \cpar C \Derives{\tau}{\{e,\ol{a}\}}{\zero} S_0 \cpar \ol{e} \col \ol{e} \cseq C_1.$ 
The blocking set $\{e,\ol{a}\}$ contains two labels:
\begin{itemize} 
 \item First, $e \in \{e,\ol{a}\}$ indicates priority, i.e., that we do not want another concurrent thread sending $\ol{e}$, because the synchronisation for absence is based on the assumption that there is no emission. 
 Specifically, 
 $S_0 \cpar C \cpar \ol{e} \Derives{\tau}{\{e,\ol{a}\}}{\ol{e}} S_0 \cpar \ol{e} \col \ol{e} \cseq C_1 \cpar \ol{e}$
 will block because $\{e,\ol{a}\} \cap \olwilA{\tau}(\ol{e}) \neq \eset$.
 Instead, the transition 
 $S_0 \cpar C \cpar \ol{e} \Derives{\tau}{\{\ol{e}\}}{C} S_1 \cpar C$ 
 which emits the signal can proceed. 
 \item Second, $\ol{a} \in \{e,\ol{a}\}$ indicates that the sending prefix $\ol{a}\col\ol{a}$ of the reader $C$ (which has entered into the synchronisation $\tau = a \cpar \ol{a}$) must not be consumed by another competing receiver (such as another signal). 
 Specifically, 
 $S_0 \cpar C \cpar S_0 \Derives{\tau}{\{e,\ol{a}\}}{a} S_0 \cpar \ol{e} \col \ol{e} \cseq C_1 \cpar S_0$
 will block. 
\end{itemize}
Now, take $S_0 \cpar C \cpar D$ where the signal $S_0$ is connected with two readers, where $C$ is as above and $D \eqdef \ol{a} \col \ol{a} \cseq \ol{e} \col \ol{e} \cseq D_1$ analogous. Both readers can go ahead and synchronise in reactions ($\tau = a \cpar \ol{a}$)
%\begin{eqnarray} 
%S_0 \cpar C \cpar D \fsstep{\tau\col\{e,\ol{a}\}[\zero\cpar\zero\cpar D]} S_0 \cpar \ol{e} \col \ol{e} \cseq C_1 \cpar D
\[   
  S_0 \cpar C \cpar D \Derives{\tau}{\{e,\ol{a}\}}{D} S_0 \cpar \ol{e} \col \ol{e} \cseq C_1 \cpar D
 \text{ and } 
  S_0 \cpar C \cpar D \Derives{\tau}{\{e,\ol{a}\}}{C} S_0 \cpar C \cpar \ol{e} \col \ol{e} \cseq D_1
\]
%\label{ex:reader-2-go-ahead}
%\end{eqnarray}
which are both strongly enabled since 
$\olwilA{\tau}(D) = \{ a \} = \olwilA{\tau}(C)$ which is disjoint from the blocking set $\{e,\ol{a}\}$.
This may seem ok, as we wanted to have multiple readers access the signal on the same port $a$. However, now the emission of $\ol{e}$ by the process that advanced first will block the reading of absence $a \cpar \ol{a}$ by the other which stayed behind. In each case symmetrically, there is only one (strongly enabled) transition left, viz.
\[ 
  S_0 \cpar \ol{e} \col \ol{e} \cseq C_1 \cpar D
 \Derives{\tau}{\{\ol{e}\}}{D} S_1 \cpar C_1 \cpar D
 \text{ and } 
 S_0 \cpar C \cpar \ol{e} \col \ol{e} \cseq D_1
 \Derives{\tau}{\{\ol{e}\}}{C} S_1 \cpar C \cpar D_1
\]
where the combined final states $C_1 \cpar D$ and $C \cpar D_1$ of the readers are behaviourally different. Thus, we have non-determinacy and the confluence is lost. 

\medskip 

The core of the problem is that the reduction contexts $D$ and $C$, respectively, of the above rendez-vous synchronisations do not account for the fact that the signal $S_0$ can be simultaneously reused by the other reader that is not involved in the transition.
We could account for this reuse by adding $S_0$ into the concurrent context as in 
\[  
  S_0 \cpar C \cpar D \Derives{\tau}{\{e,\ol{a}\}}{S_0\cpar D} S_0 \cpar \ol{e} \col \ol{e} \cseq C_1 \cpar D.
\]
Now $D$ can interact with $S_0$ in the concurrent context $S_0\cpar D$. The above transition
depends on the assumption that no other thread will send $\ol{e}$. This is not true, because now $e \in \olwilA{\tau}(S_0\cpar D)$. Thus, putting the shared signal $S_0$ into the concurrent environment is enforcing determinism by blocking the strong enabling 
properties of the above transitions.
Unfortunately, the fix is too strong. The presence of $S_0$ in the concurrent context would block even the single reader in $S_0 \cpar \ol{a} \col \ol{a}$. 
Specifically, the reduction
\[
  S_0 \cpar \ol{a} \col \ol{a} \Derives{\tau}{\{e,\ol{a}\}}{S_0} S_0
\]
would block, because $\ol{a} \in \oliA(S_0) \subseteq \olwilA{\tau}(S_0)$, and so the reduction
is not strongly enabled any more. Thus, our fix would not permit any reading of the signal $S_0$, whatsoever.

\medskip 

The solution is to tighten up the notion of strong enabling and close it under arbitrary non-clock transitions in the concurrent environment rather than just the silent steps. Technically, in \ccslm, we define a larger set $\olwilA{\ast}(E) \supseteq \olwilA{\tau}(E)$ of potential labels (Def.~\ref{def:c-enabling}) that closes under \textit{all} non-clock actions of $E$.
In this case, we do find $e \in \{ e, \ol{a} \} \cap \olwilA{\ast}(D) \neq \eset$. 
In this way, diverging reductions
are no longer enabled. On the other hand, the reduction 
\[    
  S_0 \cpar C \Derives{\tau}{\{e, \ol{a}\}}{\zero} S_0 \cpar \ol{e} \col \ol{e} \cseq C_1
\]
remains enabled, because $\olwilA{\ast}(\zero) = \eset$. 
In general, single readers $S_0 \cpar C$ where $C$ is an arbitrary sequential process, or multiple readers $S_0 \cpar C \cpar D$ where \textit{at most one} of $C$ or $D$ is writing, i.e., sending $\ol{e}$, remain enabled. 

%\qed
%\end{example} 
%

  % clean
 % clean

\end{document}